\newif\ifabstract
\abstracttrue
\abstractfalse %Uncommenting this line gives the full version
\newif\iffull
\ifabstract \fullfalse \else \fulltrue \fi
\iffull
\documentclass[11pt]{article}
\fi

\usepackage{amsfonts}
\usepackage{amssymb}
\usepackage{amstext}
\usepackage{amsmath}
\usepackage{xspace}
\usepackage{theorem}
\usepackage{graphicx}
\usepackage{url}
\usepackage{graphics}
\usepackage{colordvi}
\usepackage{colordvi}
\usepackage{subfigure}

\iffull
\advance \oddsidemargin by -0.8in
\newcommand{\myparskip}{3pt}
\parskip \myparskip
\fi

\renewcommand{\k}{\kappa}

\newcommand{\Erdos}{Erdos\xspace}
\newcommand{\Posa}{P\'{o}sa\xspace}
\newcommand{\EP}{{\Erdos-\Posa}\xspace}

\newenvironment{proofof}[1]{\noindent{\bf Proof of #1.}}%
        {\hfill\stopproof}

\def\etal{et al.\xspace}

\newcommand{\PoS}{Path-of-Sets System\xspace}
\newcommand{\ToS}{Tree-of-Sets System\xspace}

\newcommand{\ttt}[1]{\tilde T_{#1}}

\newcommand{\tC}{\tilde C}
\newcommand{\gkrv}{\ensuremath{\gamma_{\mbox{\tiny{\sc CMG}}}}}

\newcommand{\tl}{\tilde \ell}

\newcommand{\connect}{\leadsto}

\newcommand{\sconnect}{\overset{\mbox{\tiny{1:1}}}{\leadsto}}

\newcommand{\tw}{\mathrm{tw}}

\newcommand{\betaFCG}{\beta_{\mathrm{FCG}}}

\newcommand{\tGamma}{\tilde{\Gamma}}

\newcommand{\tS}{\tilde S}

\newcommand{\alphasc}{\ensuremath{\beta_{\mbox{\tiny{\sc ARV}}}}}

\newcommand{\algsc}{\ensuremath{{\mathcal{A}}_{\mbox{\textup{\scriptsize{ARV}}}}}\xspace}

\newcommand{\tpset}{\tilde{\mathcal{P}}}
\newcommand{\nset}{\mathcal{N}}

\newcommand{\tT}{\tilde{T}}
\newcommand{\G}{{\mathbf{G}}}

\newcommand{\ceil}[1]{\ensuremath{\left\lceil#1\right\rceil}}
\newcommand{\floor}[1]{\ensuremath{\left\lfloor#1\right\rfloor}}

\newcommand{\abs}[1]{\lvert #1\rvert}

\newcommand{\alphaCMG}{\ensuremath{\alpha_{\mbox{\tiny{\sc CMG}}}}}

\newcommand{\set}[1]{\left\{ #1 \right\}}
\newcommand{\sse}{\subseteq}

\newcommand{\tset}{T}

\newcommand{\subsetq}{\subseteq}

\newcommand{\pset}{{\mathcal{P}}}
\newcommand{\qset}{{\mathcal{Q}}}
\newcommand{\lset}{{\mathcal{L}}}
\newcommand{\bset}{{\mathcal{B}}}
\newcommand{\aset}{{\mathcal{A}}}
\newcommand{\cset}{{\mathcal{C}}}
\newcommand{\fset}{{\mathcal{F}}}

\newcommand{\xset}{{\mathcal{X}}}

\newcommand{\ttset}{\tilde{T}}

\newcommand{\tA}{\tilde A}
\newcommand{\tB}{\tilde B}

\newcommand{\yset}{{\mathcal{Y}}}
\newcommand{\rset}{{\mathcal{R}}}

\newcommand{\hset}{{\mathcal{H}}}

\newcommand{\sset}{{\mathcal{S}}}

\newcommand{\nots}{\overline S}
\newcommand{\be}{\begin{enumerate}}
\newcommand{\ee}{\end{enumerate}}
\newcommand{\bd}{\begin{description}}
\newcommand{\ed}{\end{description}}
\newcommand{\bi}{\begin{itemize}}
\newcommand{\ei}{\end{itemize}}

\newtheorem{theorem}{Theorem}[section]
\newtheorem{lemma}[theorem]{Lemma}
\newtheorem{observation}[theorem]{Observation}
\newtheorem{corollary}[theorem]{Corollary}

\newtheorem{claim}[theorem]{Claim}

\newtheorem{definition}{Definition}[section]
\iffull
\newenvironment{proof}{\par \smallskip{\bf Proof:}}{\hfill\stopproof}
\fi
\def\stopproof{\square}
\def\square{\vbox{\hrule height.2pt\hbox{\vrule width.2pt height5pt \kern5pt
\vrule width.2pt} \hrule height.2pt}}

\renewcommand{\phi}{\varphi}
\newcommand{\eps}{\epsilon}

\newcommand{\half}{\ensuremath{\frac{1}{2}}}

\newcommand{\poly}{\operatorname{poly}}

\newcommand{\R}{\ensuremath{\mathbb R}}
\newcommand{\Z}{\ensuremath{\mathbb Z}}

\newcommand{\U}{\Upsilon}

\newenvironment{properties}[2][0]
{
\begin{enumerate} \setcounter{enumi}{#1}}{\end{enumerate}}

\newcommand{\mynote}[1]{{\sc\bf{[#1]}}}

\newcommand{\out}{\operatorname{out}}

\newcommand{\partition}{\mathsf{PARTITION}}
\newcommand{\separate}{\mathsf{SEPARATE}}

\begin{document}

\title{
Improved Bounds for the Excluded Grid Theorem
}

\author{
Julia Chuzhoy\thanks{Toyota Technological Institute, Chicago, IL
60637. Email: {\tt cjulia@ttic.edu}. Supported in part by NSF grant CCF-1318242.}
}

%\pagenumbering{gobble}
\maketitle

\begin{abstract}
We study the Excluded Grid Theorem of Robertson and Seymour. This is a fundamental result in graph theory, that states that there is some function $f:\Z^+\rightarrow \Z^+$, such that for all integers $g>0$, every graph of treewidth at least $f(g)$ contains the $(g\times g)$-grid as a minor. Until recently, the best known upper bounds on $f$ were super-exponential in $g$. A recent work of Chekuri and Chuzhoy provided the first polynomial bound, by showing that treewidth $f(g)=O(g^{98}\poly\log g)$ is sufficient to ensure the existence of the $(g\times g)$-grid minor in any graph. In this paper we improve this bound to $f(g)=O(g^{19}\poly\log g)$. We introduce a number of new techniques, including a conceptually simple and almost entirely self-contained proof of the theorem that achieves a polynomial bound on $f(g)$.\end{abstract}

\label{----------------------------------------------sec:intro----------------------------------------}

\section{Introduction}
We study the Excluded Grid Theorem of Robertson and Seymour~\cite{RS-grid} --- a fundamental and widely used result in graph theory. Informally, the theorem states that for any undirected graph $G$, if the treewidth of $G$ is large, then $G$ contains a large grid as a minor. Treewidth is an important and extensively used graph parameter, that, intuitively, measures how close a given graph $G$ is to being ``tree-like''. The treewidth of a graph is usually defined via tree-decompositions. A valid tree-decomposition of a graph $G$ consists of a tree $\tau$, and, for every node $a\in V(\tau)$, a subset $X(a)\subseteq V(G)$ of vertices of $G$, sometimes called a \emph{bag}. For every edge $e=(u,v)\in E(G)$ there must be a node $a\in V(\tau)$, whose bag $X(a)$ contains both $u$ and $v$, and for every vertex $v\in V(G)$, the set $\xset=\set{a\in V(\tau)\mid v\in X(a)}$ of nodes of $\tau$ whose bags contain $v$ must induce a non-empty connected  sub-tree of $\tau$. The \emph{width} of a given tree decomposition $(\tau,X)$ is $\min_{a\in V(\tau)}\set{|X(a)|}-1$, and the \emph{treewidth} of a graph $G$, denoted by $\tw(G)$, is the smallest width of any valid tree-decomposition of $G$. For example, the treewidth of a tree is $1$; the treewidth of the $(g\times g)$-grid is $\Theta(g)$; and the treewidth of an $n$-vertex constant-degree expander is $\Theta(n)$. Many combinatorial optimization problems that are hard on general graphs have efficient algorithms on trees, often via the dynamic programming technique. Such algorithms can frequently be extended to bounded-treewidth graphs, usually by applying the dynamic programming-based algorithms to the bounded-width tree-decomposition $(\tau,X)$ of $G$. However, for large-treewidth graphs, a different toolkit is often needed. The Excluded Grid Theorem provides a useful insight into the structure of such graphs, by showing that every large-treewidth graph must contain a large grid as a minor. (A graph $H$ is a \emph{minor} of a graph $G$, iff $H$ can be obtained from  $G$ by a series of edge-deletion, edge-contraction, and vertex-deletion operations.) We are now ready to formally state the Excluded Grid Theorem.

\begin{theorem}\label{thm: GMT}\cite{RS-grid}
There is some function $f:\Z^+\rightarrow \Z^+$, such that for every integer $g\geq 1$, every graph of treewidth at least $f(g)$ contains the $(g\times g)$-grid as a minor.
\end{theorem}

%This theorem is especially useful, since grid graphs are easy to reason about, and so one can easily infer many useful facts about any large-treewidth graph $G$ from the existence of large grid-minor in $G$. For example, $G$ must contain a large number of disjoint cycles; conversely, the value of the Feedback Vertex Set in $G$ is large; $G$ contains a large convenient routing structure (a grid), and so on. 

The Excluded Grid Theorem plays an important role in Robertson and Seymour's seminal Graph Minor series, and it is one of the key elements in their efficient algorithm for the Node-Disjoint Paths problem (where the number of the demand pairs is bounded by a constant)~\cite{flat-wall-RS}. It is also widely used in \EP-type results (see, e.g.~\cite{Thomassen88,FominST11,RS-grid}) and in Fixed Parameter Tractability; in fact the Excluded Grid Theorem is the key tool in the bidimentionality theory~\cite{DemaineH-survey,DemaineH07}.

It is therefore important to study the best possible upper bounds on the function $f$, for which Theorem~\ref{thm: GMT} holds. Besides being a fundamental graph-theoretic question in its own right, better upper bounds on $f$ immediately result in faster algorithms and better parameters in its may applications. The original upper bound on $f$ of~\cite{RS-grid} was substantially improved by Robertson, Seymour and Thomas \cite{RobertsonST94} to $f(g) = 2^{O(g^5)}$. Diestel et al.~\cite{DiestelJGT99} (see also \cite{Diestel-book}) provide a simpler proof with a slightly weaker bound.  This was in turn improved  by Kawarabayashi and Kobayashi \cite{KawarabayashiK-grid}, and by Leaf and Seymour
\cite{LeafS12}, to $f(g) = 2^{O(g^2/\log g)}$. Finally, a recent work of Chekuri and Chuzhoy~\cite{CC14} provides the first polynomial upper bound on the function $f(g)$, by showing that Theorem~\ref{thm: GMT} holds for $f(g)=O(g^{98}\poly\log g)$.
 On the negative side, Robertson \etal \cite{RobertsonST94} show that $f(g)=\Omega(g^2\log g)$ must hold, and they conjecture that this value is sufficient. Demaine \etal \cite{DemaineHK09} conjecture that the bound of $f(g)=\Theta(g^3)$ is both necessary and sufficient.

In this paper we provide a proof of Theorem~\ref{thm: GMT} with an improved bound of $f(g)=O(g^{19}\poly\log g)$. 
The paper consists of two parts. In the first part, we provide what we call a basic construction, that achieves a weaker bound of $f(g)=O(g^{36}\poly\log g)$. This part is a full version of the extended abstract~\cite{GMT-STOC} that appeared in STOC 2015. 
The main advantage of this construction is that, unlike the proof of~\cite{CC14}, it is very simple conceptually. The proof is almost self-contained, in the following sense: we provide a self-contained proof of the theorem for bounded-degree graphs $G$. In order to handle general graphs, we need to use previously known results to reduce the maximum vertex degree of the input graph to a constant, while approximately preserving its treewidth. This is the only part of the basic construction that is not self-contained; we discuss this in more detail below. Unlike the proof of~\cite{CC14}, that relies on many known technical tools, such as the cut-matching game of Khandekar, Rao and Vazirani~\cite{KRV}, graph-reduction step preserving element-connectivity~\cite{element-connectivity,ChekuriK09}, edge-splitting~\cite{edge-connectivity}, and LP-based approximation algorithms for bounded-degree spanning tree~\cite{Singh-Lau} to name a few, the basic construction is entirely from first principles. 

The second part of the paper combines elements of the (somewhat improved and simplified) proof of~\cite{CC14} together with the basic construction, in order to achieve the final bound of $f(g)=O(g^{19}\poly\log g)$, in what we call the extended construction.

The contribution of this paper is therefore two-fold: we provide a conceptually simple framework for proving the Excluded Grid Theorem, and show that it can be used to obtain a polynomial bound on $f(g)$; and we improve the bound of~\cite{CC14} on $f(g)$ from $O(g^{98}\poly\log g)$ to $O(g^{19}\poly\log g)$. We note that unfortunately the goals of presenting a simple proof and optimizing the bound on $f(g)$ are somewhat conflicting. We have tried to provide an exposition balancing these two objectives in~\cite{GMT-STOC}. In the current paper we focus is on optimizing the bound on $f(g)$; we plan to write a separate expository article providing a simple proof of a polynomial bound on $f(g)$, with the focus on simplicity of exposition, rather than on achieving specific bounds.% Our hope is that the simple conceptual framework introduced in this paper will lead to significantly better bounds for the Excluded Grid Theorem.

There are two caveats in our proof. The first one, that we have already mentioned, is that it requires that  the input graph $G$ has a bounded degree. This can be achieved in several ways, using prior work. Reed and Wood~\cite{ReedW-grid} showed that any graph of treewidth $k$ contains a subgraph of maximum vertex degree $4$, and treewidth $\Omega(k^{1/4}/\log^{1/8}k)$. Kreutzer and Tazari~\cite{KreutzerT10} gave a constructive proof of a similar result, with slightly weaker bounds. The algorithm of Chekuri and Ene~\cite{ChekuriE13} can be used to construct a subgraph $G'$ of the input treewidth-$k$ graph $G$, such that the treewidth of $G'$ is $\Omega(k/\poly\log k)$, and maximum vertex degree bounded by some constant. Finally, Chekuri and Chuzhoy~\cite{tw-sparsifiers} have recently shown that any graph $G$ of treewidth $k$ contains a subgraph of maximum vertex degree $3$, and treewidth $\Omega(k/\poly\log k)$. Unfortunately, this latter result builds on parts of the previous proof of the Excluded Grid Theorem of~\cite{CC14}. Therefore, if one is interested in a simple self-contained proof of Theorem~\ref{thm: GMT}, one should use the result of~\cite{ReedW-grid} as a starting point. In this paper we chose instead to use the result of~\cite{tw-sparsifiers} as our starting point, for two reasons. First, it gives the best bounds on both the degree and the treewidth of the resulting graph, leading to better final bounds on $f(g)$. Second, working with graphs whose maximum vertex degree is $3$ is easier than with general constant-degree graphs, since routing on edge-disjoint and node-disjoint paths in such graphs is very similar. This saves on a number of technical steps and makes the proof easier to follow.  The second caveat is that, unlike the proof of~\cite{CC14}, that also provides an algorithm, whose running time is polynomial in $|V(G)|$ and $g$, to construct a model of the $(g\times g)$-grid minor, our proof is non-constructive. We believe that it can be turned into an algorithm whose running time is $2^{O(g)}\cdot \poly(|V(G)|)$, using methods similar to those used in~\cite{CC14}, but we have decided to keep the proof non-constructive for the sake of simplicity. It is however unlikely that our methods can give an algorithm whose running time is polynomial in both $g$ and $|V(G)|$, since we need to solve the sparsest cut problem (with $g$ terminals) exactly. We note that most applications of the Excluded Grid Theorem (e.g. in Fixed-Parameter Tractability and in \EP--type results) only use the non-constructive version of the theorem. In other results, where a constructive version is used, such as the algorithm of Robertson and Seymour for the Node-Disjoint Paths problem~\cite{flat-wall-RS}, a running time of $2^{O(g)}\cdot \poly(|V(G)|)$ for finding the grid minor is acceptable, since the rest of the algorithm inherently incurs this (and in fact much higher) running time. However, some technical ingredients of the proof (such as the construction of the \PoS) are useful in several applications, such as, for example, approximation algorithms for routing problems. Such application require a running time that is polynomial in both $|V(G)|$ and $g$, and from this viewpoint some results of~\cite{CC14} are not subsumed by this paper.

As in much prior work in this area, we use the notion of well-linkedness. We say that a set $T$ of vertices is $\alpha$-well-linked in a graph $H$, for $0<\alpha<1$, iff for any pair $T',T''\subseteq T$ of disjoint equal-sized subsets of vertices of $T$, there is a set $\qset(T',T'')$ of paths in $H$, connecting every vertex of $T'$ to a distinct vertex of $T''$, such that every edge of $H$ participates in at most $1/\alpha$ such paths. We will informally say that a set $T$ of vertices is well-linked, if $T$ is $\alpha$-well-linked for some constant $\alpha$. 
A central combinatorial object used in the proof of the Excluded Grid Theorem of~\cite{CC14}, and that we also use here, is the \PoS (see Figure~\ref{fig:pos}). We note that Leaf and Seymour~\cite{LeafS12} used a very similar, but somewhat weaker object, called a \emph{grill}. A \PoS of width $w$ and length $\ell$ consists of a sequence $\sset=(S_1,\ldots,S_{\ell})$ of $\ell$ clusters, where for each cluster $S_i\subseteq V(G)$, we are given two disjoint subsets $A_i,B_i\subseteq S_i$ of $w$ vertices each. We require that the vertices of $A_i\cup B_i$ are well-linked in $G[S_i]$. Additionally, for each $1\leq i<\ell$, the \PoS contains a set $\pset_i$ of $w$ paths, connecting every vertex of $B_i$ to a distinct vertex of $A_{i+1}$. The paths in $\bigcup_i\pset_i$ must be all mutually disjoint, and they cannot contain the vertices of $\bigcup_{i'=1}^{\ell}S_{i'}$ as inner vertices. Chekuri and Chuzhoy~\cite{CC14}, strengthening a similar result of Leaf and Seymour~\cite{LeafS12}, showed that if a graph $G$ contains a \PoS of width $\Theta(g^2)$ and length $\Theta(g^2)$, then $G$ contains the $(g\times g)$-grid as a minor. (We provide their proof in Appendix for completeness.) Therefore, in order to prove the Excluded Grid Theorem, it is now enough to show that every graph of treewidth at least $f(g)$ contains a \PoS of width $\Omega(g^2)$ and length $\Omega(g^2)$.

\begin{figure}[htb]
  \centering
  \scalebox{0.8}{\includegraphics[width=6in]{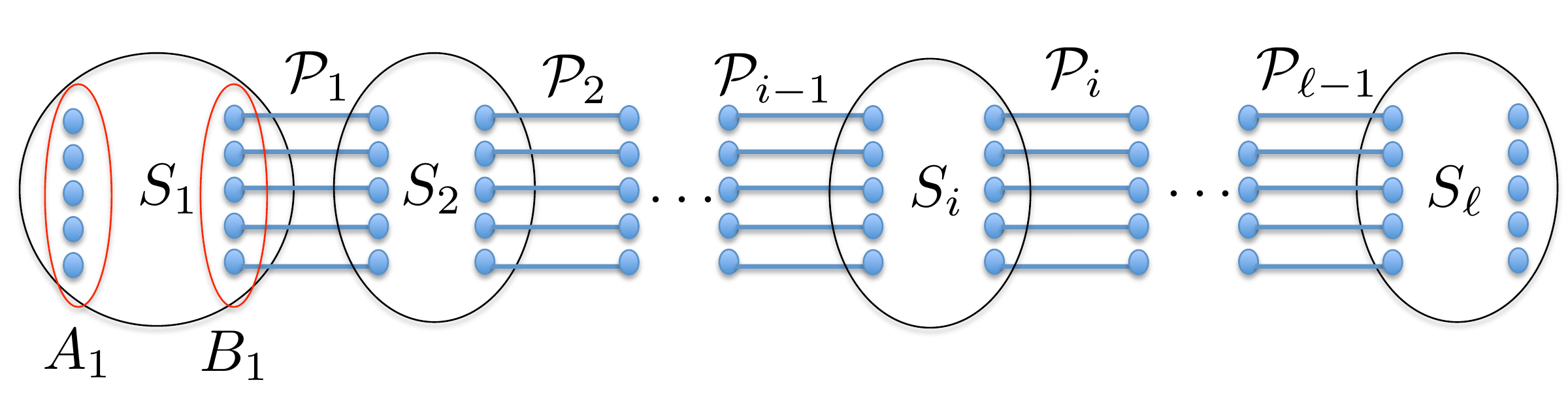}}
  \caption{Path-of-Sets System}
  \label{fig:pos}
\end{figure}

Chekuri and Chuzhoy~\cite{CC14} showed this to be true for $f(g)=O(g^{98}\poly\log g)$, and we prove it here for $f(g)=O(g^{19}\poly\log g)$. We now briefly summarize the proof of~\cite{CC14}, before we describe our proof. It is well-known (see e.g.~\cite{Reed-chapter}), that if a graph $G$ has treewidth $k$, then there is a subset $T\subseteq V(G)$ of $\Omega(k)$ vertices, such that $T$ is well-linked in $G$. Throughout the proof, we will refer to the vertices of $T$ as  \emph{terminals}. Given any cluster $C\subseteq V(G)$, we will denote by $\out(C)$ the set of edges of $G$ with exactly one endpoint in $C$, and by $\Gamma(C)$ the \emph{boundary} of $C$  --- the set of vertices of $C$ incident on the edges of $\out(C)$.

The proof of~\cite{CC14} consists of four steps. In the first step, they show that any graph $G$ of treewidth $k$ contains a large collection $\sset$ of disjoint \emph{good routers}. Informally, a good router is a cluster $C\subseteq V(G)$, such that (i) the boundary of $C$ is well-linked in $G[C]$; and (ii) there is a set $\pset(C)$ of $k^{\epsilon}$ disjoint paths, for some constant $0<\eps<1$, connecting the terminals to the vertices of $C$. The construction of the routers involves several old and new techniques, such as building a contracted graph that ``hides''  irrelevant information about $G$ by contracting some clusters; random partitions of graphs; and the so-called well-linked decompositions. In the second step, the clusters of $\sset$ are ``organized'' into a tree: that is, we construct an object, called a \ToS, that is similar to the \PoS, except that the clusters are connected via a tree-like structure instead of a path-like structure. Specifically, a \ToS of size $\ell$ and width $w$ consists of a tree $\tau$ with $\ell$ vertices; a cluster $S(v)\subseteq V(G)$ for every vertex $v\in V(\tau)$; and a set $\pset(e)$ of $w$ paths for every edge $e=(u,v)\in E(\tau)$, where every path in $\pset(e)$ connects a vertex of $S(v)$ to a vertex of $S(u)$, and is internally disjoint from $\bigcup_{x\in V(\tau)}S(x)$. We also require that all paths in $\pset=\bigcup_{e\in E(\tau)}\pset(e)$ are node-disjoint, and for every cluster $S(v)$, the endpoints of all paths in $\pset$ that lie in $S(v)$ are well-linked in $G[S(v)]$.
 %The construction of the \ToS involves carefully removing vertices of $G$ that do not belong to the clusters of $\set{S(v)}_{v\in V(\tau)}$, while preserving the connectivities between these clusters, by using graph-reduction steps preserving element-connectivity, together with standard edge-splitting. 
 If the resulting \ToS has a long root-to-leaf path, then we can use this path as the final \PoS. Otherwise, let $\sset'$ be the subset of clusters $S(v)$ that correspond to the leaves of the tree. In the third step, we repeat Step 2 on the clusters of $\sset'$ instead of the original set of clusters, and a carefully selected subgraph $G'$ of $G$, to ensure that the tree corresponding to the resulting \ToS has maximum vertex degree at most $3$. This step relies on an LP-based approximation algorithm for bounded-degree spanning trees of~\cite{Singh-Lau}. Finally, in the fourth step, we turn the resulting \ToS into a \PoS, by carefully simulating a depth-first search tour of the corresponding sub-cubic tree.

In contrast, the algorithm for our basic construction consists of only one subroutine, that, intuitively, shows that, given any \PoS of length $1$ and width $w$, we can obtain a \PoS of length $2$ and width $w/c'$, for some constant $c'$. More specifically, suppose we are given some subset $S$ of vertices of $G$, and two disjoint subsets $T_1,T_2\subseteq S$ of vertices, such that $|T_1|=w/c$ (where $c$ is some constant), $|T_2|=w$, and $(T_1\cup T_2)$ is well-linked in $G[S]$. We show that there are two disjoint clusters $X,Y\subseteq S$, a subset $E'\subseteq E(X,Y)$ of $w/c^2$ edges whose endpoints are all distinct, and two subsets $\ttset_1\subseteq X\cap T_1$ of at least $w/c^2$ vertices and $\ttset_2\subseteq Y\cap T_2$ of at least $w/c$ vertices, such that, if we denote by $\U_X$ and $\U_Y$ the endpoints of the edges of $E'$ that belong to $X$ and $Y$, respectively, then $\U_X\cup \ttt1$ is well-linked in $G[X]$, and $\U_Y\cup \ttt2$ is well-linked in $G[Y]$ (see Figure~\ref{fig: strong chain}). We call the corresponding tuple $(X,Y,\ttt1,\ttt2,E')$ a \emph{$2$-cluster chain}, and we call this procedure a \emph{splitting of a cluster}. Using this procedure, it is now easy to complete the proof of the Excluded Grid Theorem. Let $k$ be the treewidth of the input bounded-degree graph $G$. Our algorithm performs $2\log_2 g$ phases, where each phase $j$ starts with a \PoS of length $2^{j-1}$ and width $\Omega(k/c^{2(j-1)})$, and produces a \PoS of length $2^j$ and width $\Omega(k/c^{2j})$. For our initial \PoS of length $1$ and width $\Omega(k)$, we use $S_1=V(G)$, and we let $(A_1,B_1)$ be any partition of the terminals into equal-sized subsets. Clearly, after $2\log_2g$ phases, we obtain a \PoS of length $g^2$ and width $\Omega(k/g^{4\log c})$. Each phase is executed by simply splitting each cluster of the current \PoS into two, using the cluster-splitting procedure described above. We omit the technical details, that can be found in Section~\ref{sec: building PoS}.
%For each such cluster $S_i$ of the \PoS, we use the vertices of $A_i$ and $B_i$ as the sets $T_1$ and $T_2$ of terminals. Eventually, only small subsets $\ttset_1\subseteq T_1$ and $\ttset_2\subseteq T_2$ of such terminals remain in the \PoS after splitting $S_i$, and we discard from $\pset_{i-1}$ and $\pset_i$ all paths except those that contain the vertices of $\ttset_1\cup \ttset_2$ as endpoints. This step needs to be executed carefully, so that for each $i$, the paths of $\pset_i$ that are discarded due to clusters $S_i$ and $S_{i+1}$ are coordinated, so we do not end up discarding all such paths. In order to do so, we apply the splitting procedure to each of the clusters $S_1,S_2,\ldots$ in turn. Before cluster $S_i$ is processed, we discard from $\pset_{i-1}$ all paths whose endpoints were discarded while splitting $S_{i-1}$, and we discard the endpoints of such paths from $A_i$. We then apply the cluster-splitting procedure to $S_i$, using only the remaining vertices of $A_i$ as $\ttset_1$.

\begin{figure}[h]
\centering
\subfigure[Strong $2$-cluster chain.]{\scalebox{0.3}{\includegraphics{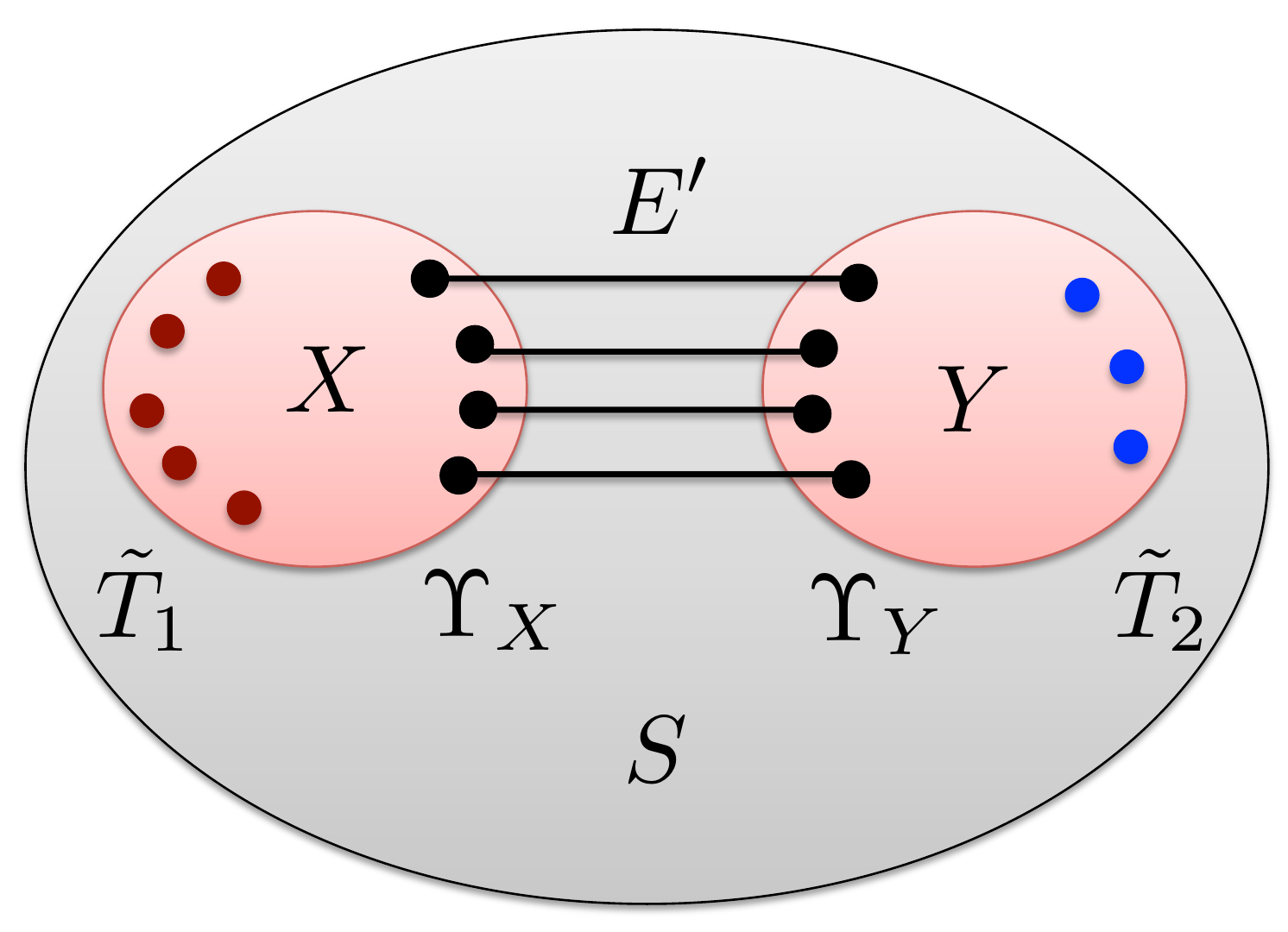}}\label{fig: strong chain}}
\hspace{1cm}
\subfigure[Weak $2$-cluster chain.]{
\scalebox{0.3}{\includegraphics{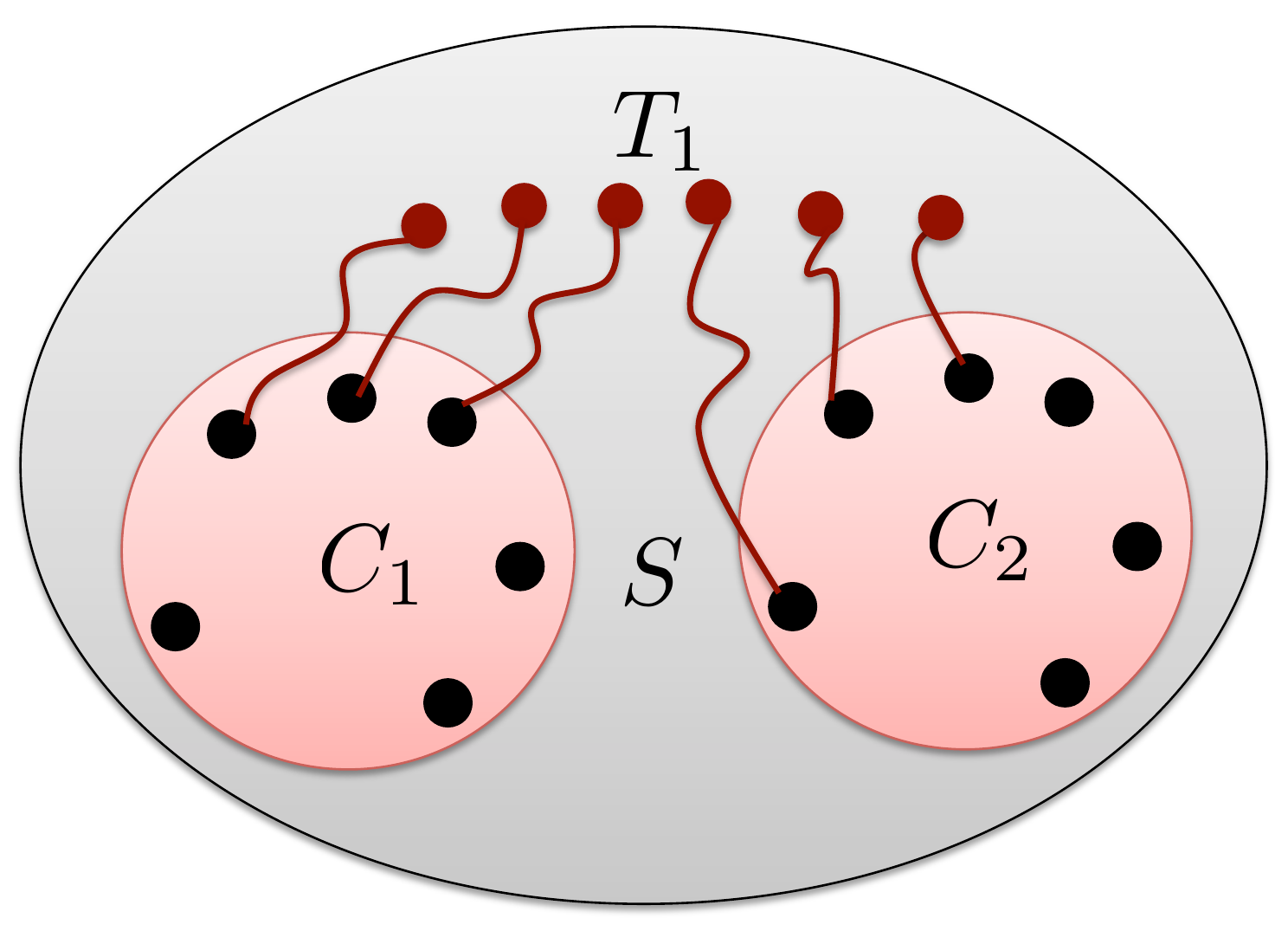}}\label{weak-two-cluster-chain}}
%\hspace{1cm}
%\subfigure[Clusters $X'$ and $Y'$.]{\scalebox{0.35}{\includegraphics{advanced-splitting3.pdf}}\label{fig: advanced-splitting3}}
\caption{Splitting a cluster.\label{fig: splitting}}
\end{figure}

%\begin{figure}[h]
%\scalebox{0.35}{\includegraphics{two-cluster-chain-cut.pdf}}
%\caption{Splitting a cluster.\label{fig: splitting}}
%\end{figure}

%\begin{figure}[h]
%\scalebox{0.4}{\includegraphics{two-cluster-chain-cut.pdf}}\caption{Splitting a cluster.\label{fig: splitting}}
%\end{figure}

We now briefly sketch our algorithm for splitting a cluster $S$. We note that this is an informal and imprecise overview, that is only intended to provide intuition. Let $k'=|T_1|$. 
We start by defining a slightly weaker object, called a \emph{weak $2$-cluster chain}. This object consists of two disjoint clusters $C_1,C_2\subseteq S\setminus(T_1\cup T_2)$, such that for $i\in\set{1,2}$, the set $\Gamma(C_i)$ of the boundary vertices of $C_i$ is well-linked in $G[C_i]$, and there is a set $\pset_i$ of $\Omega(k')$ node-disjoint paths, connecting the vertices of $C_i$ to the terminals of $T_1$, such that the paths in $\pset_1\cup \pset_2$ are disjoint from each other, and do not contain the vertices of $C_1\cup C_2$ as inner vertices (see Figure~\ref{weak-two-cluster-chain}). We show that the existence of the weak $2$-cluster chain is sufficient to guarantee the existence of the (strong) $2$-cluster chain in $G[S]$: the idea is to use the well-linkedness of the set $T_1\cup T_2$ of vertices, to carefully connect the clusters $C_1,C_2$ to each other, and to connect one of them to the set $T_2$ of vertices, by large enough collections of disjoint paths. The main technical difficulty of the proof is showing that any cluster $S$, with two disjoint subsets $T_1,T_2\subseteq S$ of vertices, where $(T_1\cup T_2)$ is well-linked in $G[S]$, contains a weak $2$-cluster chain.

In order to improve the bounds on $f(g)$, we combine a construction of the \ToS of~\cite{CC14} (that we somewhat improve and simplify) with our construction of the \PoS, and use a resulting graph and its embedding into $G$ in order to construct a grid minor.

{\bf Organization.}
We start with preliminaries in Section~\ref{sec: prelims}, and provide our basic construction in sections~\ref{sec: building PoS} and \ref{sec: splitting a cluster}. Section~\ref{sec: building PoS} introduces a general framework for constructing the \PoS, and Section~\ref{sec: splitting a cluster} focuses on splitting a cluster. This part of the paper also provides all proofs omitted in the extended abstract~\cite{GMT-STOC}. Our advanced construction appears in Sections~\ref{sec: ToS}--\ref{sec: parallel splitting}, where Section~\ref{sec: ToS} provides a somewhat more streamlined construction of the \ToS of~\cite{CC14}, Section~\ref{sec: extended construction} shows how to combine both approaches to obtain stronger bounds on $f(g)$ and Section~\ref{sec: parallel splitting} completes the proof.

\label{----------------------------------------------sec:prelims----------------------------------------}
%-----------------------------------------------------
%-----------------------------------------------------
\section{Preliminaries}\label{sec: prelims}
%-----------------------------------------------------
%-----------------------------------------------------
All logarithms in this paper are to the base of $2$.
All graphs in this paper are finite, and they do not have loops. We say that a graph is \emph{simple} to indicate that it does not have parallel edges; otherwise, parallel edges are allowed.
Given a graph $G=(V,E)$ and a subset $A\subseteq V$ of vertices, we
denote by $E_G(A)$ the set of edges with both endpoints in $A$.  For two
disjoint sets $A,B\subseteq V$, the set of edges with one endpoint in $A$ and the other in $B$ is denoted by $E_G(A,B)$. The degree of a vertex $v\in V$ is denoted by $d_G(v)$, and the set of all edges incident on $v$ is denoted by $\delta_G(v)$. 
We sometimes refer to sets of vertices as \emph{clusters}. %We say that a cluster $C$ is \emph{connected} iff $G[C]$ is a connected graph. 
Given a cluster $C\subseteq V$, we denote by $\out_G(C)$ the set of edges with exactly one endpoint in $C$, and by $\Gamma_G(C)$ the set of vertices of $C$ incident on the edges of $\out_G(C)$. We sometimes call $\Gamma_G(C)$ \emph{the boundary of $C$}. We may omit the subscript $G$ if it is clear from context.

%We usually refer to a $(g\times g)$-grid as \emph{a grid of size $g$}.

We say that a path $P$ is \emph{internally disjoint} from a set $U$ of vertices, if no vertex of $U$ serves as an inner vertex of $P$.
We say that two paths $P,P'$ are \emph{internally disjoint}, iff for every vertex $v\in V(P)\cap V(P')$, $v$ is an endpoint of both paths.
Given a set $\pset$ of paths in $G$, we denote by $V(\pset)$ the set of all vertices participating in paths in $\pset$.
Let $\pset$ be any collection of paths in graph $G$. We say that the paths in $\pset$ cause edge-congestion $\eta$, if every edge $e\in E$ is contained in at most $\eta$ paths in $\pset$. For two subsets $S,T\subseteq V(G)$ of vertices and a set $\pset$ of paths, we say that $\pset$ connects $S$ to $T$ if every path in $\pset$ has one endpoint in $S$ and another in $T$ (or it consists of a single vertex lying in $S\cap T$).

Assume that we are given two subsets $S,T\sse V$ of vertices. We denote by $\pset:S\connect T$ a collection $\pset=\set{P_v\mid v\in S}$ of paths, where  path $P_v$ has $v$ as its first vertex and some vertex of $T$ as its last vertex. Notice that each path of $\pset$ originates from a distinct vertex of $S$, and $|\pset|=|S|$. If additionally the set $\pset$ of paths causes edge-congestion at most $\eta$, then we denote this by $\pset:S\connect_{\eta}T$. Assume now that $|S|=|T|=|\pset|$, and each path in $\pset$ connects a distinct vertex of $S$ to a distinct vertex of $T$.  Then we denote $\pset:S\sconnect T$, and if the paths in $\pset$ cause edge-congestion at most $\eta$, then we denote $\pset:S\sconnect_{\eta}T$. Notice that the paths of $\pset$ are allowed to contain the vertices of $S\cup T$ as inner vertices. 
We  repeatedly use the following simple observation, whose proof appears in Appendix.

\begin{observation}\label{obs: EDP to NDP in degree-3}
Let $G$ be a graph with maximum vertex degree at most $3$, and $T_1,T_2\subseteq V(G)$ a pair of {\bf disjoint} equal-sized subsets of its vertices, such that the degree of every vertex in $T_1\cup T_2$ is at most $2$. Let $\pset: T_1\sconnect_1 T_2$ be any set of edge-disjoint paths connecting every vertex of $T_1$ to a distinct vertex of $T_2$. Then the paths in $\pset$ are node-disjoint.
\end{observation}

\subsection{Flows and Cuts}
%---------------------------------------------
%---------------------------------------------
%---------------------------------------------
In this section we define standard single-commodity flows and discuss their relationships with the corresponding notions of cuts. Most definitions and results from this section can be found in standard textbooks; we refer the reader to~\cite{Schrijver} for more details.

%\paragraph{$S$--$T$ flows, paths and cuts}
%We start with the standard (single-commodity) flow. 
Suppose we are given a graph $G=(V,E)$ with capacities $c(e)> 0$ on its edges $e\in E$, and two disjoint sets $S,T\subseteq V$ of vertices of $G$. Let $\pset$ be the set of all paths that start at $S$ and terminate at $T$. An $S$--$T$ flow $F$ is an assignment of non-negative values $F(P)$ to all paths $P\in \pset$. The \emph{value} of the flow is $\sum_{P\in \pset}F(P)$. Given a flow $F$, for each edge $e\in E$, we define a flow through $e$ to be: $F'(e)=\sum_{\stackrel{P\in \pset:}{e\in P}}F(P)$. 
The \emph{edge-congestion} of the flow is $\max_{e\in E}\set{F'(e)/c(e)}$.
We say that the flow $F$ is \emph{valid}, or that it causes no edge-congestion, if its edge-congestion is at most $1$.  
%We note that even though $|\pset|$ may be exponential in $|V|$, there are known efficient algorithms to compute a valid flow of a specified value $F$ (if it exists), and to compute a flow of maximum value. Moreover, the number of paths in $\pset$ with non-zero flow value $F(P)$ is guaranteed to be at most $|E|$ in both cases. Such flows can be computed, for example, by using an equivalent edge-based flow formulation together with Linear Programming, and a flow-path decomposition of the resulting solution (see~\cite{Schrijver} for more details). 
It is well known that if all edge capacities are integral, then whenever a valid $S$--$T$ flow of value $f$ exists in $G$, there is also a valid $S$--$T$ flow $\tilde F$ of the same value, where all values $\tilde F(P)$ are integral, and the number of non-zero values $\tilde f(P)$ is at most $|E|$. Throughout the paper, whenever the edge capacities of a given graph $G$ are not specified, we assume that they are all unit.

A \emph{cut} in a graph $G$ is a bi-partition $(A,B)$ of its vertices, with $A,B\neq \emptyset$. We sometimes use $\overline{A}$ to denote $V\setminus A$. The \emph{value} of the cut is the total capacity of all edges in $E(A,B)$ (if the edge capacities of $G$ are not specified, then the value of the cut is $|E(A,B)|$). We say that a cut $(A,B)$ \emph{separates} $S$ from $T$ if $S\subseteq A$ and $T\subseteq B$. The famous Maximum Flow -- Minimum Cut theorem states that the value of the maximum valid $S$--$T$ flow is equal to the value of the minimum cut separating $S$ from $T$ in every graph $G$. Notice that if all edges of $G$ have a unit capacity, and the value of the maximum flow from $S$ to $T$ is $f$, then the maximum number of edge-disjoint paths connecting the vertices of $S$ to the vertices of $T$ is also $f$, and if $E'$ is a minimum-cardinality set of edges, such that $G\setminus E'$ contains no path connecting a vertex of $S$ to a vertex of $T$, then $|E'|=f$.
%Given any pair $s,t\in V(G)$ of distinct vertices of $G$, a minimum $s$--$t$ cut is a cut separating $\set{s}$ from $\set{t}$ of minimum value.
When $S=\set{s}$ and $T=\set{t}$, then we sometimes refer to the $S$-$T$ flow and $S$-$T$ cut as $s$-$t$ flow and $s$-$t$ cut respectively.

Similarly to our notation for paths, a flow $F$ from the vertices of $S$ to the vertices of $T$, where every vertex of $S$ sends one flow unit, every vertex of $T$ receives one flow unit, and the edge-congestion is at most $\eta$, is denoted by $F:S\sconnect_{\eta}T$.
%Given any set $\pset'\subseteq \pset$ of paths connecting vertices of $S$ to  vertices of $T$ in $G$, we say that the paths in $\pset'$ cause edge-congestion at most $\eta$, if for every edge $e\in E$, the total number of paths in $\pset'$ containing $e$ is at most $\eta c(e)$.

A variant of the $S$--$T$ flow that we sometimes use is when the capacities are given on the graph vertices and not edges. Such a flow $F$ is defined exactly as before, except that now, for every vertex $v\in V$, we let $F'(v)=\sum_{\stackrel{P\in \pset:}{v\in P}}F(P)$, and we define the congestion of the flow to be $\max_{v\in V}\set{F'(v)/c(v)}$. If the congestion of the flow is at most $1$, then we say that it is a valid flow, or that the flow causes no vertex-congestion. When all vertex capacities are integral, there is a maximum flow $F$, such that all values $F(P)$ for all $P\in \pset$ are integral. In particular, if all vertex-capacities are $1$, and there is a valid $S$--$T$ flow of value $f$, then there are $f$ node-disjoint paths connecting vertices of $S$ to vertices of $T$.

All the definitions and results about single-commodity flows mentioned above carry over to directed graphs as well, except that cuts are defined slightly differently. As before, a cut in $G$ is a bi-partition $(A,B)$ of the vertices of $G$. The value of the cut is the total capacity of edges connecting vertices of $A$ to vertices of $B$. The Maximum Flow -- Minimum Cut theorem remains valid in directed graphs, with this definition of cuts. For every directed flow network, there exists a maximum $S$--$T$ flow, in which for every par $(e,e')$ of anti-parallel edges, at most one of these edges carries non-zero flow; if all edge capacities are integral, then there is a maximum flow that is integral and has this property. This follows from the equivalent edge-based definition of flows.
Flows in directed graphs with capacities on vertices are defined similarly. We repeatedly use the following simple observation.

\begin{observation}\label{obs: low cong flow to NDP}
Suppose we are given a graph $G$ with maximum vertex degree $d$ and unit edge capacities, two disjoint subsets $S,T\subseteq V(G)$ of vertices, and an $S$--$T$ flow $F$ of value $\kappa$, that causes edge-congestion at most $\eta\geq 1$. Then there is a collection $\qset$ of $\ceil{\frac{\kappa}{d\eta}}$ node-disjoint paths in $G$, where every path has one endpoint in $S$ and another in $T$.
\end{observation}

\begin{proof}
Let $\pset$ be the set of all paths connecting $S$ to $T$. Then for every vertex $v\in V(G)$, the total flow through $v$, $\sum_{\stackrel{P\in \pset:}{v\in V(P)}}F(P)\leq d\eta$. By sending $F(P)/(d\eta)$ flow units along every path $P\in \pset$, we obtain a flow that causes vertex-congestion at most $1$. The value of the flow is at least $\frac{\kappa}{d\eta}$. From the integrality of flow, there is a collection $\qset$ of $\ceil{\frac{\kappa}{d\eta}}$ node-disjoint paths in $G$, connecting vertices of $S$ to vertices of $T$.
\end{proof}

\subsection{Treewidth, Grids, Minors and Models}

The treewidth of a graph $G=(V,E)$ is defined via tree-decompositions.  A tree-decomposition of a graph $G$ consists of a tree
$\tau$ and a collection $\{X_v \subseteq V\}_{v \in V(\tau)}$ of vertex subsets called \emph{bags}, that have the following properties: (i) for each edge $(a,b) \in E$, there is some node $v \in V(\tau)$ with $a,b \in X_v$; and (ii) for each vertex $a \in V$, the set of all nodes of $\tau$ whose bags contain $a$ induces a non-empty (connected) subtree of $\tau$. The {\em width} of a given tree-decomposition is
$\max_{v \in V(\tau)}\set{ |X_v|} - 1$, and the \emph{treewidth} of a graph $G$, denoted by $\tw(G)$, is the width of a minimum-width tree-decomposition of $G$.

We say that a simple graph $H$ is a minor of a graph $G$, if $H$ can be
obtained from $G$ by a sequence of edge deletion, vertex deletion, and edge contraction
operations. Equivalently, a simple graph $H$ is a minor of $G$ if there is a map $\phi$, assigning to each vertex $v\in V(H)$ a subset $\phi(v)$ of vertices of $G$, and to each edge $e=(u,v)\in E(H)$ a path $\phi(e)$ connecting a vertex of $\phi(u)$ to a vertex of $\phi(v)$, such that:
\begin{itemize}
\item For each vertex $v\in V(H)$, the subgraph of $G$ induced by $\phi(v)$ is connected;
\item  If $u,v\in V(H)$ and $u\neq v$, then $\phi(u)\cap \phi(v)=\emptyset$; and
\item The paths in set $\set{\phi(e)\mid e\in E(H)}$ are node-disjoint, and they are internally disjoint from $\bigcup_{v\in V(H)}\phi(v)$.
\end{itemize}

A map $\phi$ satisfying these conditions is called \emph{a model of $H$ in $G$}. (We note that this definition is slightly different from the standard one, which requires that for each $e\in E(H)$ path $\phi(e)$ consists of a single edge; but it is immediate to verify that both definitions are equivalent.) For convenience, we sometimes refer to the map $\phi$ as an \emph{embedding} of $H$ into $G$, and specifically to $\phi(v)$ and $\phi(e)$ as the embeddings of the vertex $v\in V(H)$ and the edge $e\in E(H)$, respectively.

The $(g\times g)$-grid is a graph, whose vertex set is: $\set{v(i,j)\mid 1\leq i,j\leq g}$.
The edge set consists of two subsets: a set of \emph{horizontal edges} $E_1=\set{(v(i,j),v(i,j+1))\mid 1\leq i\leq g; 1\leq j<g}$; and a set of \emph{vertical edges} $E_2=\set{(v(i,j),v(i+1,j))\mid 1\leq i<g; 1\leq j\leq g}$. The subgraph induced by $E_1$ consists of $g$ disjoint paths, that we refer to as \emph{the rows of the grid}; the $i$th row is the row incident with $v(i,1)$. Similarly, the subgraph induced by $E_2$ consists of $g$ disjoint paths, that we refer to as \emph{the columns of the grid}; the $j$th column is the column incident with $v(1,j)$. 
We say that $G$ contains a $(g\times g)$-grid minor if some minor $H$ of $G$ is isomorphic to the $(g\times g)$-grid.

\label{------------------------------------------subsec: well-linkedness-------------------------------}
\subsection{Linkedness, Well-Linkedness, and Bandwidth Property} 
The notion of well-linkedness has played a central role in algorithms for routing problems (see e.g.~\cite{Raecke,ANF,CKS,RaoZhou,Andrews,Chuzhoy11,ChuzhoyL12,ChekuriE13}), and is also often used in graph theory. Several different variations of this notion were used in the past.  The definitions we use here are equivalent to those used in~\cite{Chuzhoy11,ChuzhoyL12,ChekuriE13,CC14}, but for convenience we define them slightly differently. %We start from the basic  $\alpha$-well-linkedness.

\begin{definition}
Given a graph $G$, a subset $T\subseteq V(G)$ of vertices, and a parameter $0<\alpha\leq 1$, we say that $\tset$ is $\alpha$-well-linked in $G$, iff for every pair of disjoint equal-sized subsets $\tset',\tset''\subseteq \tset$, there is a flow $F:\tset'\sconnect_{1/\alpha}\tset''$ in $G$.
\end{definition}

 The next observation follows immediately from the definition of well-linkedness, and from the integrality of flow.
 
 \begin{observation}\label{obs: well-linkedness properties}
 Let $G$ be a graph, $T\subsetq V(G)$ a subset of its vertices, such that $T$ is $\alpha$-well-linked in $G$, for some $0<\alpha\leq 1$. Then:
 
 \begin{itemize}
 \item  for every pair of disjoint equal-sized subsets $\tset',\tset''\subseteq \tset$, there is a set $\pset:\tset'\sconnect_{\ceil{1/\alpha}}\tset''$ of paths in $G$;
 
 \item every subset $\tset'\subseteq \tset$ is $\alpha$-well-linked in $G$;
 
 \item set $\tset$ is $\alpha'$-well-linked in $G$ for all $0<\alpha'<\alpha$; and
 
 \item for every graph $G'$ with $G\subseteq G'$, set $\tset$ is $\alpha$-well-linked in $G'$.
 \end{itemize}
 \end{observation}
 
 The next observation relates our definition to the one used in~\cite{ANF,CKS, Chuzhoy11,ChuzhoyL12,ChekuriE13,CC14}, and its proof  appears in Appendix.

\begin{observation}\label{obs: well-linkedness alt def}
Assume that we are given a graph $G$, a vertex set $\tset\subseteq V(G)$ and a parameter $0<\alpha\leq 1$, such that $\tset$ is {\bf not} $\alpha$-well-linked in $G$.
Then there is a partition $(A,B)$ of $V(G)$, with $|E(A,B)|< \alpha\cdot \min\set{|A\cap \tset|,|B\cap \tset|}$.
\end{observation}

We call the partition given in Observation~\ref{obs: well-linkedness alt def} \emph{an $\alpha$-violating partition} of $G$ with respect to $\tset$.
We also need a slightly more general definition of well-linkedness, similar to that introduced in~\cite{Chuzhoy11}. %We note that this definition is only used to optimize the parameters of the proof.

\begin{definition}
Given an integer $k'$, and a parameter $0<\alpha\leq 1$, we say that a set $\tset$ of vertices is $(k',\alpha)$-well-linked in graph $G$, iff for every pair of disjoint subsets $\tset',\tset''\subseteq \tset$, with $|\tset'|=|\tset''|\leq k'$, there is a flow  $F:\tset'\sconnect_{1/\alpha}\tset''$ in $G$.
\end{definition}

Notice that if $|\tset|\leq 2k'$, then $\tset$ is $\alpha$-well-linked in $G$ iff it is $(k',\alpha)$-well-linked in $G$. Notice also that  if a set $\tset$ of terminals is $(k',\alpha)$-well-linked in $G$, then so is every subset $\tset'\subseteq\tset$. As before, if set $\tset$ is $(k',\alpha)$-well-linked in $G$, then for every pair of disjoint subsets $\tset',\tset''\subseteq \tset$, with $|\tset'|=|\tset''|\leq k'$, there is a set  $\pset:\tset'\sconnect_{\ceil{1/\alpha}}\tset''$ of paths in $G$.
The following observation is an analogue of Observation~\ref{obs: well-linkedness alt def}, and its proof \iffull appears in Appendix.\fi\ifabstract is omitted here.\fi

\begin{observation}\label{obs: generalized well-linkedness alt def}
Assume that we are given a graph $G$, a set $\tset$ of vertices of $G$, an integer $k'>0$, and a parameter $0<\alpha\leq 1$. Assume further that $\tset$ is {\bf not} $(k',\alpha)$-well-linked in $G$.
Then there is a partition $(A,B)$ of $V(G)$, such that $|E(A,B)|< \alpha\cdot \min\set{|A\cap \tset|,|B\cap \tset|, k'}$.
\end{observation}

We call the partition given in Observation~\ref{obs: generalized well-linkedness alt def} a  \emph{$(k',\alpha)$-violating partition} of $V(G)$ with respect to $\tset$. We next define the notion of bandwidth property, somewhat similar to the one defined in~\cite{Raecke}. Recall that for a cluster $C\subseteq V(G)$, $\Gamma_G(C)$ denotes the set of vertices of $C$ incident with the edges of $\out_G(C)$.

\begin{definition}
Given a graph $G$, an integer $k'$ and a parameter $0<\alpha\leq 1$, we say that a cluster $C\subseteq V(G)$ has the $(k',\alpha)$-bandwidth property, iff $\Gamma_G(C)$ is $(k',\alpha)$-well-linked in $G[C]$. We say that it has the $\alpha$-bandwidth property, iff $\Gamma_G(C)$ is $\alpha$-well-linked in $G[C]$.
\end{definition}

The following observation is immediate from the definition of the bandwidth property.
\begin{observation}\label{obs: bandwidth gives connectivity}
Let $G$ be a connected graph, and let $C\subseteq V(G)$ be a cluster that has the $(k',\alpha)$-bandwidth property, for some integer $k'\geq 2$, and a parameter $0<\alpha\leq 1$. Then $G[C]$ is connected.
\end{observation}

We now define a stronger notion of well-linkedness, called \emph{node-well-linkedness}.

%---------------------------------------------------
%---------------------------------------------------
%---------------------------------------------------
%---------------------------------------------------

\begin{definition}
  We say that a set $\tset$ of vertices is \emph{node-well-linked} in
  $G$, iff for every pair $(\tset',\tset'')$ of disjoint equal-sized subsets of
  $\tset$, there is a collection $\pset:\tset'\sconnect\tset''$ of {\bf
    node-disjoint} paths in $G$.
\end{definition}

Notice that from Observation~\ref{obs: EDP to NDP in degree-3}, if $G$ is a graph with maximum vertex degree at most $3$, and $T$ is a set of vertices of degree at most $2$ each, then $T$ is node-well-linked in $G$ iff $T$ is $1$-well-linked in $G$.
Finally, we define the notion of linkedness between a pair of vertex subsets.

\begin{definition}
We say that two disjoint subsets $(\tset_1,\tset_2)$ of vertices of $G$ are $\alpha$-\emph{linked} for $0<\alpha\leq 1$, iff for every pair $\tset'_1\subseteq \tset_1$ and $\tset_2'\subseteq \tset_2$ of equal-sized vertex subsets, there is a flow $F:\tset_1'\sconnect_{1/\alpha} \tset_2'$ in $G$. We say that they are \emph{node-linked}, iff for every pair $\tset'_1\subseteq \tset_1$ and $\tset_2'\subseteq \tset_2$ of equal-sized vertex subsets, there is a collection of $|\tset'_1|$ node-disjoint paths connecting the vertices of $\tset'_1$ to the vertices of $\tset_2'$.
\end{definition}

Notice that as before, if $(\tset_1,\tset_2)$ are $\alpha$-linked, then for every pair $\tset'_1\subseteq \tset_1$ and $\tset_2'\subseteq \tset_2$ of equal-sized vertex subsets, there is a set $\pset:\tset_1'\sconnect_{\ceil{1/\alpha}} \tset_2'$ of paths $G$.

The following lemma summarizes an important connection between the
graph treewidth, and the size of the largest node-well-linked set of
vertices in it.

\begin{lemma}\cite{Reed-chapter}
  \label{lem:tw-wl}
  Let $k$ be the size of the largest node-well-linked vertex set in $G$. Then
  $\frac{k} 4 -1 \le \tw(G) \le k-1$.
\end{lemma}

\subsection{Boosting Well-Linkedness}
Suppose we are given a graph $G$, and a subset $T$ of its vertices, such that $T$ is $\alpha$-well-linked in $G$, for some parameter $0<\alpha\leq 1$. Boosting theorems show that there are large subsets of $T$ with much better well-linkedness properties. We start with the following simple theorem, that allows to perform basic boosting. This type of argument has been used before extensively, usually under the name of a ``grouping technique''~\cite{ANF,CKS, RaoZhou, Andrews, Chuzhoy11}.
 
\begin{theorem}\label{thm: weak well-linkedness}
Suppose we are given a connected graph $G$ with maximum vertex degree $\Delta$, and a set $\tset$ of vertices of $G$, called terminals, such that $\tset$ is $\alpha$-well-linked in $G$, for some $0<\alpha\leq 1$, and $|\tset|>1/\alpha$.
Then there is a collection $\fset$ of disjoint trees in $G$, each containing at least $\ceil{1/\alpha}$ and at most $2\Delta\cdot \ceil{1/\alpha}$ terminals. Moreover, for any set $U\subseteq \tset$ of terminals, such that for each $\tau\in \fset$, $|U\cap V(\tau)|\leq 1$, set $U$ is $1/2$-well-linked in $G$.
\end{theorem}
\begin{proof}
Let $\tau$ be any spanning tree of $G$, rooted at one of its degree-$1$ vertices. We start with $\fset=\emptyset$, and then iterate. In every iteration, we consider the current tree $\tau$. If $\tau$ contains at most $2\Delta\cdot \ceil{1/\alpha}$ terminals, then we add $\tau$ to $\sset$ and terminate the algorithm. Otherwise, let $v$ be the lowest vertex of $\tau$, such that the sub-tree $\tau_v$ of $\tau$ rooted at $v$ contains at least $\ceil{1/\alpha}$ terminals. Since the maximum vertex degree in $G$ is bounded by $\Delta$, $T_v$ contains at most $\Delta \cdot \ceil{1/\alpha}$ terminals. We add $\tau_v$ to $\fset$, delete all vertices of $\tau_v$ from $\tau$, and continue to the next iteration. Notice that $\tau$ is still guaranteed to contain at least $\Delta\cdot\ceil{1/\alpha}$ terminals.
This completes the construction of the set $\fset$ of trees.

Assume now that we are given any set $U\subseteq \tset$ of terminals, where for each $\tau\in \fset$, $|U\cap V(\tau)|\leq 1$. Let $(U',U'')$ be any pair of equal-sized subsets of $U$. It is enough to show that there is a flow $F: U'\sconnect_2 U''$ in $G$. Let $|U'|=\kappa$.

For every terminal $t\in U$, let $\tau_t\in \fset$ be the tree containing $t$, and let $L(\tau)$ be any set of $\ceil{1/\alpha}$ terminals of $\tau_t$. Let $X=\bigcup_{t\in U'}L(\tau_t)$, and let $Y=\bigcup_{t\in U''}L(\tau_t)$. Then $|X|=|Y|=\kappa \cdot \ceil{1/\alpha}$. Since the set $\tset$ of terminals is $\alpha$-well-linked, there is a flow $F':X\sconnect_{1/\alpha} Y$ in $G$. Scaling this flow down by factor $\ceil{1/\alpha}$, we obtain a new flow $F''$ from $X$ to $Y$, where every terminal in $X$ sends $1/\ceil{1/\alpha}$ flow units, every terminal in $Y$ receives $1/\ceil{1/\alpha}$ flow units, and the edge-congestion due to $F''$ is at most $1$.

The final flow $F$ is a concatenation of three flows: $F_1,F'',F_2$. Flow $F_1$ is defined as follows. Every terminal $t\in U'$ sends one flow unit to the $\ceil{1/\alpha}$ terminals of $L(\tau_t)$, splitting the flow evenly among them, so every terminal in $L(\tau_t)$ receives $1/\ceil{1/\alpha}$ flow units. The flow is sent along the edges of the tree $\tau_t$. The flow $F_1$ is the union of all such flows from all terminals $t\in U'$. The flow $F_3$ is defined similarly with respect to $U''$, except that we reverse the direction of the flow. The final flow is a concatenation of $F_1,F'',F_2$. It is easy to see that this is a flow from $U'$ to $U''$, where every terminal in $U'$ sends one flow unit and every terminal in $U''$ receives one flow unit. Since the trees in $\fset$ are edge-disjoint, the congestion due to $F$ is bounded by $2$.
\end{proof}

The above claim gives a straightforward way to boost the well-linkedness of a given set $\tset$ of terminals to $\half$. However, we need a stronger result: we would like to find a large subset $\tset'\subseteq \tset$, such that $\tset'$ is {\bf node-well-linked} in $G$. The following theorem, that was proved in~\cite{CC14}, allows us to achieve this. For completeness, we provide its proof in Appendix. 

\begin{theorem}\label{thm: grouping} \cite{CC14}
  Suppose we are given a connected graph $G=(V,E)$ with maximum vertex degree at most $\Delta$, and a set $\tset\subseteq V$ of $\kappa$ vertices called terminals, such that $\tset$ is $\alpha$-well-linked in $G$, for some $0<\alpha\leq 1$. Then there is a subset $\tset'\subseteq \tset$ of $\ceil{\frac{3\alpha\kappa }{10\Delta}}$ terminals, such that $\tset'$ is node-well-linked in $G$.
\end{theorem}

The next theorem allows us to achieve node-linkedness property between a pair of subsets of terminals. The theorem was proved in~\cite{CC14}, and its proof is included in Appendix for completeness.

%----------------------------------------------
%----------------------------------------------
\begin{theorem}\label{thm: linkedness from node-well-linkedness}~\cite{CC14}
Suppose we are given a graph $G$ with maximum vertex degree at most $\Delta$, and two disjoint subsets $\tset_1,\tset_2$ of vertices of $G$, with $|\tset_1|,|\tset_2|\geq \kappa$, such that $\tset_1\cup \tset_2$ is $\alpha$-well-linked in $G$, for some $0<\alpha<1$, and each one of the sets $\tset_1,\tset_2$ is node-well-linked in $G$. Let $\tset'_1\subseteq \tset_1$, $\tset_2'\subseteq\tset_2$, be any pair of subsets with $|\tset_1'|=|\tset_2'|\leq \frac{\alpha\kappa}{2\Delta}$. Then $(\tset'_1,\tset_2')$ are node-linked in $G$.
\end{theorem}

%---------------------------------------------------
%---------------------------------------------------
%---------------------------------------------------
%---------------------------------------------------
\label{---------------------------------------------subsec: balanced cuts-------------------------------------------}
\subsection{Balanced Cuts}

\begin{definition}
Let $G$ be a graph, and let $T\subseteq V(G)$ be a subset of its vertices. Given a parameter $0<\rho\leq 1/2$, a partition $(A,B)$ of $V(G)$ is called a \emph{$\rho$-balanced cut} of $G$ with respect to $T$, iff $|A\cap T|,|B\cap T|\geq \rho |T|$. It is called a \emph{minimum $\rho$-balanced cut} of $G$ with respect to $T$, if it minimizes $|E(A,B)|$ among all $\rho$-balanced cuts, and subject to this, minimizes $\min\set{|A\cap T|,|B\cap T|}$.
\end{definition} 

We will use the following lemma, whose proof uses standard techniques.

%---------------------------------------
%---------------------------------------
%---------------------------------------
%---------------------------------------
%---------------------------------------
%---------------------------------------
\begin{lemma}\label{lemma: balanced cut large piece wl}
Let $G=(V,E)$ be a graph, and $C\subseteq V$ any cluster that has the $\alpha$-bandwidth property, for some $0<\alpha\leq 1$. Let $(A,B)$ be the minimum $\rho$-balanced cut of $G[C]$ with respect to $\Gamma(C)$, for some $0<\rho\leq 1/4$, and assume that $|\Gamma(C)\cap A|\geq |\Gamma(C)\cap B|$. Then $A$ has the $\alpha/(2+\alpha)$-bandwidth property.
\end{lemma}

%---------------------------------------
%---------------------------------------
%---------------------------------------
%---------------------------------------
%---------------------------------------
%---------------------------------------

\begin{proof}
%Let $\Gamma\subseteq A$ be the subset of vertices incident on the edges of $E(A,B)$.
Consider any partition $(X,Y)$ of $A$. Let $S_X=\Gamma(C)\cap X$, $S_Y=\Gamma(C)\cap Y$, and $E'=E(X,Y)$. Let $U_X\subseteq X, U_Y\subseteq Y$ be the subsets of vertices incident on the edges of $E(A,B)$ (see Figure~\ref{fig: balanced-well-linked}). From Observation~\ref{obs: well-linkedness alt def}, it is enough to show that $|E'|\geq\frac{\alpha}{2+\alpha}\min\set{|X\cap \Gamma(A)|,|Y\cap \Gamma(A)|}=\frac{\alpha}{2+\alpha}\min\set{|S_X\cup U_X|,|S_Y\cup U_Y|}$. 
Our first observation is that $|E'|\geq \min\set{|E(X,B)|,E(Y,B)|}$. Indeed, assume otherwise. Assume w.l.o.g. that $|S_X|\leq |S_Y|$, and consider the cut $(B\cup X,Y)$. Then it is easy to see that this is a $\rho$-balanced cut of $G[C]$ with respect to $\Gamma(C)$, since $\rho\leq 1/4$. Moreover, $|E(B\cup X,Y)|\leq |E(A,B)|-|E(X,B)|+|E'|<|E(A,B)|$, contradicting the fact that $(A,B)$ is the minimum $\rho$-balanced cut with respect to $\Gamma(C)$. Therefore, $|E'|\geq \min\set{|E(X,B)|,E(Y,B)|}$ must hold.

We assume without loss of generality that $|E(X,B)|\leq |E(Y,B)|$, so $|E'|\geq |E(X,B)|$, and
we consider two cases.

The first case happens when $|S_X|\leq |\Gamma(C)|/2$. In this case, since $C$ has the $\alpha$-bandwidth property, $|\out(X)|=|E'|+|E(X,B)|\geq \alpha |S_X|$. Since $|E'|\geq |E(X,B)|$, we conclude that:

\[(2+\alpha)|E'|\geq |E'|+|E(X,B)|+\alpha|E(X,B)|\geq \alpha |S_X|+\alpha |E(X,B)|\geq \alpha (|S_X|+|U_X|),\]

 and so $|E'|\geq \frac{\alpha}{2+\alpha}(|S_X\cup U_X|)$.

The second case happens when $|S_X|>|\Gamma(C)|/2$. Then $(X,B\cup Y)$ is a $\rho$-balanced cut of $G[C]$ with respect to $\Gamma(C)$, and $|E(X,B\cup Y)|\leq |E(A,B)|-|E(Y,B)|+|E'|$. From the minimality of the cut $(A,B)$, $|E'|\geq |E(Y,B)|$ must hold. From the $\alpha$-well-linkedness of $\Gamma(C)$, we get that $|\out(Y)|=|E'|+|E(Y,B)|\geq \alpha |S_Y|$, and so:

\[(2+\alpha)|E'|\geq |E'|+|E(Y,B)|+\alpha |E(Y,B)|\geq \alpha |S_Y|+\alpha |E(Y,B)|\geq \alpha(|S_Y|+|U_Y|),\]

 and $|E'|\geq \frac{\alpha}{2+\alpha}(|S_X|+|U_X|)\geq \frac{\alpha}{2+\alpha}(|S_X\cup U_X|)$.

\begin{figure}
\scalebox{0.35}{\includegraphics{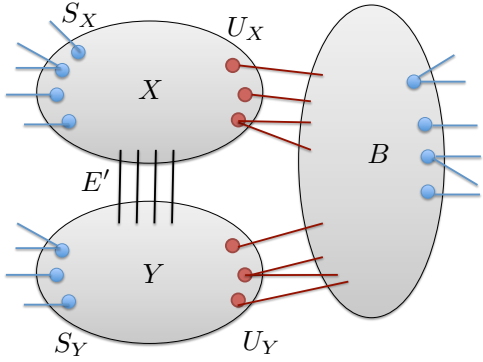}}\caption{Illustration for the proof of Lemma~\ref{lemma: balanced cut large piece wl}.}\label{fig: balanced-well-linked}
\end{figure}

\end{proof}

\label{---------------------------------------------subsec: degree reduction----------------------------}
\subsection{Treewidth and Degree Reduction}

Our proof of Theorem~\ref{thm: GMT} assumes that the maximum vertex degree of the input graph $G$ is bounded by a constant. There are several known results, that, given a graph $G$ of treewidth $k$, find a subgraph $G'$ of $G$, whose maximum vertex degree is bounded by a constant, and whose treewidth is close to $\tw(G)$. %Of course, the Excluded Grid Theorem itself provides such a graph $G'$, with maximum vertex degree at most $4$ (which can in fact be easily reduced to $3$), and the treewidth of $G'$ is $k^{\epsilon}$. However, this $\epsilon$ is too small, and using this approach here would not result in improved bounds for the Excluded Grid Theorem.
For example, Reed and Wood~\cite{ReedW-grid} have shown that any graph of treewidth $k$ contains a subgraph of maximum vertex degree at most $4$, and treewidth $\Omega(k^{1/4}/\log^{1/8}k)$. The algorithm of Chekuri and Ene~\cite{ChekuriE13} can be used to construct a subgraph $G'$ of $G$ of treewidth $k/\poly\log k$, and maximum vertex degree bounded by some constant. %This result builds on a construction of an object that is similar to (but is somewhat weaker than) a \PoS of width $\Theta(\log^2k)$ and height $k/\poly\log k$, that was introduced in~\cite{Chuzhoy11}. 
We use the following stronger result of~\cite{tw-sparsifiers}:

\begin{theorem}\cite{tw-sparsifiers}\label{thm: degree reduction}
Let $G$ be a graph of treewidth $k$. Then there is a subgraph $G'$ of $G$, whose maximum vertex degree is $3$, and $\tw(G')=\Omega(k/\poly\log k)$. Moreover, there is a set $\tset\subseteq V(G')$ of $\Omega(k/\poly\log k)$ vertices, such that $\tset$ is node-well-linked in $G'$, and each vertex of $\tset$ has degree $1$ in $G'$.
\end{theorem}
 
% \paragraph{Remark} The results of~\cite{tw-sparsifiers} only guarantee that the set $\tset\subseteq V(G')$ of vertices is $1$-well-linked in $G'$. However, since the maximum vertex degree in $G'$ is bounded by $3$, from Theorem~\ref{thm: grouping}, there is a set $\tset'\subseteq \tset$ of $\Omega(|\tset|)$ vertices that are node-well-linked in $G'$.

The starting point of the above theorem is a \PoS of length $\Omega(\poly\log k)$ and width $\Omega(k/\poly\log k)$, whose existence follows from~\cite{CC14}. %Theorem~\ref{thm: degree reduction} obtains the best bounds both with respect to the treewidth of $G'$, and the maximum vertex degree in $G'$. %The drawback is that it builds on the previous proof of the Excluded Grid Minor theorem (though it does not use the full power of the proof, as the \PoS it uses has length $\poly\log k$, as opposed to the polynomial length provided in~\cite{CC14}). Therefore, if one is interested in optimizing the bounds of Theorem..., then the best starting point is Theorem~\ref{thm: degree reduction}, but if one is interested in a simple and self-contained proof of Theorem..., then one should probably use the results of~\cite{ReedW-grid}. 
We chose to use Theorem~\ref{thm: degree reduction} as our starting point, since it provides the best parameters, and, due to Observation~\ref{obs: EDP to NDP in degree-3}, degree-$3$ graphs are somewhat easier to work with. But our proof can work as well using the result of~\cite{ReedW-grid} as a starting point instead.

\label{---------------------------------------------subsec: PoS----------------------------}
%---------------------------------------------------
%---------------------------------------------------

\subsection{A Path-of-Sets System and the Grid Minor}
A central combinatorial object that we use  is the \PoS, introduced in~\cite{CC14}. A closely related object, called a grill, was previously defined by Leaf and Seymour~\cite{LeafS12}.

%A \PoS of width $r$ and height $h$ is simply an $h$-wide embedding of a path $P$ of length $r$ into $G$. We denote the vertices of $P$ by $v_1,\ldots,v_r$, and assume that they appear on $P$ in this order. For simplicity, for each $1\leq i\leq r$, $S_{v_i}$ is denoted by $S_i$. The set $\pset(e)$ of paths corresponding to the edge $e=(v_i,v_{i+1})$ is denoted by $\pset_i$. The set of vertices of $S_i$ where the paths of $\pset_i$ originate is denoted by $A_i$, and the set of vertices of $S_{i+1}$ where the paths of $\pset_i$ terminate are denoted by $B_i$. We also assume that we are additionally given a set $A_1$ of $h$ vertices in $S_1\setminus B_1$, and a set $B_h$ of $h$ vertices in $S_1\setminus B_r$.  (see Figure~\ref{fig:pos}).

\begin{definition}
A \PoS $(\sset,\bigcup_{i=1}^{\ell-1}\pset_i)$ of width $w$ and length $\ell$ in a graph $G$ consists of:
\begin{itemize}
 
\item A sequence $\sset=(S_1,\ldots,S_{\ell})$ of $\ell$ disjoint vertex subsets of $G$, where for each $i$, $G[S_i]$ is connected;

\item For each $1\leq i\leq \ell$, two disjoint sets $A_i,B_i\subseteq S_i$ of $w$ vertices each; the vertices of $A_1\cup B_{\ell}$ must have degree at most $2$ in $G$; and

\item For each $1\leq i<\ell$, a set $\pset_i:B_i\sconnect A_{i+1}$ of $w$ paths, such that all paths in $\bigcup_i\pset_i$ are mutually node-disjoint, and do not contain the vertices of $\bigcup_{S_j\in \sset}S_j$ as inner vertices (see Figure~\ref{fig:pos}).
\end{itemize}

 We refer to vertex sets $A_1$ and $B_{\ell}$ as the \emph{anchors} of the \PoS.
 
We say that it is an \emph{$\alpha$-weak} \PoS, if for all $1\leq i\leq \ell$, $A_i\cup B_i$ is $\alpha$-well-linked in $G[S_i]$; we say that it is a \emph{good} \PoS, if for all $1\leq i\leq \ell$, $B_i$ is $1$-well-linked in $G[S_i]$, and $(A_i,B_i)$ are $\half$-linked in $G[S_i]$. Finally, we say that it is a \emph{perfect} \PoS, if for each $1\leq i\leq \ell$, $A_i$ is node-well-linked in $G[S_i]$, $B_i$ is node-well-linked in $G[S_i]$, and $(A_i,B_i)$ are node-linked in $G[S_i]$.
\end{definition}

%We say that it is an \emph{$\alpha$-weak} \PoS, if for all $1\leq i\leq r$, $A_i\cup B_i$ is $\alpha$-well-linked in $G[S_i]$; we say that it is a \emph{good} \PoS, if for all $1\leq i\leq r$, $B_i$ is $1$-well-linked in $G[S_i]$, and $(A_i,B_i)$ are $\half$-linked in $G[S_i]$. Finally, we say that it is a \emph{perfect} \PoS, if for each $1\leq i\leq r$, $A_i$ is node-well-linked in $G[S_i]$, $B_i$ is node-well-linked in $G[S_i]$, and $(A_i,B_i)$ are $1$-linked in $G[S_i]$.

\iffalse
%-------------------------------------
%-------------------------------------
%-------------------------------------
%-------------------------------------
The following theorem was proved in~\cite{CC14}.

\begin{theorem}[Theorem 3.2 in~\cite{CC14}] 
\label{thm: strong PoS system}
Let $G$ be any graph of treewidth $k$, and let $h,r>1$ be integral
parameters, such that for some large enough constants $c$ and $c'$,
$k/\log^{c'}k>chr^{48}$. Then there is an efficient randomized
algorithm, that, given $G,h$ and $r$, w.h.p. computes a strong
\PoS of height $h$ and width $r$ in $G$.
\end{theorem}

%-------------------------------------
%-------------------------------------
%-------------------------------------
%-------------------------------------
\fi

The following theorem allows us to turn an $\alpha$-weak \PoS into a good one, and eventually into a perfect one, with only a small loss in the system's width. The proof appears in Appendix. A simpler proof, with somewhat weaker parameters, follows easily from Theorems~\ref{thm: grouping} and \ref{thm: linkedness from node-well-linkedness}.

\begin{theorem}\label{thm from weak to perfect PoS}
Let $G$ be a graph with maximum vertex degree at most $3$, and suppose we are given an $\alpha$-weak \PoS of width $w$ and length $\ell$ in $G$, where $0<\alpha\leq 1$, and $1/\alpha$ is an integer. Then $G$ contains a good \PoS of width $\ceil{\alpha w/4}$ and length $\ell$, and it contains a perfect \PoS of width at least $\alpha w/c^*$, for some constant $c^*$, and length $\ell$. Moreover, if we denote by $A_1,B_{\ell}$ the anchors of the original \PoS, and by $A_1',B_{\ell}'$ the anchors of the new \PoS, then in each case $A_1'\subseteq A_1$ and $B_{\ell}'\subseteq B_{\ell}$.
\end{theorem}

The following theorem was implicitly proved in~\cite{CC14}. We include its proof in Appendix for completeness.

\begin{theorem}\label{thm: find grid minor or good linkage}
Let $G=(V,E)$ be a connected graph, and let $A,B\subseteq V$ be two disjoint subsets of vertices of $G$, with $|A|=|B|=w$, such that $(A,B)$ are node-linked in $G$. Then for all integers $h_1,h_2>1$ with $(16h_1+10)h_2\leq w$, either $G$ contains the $(h_1\times h_1)$-grid minor, or there is a collection $\pset$ of $h_2$ node-disjoint paths, connecting vertices of $A$ to vertices of $B$, such that for every pair $P,P'\in \pset$ of paths with $P\neq P'$, there is a path $\beta_{P,P'}\subseteq G$, connecting a vertex of $P$ to a vertex of $P'$, where $\beta_{P,P'}$ is internally disjoint from $\bigcup_{P''\in \pset}V(P'')$.
\end{theorem}

Given a \PoS $(\sset,\bigcup_{i=1}^{\ell-1}\pset_i)$ in $G$, we say that $G'$ is a subgraph of $G$ spanned by the \PoS, if $G'$ is the union of $G[S_i]$ for all $1\leq i\leq \ell$ and all paths in $\bigcup_{i=1}^{\ell-1}\pset_i$. 

The following corollary of Theorem~\ref{thm: find grid minor or good linkage} allows us to obtain a grid minor from a Path-of-Sets system.  Its proof has implicitly appeared in~\cite{CC14}, and is included in Appendix for completeness.

\begin{corollary}\label{cor: paths from the path-set system}\cite{CC14}
Let $G$ be any graph, and let $(\sset,\bigcup_{i=1}^{\ell-1}\pset_i)$ be a perfect \PoS of length $\ell\geq 2$ and width $w$ in $G$. Let $h_1,h_2$ be integers with  $(16h_1+10)h_2\leq w$, and let $G'$ be a subgraph of $G$ spanned by the \PoS. Then either $G'$ contains the $(h_1\times h_1)$-grid minor, or there is a collection $\qset$ of $h_2$ node-disjoint paths in $G'$, connecting vertices of $A_1$ to vertices of $B_{\ell}$, such that for all $1\leq i\leq \ell$, for every path $Q\in \qset$, $S_i\cap Q$ is a path, and moreover, for every $1\leq j\leq \floor{\ell/2}$, for every pair $Q,Q'\in \qset$ of paths, there is a path $\beta_{2i}(Q,Q')$ contained in $G'[S_{2i}]$, connecting a vertex of $Q$ to a vertex of $Q'$, such that $\beta_{2i}(Q,Q')$ is internally disjoint from all paths in $\qset$.
\end{corollary}

The following corollary, that was also proved in~\cite{CC14}, completes the construction of the grid minor. We include its proof for completeness in Appendix.
The corollary slightly improves upon a similar result of~\cite{LeafS12}.
\begin{corollary}\label{cor: from path-set system to grid minor}\cite{CC14}
  Let $G$ be any graph, $g>1$ an integer, and let $\left(\sset,\bigcup_{i=1}^{\ell-1}\pset_i\right)$ be a perfect \PoS of width  $w=16g^2+10g$ and length $\ell=2g(g-1)$ in $G$. Then $G$ contains the $\left (g\times g\right )$-grid as a minor.
\end{corollary}

\label{-------------------------------------------sec: building PoS--------------------------------}
%----------------------------------------------------------------------------------------------
%----------------------------------------------------------------------------------------------
%----------------------------------------------------------------------------------------------
\section{Constructing a Path-of-Sets System}\label{sec: building PoS}

The next theorem is central to our construction of the \PoS.

\begin{theorem}\label{thm: phase execution}
Suppose we are given a graph $G$ with maximum vertex degree $3$, and a good \PoS $\left (\sset=(S_1,\ldots,S_{\ell}),\bigcup_{i=1}^{\ell-1}\pset_i\right )$ of length $\ell$ and width $w$, where $w\geq 12000$ is an integral power of $2$. Let $A_1\subseteq S_1$, $B_{\ell}\subseteq S_{\ell}$ denote the anchors of the \PoS.
Then there is a good \PoS $\left (\sset'=(S'_1,\ldots,S'_{2\ell}),\bigcup_{i=1}^{2\ell-1}\pset'_i\right)$ of length $2\ell$ and width $w/2^{17}$ in $G$.
Moreover, if $A'_1\subseteq S_1'$, $B_{2\ell}'\subseteq S'_{2\ell}$ denote the anchors of this new \PoS, then $A'_1\subseteq A_1$ and $B'_{2\ell}\subseteq B_{\ell}$.
\end{theorem}

We prove Theorem~\ref{thm: phase execution} below, after we discuss some of its consequences here. First, we obtain the following immediate corollary.

\begin{corollary}\label{cor: PoS parameters}
Let $G$ be any graph with maximum vertex degree $3$, and $\tset_1,\tset_2\subseteq V(G)$ any two disjoint subset of vertices of $G$ of cardinality $k$ each, such that $(\tset_1, \tset_2)$ are node-linked in $G$, each of $\tset_1,\tset_2$ is node-well-linked in $G$, and the degree of every vertex in $T_1\cup T_2$ is at most $2$ in $G$. Let $w,\ell>1$ be integers, where $\ell$ is an integral power of $2$, and assume that for some large enough constant $c_p$, $k\geq c_pw\ell^{17}$. Then there is a perfect \PoS  $\left (\sset=(S_1,\ldots,S_{\ell}),\bigcup_{i=1}^{\ell-1}\pset_i\right )$ of length $\ell$ and width $w$ in $G$. Moreover, if $A_1\subseteq S_1$, $B_{\ell}\subseteq S_{\ell}$ are the anchors of this \PoS, then $A_1\subseteq \tset_1$ and $B_{\ell}\subseteq \tset_2$.
\end{corollary}

\begin{proof}
Let $k'$ be the largest integral power of $2$ with $k'\leq k/2$, and let $\hat A_1\subseteq \tset_1,\hat B_1\subseteq \tset_2$ be arbitrary disjoint subsets of vertices of cardinality $k'$ each. Let $S_1=V(G)$. Then $S_1$, together with $\hat A_1$ and $\hat B_1$ playing the role of the sets $A_1$ and $B_1$ define a good \PoS of width $k'$ and length $1$.

For $0\leq j\leq \log \ell$, let $\ell_j=2^j$ and $w_j=\frac{k'}{2^{17j}}$. We assume that $c_P$ is large enough, so $w_{\log \ell}\geq 12000$. We perform $\log_2 \ell$ phases. The input to phase $j$, for $1\leq j\leq \log\ell$, is a good \PoS of length $\ell_{j-1}$ and width $w_{j-1}$, and the output is a good \PoS of length $\ell_j$ and width $w_j$.  The anchors of each \PoS are contained in $\hat A_1$ and $\hat B_1$ respectively. The input to the first phase is the \PoS length $1$ and width $k'$, constructed above. 
Every phase is executed by applying Theorem~\ref{thm: phase execution} to the current \PoS, and using its output as the next \PoS.
Clearly, after $\log_2\ell$ iterations, we will obtain a good \PoS of length $\ell$ and width $\frac{k'}{2^{17\log_2\ell}}=\frac{k'}{\ell^{17}}\geq c_pw/4$. 
Let $A'_1$, $B'_{\ell}$ be the anchors of the resulting \PoS. Then it is easy to verify that $A'_1\subseteq \hat A_1$ and $B'_{\ell}\subseteq \hat B_1$.

Our final step is to apply Theorem~\ref{thm from weak to perfect PoS} to this last \PoS, in order to obtain a perfect \PoS of width $w$ and length $\ell$. If we denote by $A''_1$ and $B''_{\ell}$ the anchors of this final \PoS, then we are guaranteed that $A''_1\subseteq A_1'\subseteq \hat A_1\subseteq \tset_1$, and $B''_{\ell}\subseteq B_{\ell}'\subseteq \hat B_1\subseteq \tset_2$.
\end{proof}

We are now ready to complete the proof of Theorem~\ref{thm: GMT} with the weaker bound of $f(g)=O(g^{36\poly\log g})$, using Corollary~\ref{cor: PoS parameters}.  Let $G$ be any graph of treewidth $\kappa=\Omega(g^{36}\poly\log g)$.  We use Theorem~\ref{thm: degree reduction} to obtain a subgraph $G'$ of $G$, whose maximum vertex degree is $3$, together with a set $\tset$ of $\kappa^*=\Omega(\kappa/\poly\log \kappa)$ terminals, such that the terminals of $\tset$ are node-well-linked in $G'$, and the degree of every terminal is $1$. We assume that $\kappa^*\geq 2^{80}c_pg^{36}$, where $c_P$ is the constant from  Corollary~\ref{cor: PoS parameters}. From Corollary~\ref{cor: PoS parameters}, there is a perfect \PoS of width $16g^2+10g$ and length $2g(g-1)$ in $G'$, and from Corollary~\ref{cor: from path-set system to grid minor}, $G'$ (and hence also $G$) contains the $(g\times g)$-grid as a minor.

We now focus on proving Theorem~\ref{thm: phase execution}. The central combinatorial object that we use is a two-cluster chain (that can intuitively be thought of as a \PoS of width $2$, except that the sizes of $A_1,B_1,A_2,B_2$ are no longer uniform).

\begin{definition} 
Let $G$ be a graph, $\tset_1,\tset_2$ two disjoint sets of vertices,  with $|\tset_1|=k$ and $|\tset_2|=k'=k/64$, where $k\geq 12000$ is an integral power of $2$. A $2$-cluster chain $(X,Y,\ttset_1,\ttset_2,E')$ consists of:

\begin{itemize}
\item two disjoint clusters $X,Y\subseteq V(G)$;

\item a subset $\ttset_1\subseteq \tset_1\cap X$, with $|\ttset_1|=k'$, and a subset $\ttset_2\subseteq\tset_2\cap Y$, with $|\ttset_2|=k/512$;

\item a set $E'\subseteq E(X,Y)$ of $k/512$ edges, whose endpoints are all distinct;

Let $\U_X\subseteq X$ be the subset of vertices of $X$ incident on the edges of $E'$, and let $\U_Y\subseteq Y$ be the subset of vertices of $Y$ incident on the edges of $E'$. Then:

\item $\ttset_1\cup \U_X$ is $(k/512,\alpha^*)$-well-linked in $G[X]$ and $\ttset_2\cup \U_Y$ is $(k/512,\alpha^*)$-well-linked in $G[Y]$, for $\alpha^*=1/64$. (See Figure~\ref{fig: strong chain}).
\end{itemize}
\end{definition}

%\begin{figure}[h]
%\scalebox{0.4}{\includegraphics{two-cluster-chain-cut.pdf}}\caption{A 2-cluster chain. Here, $|\ttset_1|=k/128$, $|\ttset_2|=\floor{k'/11}$, $|E'|=k/128$; set $\ttset_1\cup\U_X$ is $\alpha'$-well-linked in $G[X]$, and $\ttset_2\cup \U_Y$ is $\alpha'$-well-linked in $G[Y]$.\label{fig: 2-cluster-chain}}
%\end{figure}
The main technical ingredient of the proof is the following theorem, that is proved in Section~\ref{sec: splitting a cluster}.

\begin{theorem}\label{thm: main advanced splitting}
Suppose we are given a graph $G$, with maximum vertex degree at most $3$, and two disjoint subsets of vertices, $T_1$ of size $k$ (where $k\geq 12000$ is an integral power of $2$), and $T_2$ of size $k'=k/64$, such that the degree of every vertex in $T_1\cup T_2$ is $1$ in $G$, the vertices of $T_1$ are $1$-well-linked, and $(T_1,T_2)$ are $\half$-linked in $G$. Then there is a $2$-cluster chain in $G$.
\end{theorem}

We are now ready to complete the proof of Theorem~\ref{thm: phase execution}. Let $\pset=\bigcup_{i=1}^{\ell-1}\pset_i$. For convenience, for each path $P\in \pset$, we delete all edges and inner vertices of $P$ from the graph, and instead add a new vertex $t_P$, that connects to the two endpoints $u,v$ of $P$. Let $P'=(u,t_P,v)$ be the resulting path. We also add a new set $Z_0'$ of $w/64$ new vertices, each of which connects to a distinct vertex of $A_1$, and a new set $Z_{\ell}$ of $w$ vertices, each of which connects to a distinct vertex of $B_{\ell}$.
We denote the resulting graph by $G'$. 

For each $1\leq i\leq \ell-1$, let $Z_i=\set{t_P\mid P\in \pset_i}$. We perform $\ell$ iterations, where the $i$th iteration splits cluster $S_i$. We assume that for each $1\leq i\leq \ell$, when iteration $i$ starts, we are given a subset $Z'_{i-1}\subseteq Z_{i-1}$ of $w/64$ vertices (where at the beginning of the first iteration we use the set $Z_0'$ we have just defined). In the $i$th iteration, we apply Theorem~\ref{thm: main advanced splitting} to graph $G_i=G'[S_i\cup Z'_{i-1}\cup Z_i]$, with $T_1=Z_i$, and $T_2=Z'_{i-1}$. Since the \PoS is good, it is easy to see that $T_1$ is $1$-well-linked in $G_i$,  $(T_1,T_2)$ are $\half$-linked, and from our definition of the graph $G'$, all vertices of $T_1\cup T_2$ have degree $1$ in $G_i$. Moreover, it is easy to verify that all vertex degrees in $G_i$ are bounded by $3$. Let $(X_i,Y_i)$ be the resulting pair of clusters, $E_i=E'$ the corresponding set of edges, and $Z_{i-1}''=\ttt2$, $Z_{i}'=\ttt1$ the corresponding vertex subsets. We then continue to the next iteration. Consider the final collection $(Y_1,X_1,\ldots,Y_{\ell},X_{\ell})$ of clusters obtained after $\ell$ iterations. Then for each $0\leq i\leq \ell$, $|Z_i''|=w/512$. We build an $\alpha^*$-weak path-of sets system $(\sset',\bigcup_{i=1}^{2\ell-1}\pset_i')$ in the original graph $G$, as follows. Let $Z=Z'_0\cup (\bigcup_{i=1}^{\ell}Z_i)$.
We let $\sset'=(S_1',\ldots,S_{2\ell}')$, where for $1\leq i\leq \ell$, $S_{2i-1}'=Y_i\setminus Z$ and $S_{2i}'=X_i\setminus Z$. For $1\leq i\leq \ell$, let $\pset'_{2i-1}$ be any subset of $w/512$ edges of $E_i$. For $1\leq i<\ell$, we let $\pset'_{2i}$ be the set of paths $P\in \pset_i$, where $t_P\in Z''_i$. We let $A_1'\subseteq A_1$ contain all vertices that are neighbors of the vertices of $Z_0''$ in $G'$, and similarly $B_{2\ell}'\subseteq B_{\ell}$ is any subset of $w/512$ neighbors of the vertices of $Z_{\ell}'$. For $1\leq i<2\ell$, we let $B_i'$ be the set of the endpoints of the paths in $\pset_i'$ that lie in $S'_i$, and we let $A_{i+1}'$ be the set of their endpoints that lie in $S'_{i+1}$. 

Fix some $1\leq i\leq \ell$, and recall that $(Z''_{i-1},B'_{2i-1})$ were $\alpha^*$-well-linked in $G[Y_i]=G[S_{2i-1}']$. Since the vertices of $Z\cap V(G_i)$ have degree $1$ each in $G_{i}$, and since from our definition $A_{2i-1}$ is the set of the neighbors of the vertices in $Z''_{i-1}$ in graph $G_i$, it is easy to verify that $A_{2i-1}\cup B_{2i-1}$ is $(w/512,\alpha^*)$-well-linked in $G[S_{2i-1}']$. Since $|A_{2i-1}\cup B_{2i-1}|=2\cdot w/512$, set $A_{2i-1}\cup B_{2i-1}$ is $\alpha^*$-well-linked in $G[S_{2i-1}']$.  Using a similar reasoning, $A_{2i}\cup B_{2i}$ is $\alpha^*$-well-linked in $G[S_{2i}]$. Notice that $A_1'\subseteq A_1$ and $B_{2\ell}'\subseteq B_{2\ell}$.
So far, we have obtained an $\alpha^*$-weak \PoS of length $2\ell$ and width $w/512$, where $\alpha^*=1/64$. From Theorem~\ref{thm from weak to perfect PoS}, we can now obtain a good \PoS of length $2\ell$ and width at least $\frac{w}{512}\cdot \frac{\alpha^*}{4}= \frac{w}{2^{17}}$.
Moreover, if $A_1'',B_{2\ell}''$ are the anchors of this final \PoS, then $A_1''\subseteq A_1'\subseteq A_1$, and $B_{2\ell}''\subseteq B_{2\ell}\subseteq B_{\ell}$.

%----------------------------------------------------------------------------------------------
%------%--------------------------------------------------------------------------------------------------------------------------------------------
%-------------------------------------------------------------------------------------------
%-------------------------------------------------------------------------------------------
%-------------------------------------------------------------------------------------------
\label{-------------------------------------------sec: Splitting a cluster--------------------------------}
%----------------------------------------------------------------------------------------------
%----------------------------------------------------------------------------------------------
%----------------------------------------------------------------------------------------------
\section{Splitting a Cluster}\label{sec: splitting a cluster}
%----------------------------------------------------------------------------------------------
%----------------------------------------------------------------------------------------------
%----------------------------------------------------------------------------------------------
The goal of this section is to prove Theorem~\ref{thm: main advanced splitting}. We denote $T=T_1\cup T_2$, and we call the vertices of $T$ \emph{terminals}. Recall that $|T_1|=k$, $|T_2|=k'=k/64$, and we denote $k''=k/512$.

 Let $G^*$ be a minimal (with respect to edge- and vertex-deletion) subgraph of $G$, in which $T_1$ is $(k/4,1)$-well-linked, and $(T_1,T_2)$ are $\half$-linked. For each terminal $t\in T$, we subdivide the unique edge incident on $t$ by a new vertex $v_t$. It is easy to see that a $2$-cluster chain in $G^*$ immediately defines a $2$-cluster chain in $G$. 
From now on we will be working with graph $G^*$, and for simplicity of notation, we denote $G^*$ by $G$.
Our goal is to show that $G$ contains a $2$-cluster chain. Notice that from the minimality of $G$, it is a connected graph.

Given a cluster $C\subseteq V(G)\setminus T$, we denote by $\pset(C)$ the maximum-cardinality set of node-disjoint paths connecting the terminals in $T$ to $\Gamma(C)$, and we denote $p(C)=|\pset(C)|$. We assume w.l.o.g. that the paths in $\pset(C)$ are internally disjoint from $C\cup T$. The following lemma can be seen as  a variation of the Deletable Edge Lemma of~ Chekuri, Khanna and Shepherd~\cite{deletable-edge-original} (the proof of their original lemma can be found in~\cite{deletable-edge}), though it is somewhat simpler. \ifabstract The proof is omitted from this extended abstract. \fi

%Need to ensure that a set of paths connecting $T_2$ to $T_1$ remains. We'll assume that $C\cap T_2=\emptyset$, and every vertex in $\Gamma(C)$ has degree $3$ (otherwise can pull them out without changing well-linkedness). Define $Y'$ exactly as before, and consider the set $\qset$ of paths from $T_2$ to $T_1$. Then $Y'$ contains at least $k'+1$ vertices of $\Gamma(C)$, each has degree $3$. The paths in $\qset$ can use at most two edges per each such vertex, so at least $k'+1$ edges incident on vertices of $\Gamma(C)\cap Y'$ remain unused, and at most $k'$ of them belong to $\out(Y')$, so we can take one edge with both endpoints in $Y'$.

\begin{lemma}\label{lem: deletable edge}
Let $C\subseteq V(G)\setminus T$ be a cluster of $G$, such that $C$ has the $(k/4,1)$-bandwidth property. Then $p(C)=|\Gamma(C)|$.\end{lemma}

\iffull
\begin{proof}
Let $r=p(C)$.
From Menger's theorem, there is a set $R$ of $r$ vertices, separating $T$ from $\Gamma(C)$ in $G$. Since the degree of every terminal is $1$, we can assume without loss of generality that $R\cap T=\emptyset$.
Each path in $\pset(C)$ contains exactly one vertex of $R$, and, since the paths in $\pset(C)$ are internally disjoint from $C$, $R\cap (C\setminus \Gamma(C))=\emptyset$. Let $U$ be the union of all connected components of $G\setminus R$ containing the vertices of $C$. Then $C\subseteq U\cup R$, and $T\cap U=\emptyset$. Let $G'=G[U\cup R]$. 

Notice that the paths in $\pset(C)$ define a collection $\pset': R\connect \Gamma(C)$ of node-disjoint paths, that are internally disjoint from $C$, and are contained in $G'$. Let $G''=G[C]\cup \pset'$. Then, since $\Gamma(C)$ is $(k/4,1)$-well-linked in $G[C]$, and $G[C]\subseteq G''$, while the paths in $\pset'$ are node-disjoint, is it is easy to see that  $R$ is $(k/4,1)$-well-linked in $G''$.

Let $R'\subseteq \Gamma(C)$ be the set of vertices where the paths of $\pset(C)$ terminate. We claim that $R'=\Gamma(C)$. Assume otherwise. Then there is some edge $e\in \out(C)$, that lies in $G'$, and does not belong to any path in $\pset(C)$. Notice that $e$ does not belong to $G''$, and so $R$ remains $(k/4,1)$-well-linked in $G'\setminus\set{e}$. The following claim will then finish the proof.

\begin{claim}\label{claim: wl}
Set $T_1$ is $(k/4,1)$-well-linked in $G\setminus\set{e}$, and $(T_1,T_2)$ are \half-linked in $G\setminus \set{e}$.
\end{claim}

From the above claim, and the minimality of $G$, we conclude that $R'=\Gamma(C)$, and $|\Gamma(C)|=p(C)$. It now remains to prove Claim~\ref{claim: wl}.

\begin{proof}
We first prove that $T_1$ is  $(k/4,1)$-well-linked in $G\setminus\set{e}$.
Consider any pair $T',T''\subseteq T_1$ of disjoint equal-sized subsets of $T_1$, with $|T'|=|T''|\leq k/4$. Since $T_1$ is $(k/4,1)$-well-linked in $G$, there is a set $\qset: T'\sconnect T''$ of edge-disjoint paths in $G$. We view the paths in $\qset$ as directed from $T'$ to $T''$.
 From Observation~\ref{obs: EDP to NDP in degree-3}, and since all terminals have degree $1$, the paths in $\qset$ are node-disjoint. We partition $\qset$ into two subsets: $\qset_1$ contains all paths that do not contain the edges of $G'$, and $\qset_2$ contains all remaining paths. We now define a pair $R_1,R_2$ of disjoint equal-sized subsets of $R$, and two new sets $\qset_2',\qset_2''$ of paths as follows: for each path $Q\in \qset_2$, consider the first edge $(v,v')\in E(G')$ lying on $Q$, where $v$ appears before $v'$ on $Q$. We add $v$ to $R_1$, and we let $Q_1$ be the sub-path of $Q$ between its first vertex and $v$. Similarly, consider the last edge $(u',u)\in E(G')$ lying on $Q$, where $u'$ appears before $u$ on $Q$. We add $u$ to $R_2$, and we let $Q_2$ be the sub-path of $Q$ between $u$ and its last vertex.  (Notice that $u\neq v$, and $u,v\in R$).
Let $\qset_2'=\set{Q_1\mid Q\in \qset_2}$ and $\qset_2''=\set{Q_2\mid Q\in \qset_2}$. Observe that the paths in $\qset_1\cup \qset_2'\cup \qset_2''$ do not use the edges of $G'$. Notice also that $|R_1|=|R_2|=|\qset_2|\leq k/4$. Since $R$ is $(k/4,1)$-well-linked in $G'\setminus \set{e}$, there is a set $\rset: R_1\sconnect_1R_2$ of edge-disjoint paths in $G'\setminus\set{ e}$. By concatenating the paths in $\qset_2',\rset$, and $\qset_2''$, and taking the union with the paths in $\qset_1$, we obtain a collection of edge-disjoint paths, connecting every vertex in $T'$ to a distinct vertex in $T''$, in graph $G\setminus \set{e}$.

We now prove that $(T_1,T_2)$ are $\half$-linked in $G\setminus\set{e}$.
Consider any pair $T'\subseteq T_1,T''\subseteq T_2$ of equal-sized subsets. Since $(T_1,T_2)$ are $\half$-linked in $G$, there is a set $\qset: T'\sconnect_2 T''$ of paths in $G$. We view the paths in $\qset$ as directed from $T'$ to $T''$.

We partition $\qset$ into two subsets: $\qset_1$ contains all paths that do not contain the edges of $G'$, and $\qset_2$ contains all remaining paths. For each path $Q\in \qset_2$, we define two sub-paths $Q_1,Q_2$ of $Q$ as follows. Let $e$ be the first vertex of $G'$ on $Q$, and let $e'$ be the last edge of $G'$ on $Q$. Then $Q_1$ is the sub-path of $Q$ from its start vertex to an endpoint of $e$ (excluding $e$), and $Q_2$ is the sub-path of $Q$ from an endpoint of $e'$ to the last vertex of $Q$ (excluding $e'$). Let $\qset'_2=\set{Q_1\mid Q\in \qset_2}$, and $\qset''_2=\set{Q_2\mid Q\in \qset_2}$, so $|\qset'_2|=|\qset_2''|=|\qset_2|$. Since the degree of every vertex in $G$ is at most $3$, it is easy to see that every vertex in $R$ may serve as an endpoint of at most two paths in $\qset_2'\cup \qset_2''$. Indeed, consider any vertex $v\in R$. At least one edge $e$ incident on $v$ must lie in $G'$, and at least one edge $e'$ incident on $v$ does not belong to $G'$. There is at most one additional edge incident on $v$, that we denote by $e''$. If $e''\in E(G')$, then at most two paths of $\qset_2$ may contain the edge $e'$, and each of these paths may contribute at most one path to $\qset_2'\cup \qset_2''$ that has $v$ as its endpoint. No other paths in $\qset_2'\cup \qset_2''$ may have $v$ as their endpoint. Otherwise, if $e''\not\in E(G')$, then at most two paths of $\qset_2$ may contain the edge $e$, and each of these paths may contribute at most one path  to $\qset_2'\cup \qset_2''$ that has $v$ as its endpoint. No other paths in $\qset_2'\cup \qset_2''$ may have $v$ as their endpoint. In either case, $v$ may serve as an endpoint of at most two paths in $\qset_2'\cup \qset_2''$. We view each path of $\qset_2'\cup \qset_2''$ as directed toward its endpoint that lies in $R$.

We now define a partition of $R$ into five subsets, $R_0,R_1,R_2,R_3,R_4$: Set $R_0$ contains all vertices $v\in R$, such that either no path of $\qset_2'\cup \qset_2''$ terminates at $v$, or exactly one path of $\qset_2'$ and exactly one path of $\qset_2''$ terminate at $v$. Set $R_1$ contains all vertices $v$, such that exactly two paths of $\qset_2'$ terminate at $v$, and set $R_2$ contains all vertices $v$, such that exactly one path of $\qset_2'$, and no paths of $\qset_2''$ terminate at $v$. Similarly, $R_3$ contains all vertices $v\in R$, such that exactly two paths of $\qset_2''$ terminate at $v$, and $R_4$ contains all vertices $v\in R$, such that exactly one path of $\qset_2''$, and no path of $\qset_2'$ terminate at $v$. Since $|\qset_2'|=|\qset_2''|$, it is easy to see that $2|R_1|+|R_2|=2|R_3|+|R_4|$. Our goal is to construct a set $\rset$ of paths in $G'\setminus \set{e}$, connecting the vertices of $R_1\cup R_2$ to the vertices of $R_3\cup R_4$, such that
each vertex in $R_1$ has exactly two paths originating from it, and each vertex in $R_2$ has exactly  one path originating from it; similarly, each vertex in $R_3$ has  exactly two paths terminating at it, and each vertex in $R_4$ has exactly one path terminating at it. Moreover, we will ensure that every edge of $G'$ participates in at most two paths in $\rset$. It is then easy to see that the union of the paths in $\qset_1,\qset_2',\qset_2''$ and $\rset$ gives the desired set $\qset': T'\sconnect_2 T''$ of paths in $G\setminus \set{e}$.

We construct the set $\rset$ of paths in two steps. Start with $A=R_1$ and $B=R_3$. If $|A|>|B|$, add $|A|-|B|$ vertices of $R_4$ to $B$; otherwise add $|B|-|A|$ vertices of $R_2$ to $A$ (we can do so since $2|R_1|+|R_2|=2|R_3|+|R_4|$). Since set $R$ is $(k/4,1)$-well-linked in $G'\setminus\set{e}$, and $|A|=|B|\leq |T_2|<k/4$, there is a set $\rset_1:A\sconnect_1B$ of paths in $G'\setminus \set{e}$.

Let $A'=R_1\cup (R_2\setminus A)$, and $B'=R_3\cup(R_4\setminus B)$. It is easy to see that $|A'|=|B'|$. Since set $R$ is $(k/4,1)$-well-linked in $G'\setminus\set{e}$, and $|A'|=|B'|\leq |T_2|<k/4$, there is a set $\rset_2:A'\sconnect_1B'$ of paths in $G'\setminus \set{e}$. We then set $\rset=\rset_1\cup \rset_2$. Combining the paths in $\qset_1,\qset_2',\qset_2''$ and $\rset$, we obtain a collection $\qset':T'\sconnect_2 T''$ of paths in $G\setminus\set{e}$.
%
\iffalse
Let $(X,Y)$ be a partition of $V(G)$ with $T\subseteq X$, $C\subseteq Y$, minimizing $|E(X,Y)|$. From the duality of cuts and flows, $|E(X,Y)|\geq r$ must hold, and there is a set $\pset$ of $|E(X,Y)|$ edge-disjoint paths connecting the vertices of $T$ to the vertices of $\Gamma(C)$ in $\G$, where the paths of $\pset$ are internally disjoint from $C$. From Observation~\ref{obs: EDP to NDP in degree-3}, and since the degrees of all terminals are $1$, and $C$ has the boundary property, the paths in $\pset$ are also node-disjoint, so $|\pset|=|E(X,Y)|=r$. Each edge of $E(X,Y)$ must lie on exactly one path of $\pset$, and so the endpoints of the edges in $E(X,Y)$ are completely disjoint. Let $R\subseteq X$ be the set of vertices that serve as endpoints of the edges of $E(X,Y)$, and let $G'=G[Y]\cup E(X,Y)$. 
\fi
\end{proof}
\end{proof}
\fi

\iffull
We obtain the following immediate corollary of Lemma~\ref{lem: deletable edge}, since $p(C)\leq |T|=k+k'$ for any cluster $C\subseteq V(G)\setminus T$.

\begin{corollary}\label{cor: deletable edge}
Let $C\subseteq V(G)\setminus T$ be any cluster of $G$ containing non-terminal vertices only, that has the $(k/4,1)$-bandwidth property. Then $|\Gamma(C)|\leq k+k'$.
\end{corollary}
\fi

\iffalse
\begin{proof}
Assume otherwise. Notice that the largest number of node-disjoint paths connecting the terminals of $T$ to the vertices of $\Gamma(C)$ is bounded by $r\leq |T|=k+k'$. Since we assumed that $\kappa>k+k'$, $\kappa/3>k/3+k'/3>k/4$. Therefore, $\Gamma(C)$ is $(k/4,1)$-well-linked, and hence $(z,1)$-well-linked in $G^*[C]$, for $z=\min\set{k/4,\kappa/2}$. From Lemma~\ref{lem: deletable edge}, $|\Gamma(C)|\leq r\leq k+k'$, a contradiction.
\end{proof}
\fi

%-------------------------------------------------------
%---------------------------------------------------------
\label{-------------------------------------weak 2-cluster chain-------------------------}
\subsection{Weak $2$-Cluster Chain}\label{subsec: weak 2-cluster chain}
In this section we define a weak $2$-cluster chain, which is somewhat weaker than the $2$-cluster chain defined in Section~\ref{sec: building PoS}. We then show that if $G$ contains a weak $2$-cluster chain, then it must contain a (strong) $2$-cluster chain.
%-------------------------------------------------------
%---------------------------------------------------------
\begin{definition}  A weak $2$-cluster chain consists of two disjoint clusters $X'$ and $Y'$, and a set $\pset=\pset_1\cup \pset_2$ of node-disjoint paths, such that:

\begin{itemize}

\item $T\cap X',T\cap Y'=\emptyset$; 
\item $X'$ has the $(k'',\alpha^*)$-bandwidth property, and $Y'$ has the $(k'',\alpha^*)$-bandwidth property in $G$; and

\item $|\pset_1|=|\pset_2|=2k'$; paths in $\pset_1$ connect vertices in $T_1$ to vertices in $X'$, and paths in $\pset_2$ connect vertices in $T_1$ to vertices in $Y'$. Moreover, the paths in $\pset$ are internally disjoint from $X'\cup Y'$.
\end{itemize}
\end{definition}

We will refer to the $2$-cluster chain defined in Section~\ref{sec: building PoS} as a strong $2$-cluster chain from now on.
In the next theorem we show how to obtain a strong $2$-cluster chain from a weak one. \ifabstract{The proof is omitted from this extended abstract.}\fi

\begin{theorem}\label{thm: from weak to strong $2$-cluster chain}
If $G$ contains a weak $2$-cluster chain, then it contains a strong $2$-cluster chain.
\end{theorem}

%------------------------------------------------
%------------------------------------------------
\iffull
%------------------------------------------------
%------------------------------------------------
\begin{proof}
Our proof extensively uses the following re-routing lemma.
Suppose we are given a directed graph $\hat G$, a set $U\subseteq V(\hat G)$ of its vertices, and an additional vertex $s\in V(\hat G)\setminus U$. A set $\xset$ of directed paths that originate at the vertices of $U$ and terminate at $s$ is called a set of $U$-$s$ paths. We say that the paths in $\xset$ are \emph{nearly disjoint}, if except for vertex $s$ they do not share any other vertices.
The following lemma was proved by Conforti, Hassin and Ravi~\cite{CHR}. We provide a simpler proof, suggested to us by Paul Seymour \cite{PS-comm} in the Appendix.

\begin{lemma}\cite{CHR}
\label{lemma: re-routing of vertex-disjoint paths} 
There is an efficient algorithm, that, given a directed graph $\hat G$, two subsets $U_1,U_2$ of its vertices, and an additional vertex $s\in V(\hat G)\setminus(U_1\cup U_2)$, together with a set $\xset_1$ of $\ell_1$ nearly disjoint $U_1$-$s$ paths and a set $\xset_2$ of $\ell_2$ nearly disjoint $U_2$-$s$ paths in $\hat G$, where $\ell_1>\ell_2\geq 1$, finds a set $\xset'$ of $\ell_1$ nearly-disjoint $(U_1\cup U_2)$-$s$ paths, and  a partition $(\xset'_1,\xset_2')$ of $\xset'$, such that $|\xset_2'|=\ell_2$, the paths of $\xset_2'$ originate from $U_2$, and $\xset_1'\subseteq \xset_1$.\end{lemma}

Let $(X',Y',\pset)$ be a weak $2$-cluster chain in $G$, together with the corresponding partition $(\pset_1,\pset_2)$ of $\pset$.
The idea of the proof is to construct three large sets of paths: set $\qset$, connecting the vertices of $X'$ to the vertices of $Y'$,  set $\rset$, connecting the terminals in $\tset_2$ to the vertices of $X'$ or $Y'$, and a subset $\pset'\subseteq \pset$, containing many paths from both $\pset_1$ and $\pset_2$. We also ensure that the paths in $\qset\cup \rset\cup \pset'$ are all disjoint. This is done by exploiting the connectivity properties of the terminals in $G$, and by carefully applying Lemma~\ref{lemma: re-routing of vertex-disjoint paths} several times. Constructing the strong $2$-cluster chain is then immediate. The rest of the proof consists of three steps. In the first step, we define the set $\qset$ of paths connecting the vertices of $X'$ to the vertices of $Y'$. In the second step, we define the set $\rset$ of paths, connecting the terminals of $T_2$ to the vertices of $X'\cup Y'$. In the third step, we complete the construction of the strong $2$-cluster chain.

\paragraph{Step 1: Connecting $X'$ to $Y'$.}

Our first step is  to construct a set $\qset$ of paths, connecting the vertices of $X'$ to the vertices of $Y'$. Let $\tset_1',\tset_1''\subseteq \tset_1$ be the sets of terminals where the paths of $\pset_1$ and $\pset_2$ originate, respectively, so $|\tset_1'|=|\tset_1''|=2k'$. Since the terminals of $\tset_1$ are $(k/4,1)$-well-linked in $G$, and $2k'\leq k/4$, there is a set $\tilde{\qset}:\tset_1'\sconnect_1\tset_1''$ of edge-disjoint paths in $G$. From Observation~\ref{obs: EDP to NDP in degree-3}, the paths in $\qset$ are node-disjoint. By concatenating the paths in $\pset_1,\tilde{\qset}$ and $\pset_2$, and sending $1/2$ flow units on each such path, we obtain a flow of value $k'$ from the vertices of $\Gamma(X')$ to the vertices of $\Gamma(Y')$, with node-congestion at most $1$. From the integrality of flow, there is a set $\qset_0$ of $3k'/4$ node-disjoint paths, connecting the vertices of $\Gamma(X')$ to the vertices of $\Gamma(Y')$. By appropriately truncating each such path, we can ensure that each path connects a vertex of $\Gamma(X')$ to a vertex of $\Gamma(Y')$, and is internally disjoint from $X'\cup Y'$. 
Our next step is to re-route the paths in $\qset_0$, so that they are disjoint from a large fraction of paths in sets $\pset_1$ and $\pset_2$. We do so in two steps.

First, we construct a directed graph $G_1$ as follows. We start by taking the union of the paths in $\qset_0\cup \pset_1$. If an edge $e$ appears on a path in $\qset_0$ and on a path in $\pset_1$, then we include two copies of $e$ in $G_1$ - one for each of the two paths. For each path $P\in \qset_0\cup \pset_1$, we direct all edges of $P$ toward the endpoint of $P$ lying in $\Gamma(X')$. Finally, we add a new vertex $s$, and connect all vertices of $\Gamma(X')$ that belong to $G_1$ to $s$ with a directed edge. We extend each path in $\qset_0\cup \pset_1$ by adding one edge to it, so that it now terminates at $s$. We can now apply Lemma~\ref{lemma: re-routing of vertex-disjoint paths} to graph $G_1$, with $\xset_1=\pset_1$ and $\xset_2=\qset_0$. As a result, we obtain a subset $\pset_1'\subseteq \pset_1$ of at least $|\pset_1|-|\qset_0|=2k'-3k'/4=5k'/4$ paths, and a collection $\hat{\qset}$ of paths in $G_1$, from vertices of $\Gamma_G(Y')$ to $S$. Moreover, the paths in $\pset_1'\cup \hat{\qset}$ are node-disjoint in $G_1$, except for sharing their last vertex $s$. Consider the corresponding sets $\pset_1',\hat{\qset}$ of paths in the original graph $G$. The paths in $\pset_1'\cup \hat{\qset}$ are node-disjoint; set $\pset_1'$ contains at least $5k'/4$ paths, connecting  vertices of $T_1'$ to  vertices of $\Gamma(X')$, such that the paths in $\pset_1'$ are disjoint from $Y'$, and internally disjoint from $X'$; set $\hat{\qset}$ contains $3k'/4=6k''$ paths, connecting vertices of $\Gamma(X')$ to vertices of $\Gamma(Y')$, such that the paths in $\hat{\qset}$ are internally disjoint from $X'\cup Y'$.

We next repeat the same procedure with $\hat{\qset}$ and $\pset_2$: we construct a directed graph $G_2$, by taking a union of the paths in $\hat{\qset}\cup \pset_2$ (if an edge $e$ appears on two such paths, then we include two copies of $e$ - one for each path). For each path $P\in \hat{\qset}\cup \pset_2$, we direct all edges of $P$ toward the endpoint of $P$ lying in $\Gamma(Y')$. We add a new vertex $s$, and connect all vertices of $\Gamma(Y')$ that belong to $G_2$ to $s$ with a directed edge. We extend each path in $\hat{\qset}\cup \pset_2$ by adding one edge to it, so it now terminates at $s$. As before, we apply Lemma~\ref{lemma: re-routing of vertex-disjoint paths} to graph $G_2$, with $\xset_1=\pset_2$ and $\xset_2=\hat{\qset}$. As a result, we obtain a subset $\pset_2'\subseteq \pset_2$ of at least $5k'/4$ paths, and a collection $\qset$ of paths in $G_2$ from $\Gamma_G(X')$ to $\Gamma_G(Y')$. Moreover, the paths in $\pset_2'\cup \qset$ are node-disjoint in $G_2$, except for sharing their last vertex $s$. Consider the corresponding sets $\pset_2',\qset$ of paths in the original graph $G$. Then the paths in $\pset_1'\cup\pset_2'\cup \hat{\qset}$ are node-disjoint (as the vertices lying on the paths in $\pset_1'$ do not participate in graph $G_2$); $|\pset_1'|,|\pset_2'|\geq 5k'/4$; set $\qset$ contains $6k''$ paths, connecting  vertices of $\Gamma(X')$ to  vertices of $\Gamma(Y')$, and all paths of $\qset$ are internally disjoint from $X'\cup Y'$.

\paragraph{Step 2: Connecting the Terminals of $\mathbf{T_2}$ to $\mathbf{X'\cup Y'}$.}
In this step we construct a set $\rset$ of paths, connecting terminals of $T_2$ to vertices of $\Gamma(X')\cup \Gamma(Y')$. Recall that $|T_1'|=2k'$. Let $T'\subseteq T_1'$ be any subset of $k'$ terminals. Since $(T_1,T_2)$ are $\half$-linked in $G$, there is a set $\tilde{\pset}:T_2\sconnect_2 T'$ of paths in $G$. By concatenating the paths of $\tilde{\pset}$ and the paths of $\pset_1$ that originate at the vertices of $T'$, and sending $1/4$ flow units on each such path, we obtain a flow of value $k'/4$ from $T_2$ to $X'$, with node-congestion at most $1$. From the integrality of flow, there is a collection $\rset_0$ of $k'/4$ node-disjoint paths from the terminals of $T_2$ to the vertices of $X'$ in $G$. By suitably truncating the paths in $\rset_0$, we obtain a collection of $k'/4=2k''$ node-disjoint paths, connecting the terminals of $T_2$ to the vertices of $\Gamma(X')\cup \Gamma(Y')$, such that the paths in $\rset_0$ are internally disjoint from $X'\cup Y'$. Next, we re-route the paths in $\rset_0$, so that they are disjoint from many paths in $\pset_1'\cup \pset_2'\cup \qset$.

We do so by constructing a directed graph $G_3$: Start with the union of the paths $\pset_1'\cup\pset_2'\cup \qset\cup \rset_0$. Note that an edge may belong to up to two such paths. If some edge $e$ belongs to two such paths, we make two copies of $e$ - one for each path. We then direct the edges on each path $P\in \pset_1'\cup \pset_2'\cup \rset_0$ toward the endpoint of $P$ lying in $X'\cup Y'$. Consider now some path $Q\in \qset$, and let $x\in X'$, $y\in Y'$ be its two endpoints. Let $e=(u,v)\in E(Q)$ be any edge of $Q$, with $u$ lying closer to $x$ on $Q$ than $v$. We subdivide the edge $(u,v)$ by two new vertices $a_Q,b_Q$, thus obtaining a new path $(u,a_Q,b_Q,v)$ replacing $e$. Let $Q'$ denote the resulting path, let $Q_1$ be the sub-path of $Q'$ from $x$ to $a_Q$, and let $Q_2$ be the sub-path of $Q'$ from $b_Q$ to $y$. We direct all edges of $Q_1$ along the path toward $x$, and all edges of $Q_2$ along the path toward $y$. Let $\tilde{\qset}=\set{Q_1,Q_2\mid Q\in \qset}$. Finally, we add a destination vertex $s$ and connect every vertex of $\Gamma(X')\cup\Gamma(Y')$ that belongs to $G_3$ to $s$ with a directed edge. We extend all paths in $\pset_1'\cup \pset_2'\cup \rset_0\cup \tilde{\qset}$, so that each such path now terminates at $s$. Let $\xset_1=\pset_1'\cup \pset_2'\cup \tilde{\qset}$, and $\xset_2=\rset_0$ be the resulting sets of paths in the new graph $G_3$. Then the paths in $\xset_1$ are all node-disjoint, except for sharing the destination vertex $s$, and the same holds for the paths in $\xset_2$. We apply Lemma~\ref{lemma: re-routing of vertex-disjoint paths} to obtain the sets $\xset'_1\subseteq \xset_1$ and $\rset=\xset_2'$ of paths, such that all paths in $\xset'_1\cup \rset$ are node-disjoint, except for sharing the destination vertex $s$.

Let $\pset_1''=\xset'_1\cap \pset_1'$, $\pset_2''=\xset'_1\cap \pset_2'$. We construct a set $\qset^*\subseteq\qset$ of paths as follows: for each path $Q\in \qset$, if both $Q_1,Q_2\in \xset'_1$, then add $Q$ to $\qset^*$. Notice that $|\qset^*|\geq |\qset|-|\rset_0|\geq 6k''-2k''\geq k''$. We discard paths from $\qset^*$ until $|\qset^*|=k''$ holds. Similarly, $|\pset_1''|\geq |\pset_1'|-|\rset_0|\geq 5k'/4-2k''=k'$, and $|\pset_2''|\geq k'$. 
Consider now the sets $\pset_1'',\pset_2'',\qset^*$ and $\rset$ of paths in the original graph $G$. Then all paths in $\pset_1''\cup \pset_2''\cup \qset^*\cup \rset$ are node-disjoint, and internally disjoint from $X'\cup Y'$; set $\rset$ contains $2k''$ paths connecting the terminals of $T_2$ to the vertices of $\Gamma(X')\cup \Gamma(Y')$; set $\qset^*$ contains $k''$ paths connecting the vertices of $\Gamma(X')$ to the vertices of $\Gamma(Y')$; set $\pset_1''$ contains at least $k'$ paths connecting the terminals of $T_1'$ to the vertices of $\Gamma(X')$, and set $\pset_2''$ contains at least $k'$ paths connecting the terminals of $T_1''$ to the vertices of $\Gamma(Y')$.

% Let $\rset',\rset''\subseteq \rset$ be the subsets of paths terminating at $X'$ and $Y'$ respectively, and assume w.l.o.g. that $|\rset''|\geq |\rset'|$, so $|\rset'|\leq k''$, and $|\rset''|\geq k''$. We claim that $|\pset_1''|\geq k'$. Indeed, assume that $|\pset_1''|<k'$. Then $|\pset'_1\setminus \xset'|>5k'/4-k'= k''$. From Lemma~\ref{lemma: re-routing of vertex-disjoint paths}, for each path $P\in\pset_1'\setminus \xset'$, there is a distinct path $\hat{R}_P\in \rset^*$, which is a concatenation of a prefix of the path $R_P\in \rset_0$ and a suffix of $P$, containing the penultimate vertex of $P$ in graph $G_3$. In other words, $\hat{R}_P$ is internally disjoint from $X'\cup Y'$, and it terminates at a vertex of $\Gamma(X')$.  Therefore, $|\rset'|>k''$, a contradiction.

\paragraph{Step 3: Constructing the $2$-Cluster Chain.}

We are now ready to define the $2$-cluster chain.   Let $\rset',\rset''\subseteq \rset$ be the subsets of paths terminating at $X'$ and $Y'$ respectively, and assume w.l.o.g. that $|\rset''|\geq |\rset'|$, so $|\rset''|\geq k''$. The set $E'$ of edges contains exactly one edge $e_Q$ from each path $Q\in \qset^*$, where the edge is chosen arbitrarily. For each path $Q\in \qset^*$, let $Q_1$ and $Q_2$ be the two sub-paths obtained from $Q$ by deleting the edge $e_Q$ (where each sub-path contains an endpoint of $e_Q$). We assume that $Q_1$ contains a vertex of $X'$ and $Q_2$ contains a vertex of $Y'$. Let $\qset'=\set{Q_1\mid Q\in \qset^*}$, and $\qset''=\set{Q_2\mid Q\in \qset^*}$.

We let $\ttset_2\subseteq \tset_2$ be the set of $k''$ vertices, where the paths of $\rset''$ originate, and we let $\ttset_1\subseteq \tset_1$ be the set of $k'$ vertices, where the paths of $\pset_1''$ originate.
Cluster $X$ is the union of $X'$, $V(\pset_1'')$ and $V(\qset')$, while cluster $Y$ is the union of $Y'$, $V(\rset'')$ and $V(\qset'')$.
Let $\Upsilon_X\subseteq X$ be the set of vertices of $X$ that serve as endpoints of the edges in $E'$, and define $\Upsilon_Y\subseteq Y$ similarly. Recall that $X'$ has the $(k'',\alpha^*)$-bandwidth property in $G$. Moreover, we have a set $\pset_1''\cup \qset'_2$ of node-disjoint paths, connecting every vertex of $\ttset_1\cup \Upsilon_X$ to some vertex of $\Gamma(X')$, and these paths are internally disjoint from $X'$. It is immediate to see that the vertices of $\ttset_1\cup \Upsilon_X$ are $(k'',\alpha^*)$-well-linked in $G[X]$. Similarly, using the $(k'',\alpha^*)$-bandwidth property of $Y'$, and the sets $\rset_2''\cup \qset_2''$ of paths, we conclude that $\ttset_2\cup \Upsilon_Y$ is $(k'',\alpha^*)$-well-linked in $G[Y]$.
\end{proof}

%------------------------------------------------
%------------------------------------------------
\fi
%------------------------------------------------
%------------------------------------------------

\label{----------------------------------Good and Perfect Clusters---------------------------------------------}
\subsection{Good Clusters and Perfect Clusters}\label{subsec: good and perfect clusters}

\begin{definition} Let $C\subseteq V(G)\setminus T$ be any cluster containing non-terminals vertices only, and let $(A,B)$ be the minimum $1/4$-balanced cut in $G\setminus C$ with respect to $T_1$. 
We say that $C$ is a \emph{good} cluster, iff $|E(A,B)|< k/28$.  We say that it is a \emph{perfect} cluster, iff $k/28\leq |E(A,B)|\leq 7k/32$.
\end{definition}

\ifabstract
The proof of the following theorem uses standard techniques, namely variations of the so-called well-linked decompositions, and it is omitted from this extended abstract.
\fi
\iffull
The following theorem shows that a perfect cluster can be used to construct a $2$-cluster chain.
\fi

\begin{theorem}\label{thm: perfect cluster gives chain}
If there is a perfect cluster $C\subseteq V(G)\setminus T$, such that $|\out(C)|\leq k+k'+1$, and $C$ has the $(k'',\alpha^*)$-bandwidth property, then $G$ contains a strong $2$-cluster chain.
\end{theorem}

%------------------------------------------------
%------------------------------------------------
\iffull
%------------------------------------------------
%------------------------------------------------
\begin{proof}
From Theorem~\ref{thm: from weak to strong $2$-cluster chain}, it is enough to show that $G$ contains a weak $2$-cluster chain.  Let $(A,B)$ be the $1/4$-minimum balanced cut of $G\setminus C$ with respect to $T_1$, where $|A\cap T_1|\leq |B\cap T_1|$, 
so that $k/28\leq |E(A,B)|\leq 7k/32$.

Since the terminals in $T_1$ are $(k/4,1)$-well-linked in $G$, there is a set $\pset$ of $k/4$ node-disjoint paths between the vertices of $T_1\cap A$ and the vertices of $T_1\cap B$. At most $|E(A,B)|$ of such paths can use the edges of $E(A,B)$, and the remaining paths must intersect the cluster $C$. Let $\pset_2\subseteq \pset$ be the subset of paths that do not contain the edges of $E(A,B)$, truncated at the first vertex of $C$ on the path (where we view the paths as directed from $T_1\cap A$ to $T_1\cap B$). Then $\pset_2$ contains at least $k/4-|E(A,B)|\geq k/4-7k/32\geq k/32$ paths, connecting the vertices of $T_1\cap A$ to the vertices of $C$, and they are internally disjoint from $C\cup B$. Notice that $|E(A,C)|\geq |\pset_2|\geq k/4-|E(A,B)|$, and so $|E(B,C)|\leq |\out(C)|-|E(A,C)|\leq k+k'+1-k/4+|E(A,B)|\leq 49k/64+1+|E(A,B)|$, since $k'=k/64$.
The following lemma is central to the proof of the theorem.

\begin{lemma}\label{lem: partitioning alg}
There is a cluster $S^*\subseteq B$, such that $|T_1\cap S^*|\geq k/4$, and $\Gamma(S^*)\cup (T\cap S^*)$ is $(k'',\alpha^*)$-well-linked in $G[S^*]$. \end{lemma}

Before we prove the lemma, we show that Theorem~\ref{thm: perfect cluster gives chain} follows from it, by constructing a weak $2$-cluster chain in $G$. We set $Y'=C$ and $X'=S^*\setminus T$. Since each terminal has degree $1$, and $\Gamma(S^*)\cup(S^*\cap T)$ is $(k'',\alpha^*)$-well-linked in $G[S^*]$, it is easy to see that $\Gamma(X')$ is $(k'',\alpha^*)$-well-linked in $G[X']$, so $X'$ has the $(k'',\alpha^*)$-bandwidth property. We let $\pset_1$ be the set of edges with one endpoint in $T_1\cap S^*$, and another in $X'$, and the set $\pset_2$ of paths connecting the vertices of $T_1$ to the vertices of $Y'$ remains the same. It is easy to see that $(X',Y',\pset_1,\pset_2)$ is a valid weak $2$-cluster chain. It now remains to prove Lemma~\ref{lem: partitioning alg}.

\begin{proofof}{Lemma \ref{lem: partitioning alg}} We show an algorithm to compute the cluster $S^*$ with the required properties.
Throughout the algorithm, we maintain a partition $\sset$ of $B$ into clusters with a special cluster $S\in \sset$, and a set $E'=\left(\bigcup_{C'\in \sset}\out(C')\right )\setminus\out(B)$ of edges. At the beginning, $S=B$, $\sset=\set{S}$, and $E'=\emptyset$. While the set $\Gamma(S)\cup (T\cap S)$ of vertices is not $(k'',\alpha^*)$-well-linked in $G[S]$, let $(Z,Z')$ be a $(k'',\alpha^*)$-violating partition, that is, $|E(Z,Z')|<\alpha^*\cdot\min\set{|Z\cap (\Gamma(S)\cup T|),|Z'\cap (\Gamma(S)\cup T)|,k''}$. We remove $S$ from $\sset$ and add $Z$ and $Z'$ to $\sset$ instead. We also add the edges of $E(Z,Z')$ to $E'$. Finally, if $|Z\cap T_1|\geq |Z'\cap T_1|$, then we set $S=Z$, and otherwise we set $S=Z'$. We then continue to the next iteration. 

It is clear that at the end of the algorithm, if we denote the final cluster $S$ by $S^*$, then $\Gamma(S^*)\cup (T\cap S^*)$ is $(k'',\alpha^*)$-well-linked in $G[S^*]$. It now remains to show that $|T_1\cap S^*|\geq k/4$. We will show something stronger: namely, that $|T_1\cap S^*|\geq |T_1\cap B|/2$.

In order to analyze the algorithm, we maintain, throughout the algorithm's execution, non-negative budgets $\beta(v)$ for all vertices $v\in B$. The budgets are determined as follows. For each vertex $v\in \Gamma(S)\cup (S\cap T)$, we set $\beta(v)=\alpha^*/(1-\alpha^*)$, and all other vertices have budget $0$. We need the following claim.

\begin{claim}\label{claim: budget bound}
Throughout the execution of the algorithm, the following invariant holds:
\[\sum_{v\in B}\beta(v)+|E'|\leq \frac{\alpha^*}{1-\alpha^*}\left (\frac{49k}{32}+2|E(A,B)|+1\right ).\]
\end{claim}

\begin{proof}
 Observe first, that the invariant holds at the beginning of the algorithm, since $|\Gamma(B)|\leq |E(B,C)|+|E(A,B)|$, and so:

\[\begin{split} \sum_{v\in B}\beta(v)+|E'|&\leq \frac{\alpha^*}{1-\alpha^*}\left(|E(B,C)|+|E(A,B)|+|T_1\cap B|+|T_2|\right )\\
&\leq \frac{\alpha^*}{1-\alpha^*}\left (\frac{49k}{64}+2|E(A,B)|+1+\frac{3k}{4}+k'\right )\\
&\leq \frac{\alpha^*}{1-\alpha^*}\cdot\left(\frac{49k}{32}+1+2|E(A,B)|\right ),\end{split}\]

since $k'=k/64$.

%This reduces to $|E(A,B)|\geq \frac{(3k+2k'+2)\alpha}{1-5\alpha}$. Taking $|E(A,B)|\leq 3.26\alpha$ is enough when $\alpha=1/128$ or less, and $k'<k/16$. (for $k'<k/32$, can do with $3.2\alpha$).

Assume now that the invariant holds at the beginning of the current iteration, where we split $S$ into $Z$ and $Z'$. Assume w.l.o.g. that $|Z\cap T_1|\leq |Z'\cap T_1|$, so $S=Z'$. Let $L=Z\cap (\Gamma(S)\cup T)$, and let $R\subseteq Z'$ be the set of vertices incident on the edges of $E(Z,Z')$. Then the budget of every vertex in $L$ decreases by $\alpha^*/(1-\alpha^*)$, and the budget of every vertex in $R$ increases by the same amount. All other vertex budgets remain the same. Therefore, in total, the budgets of the vertices in $L$ decrease by $|L|\cdot \alpha^*/(1-\alpha^*)$, and the budgets of the vertices in $R$, and the cardinality of $E'$ increase by:

\[\frac{\alpha^*}{1-\alpha^*}|R|+|E(Z,Z')|\leq |E(Z,Z')|\left (1+\frac{\alpha^*}{1-\alpha^*}\right )\leq \frac{\alpha^* |L|}{1-\alpha^*},\]

and so the invariant continues to hold.
\end{proof}

We now claim that when the algorithm terminates, $|S\cap T_1|\geq |B\cap T_1|/2$. Assume otherwise. 
Let $i$ be the largest integer, such that at the beginning of iteration $i$, $|S\cap T_1|\geq |B\cap T_1|/2$ held. Let $E_1$ be the set $E'$ of edges at the beginning of iteration $i$, and let $E_2$ be the set $E'$ of edges at the end of iteration $i$. Then $|E_2|\leq |E_1|+\alpha^*k''=|E_1|+\alpha^*k/512$.

Consider the set $S$ at the beginning of iteration $i$. Since the terminals in $T_1$ are the $(k/4,1)$-well-linked, and $|S\cap T_1|\geq |B\cap T_1|/2\geq k/4$, there is a set $\rset$ of at least $k/4$ edge-disjoint paths, connecting vertices of $S\cap T_1$ to vertices of $A\cap T_1$, whose endpoints are all distinct. From Observation~\ref{obs: EDP to NDP in degree-3}, the paths in $\rset$ are node-disjoint, and, since each such path must contain a vertex of $\Gamma(S)$,  $|\Gamma(S)|\geq k/4$ must hold. Therefore, $\sum_{v\in B}\beta(v)\geq (|T_1\cap S|+|\Gamma(S)|)\cdot \frac{\alpha^*}{1-\alpha^*}\geq \frac{k}{2}\cdot \frac{\alpha^*}{1-\alpha^*}$ at the beginning of iteration $i$. From Claim~\ref{claim: budget bound}, we conclude that:

\[\begin{split}
|E_2|&\leq |E_1|+\frac{\alpha^* k}{512}\\
&\leq \frac{\alpha^*}{1-\alpha^*}\left (\frac{49k}{32}+2|E(A,B)|+1\right )-\frac{k}{2}\cdot \frac{\alpha^*}{1-\alpha^*}+\frac{\alpha^* k}{512}\\
&=\frac{1}{63}\left(\frac{33k}{32}+1\right)+\frac{k}{64\cdot 512}+\frac{2|E(A,B)|}{63}\\
&< \frac{|E(A,B)|} 2,
\end{split}\]

since $|E(A,B)|\geq k/28$.

Consider the set $\sset$ of clusters at the end of iteration $i$. Since we have assumed that $|S\cap T_1|\leq |B\cap T_1|/2$, and since in every iteration, whenever $S$ was partitioned into two clusters $Z,Z'$, we chose to continue with the cluster containing more terminals of $T_1$, we conclude that $|C'\cap T_1|\leq |B\cap T_1|/2$ for every cluster $C'\in \sset$.

For each cluster $C'\in \sset$, we let $d_1(C')=|\out(C')\cap E(A,B)|$, and $d_2(C')=|\out(C')\setminus E(A,B)|$. Then $\sum_{C'\in \sset}d_2(C')=2|E'|<|E(A,B)|$. Therefore, there is some cluster $S'\in \sset$, with $d_1(S')>d_2(S')$. However, the partition $(A\cup S',B\setminus S')$ is a $1/4$-balanced cut in $G\setminus C$ with respect to $T_1$, and $|E(A\cup S',B\setminus S')|=|E(A,B)|-d_1(S')+d_2(S')<|E(A,B)|$, contradicting the choice of $(A,B)$. We conclude that at the end of the algorithm, $|S\cap T_1|\geq |B\cap T_1|/2$ must hold.
\end{proofof}
\end{proof}

%------------------------------------------------
%------------------------------------------------
\fi
%------------------------------------------------
%------------------------------------------------

We will also use the following simple observation\ifabstract{, whose proof is omitted from this extended abstract}\fi.

\begin{observation}\label{obs: cutting good cluster}
Let $C$ be a good cluster, and let $(C',C'')$ be any partition of $C$, where $|\out(C)\cap \out(C'')|\leq 27k/80$. Then $C'$ is a good or a perfect cluster.
\end{observation}

\iffull
\begin{proof}
It is enough to prove that there is some $1/4$-balanced cut $(A',B')$ in $G\setminus C'$, with respect to $T_1$, such that $|E(A',B')|\leq 7k/32$.
Let $(A,B)$ be the minimum $1/4$-balanced cut in $G\setminus C$ with respect to $T_1$, and recall that $|E(A,B)|\leq k/28$. We now define a  balanced cut $(A',B')$ in graph $G\setminus C'$, with respect to $T_1$, as follows. Start with $A'=A$ and $B'=B$. If $|E(A,C'')|\geq |E(B,C'')|$, add the vertices of $C''$ to $A'$; otherwise add them to $B'$. Notice that $|E(A',B')|\leq |E(A,B)|+|\out(C)\cap \out(C'')|/2\leq  k/28+27k/160 <7k/32$.
\end{proof}

\fi

\label{-----------------------------------------------The alg--------------------------------------------}

\subsection{Splitting the Cluster}\label{subsec: the alg}
We are now ready to prove that $G$ contains a $2$-cluster chain.% From Theorem~\ref{thm: perfect cluster gives chain} if $G$ contains a perfect cluster $C\subseteq V(G)\setminus T$ that has the $(k'',\alpha^*)$-bandwidth property, and $|\out(C)|\leq k+k'+1$, then $G$ contains a $2$-cluster chain. Therefore, from now on we assume that $G$ contains no such cluster $C$.

%Note that there is a good cluster $C\subseteq V(G)\setminus T$ with the $(k/4,1)$-bandwidth property, and $|\out(C)|\leq k+k'
%$ - the cluster $V(G)\setminus T$. 

We are interested in a cluster $C\subseteq V(G)\setminus T$, with the following properties: $C$ is a good or a perfect cluster, and it has the $1/23$-bandwidth property. Among all such clusters $C$, let $C^*$ be the one minimizing $|\out(C^*)|$, and subject to this, minimizing $|C^*|$.
We note that $|\out(C^*)|\leq k+k'$ must hold, since $V(G)\setminus T$ is a good cluster, and, from the well-linkedness properties of the terminals, it is not hard to see that it has the $1/23$-bandwidth property.
We need the following two claims.

\begin{claim}\label{claim: boundary prop}
If $C^*$ is a good cluster, then every vertex $v\in \Gamma(C^*)$ is incident on exactly one edge of $\out(C^*)$.
\end{claim}

\begin{proof}
Assume otherwise, and let $v\in \Gamma(C^*)$ be incident on more than one edge of $\out(C^*)$. Since maximum vertex degree in $G$ is $3$, and $G[C^*]$ is connected, due to the $1/23$-bandwidth property of $C^*$, $v$ has exactly one neighbor $u\in C^*$. Consider the cluster $C'=C^*\setminus\set{v}$. Then $|\out(C')|<|\out(C^*)|$, and it is easy to see that $C'$ has the $1/23$-bandwidth property. Moreover, from Observation~\ref{obs: cutting good cluster}, $C'$ is a good or a perfect cluster, contradicting the choice of $C^*$.
\end{proof}

\begin{claim}\label{claim: many paths}
$p(C^*)=|\Gamma(C^*)|$.
\end{claim}

\begin{proof} Assume otherwise. Intuitively, if $p(C^*)<|\Gamma(C^*)|$, then there is a small cut separating $\Gamma(C^*)$ from the terminals in $T$. We use this cut to define a new cluster $C'$, such that $C'$ is either a good or a perfect cluster, and it has the $1/23$-bandwidth property, while $|\out(C')|<|\out(C^*)|$, contradicting the choice of $C^*$.

Let $p=p(C^*)$. From Menger's theorem, there is a tri-partition $(X,Y,Z)$ of $V(G)$, such that $|Y|=p$, $Y$ separates $X$ from $Z$ in $G$, $C^*\subseteq Y\cup Z$, and $T\subseteq X\cup Y$. Among all such tri-partitions, we choose the one minimizing $|Y|+|Z|$. 
As each terminal has degree $1$ in $G$, it is easy to see that $Y\cap T=\emptyset$, and so $T\subseteq X$. Recall that $\pset(C^*)$ is the largest-cardinality set of node-disjoint paths connecting the terminals of $T$ to $C^*$, and the paths in $\pset(C^*)$ are internally disjoint from $C^*$. Therefore,  $Y$ contains exactly one vertex from each path in $\pset(C^*)$, and so $Y\cap C^*\subseteq \Gamma(C^*)$. We let $C'$ be the set of vertices of the connected component of $G[Y\cup Z]$, containing $C^*$. Notice that $\Gamma(C')\subseteq Y$, and so $|\Gamma(C')|<|\Gamma(C^*)|$. Since $C^*\subseteq C'$, and $C^*$ is a good or a perfect cluster, it is easy to see that $C'$ is also a good or a perfect cluster. Moreover, the set $\pset(C^*)$ of paths defines a collection $\pset'$ of node-disjoint paths, connecting every vertex of $\Gamma(C')$ to some vertex in $\Gamma(C^*)$, such that the paths in $\pset'$ are internally disjoint from $C^*$. Using the fact that $C^*$ has the $1/23$-bandwidth property, it is easy to see that $C'$ also has the $1/23$-bandwidth property.

In order to reach a contradiction, it is now enough to show that $|\out(C')|<|\out(C^*)|$. We partition the edges of $\out(C')$ into two subsets: set $E_1$ contains all edges incident on the vertices of $Y\cap \Gamma(C^*)$, and set $E_2$ contains all remaining edges. Similarly, we partition the edges of $\out(C^*)$ into two subsets: set $E_1'$ contains all edges incident on the vertices of $Y\cap \Gamma(C^*)$, and set $E_2$ contains all remaining edges. Observe first that the edges of $E_1$ and $E_1'$ are incident on the same subset of vertices: $Y\cap \Gamma(C^*)$, and, since $C^*\subseteq C'$, it is easy to see that $|E_1|\leq |E_1'|$. 

Let $Y'=Y\cap \Gamma(C^*)$.
Since $|\Gamma(C')|<|\Gamma(C^*)|$, and $Y'\subseteq \Gamma(C')$, $|\Gamma(C')\setminus Y'|<|\Gamma(C^*)\setminus Y'|$. Every vertex of $\Gamma(C^*)\setminus Y'$ has at least one edge incident to it in $E_2'$, and every edge of $E_2$ is incident on some vertex of $\Gamma(C')\setminus Y'$. Therefore, it is enough to show that every vertex in $\Gamma(C')\setminus Y'$ is incident on exactly one edge of $E_2$. Assume otherwise, and let $v\in \Gamma(C')\setminus Y'$ be any vertex incident on at least two edges of $E_2$. Since $v\not \in Y'$, it does not belong to $\Gamma(C^*)$, or to $C^*$. Moreover, $v$ has at most one neighbor in $Z$ - denote it by $u$. Therefore, we can obtain a new tri-partition $(X\cup \set{v},(Y\setminus \set{v})\cup\set{u},Z\setminus \set{u})$ separating the terminals from $C^*$, contradicting the choice of the partition $(X,Y,Z)$. We conclude that $|E_2|\leq |\Gamma(C')\setminus Y'|<|\Gamma(C^*)\setminus Y'|\leq |E_2'|$, and $|\out(C')|<|\out(C^*)|$, contradicting the choice of $C^*$.
\end{proof}

If $C^*$ is a perfect cluster, then, from Theorem~\ref{thm: perfect cluster gives chain}, we obtain a $2$-cluster chain in $G$. Therefore, we assume from now on that $C^*$ is a good cluster. We use the following theorem to finish our proof. Its statement is slightly stronger than what we need, but this stronger statement will be used in the proof itself.
 
\begin{theorem}\label{thm: iteration}
Let $C\subseteq V(G)\setminus T$ be a good cluster with $|\out(C)|\leq k+k'+1$, and $p(C)\geq |\Gamma(C)|-1$, such that $C$ has the $1/23$-bandwidth property, and every vertex of $\Gamma(C)$ is incident on exactly one edge of $\out(C)$. Then either there is a  strong $2$-cluster chain in $G$, or there is a good or a perfect cluster $C'\subsetneq C$, with $|\out(C')|\leq |\out(C)|$, such that $C'$ has the $1/23$-bandwidth property.
\end{theorem}

From Theorem~\ref{thm: iteration}, either $G$ has a strong $2$-cluster chain, or there is a good or a perfect cluster $C'\subsetneq C^*$ with $|\out(C')|\leq |\out(C)|$, such that $C'$ has the $1/23$-bandwidth property. The latter is impossible from the definition of $C^*$, so $G$ must contain a $2$-cluster chain.
From now on we focus on proving Theorem~\ref{thm: iteration}.

\subsection*{Proof of Theorem~\ref{thm: iteration}}

We start with the following two theorems, whose proofs use standard techniques, \iffull and are deferred to the Appendix.\fi \ifabstract and are omitted from this extended abstract.\fi

%---------------------------------------------------
\begin{theorem}\label{thm: good cluster w small boundary is enough}
If there is a good cluster $C'\subseteq C$, with $|\out(C')|\leq k+k'+1$ and $|\Gamma(C')|\leq 7k/8$, such that $C'$ has the $(k'',\alpha^*)$-bandwidth property, then there is a strong $2$-cluster chain in $G$.
\end{theorem}
%---------------------------------------------------

%---------------------------------------------------
\begin{theorem}\label{thm: balanced-cut-large-smaller-side}
Suppose there is some value $\rho$, such that $\rho|\Gamma(C)|\leq (k+k'+1)/4$, and a minimum $\rho$-balanced cut $(A,B)$ of $C$ with respect to $\Gamma(C)$, such that $|\Gamma(A)\cap \Gamma(C)|\geq |\Gamma(B)\cap \Gamma(C)|\geq 27k/80$. Then there is a cluster $C'\subseteq B$ that has the $(k'',\alpha^*)$-bandwidth property, and $|\Gamma(C')\cap \Gamma(C)|> 3k/64$.
\end{theorem}
%---------------------------------------------------

If $|\Gamma(C)|\leq 7k/8$, then from Theorems~\ref{thm: good cluster w small boundary is enough} and~\ref{thm: from weak to strong $2$-cluster chain}, there is a strong $2$-cluster chain in $G$. We assume from now on that $|\Gamma(C)|>7k/8$.
Let $\alpha$ be the largest value for which $C$ has the $\alpha$-bandwidth property, so $\alpha\geq 1/23$.
 We distinguish between three cases. The first case is when $\alpha< 1/5$; the second case is when $\alpha\geq1/5$ but $C$ does not have the $(k/4,1)$-bandwidth property, and the third case is when $C$ has the $(k/4,1)$-bandwidth property.

\paragraph{Case 1: $\mathbf{\alpha< 1/5}$.}
Let $(Z,Z')$ be the minimum $1/4$-balanced cut of $C$ with respect to $\Gamma(C)$, where $|Z\cap \Gamma(C)|\geq |Z'\cap \Gamma(C)|$. We consider three sub-cases.

\noindent {\bf Subcase 1a.} This first subcase happens if $|E(Z,Z')|>\frac{\alpha}{1-\alpha}|\Gamma(C)|$. 
Observe that for $0<\alpha'<1$, function $\frac{\alpha'}{1-\alpha'}$ monotonously increases in $\alpha'$. So there is some value $\alpha<\alpha'<1$, such that $|E(Z,Z')|>\frac{\alpha'}{1-\alpha'}|\Gamma(C)|$. The following lemma uses standard techniques, and its proof \iffull appears in Appendix.\fi \ifabstract is omitted from this extended abstract.\fi

\begin{lemma}\label{lem: balanced-cut-simple}
There is a cluster $C'\subsetneq C$, such that $|\out(C')|<|\out(C)|$, $|\Gamma(C)\cap \Gamma(C')|\geq 3|\Gamma(C)|/4$,  and
$C'$ has the $\alpha'$-bandwidth property. \end{lemma}

Let $C'$ be the cluster given by Lemma~\ref{lem: balanced-cut-simple}, and let $C''=C\setminus C'$. Then $|\Gamma(C'')\cap \Gamma(C)|\leq |\Gamma(C)|/4\leq (k+k'+1)/4$. Since every vertex of $\Gamma(C)$ is incident on exactly one edge of $\out(C)$, we get that $|\out(C'')\cap \out(C)|\leq |\Gamma(C'')\cap \Gamma(C)|\leq (k+k'+1)/4\leq 27k/80$. Therefore, from Observation~\ref{obs: cutting good cluster}, $C'$ is a good or a perfect cluster. Since $\alpha'>\alpha\geq 1/23$, $C'$ is a valid output for the theorem. Notice that $|\out(C')|<|\out(C)|$.

\noindent{\bf Subcase 1b.} This case happens if $|E(Z,Z')|\leq \frac{\alpha}{1-\alpha}|\Gamma(C)|$, and $Z$ is a good or a perfect cluster. In this case, from Lemma~\ref{lemma: balanced cut large piece wl}, cluster $Z$ has the $\alpha'$-bandwidth property, for $\alpha'=\frac{\alpha}{2+\alpha}$. Moreover, since $\alpha< 1/5$, $|E(Z,Z')|< |\Gamma(C)|/4\leq |\Gamma(C)\cap Z'|$, and so  $|\out(Z)|<|\out(C)|\leq k+k'$.

If $\alpha\geq 1/11$, then $\alpha'=\alpha/(2+\alpha)\geq 1/23$, and we return $C'=Z$. 
%Notice that $|\out(Z)|\leq |\out(Z\cap C)|+|E(Z,Z')|<|\out(Z\cap C)|+\frac{\alpha}{1-\alpha}|\Gamma(C)|\leq |\out(Z\cap C)|+\frac{|\Gamma(C)|}{4}<|\out(C)|$ in this case.
%
Otherwise, if $\alpha<1/11$, then $\alpha'=\alpha/(2+\alpha)\geq 1/64$, since $\alpha\geq 1/23$, and:

\[
\begin{split}
|\Gamma(Z)|&\leq \frac{3|\Gamma(C)|} 4+|E(Z,Z')|
\leq \frac{3|\Gamma(C)|} 4+ \frac{\alpha}{1-\alpha}|\Gamma(C)|\\
&\leq \frac{17|\Gamma(C)|}{20}
\leq \frac{17(k+k'+1)}{20}
\leq \frac{7k}{8},\end{split}\]

since $k'=k/64$. If $Z$ is a perfect cluster, then from Theorem~\ref{thm: perfect cluster gives chain}, we obtain a $2$-cluster chain. Otherwise, $Z$ is a good cluster, and we obtain the $2$-cluster chain by applying Theorem~\ref{thm: good cluster w small boundary is enough} to cluster $Z$.

\noindent{\bf Subcase 1c.} This  case happens if $|E(Z,Z')|\leq \frac{\alpha}{1-\alpha}|\Gamma(C)|$, but $Z$ is not a good or a perfect cluster. From Observation~\ref{obs: cutting good cluster}, $|\out(Z')\cap \out(C)|\geq 27k/80$ must hold, and, since every vertex of $\Gamma(C)$ is incident on exactly one edge of $\out(C)$, we get that $|\Gamma(Z')\cap \Gamma(C)|\geq |\out(Z')\cap \out(C)|\geq 27k/80$.

In this case, we construct a $2$-cluster chain in $G$. We let $X'=Z$. Since $C$ has the $\alpha\geq 1/23$-bandwidth property, from Lemma~\ref{lemma: balanced cut large piece wl}, $X'$ has the $\alpha'$-bandwidth property, for $\alpha'=\alpha/(2+\alpha)\geq \alpha^*$. Therefore, $X'$ has the $(k'',\alpha^*)$-bandwidth property. We then apply Theorem~\ref{thm: balanced-cut-large-smaller-side} to the partition $(Z,Z')$, to obtain a cluster $C'\subseteq Z'$, that has the $(k'',\alpha^*)$-bandwidth property, and we set $Y'=C'$. In order to define the sets $\pset_1$ and $\pset_2$ of paths, observe that all but at most one vertices of $\Gamma(C)$ have a path of $\pset(C)$ terminating at them, since $p(C)\geq |\Gamma(C)|-1$. Since we have assumed that $|\Gamma(C)|\geq 7k/8$, $|\Gamma(Z)\cap \Gamma(C)|\geq |\Gamma(C)|/2\geq 7k/16$. Therefore, at least $7k/16-1\geq k/32+k'$ paths of $\pset(C)$ terminate at the vertices of $\Gamma(Z)$, and all paths in $\pset(C)$ are internally disjoint from $C$. We let $\pset_1$ be any subset of $k/32$ paths terminating at the vertices of $Z$, that originate from the vertices of $T_1$. Recall that $|\Gamma(C')\cap \Gamma(C)|\geq 3k/64+1\geq k/32+k'+1$. Therefore, at least $k/32$ paths of $\pset(C)$ originate from the vertices of $T_1$ and terminate at the vertices of $C'$. We let $\pset_2\subseteq \pset(C)$ be any set of $k/32$ such paths. It is now easy to see that $(X',Y',\pset_1,\pset_2)$ is a valid weak $2$-cluster chain.

\paragraph{Case 2: $\mathbf{\alpha\geq 1/5}$, but $\mathbf{C}$ does not have the $\mathbf{(k/4,1)}$-bandwidth property.}
We say that a partition $(Z,Z')$ of $C$ is a \emph{sparse cut}, iff the following condition holds: 
\[|E(Z,Z')|<\min\set{|\Gamma(C)\cap Z|,|\Gamma(C)\cap Z'|}.\]

 Since we have assumed that $C$ does not have the $(k/4,1)$-bandwidth property, there is some sparse cut $(Z,Z')$ of $C$, with $|E(Z,Z')|<k/4$. Let $r$ be the smallest value $|E(Z,Z')|$ among all sparse cuts $(Z,Z')$ of $C$, so $r<k/4$, and let $\rho'=(r+1)/|\Gamma(C)|$. Finally, let $(A',B')$ be the minimum $\rho'$-balanced cut of $C$ with respect to $\Gamma(C)$, and assume w.l.o.g. that $|\Gamma(C)\cap A'|\geq |\Gamma(C)\cap B'|$. From the above discussion, $|E(A',B')|=r<k/4$. We need the following claim\ifabstract, whose proof is omitted from this extended abstract\fi.

\begin{claim}\label{claim: A is well-linked}
Set $A'$ has the $1/11$-bandwidth property.
\end{claim}

%---------------------------------------
%---------------------------------------
\iffull
%---------------------------------------
%---------------------------------------
\begin{proof}
Consider any partition $(J,J')$ of $A'$. Let $S=\Gamma(C)\cap J$,  $S'=\Gamma(C)\cap J'$, and let $U\subseteq J$, $U'\subseteq J'$ be the subsets of vertices incident on the edges of $E(A',B')$. Let $E'=E(J,J')$. It is enough to prove that $|E'|\geq\frac 1 {11} \min\set{|S\cup U|,|S'\cup U'|}$. Assume w.l.o.g. that $|S'|\leq |S|$, and consider two cases (see Figure~\ref{fig: balanced-well-linked2}).

\begin{figure}
\scalebox{0.35}{\includegraphics{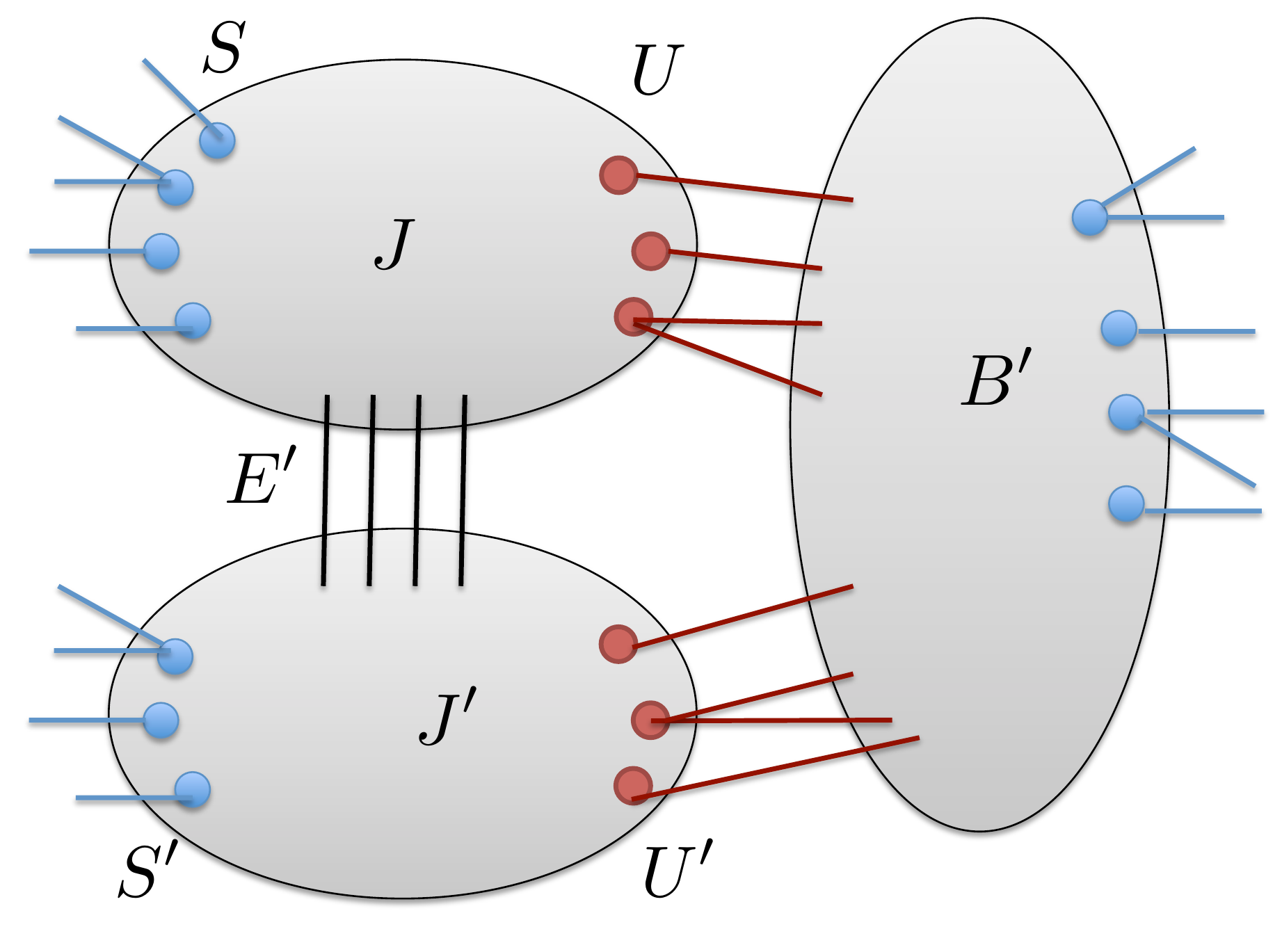}}\caption{Illustration for the proof of Claim~\ref{claim: A is well-linked}.\label{fig: balanced-well-linked2}}
\end{figure}

The first case is when $|S|\geq r+1$. In this case, $(J,C\setminus J)$ is a $\rho'$-balanced cut of $C$ with respect to $\Gamma(C)$, and so $|E(J,C\setminus J)|=|\out(J)|=|E'|+|E(J,B')|\geq r$ must hold. Since $|E(J,B')|+|E(J',B')|=r$, it follows that $|E'|\geq |E(J',B')|$. From the $1/5$-bandwidth property of $C$, $|\out(J')|=|E'|+|E(J',B')|\geq |S'|/5$, and so 

\[\frac{|S'|+|U'|}{5}\leq \frac{|S'|+|E(J',B')|}{5}\leq |E'|+\frac 6 5|E(J',B')|
\leq \frac{11}{5}|E'|.\]
We conclude that $|E'|\geq (|S'|+|U'|)/11$ in this case.

We now assume that $|S|<r+1$, and so $|S'|<r+1$ as well. From our definition of $r$, $S$ is not a sparse cut, and neither is $S'$. In particular:

\[|E'|+|E(J,B')|\geq |S|,\]

and

\[|E'|+|E(J',B')|\geq |S'|.\]

Combining these two inequalities together, we get that $2|E'|+|E(A',B')|\geq |S|+|S'|$. However, $|E(A',B')|\leq r\leq k/4\leq |\Gamma(C)|/3$ (since $|\Gamma(C)|\geq 7k/8$), while $|S|+|S'|=|\Gamma(C)\cap A'|\geq |\Gamma(C)|/2$. We conclude that $|E'|\geq \half|\Gamma(C)|\left(\frac 1 2 -\frac 1 3\right)=\frac{1}{12}|\Gamma(C)|$. Notice that, since $(A',B')$ is a sparse cut, $|E(A',B')|<|\Gamma(C)\cap B'|$, and so $|\Gamma(A')|<|\Gamma(C)|$. Therefore, $\min\set{|S\cup U|, |S'\cup U'|}\leq |\Gamma(A')|/2<|\Gamma(C)|/2$, and $|E'|\geq \frac{1}{6}\min\set{|S\cup U|, |S'\cup U'|}$.
\end{proof}
%---------------------------------------
%---------------------------------------
\fi
%---------------------------------------
%---------------------------------------

We now consider two subcases. The first subcase happens when $A'$ is a good or a perfect cluster. It is easy to see that $|\out(A')|<|\out(C)|\leq k+k'$ in this case, and we return $C'=A'$. 

The second subcase happens when $A'$ is not a good or a perfect cluster.  
In this case, we construct a $2$-cluster chain in $G$, similarly to Case 1c. We let $X'=A'$. From the above discussion $X'$ has the $1/11\geq \alpha^*$-bandwidth property, and so it has the $(k'',\alpha^*)$-bandwidth property. From Observation~\ref{obs: cutting good cluster}, $|\out(B')\cap \out(C)|\geq 27k/80$ must hold, and, since every vertex of $\Gamma(C)$ is incident on exactly one edge of $\out(C)$, we get that $|\Gamma(B')\cap \Gamma(C)|\geq |\out(B')\cap \out(C)|\geq 27k/80$.

 We then apply Theorem~\ref{thm: balanced-cut-large-smaller-side} to the partition $(A',B')$, to obtain a cluster $C'\subseteq B'$, that has the $\alpha^*$-bandwidth property, and we set $Y'=C'$. In order to define the sets $\pset_1$ and $\pset_2$ of paths, observe that all but at most one vertices of $\Gamma(C)$ have a path of $\pset(C)$ terminating at them, since $p(C)\geq |\Gamma(C)|-1$. Since we have assumed that $|\Gamma(C)|\geq 7k/8$, $|\Gamma(A')\cap \Gamma(C)|\geq |\Gamma(C)|/2\geq 7k/16$. Therefore, at least $7k/16-1\geq k/32+k'$ paths of $\pset(C)$ terminate at the vertices of $\Gamma(A')$, and all paths in $\pset(C)$ are internally disjoint from $C$. We let $\pset_1$ be any subset of $k/32$ paths terminating at the vertices of $A'$, that originate from the vertices of in $T_1$. Recall that $|\Gamma(C')\cap \Gamma(C)|\geq 3k/64+1\geq k/32+k'+1$. Therefore, at least $k/32$ paths of $\pset(C)$ originate from the vertices of $T_1$ and terminate at the vertices of $C'$. We let $\pset_2\subseteq \pset(C)$ be any such set of $k/32$ paths. It is now easy to see that $(X',Y',\pset_1,\pset_2)$ is a valid weak $2$-cluster chain.

\paragraph{Case 3: $\mathbf{C}$ has the $\mathbf{(k/4,1)}$-bandwidth property.}
Observe that in Cases 1 and 2, whenever we did not construct a $2$-cluster chain, we returned a good or a perfect cluster $C'\subsetneq C$ with the $1/23$-bandwidth property, such that $|\out(C')|<|\out(C)|$. We will use this fact later.

From Lemma~\ref{lem: deletable edge}, $p(C)=|\Gamma(C)|$, and so $|\Gamma(C)|\leq k+k'$. Since every vertex of $\Gamma(C)$ is incident on exactly one edge of $\out(C)$, we get that $|\out(C)|\leq k+k'$. We claim that there is some vertex $v\in \Gamma(C)$, such that $v$ is a non-separating vertex for $G[C]$. Indeed, assume otherwise, and let $\cset$ be the set of all connected components of $G[C]\setminus\Gamma(C)$. Then there must be some component $R\in \cset$, so that exactly one vertex $v\in \Gamma(C)$ has an edge connecting $v$ to a vertex of $R$. Let $u\in V(R)$ be any vertex. Since $V(R)\cap T=\emptyset$, and $v$ separates $R$ from $T$, it is easy to see that $T_1$ remains $(k/4,1)$-well-linked, and $(T_1,T_2)$ remain $\half$-linked in $G\setminus\set{u}$, contradicting the minimality of $G$. Let $v$ be any vertex in $\Gamma(C)$, such that $v$ is not a separator vertex for $G[C]$. Consider the cluster $C'=C\setminus\set{v}$.

We start by observing that $C'$ has the $1/23$-bandwidth property, in the following claim, whose proof uses standard techniques and is \iffull deferred to the Appendix.\fi \ifabstract omitted from this extended abstract.\fi

\begin{claim}\label{claim: C' has bw prop}
Cluster $C'$ has the $1/23$-bandwidth property.
\end{claim}

From Observation~\ref{obs: cutting good cluster}, cluster $C'$ is either a good or a perfect cluster. 
Moreover, it is easy to see that $|\out(C')|\leq |\out(C)|+1\leq k+k'+1$, and $p(C')\geq |\Gamma(C')|-1$ (there are at most two vertices that belong to $\Gamma(C')\setminus \Gamma(C)$ - the neighbors of $v$ in $C$; we can extend the path of $\pset(C)$ terminating at $v$ to terminate at one of these vertices). If $C'$ is a perfect cluster, then from Theorem~\ref{thm: perfect cluster gives chain}, $G$ contains a $2$-cluster chain. Therefore, we assume that $C'$ is a good cluster. We now consider three subcases.

The first sub-case happens when $C'$ does not have the $(k/4,1)$-bandwidth property, and every vertex of $\Gamma(C')$ is incident on exactly one edge of $\out(C')$. In this case, cluster $C'$ is a valid input to Theorem~\ref{thm: iteration}, where it falls under Case 1 or Case 2. In each of these cases, we either showed that $G$ contains a $2$-cluster chain, or produced a good or a perfect cluster $C''\subsetneq C'$, that has the $1/23$-bandwidth property, and $|\out(C'')|<|\out(C')|\leq |\out(C)|+1$, so $|\out(C'')|\leq |\out(C)|$. We can then return the cluster $C''$.

The second sub-case happens when at least one vertex $u$ of $\Gamma(C')$ is incident on two edges of $\out(C')$. In this case, we consider the cluster $C''=C'\setminus \set{u}$. It is easy to see that $C''$ still has the $1/23$-bandwidth property, and from Observation~\ref{obs: cutting good cluster}, it is a good or a perfect cluster. Moreover, $C''\subsetneq C$, and $|\out(C'')|<|\out(C')|\leq |\out(C)|+1$, so $|\out(C'')|\leq |\out(C)|$. We then return cluster $C''$.

The third sub-case happens when every vertex of $\Gamma(C')$ is incident on exactly one edge of $\out(C')$, and $C'$ has the $(k/4,1)$-bandwidth property. In this case, from Lemma~\ref{lem: deletable edge}, $p(C')=|\Gamma(C')|=|\out(C')|$. However, since $C'\subsetneq C$, $p(C')\leq |\Gamma(C)|$ must hold. Therefore, $|\out(C')|\leq |\Gamma(C)|\leq |\out(C)|$, and we return $C'$.

\label{------------------------------------------------------End of New Proof----------------------------------}

\label{-------------------------------------------Tree-of-Sets System---------------------------------------------------}
\section{Tree-of-Sets System}\label{sec: ToS}

A \ToS is defined very similarly to the \PoS, except that instead of being organized into a path-like structure, the clusters are organized into a tree-like structure. We would also like to ensure that the resulting \ToS is \emph{anchored} - that is, if the original graph $G$ contains a set $\tset$ of terminals that are sufficiently well-linked, then there is some large enough set $\pset^*$ of node-disjoint paths, connecting one of the clusters of the \ToS to the terminals, so that the paths in $\pset^*$ are internally disjoint from all vertices participating in the \ToS. In order to make the notation convenient, we add one special vertex $v_0$ to the tree corresponding to the \ToS, whose vertex set $S(v_0)$ only contains a subset of the terminals, and we do not impose well-linkedness requirements on this cluster. We now formally define an anchored \ToS.

%------------------------------------------------------------------------
%------------------------------------------------------------------------
%------------------------------------------------------------------------
%------------------------------------------------------------------------
%------------------------------------------------------------------------

\begin{definition} Let $G$ be any graph, and let $\tset\subseteq V(G)$ be any subset of vertices of $G$ called terminals. A $T$-anchored \ToS of width $w$ and size $\ell$ in $G$ consists of the following:

\begin{itemize}
\item a tree $\tau$ with $|V(\tau)|=\ell+1$ vertices, whose root vertex $v_0$ has degree $1$;

\item for every vertex $v\in V(\tau)$, a subset $S(v)\subseteq V(G)$ of vertices of $G$, such that, if $v=v_0$ then $S(v)\subseteq \tset$, and otherwise $G[S(v)]$ is connected. Moreover, all resulting vertex sets $\set{S(v)}_{v\in V(\tau)}$ are disjoint; and

\item for every edge $e=(u,v)\in E(\tau)$, a set $\pset(e)$ of $w$ node-disjoint paths in graph $G$, where each path $P\in \pset(e)$ connects a vertex of $S(u)$ to a vertex of $S(v)$, and it is internally disjoint from $\bigcup_{x\in V(\tau)}S(x)$. The paths in $\pset=\bigcup_{e\in E(\tau)}\pset(e)$ must be all mutually disjoint.
\end{itemize}

For every vertex $v\in V(\tau)$ and edge $e\in \delta_{\tau}(v)$, let $U(v,e)\subseteq S(v)$ be the set of all endpoints of the paths of $\pset(e)$ that belong to $S(v)$. Let $U(v)=\bigcup_{e\in \delta_{\tau}(v)}U(v,e)$. We say that the \ToS is perfect, if for every vertex $v\in V(\tau)\setminus \set{v_0}$:

\begin{itemize}
\item for every edge $e\in \delta_{\tau}(v)$, set $U(v,e)$ of vertices is node-well-linked in $G[S(v)]$; 

\item for every pair $e\neq e'\in \delta_{\tau}(v)$ of edges, sets $U(v,e),U(v,e')$ of vertices are node-linked in $G[S(v)]$; and
\item the set $U(v)$ is $1/5$-well-linked in $G[S(v)]$.
\end{itemize}
\end{definition}

%\mynote{Should be rooted: the root of the tree can connect to terminals w many paths, and the paths are disjoint from the system.}
%A Tree-of-Sets System of width $h$ and height $r$ is simply an $h$-wide embedding of any tree $\tau^*$ with $|V(\tau^*)|=r$ into the graph $G$, where the embedding is required to be perfect\footnote{We note that in some previous work, e.g.~\cite{CC14}, Tree-of-Sets System only referred to an $h$-wide embedding of a sub-cubic tree. For our purposes any tree works, so we use this slightly different definition.}. 
The main theorem of this section is the following.

\begin{theorem}\label{thm: build a ToS system} There is a universal constant $c$ such that the following holds.
Suppose we are given any graph $G$ with maximum vertex degree at most $d$, and a set $\tset$ of $k>2$ vertices called terminals, such that the terminals are $1$-well-linked in $G$. Then for all integers $\ell\geq 2,w\geq d$, with $w\ell^5\leq \frac{k}{cd^8}$, there is a $T$-anchored perfect \ToS of size $\ell$ and width $w$ in $G$.
\end{theorem}

The proof roughly follows the outline of the proof of the Excluded Grid Theorem of~\cite{CC14}, but it is much simpler (partly because we only provide a non-constructive version here, partly because we stop short of turning the Tree-of-Sets System into a Path-of-Sets System, and partly because some arguments have been simplified). As in~\cite{CC14}, we follow the bottom-up approach, that is, we first build a collection of $\ell$ disjoint \emph{good routers} - clusters whose boundaries are reasonably well-linked, such that each cluster can send a large amount of flow to the terminals. We then ``organize'' these good routers into a \ToS. We note that in~\cite{CC14}, the good routers are first organized into a \ToS, then a subset of the clusters is organized into a Tree-of-Sets System where the corresponding tree is sub-cubic, and finally they are organized into a \PoS. Much of the loss in the width and length parameters of the \PoS occurs in the last two steps. Here we both exploit the fact that we only need to construct a \ToS, and simplify parts of the proof of~\cite{CC14}. We break the proof of Theorem~\ref{thm: build a ToS system} into two steps. The first step, described in Section~\ref{subsec: find the routers}, constructs the good routers, and the second step, that appears in Section~\ref{subsec: build the ToS System}, organizes them into a \ToS.

Throughout the proof, we denote $n=|V(G)|$.

%------------------------------------------------------------------------
%------------------------------------------------------------------------
%------------------------------------------------------------------------
\subsection{Building the Routers}\label{subsec: find the routers}

%------------------------------------------------------------------------
%------------------------------------------------------------------------
%------------------------------------------------------------------------

The main objects studied in this section are good routers, that are defined below.
We use a parameter $w'=\frac{cwd^8\ell^3}{384}$, and we assume that $c$ is a large enough even integer, so $w'>w$, and $w'$ is even. Notice that that $k\geq 384w'\ell^2$. We set $\alpha=\frac{1}{2^{12}\cdot \ell}$.

\begin{definition}
A subset $S\subseteq V(G)$ of vertices is called a \emph{good router}, iff (i) $|S\cap \tset|\leq |\tset|/2$; (ii) there is a set of $w'/2$ edge-disjoint paths connecting the vertices of $\Gamma(S)$ to the vertices of $T\setminus S$ in $G$, and (iii) $\Gamma(S)$ is $(w',\alpha)$-well-linked in $G[S]$.
\end{definition}

The main result of this section is the following theorem.

\begin{theorem}\label{thm: build good routers}
Given a graph $G$ and a set $\tset$ of terminals as in Theorem~\ref{thm: build a ToS system}, there is a collection $\sset=\set{S_1,\ldots,S_{\ell}}$ of disjoint good routers in $G$, such that $|T\setminus(\bigcup_{i=1}^{\ell}S_i)|\geq |T|/(\ell+1)$.
\end{theorem}

The rest of this section is dedicated to proving Theorem~\ref{thm: build good routers}. As in previous work, we use good clusterings and corresponding contracted graphs, but they are defined slightly differently.

\subsubsection{Vertex Clusterings and Contracted Graphs} 
We say that a cluster $C\subseteq V(G)$ is \emph{large} if $|\out(C)|\geq w'$, and we say that it is small otherwise.

\begin{definition}
A subset $C\subseteq V(G)$ of vertices is called a \emph{good cluster} iff $G[C]$ is connected; $C$ is a small cluster, and $|C\cap T|\leq |T|/2$. 
\end{definition}

A partition $\cset$ of the vertices of $G$ into disjoint subsets is called a \emph{clustering of $G$}, iff for each $C\in \cset$, $G[C]$ is connected. We say that a clustering $\cset$ of $G$ is \emph{good} iff every cluster of $\cset$ is a good cluster.

We note that, since $w'\geq w\geq d$ and $|T|>2$, a cluster $C$ consisting of a single vertex is a good cluster, and therefore there is a good clustering $\cset$ of $V(G)$, where each cluster contains a single vertex of $G$. Given any good clustering $\cset$ of $V(G)$, the corresponding \emph{contracted graph} $G_{\cset}$ is obtained from $G$ by contracting the vertices of every cluster $C\in \cset$ into a supernode $v_C$, and deleting all loops (the parallel edges stay in the graph). We need the following claim.

\begin{claim}\label{claim: m is large}
If $G_{\cset}$ is a contracted graph corresponding to some good clustering $\cset$, then $|E(G_{\cset})|\geq k/3$.
\end{claim}

\begin{proof}
If there is some cluster $C\in \cset$, with $|C\cap T|\geq k/3$, then, since the terminals are $1$-well-linked in $G$, and $|C\cap T|\leq k/2$, there must be at least $k/3$ edges in $\out(C)$, and hence in $E(G_{\cset})$.

Assume now that for every cluster $C\in \cset$, $|C\cap T|< k/3$. Then we can find a partition $(\aset,\bset)$ of $\cset$, so that $\sum_{C\in \aset}|T\cap C|,\sum_{C\in \bset}|T\cap C|\geq k/3$, using a simple greedy algorithm: start with $\aset=\bset=\emptyset$, and process the clusters $C\in \cset$ one-by-one. When $C$ is processed, we add it to $\aset$ if $\sum_{C\in \aset}|T\cap C|<\sum_{C\in \bset}|T\cap C|$, and we add it to $\bset$ otherwise. Since $|C\cap T|< k/3$ for all $C\in \cset$, at the end of this procedure, $\left |\sum_{C\in \aset}|T\cap C|-\sum_{C\in \bset}|T\cap C|\right |\leq k/3$, and so $\sum_{C\in \aset}|T\cap C|,\sum_{C\in \bset}|T\cap C|\geq k/3$. We now let $A=\bigcup_{C\in \cset}C$ and $B=\bigcup_{C\in \cset}C$. From the well-linkedness of the terminals, $|E_G(A,B)|\geq k/3$, and all edges of $E_G(A,B)$ belong to $G_{\cset}$.
\end{proof}

Given {\bf any} clustering $\cset$ of $G$, we define a potential $\phi(\cset)$ for this clustering, somewhat similarly to~\cite{ChuzhoyL12,CC14}. The idea is that $\phi(\cset)$ will serve as a reasonably tight bound on the number of edges connecting the different clusters in $\cset$. At the same time, it is designed in a way that allows us to perform a number of useful operations on the current clustering, without increasing the potential.

We start by defining potentials $\rho(z)$ for integers $z>0$. Let $n_0=w'$, and for $i>0$, let $n_i=\left(\frac 3 2\right )^i n_0$.  We set $\rho(n_0)=4\alpha$, and for $i>0$, $\rho(n_i)=4\frac{\alpha w'}{n_i}+\rho(n_{i-1})$. Notice
that for all $i$, $\rho(n_i)\leq 12\alpha$.  Next, we partition all integers $z\geq n_0$ into sets $Z_1,Z_2,\ldots$, where set $Z_i$ contains all
integers $z$ with $n_{i-1}\leq z< n_i$. For $z\in Z_i$, we define
$\rho(z)=\rho(n_{i-1})$, and for all $z<w'$, we define $\rho(z)=0$. Clearly, for all $z$, $\rho(z)\leq 12\alpha$.

Given a clustering $\cset$ of $G$, we define a potential $\phi(e)$ for every edge of $G$ with respect to $\cset$, as follows. Let $e=(u,v)$, and let $C,C'\in \cset$ such that $v\in C$ and $u\in C'$. If $C=C'$, then $\phi(e)=0$. Otherwise, let $|\out(C)|=z$ and $|\out(C')|=z'$. Then we set $\phi(e)=1+\rho(z)+\rho(z')$.
The following observation is now immediate.

\begin{observation}
 For any partition $\cset$ of the vertices of $G$ and for any edge $e=(u,v)\in E(G)$, if $u,v$ belong to the same cluster of $\cset$, then $\phi(e)=0$. Otherwise,  $1\leq \phi(e)\leq 1.1$.   
\end{observation}

Finally, we define a potential $\phi(\cset)$ of any clustering $\cset$ of $G$, as the sum of all potentials of the edges $e\in E$.

Assume that we are given any clustering $\cset$ of $G$. We define two operations,  each of which produces a valid new  clustering of $G$, whose potential is strictly smaller than $\phi(\cset)$. We note that this part is almost identical to what appeared in~\cite{ChuzhoyL12,CC14}.

\paragraph{Action 1: Partitioning a large cluster.}
Suppose we are given a large cluster $C$, and let $\Gamma=\Gamma_G(C)$. Recall that a partition $(X,Y)$ of $C$ is a $(w',\alpha)$-violating partition with respect to $\Gamma$, iff $|E(X,Y)|<\alpha\cdot \min\set{|X\cap\Gamma|,|Y\cap \Gamma|,w'}$.

Suppose we are given any clustering $\cset$ of $G$, a large cluster $C\in \cset$, and a $(w',\alpha)$-violating partition $(X,Y)$ of $C$. 
In order to perform this operation, we remove $C$ from $\cset$. For every connected component $R$ of $G[X]$, we add $V(R)$ to $\cset$, and we do the same with every connected component of $G[Y]$, to obtain a new clustering $\cset'$.
We denote this operation by $\partition(C,X,Y)$.

\begin{claim}\label{claim: bound on potential for partition}
Let $\cset'$ be the outcome of operation $\partition(C,X,Y)$. Then $\phi(\cset')<\phi(\cset)-\frac{1}{10n}$. 
\end{claim}

\begin{proof}
Assume without loss of generality that $|\out(X)|\leq |\out(Y)|$.  Let
$z=|\out(C)|$, $z_1=|\out(X)|$, $z_2=|\out(Y)|$, so $z_1=|\out(X)\cap \out(C)|+|E(X,Y)|< |\out(C)|/2+\alpha |\Gamma(C)|\leq 
2z/3$. Assume that $z\in Z_i$. Then either $z_1\in Z_{i'}$ for $i'\leq
i-1$, or $z_1<n_0$. The potential of the edges in $\out(Y)\cap
\out(C)$ does not increase. The only other changes in the potential
are the following: the potential of each edge in $\out(X)\cap \out(C)$
decreases by at least $\rho(z)-\rho(z_1)$, and the potential of every edge in
$E(X,Y)$ increases from $0$ to at most $1.1$. We consider two cases.

First, if $z_1<n_0$, then $\rho(z_1)=0$, while $\rho(z)\geq 4\alpha$. So the
potential of each edge in $\out(X)\cap \out(C)$ decreases by at least
$4\alpha$, and the overall decrease in potential due to these edges is
at least $4\alpha|\out(X)\cap \out(C)|$. The total increase in
potential due to the edges in $E(X,Y)$ is bounded by
$1.1|E(X,Y)|<1.1\alpha|\Gamma\cap X|\leq 1.1\alpha|\out(X)\cap \out(C)|$,
so the overall potential decreases by at least $2\alpha|\out(X)\cap
\out(C)|>\frac{1}{10n}$.

The second case is when $z_1\geq n_0$. Assume that $z_1\in
Z_{i'}$. Then $n_{i'}\leq 3z_1/2$, and, since $i'\leq i-1$ must hold,
$\rho(z)\geq \frac{4\alpha w'}{n_{i'}}+\rho(n_{i'-1})=\frac{4\alpha
  w'}{n_{i'}}+\rho(z_1)\geq \frac{8\alpha w'}{3z_1}+\rho(z_1)$. So
the potential of each edge in $\out(X)\cap \out(C)$ decreases by at
least $\frac{8\alpha w'}{3z_1}$, and the total decrease in potential
due to these edges is at least $\frac{8\alpha
  w'}{3z_1}\cdot|\out(X)\cap \out(C)|\geq \frac{4\alpha w'}{3}$,
since $|\out(X)\cap \out(C)|\geq z_1/2$. The total increase in the
potential due to the edges in $E(X,Y)$ is bounded by
$1.1|E(X,Y)|<1.1\alpha w'$, since $|E(X,Y)|\leq \alpha
w'$. Overall, the total potential decreases by at least
$0.2\alpha w'>\frac{1}{10n}$.
\end{proof}

\paragraph{Action 2: Separating a large cluster.}
Let $\cset$ be any clustering of $G$, and let $C\in \cset$ be a large cluster in $\cset$. Assume further that we are given a partition $(A,B)$ of $V(G)$, with $C\sse A$, $|\tset\cap B|\geq |\tset|/2$, and $|E_G(A,B)|< w'/2$. We perform the following operation, that we denote by $\separate(C,A)$.

Consider any cluster $S\in \cset$. If $S\setminus A\neq \emptyset$, and $|\out(S\setminus A)|>|\out(S)|$, then we modify $A$ by removing all vertices of $S$ from it. Notice that in this case, the number of edges in $E(S)$ that originally contributed to the cut $(A,B)$,  $|E(S\cap A,S\cap  B)|>|\out(S)\cap E(A)|$ must hold, so $|\out(A)|$ only goes down as a result of this modification. We assume from now on that if $|S\setminus A|\neq \emptyset$, then $|\out(S\setminus A)|\leq |\out(S)|$. In particular, if $S$ is a small cluster, and $S\setminus A\neq\emptyset$, then $S\setminus A$ is also a small cluster.

We construct a new clustering $\cset'$ of $G$ as follows. First, for every connected component $R$ of $G[A]$, we add $V[R]$ to $\cset$. Notice that all these clusters are small, as $|\out(A)|=|E(A,B)|<w'/2$. Next, for every cluster $S\in\cset$, such that $S\setminus A\neq \emptyset$, for every connected component $R$ of $G[S\setminus A]$, we add $V[R]$ to $\cset'$.  
Notice that for every cluster $S\in \cset'$, either $S\subseteq A$, or there is some cluster $S'$ in the original partition $\cset$ with $S\subseteq S'$ and $|\out(S')|\geq |\out(S)|$. In particular, 
if $S\in\cset'$ is a large cluster, then $S\not\subseteq A$, and the cluster $S'\in \cset$ containing $S$ is also a large cluster. It is easy to see that $\cset'$ is a valid clustering of $G$.

%The following claim was proved in~\cite{ChuzhoyL12}.
\begin{claim}\label{claim: bound on potential for separation}
 Let $\cset'$ be the outcome of operation  $\separate(C,A)$. Then $\phi(\cset')\leq \phi(\cset)-1$. 
 \end{claim}

 \begin{proof}

We can bound the changes in the potential as follows:

\begin{itemize}
\item Every edge in $\out(A)$ contributes at most $1.1$ to the potential of $\cset'$, and there are at most $ \frac{w'-1} 2$ such edges. %These are the only edges whose potential in $\cset'$ may be
%  greater than their potential in $\cset$.

\item Every edge in $\out(C)$ contributed at least $1$ to the
  potential of $\cset'$, and there are at least $w'$ such edges,
  since $C$ is a large cluster.
\end{itemize}

For every other edge $e$, the potential of $e$ does not increase. Indeed, let $e=(u,v)$, where $u\in S_1$, $v\in S_2$, with $S_1,S_2\in \cset'$, and $S_1,S_2\not\subseteq A$. Then there are clusters $S_1',S_2'\in \cset$, with $S_1\subseteq S_1'$ and $S_2\subseteq S_2'$. Notice that $S_1'\neq S_2'$, since $S_1$ and $S_2$ correspond to connected components of $S_1'$ and $S_2'$, respectively, and so no edge can connect them. From our construction of $\cset'$, $|\out(S_1)|\leq |\out(S_1')|$ and $|\out(S_2)|\leq |\out(S_2')|$, so the potential of $e$ cannot increase.
Therefore, the total decrease in the potential is at least $w'-\frac{1.1(w'-1)}2\geq 1$.
\end{proof}

\iffalse
To summarize, given any acceptable clustering $\cset$ of the vertices
of $G$, let $E'$ be the set of edges whose endpoints belong to
distinct clusters of $\cset$. Then $|E'|\leq \phi(\cset)\leq
1.1|E'|$. So the potential is a good estimate on the number of edges
connecting the different clusters. We have also defined two actions on
large clusters of $\cset$: $\partition(C,X,Y)$ can be performed if we are given a large cluster $C\in \cset$ and an
$(h_0,\alpha)$-violating partition $(X,Y)$ of $C$, and $\separate(C,A)$, where
$(A,V(G)\setminus A)$ is a partition of $V(G)$ with $|\out(A)|< h_0/2$, separating a large cluster $C$
from the terminals. Each such action returns a new acceptable
clustering, whose potential goes down by at least $1/n$. Both operations ensure that if $S$ is any large cluster in the new clustering, then there is some
large cluster $S'$ in the original clustering with $S\subseteq S'$.

\fi

Among all good clusterings $\cset$ of $G$, we select one minimizing $\phi(\cset)$. We denote $H=G_{\cset}$, and $m=|E(H)|$. From Claim~\ref{claim: m is large}, $m\geq k/3$. We now show how to find the family $\sset$ of $\ell$ disjoint good routers. This is done in two steps. First, we compute a collection $X_1,\ldots,X_{\ell+1}$ of disjoint subsets of vertices in $H$, such that for each $1\leq j\leq \ell+1$, $|E_H(X_j)|\geq \frac{|\out_H(X_j)|}{128\ell}$. For each $1\leq j\leq \ell+1$, let $X'_j\subseteq V(G)$ be obtained by un-contracting all clusters corresponding to the vertices of $X_j$, so $X'_j=\bigcup_{v_C\in X_j}C$. We then find, for each $1\leq j\leq \ell+1$ with $|T\cap X'_j|\leq |T|/2$, a good router $S_j\subseteq X_j'$.

%------------------------------------------------------------------------
%------------------------------------------------------------------------
\subsubsection{Initial Partition of $V(H)$}

%------------------------------------------------------------------------
%------------------------------------------------------------------------
%------------------------------------------------------------------------
In this section the goal is to prove the following theorem.

\begin{theorem}\label{thm: partition into clusters w many edges inside and small boundary}
There is a collection $X_1,\ldots,X_{\ell+1}$ of disjoint subsets of vertices of $H$, such that for each $1\leq j\leq \ell+1$, $|E_H(X_j)|\geq \frac{|\out_H(X_j)|}{128\ell}$.
\end{theorem}
%------------------------------------------------------------------------
%------------------------------------------------------------------------
%------------------------------------------------------------------------

\begin{proof}
Let $\tl$ be the smallest integral power of $2$ greater than $2\ell$. Our starting point is the following lemma.

\begin{lemma}\label{lemma: partition into two}
Let $G'$ be any graph with maximum vertex degree at most $\Delta$ and $|E(G')|=m'$. Then there is a partition $(A,B)$ of $V(G')$, such that $|E(A)|,|E(B)|\geq \frac{m'}{4}-\Delta$.
\end{lemma}

\begin{proof}
For every vertex $v\in V(G')$, let $d(v)$ denote its degree. For a subset $S\subseteq V(G')$ of vertices, let $D(S)=\sum_{v\in S}d(v)$. Notice that there is a partition $(A,B)$ of $V(G')$, such that $|D(A)-D(B)|\leq \Delta$. Such a partition can be computed by a greedy algorithm: start with $A=B=\emptyset$, and process the vertices of $G'$ one-by-one. When $v$ is processed, add it to $A$ if $D(A)<D(B)$ currently, and add it to $B$ otherwise. It is easy to see that at the end of this procedure, we obtain a partition $(A,B)$ of $V(G')$ with $|D(A)-D(B)|\leq \Delta$.

Among all partitions $(A,B)$ of $V(G')$ with $|D(A)-D(B)|\leq 2\Delta$, choose one minimizing $|E(A,B)|$, and assume w.l.o.g. that $D(A)\geq D(B)$. Notice that $D(A)=2|E(A)|+|E(A,B)|$, and $D(B)=2|E(B)|+|E(A,B)|$. Since $|D(A)-D(B)|\leq 2\Delta$, $|E(A)|-|E(B)|\leq \Delta$.

For every vertex $v\in A$, let $d_1(v)$ be the number of edges incident on $v$ whose other endpoint belongs to $A$, and let $d_2(v)$ be the number of edges incident on $v$ whose other endpoint belongs to $B$. We claim that $d_1(v)\geq d_2(v)$ for every vertex $v\in A$. Indeed, assume otherwise, and consider the partition $(A',B')$ of $V(G')$, where $A'=A\setminus\set{v}$ and $B'=B\cup\set{v}$. It is easy to see that $|E(A',B')|<|E(A,B)|$, while $|D(A')-D(B')|\leq 2d(v)\leq 2\Delta$, a contradiction.

Therefore, $|E(A)|=\half\sum_{v\in A}d_1(v)\geq \half\sum_{v\in A}d_2(v)=\half |E(A,B)|$.

Altogether, $m'=|E(A)|+|E(B)|+|E(A,B)|\leq 4|E(A)|$, and $|E(A)|\geq m'/4$. From the above discussion, $|E(B)|\geq |E(A)|-\Delta\geq m'/4-\Delta$.
\end{proof}

We perform $\log \tl$ iterations, where in iteration $i$ we start with some partition $\xset_{i-1}$ of $V(H)$ into $2^{i-1}$ subsets (the initial set, $\xset_0$ contains a single set $X=V(H)$). An iteration is executed by applying Lemma~\ref{lemma: partition into two} to each graph $H[X]$, for $X\in \xset_i$ in turn, obtaining a partition $(X',X'')$ of $X$, such that $|E_H(X')|,|E_H(X'')|\geq |E_H(X)|/4-w'$ (since maximum vertex degree in $H$ is bounded by $w'$). We then add all resulting clusters $X',X''$ for $X\in \xset_{i-1}$ to the new partition $\xset_{i}$, that becomes an input to the next iteration. It is easy to see that the final partition $\xset^*=\xset_{\log \tl}$ contains $\tl$ vertex subsets, and for each set $X\in \xset^*$, $|E(X)|\geq \frac{m}{4^{\log \tl}}-2w'\geq \frac{m}{\tl^2}-2w'\geq \frac{m}{64\ell^2}$, since $\tl\leq 4\ell+2$, and $w'\leq \frac{k}{384 \ell^2}\leq \frac{m}{128\ell^2}$.

Since there are $m$ edges in $H$, $\sum_{X\in \xset^*}|\out_H(X)|\leq 2m$. We say that a set $X\in \xset^*$ is bad if $|\out_H(X)|> 4m/\tl$. Clearly, at most $\tl/2$ clusters are bad. Let $\xset\subseteq \xset^*$ be the subset of clusters that are not bad. Then for each cluster $X\in \xset$, $|\out(X)|\leq 4m/\tl\leq 2m/\ell$, while $|E(X)|\geq \frac{m}{64\ell^2}$, so $|E(X)|\geq \frac{|\out(X)|}{128\ell}$ as required. We discard sets from $\xset$ as necessary, until $|\xset|=\ell+1$ holds.
\end{proof}

\subsubsection{Finding the Routers}

Let $X_1,\ldots,X_{\ell+1}$ be the subsets of $V(H)$ computed in Theorem~\ref{thm: partition into clusters w many edges inside and small boundary}.
For each $1\leq j\leq \ell+1$, let $X'_j\subseteq V(G)$ be obtained from $X_j$ by un-contracting all its clusters, that is, $X'_j=\bigcup_{v_C\in X_j}C$. We assume without loss of generality that $X_{\ell+1}'$ contains the most terminals of $T$ among all sets $X'_j$, that is, $|X'_j\cap T|\leq |X'_{\ell+1}\cap T|$ for all $1\leq j\leq \ell$. In particular, for each $1\leq j\leq \ell$, $|X'_j\cap T|\leq |T|/2$, and $|T\setminus(\bigcup_{j=1}^{\ell}X'_j|)\geq |T|/(\ell+1)$. 
Our final step, that finishes the proof of Theorem~\ref{thm: build good routers}, is summarized in the following theorem.

%--------------------------------------------------
%--------------------------------------------------
%--------------------------------------------------
\begin{theorem}\label{thm: a router in each set}
For each $1\leq j\leq \ell$, there is a good router $S_j\subseteq X_j'$.
\end{theorem}
%--------------------------------------------------
%--------------------------------------------------
%--------------------------------------------------
\begin{proof}
Throughout the proof, we fix some $1\leq j\leq \ell$.  Recall that $|E_H(X_j)|\geq \frac{|\out_H(X_j)|}{128\ell}$, and $|X_j\cap T|\leq |T|/2$. Clearly, $|E_G(X'_j)|\geq |E_H(X_j)|\geq \frac{|\out_H(X_j)|}{128\ell} = \frac{|\out_G(X_j')|}{128\ell}$. Let $\cset'$ be the clustering of $G$, obtained as follows: first, we add to $\cset'$ all clusters $C\in \cset$ with $v_C\not\in X_j$; next, for every connected component $R$ of $G[X'_j]$, we add $V(R)$ to $\cset'$. It is easy to see that $\cset'$ is a valid clustering of $G$ (though it may not be a good clustering). Moreover, if $C\in \cset'$ is a large cluster, then $C\subseteq X'_j$. 

\begin{claim}
$\phi(\cset')\leq \phi(\cset)-0.5$.
\end{claim}

\begin{proof}
The changes of the potential from $\cset$ to $\cset'$ can be bounded as follows:

\begin{itemize}
\item The edges in $E_{H}(X_j)$ contribute at least $1$ each to $\phi(\cset)$ and contribute $0$ to $\phi(\cset'_j)$.

\item The potential of the edges in $\out_G(X'_j)$ may increase. The increase is at most $\rho(n)\leq 12\alpha$ per edge. So the total increase is at most $12\alpha|\out_{H}(X_j)|\leq 12\alpha\cdot 128 \ell |E_H(X_j)|<\frac{|E_{H}(X_j)|}{2}$, since $\alpha=\frac{1}{2^{12}\ell}$. These are the only edges whose potential may increase.
\end{itemize}

Overall, the decrease in the potential is at least $\frac{|E_{H}(X_j)|}{2}\geq 0.5$.
\end{proof}

We now perform a number of iterations that modify the clustering $\cset'$, while maintaining the following invariants:

\begin{itemize}
\item the potential of $\cset'$ decreases by at least $\frac{1}{10n}$ after each iteration;

\item  $\cset'$ remains a valid clustering of $G$, and in particular, for each $C\in \cset'$, $G[C]$ is connected;

\item for each $C\in \cset'$, $|T\cap C|\leq |T|/2$; and

\item if $C\in \cset'$ is a large cluster, then $C\subseteq X'_j$.
\end{itemize}

Notice that if the above invariants are maintained, then we are guaranteed that throughout the algorithm, there is always some large cluster in $\cset'$. Indeed, if every cluster in $\cset'$ is small, then it is easy to verify that $\cset'$ is a good clustering with $\phi(\cset')<\phi(\cset)$, contradicting the choice of $\cset$. The invariants are clearly true for the initial clustering $\cset'$.

Each iteration is executed as follows. Let $S\in \cset'$ be any large cluster. If $\Gamma(S)$ is not $(w',\alpha)$-well-linked in $G[S']$, then from Observation~\ref{obs: generalized well-linkedness alt def}, there is an $(w',\alpha)$-violating partition $(X,Y)$ of $S$. We perform operation $\partition(S,X,Y)$ to obtain the new clustering of the vertices of $G$. It is easy to verify that all invariants continue to hold. 

Assume now that $\Gamma(S)$ is $(w',\alpha)$-well-linked in $G[S']$. If there is a set of $w'/2$ edge-disjoint paths connecting the vertices of $\Gamma(S)$ to the vertices of $\tset\setminus S$ in $G$, then $S$ is a good router and we are done. Otherwise, there is a partition $(A,B)$ of $V(G)$, with $S\subseteq A$, $|\tset\cap B|\geq |\tset|/2$, and $|E(A,B)|< w'/2$. We then apply operation $\separate(S,A,B)$ to obtain the new clustering of the vertices of $G$. All invariants again continue to hold, from the analysis of the operation above. 

Since the potential of the clustering $\cset'$ decreases by at least $\frac{1}{10n}$ after each iteration, we are guaranteed that this algorithm will terminate with a good router.
\end{proof}

%------------------------------------------------------------------------
%------------------------------------------------------------------------
%------------------------------------------------------------------------
\subsection{Constructing the Tree-of-Sets System}\label{subsec: build the ToS System}

%------------------------------------------------------------------------
%------------------------------------------------------------------------
%------------------------------------------------------------------------
Let $\sset=\set{S_1,\ldots,S_{\ell}}$ be the set of the good routers constructed in the previous step. Recall that for each $1\leq i\leq \ell$, $\Gamma(S_i)$ is $(w',\alpha)$-well-linked in $G[S_i]$, and there is a set $\pset(S_i)$ of $w'/2$ edge-disjoint paths connecting the vertices of $\Gamma(S_i)$ to the vertices of $\tset\setminus S_i$ in $G$. For convenience, we denote $w'/(4d)=w_1$. Over the course of the algorithm, we will define new parameters $w_2,w_3,\ldots$, as we construct sets of paths whose cardinalities become progressively smaller.

This part consists of three steps. First, we show that we can simultaneously connect every cluster $S_i\in \sset\setminus\set{S_1}$ to cluster $S_1$ by a large number of node-disjoint paths in graph $G'$, which is obtained from $G$ by contracting every cluster of $\sset$. At the same time, we will also connect $S_1$ to the terminals of $T\setminus(\bigcup_{j=1}^{\ell}S_j)$ by a large number of disjoint paths. Next, we organize the clusters into a Tree-of-Sets system, except that it may not be perfect. Finally, we boost the well-linkedness inside each cluster to obtain a perfect $T$-anchored Tree-of-Sets System.

\paragraph{Step 1: Initial Paths}
Since the maximum vertex degree in $G$ is at most $d$, for each $1\leq i\leq \ell$, there is a subset $\pset'(S_i)\subseteq \pset(S_i)$ of $w_1$ paths, whose endpoints are all distinct.
For all $1\leq i\leq \ell$, let $\tset_i\subseteq \tset$ be the set of terminals where the paths of $\pset'(S_i)$ terminate. Since the terminals are $1$-well-linked in $G$, there is a set $\qset_i$ of $w_1$ edge-disjoint paths, connecting the terminals of $\tset_i$ to the terminals of $\tset_1$. By combining the paths in sets $\pset'(S_i),\qset_i$, and $\pset'(S_1)$, we obtain a set $\pset''(S_i)$ of $w_1$ paths in $G$ that connect the vertices of $S_i$ to the vertices of $S_1$ with edge-congestion at most $3$.
Recall that $|T\setminus(\bigcup_{j=1}^{\ell}S_j|)\geq |T|/(\ell+1)>w'>w_1$. Let $T^*\subseteq T\setminus(\bigcup_{j=1}^{\ell}S_j)$ be any set of 
$w_1$ such terminals. As before, there is a set $\qset_1$ of $w_1$ edge-disjoint paths, connecting the terminals of $\tset^*$ to the terminals of $\tset_1$. By concatenating the paths of $\qset_1$ and $\pset'(S_1)$, we obtain a set $\tpset$ of $w_1$ paths connecting the terminals of $\tset^*$ to the vertices of $S_1$ with edge-congestion at most $2$.

Next, we build a directed node-capacitated flow network $\nset$. Start with graph $G$, bi-direct all its edges, and contract every cluster $S_i\in \sset$ into a super-node $v_i$. Set the capacities of all super-nodes $v_i$ to be infinite, and the capacities of all other vertices to $1$. We assume without loss of generality that no edge connects any pair of super-nodes, as each such edge can be subdivided by a capacity-$1$ vertex. Add a source vertex $s$ of infinite capacity, and $\ell$ additional vertices $s^*,s_2,\ldots,s_{\ell}$, each of which has capacity $\floor{w_1/(3d\ell)}$. For each $2\leq i\leq \ell$, we connect $s_i$ to $v_i$ via a directed edge. Additionally, we connect $s^*$ to every vertex in $\tset^*$ with a directed edge. Finally, connect $s$ to all vertices in $\set{s^*,s_2,\ldots,s_{\ell}}$ by directed edges. We use $s$ as our source vertex and vertex $v_1$ as the destination vertex for the single-commodity maximum flow that we compute in $\nset$. 

Consider the set $\pset'=\left(\bigcup_{i=2}^{\ell}\pset''(S_i)\right )\cup \tpset$ of paths in graph $G$. Recall that for $2\leq i\leq \ell$, set $\pset''(S_i)$ contains $w_1$ paths, connecting vertices of $S_i$ to vertices of $S_1$, with edge-congestion at most $3$. Therefore, the paths in $\pset''(S_i)$ cause vertex-congestion at most $3d$ in $G$. The paths in $\tpset$ connect the terminals of $\tset^*$ to the vertices of $S_1$, with total edge-congestion at most $2$, and therefore total vertex-congestion at most $2d$. Altogether, the paths in $\pset'$ cause total vertex-congestion at most $3d\ell$ in $G$. The set $\pset'$ of paths naturally defines a corresponding set of paths in the flow network $\nset$. By sending $1/(3d\ell)$ flow units along each path in $\pset'$ (and lowering the flow on some paths as needed), we obtain a valid $s$--$v_1$ flow in $\nset$ of value $\ell\cdot \floor{\frac{w_1}{3d\ell}}$ in $\nset$. We can then obtain an integral $s$--$v_1$ flow $F$ in $\nset$ of the same value. Since this is a single-source/single-sink flow, we can assume that it is acyclic, that is, there is no directed cycle $C$ in $\nset$, where every edge of $C$ carries non-zero flow. We can therefore define an ordering of the vertices in $\set{v_2,\ldots,v_{\ell}}$, where whenever there is a directed path $P$ in $\nset$ from $v_i$ to $v_j$, with every edge of $P$ carrying non-zero flow, $v_j$ appears before $v_i$ in this ordering. By re-indexing the vertices, we can assume that for $j<i$, $v_j$ appears before $v_i$ in the ordering.

The flow $F$ defines, for every vertex $v_i$, for $2\leq i\leq \ell$, a collection $\pset'''(S_i)$ of $\floor{w_1/(3d\ell)}$ paths in $\nset$, connecting $v_i$ to $v_1$, and an additional collection $\tpset'$ of $\floor{w_1/(3d\ell)}$ paths in $\nset$, connecting terminals of $T^*$ to $v_1$. Moreover, all paths in set $\pset''=\left(\bigcup_{i=2}^{\ell}\pset'''(S_i)\right ) \cup \tpset'$ are edge-disjoint and vertex-disjoint in $\nset$, except for possibly sharing the super-nodes $v_j$, for $2\leq j\leq \ell$.

\paragraph{Step 2: Building the Tree.}

Let $w_2=\floor{\frac{w_1}{6d\ell ^2}}=\Omega(cw d^6\ell )$.
The vertices of the tree $\tau^*$ are $\set{v_0,v_1,\ldots,v_{\ell}}$. For $1\leq i\leq \ell$, we set $S(v_i)=S_i$. Set $S(v_0)$ is defined later, and it will contain some terminals of $\tset^*$. In order to define the edges of the tree, we consider the vertices $v_2,\ldots,v_{\ell}$ in this order. Since $v_2$ is the first vertex in our ordering, the paths of $\pset'''(v_2)$ are disjoint from the vertices $v_3,\ldots,v_{\ell}$. In particular, they define a set $\hat{\pset}(v_2)$ of at least $w_2$ paths in graph $G$, connecting the vertices of $S_2$ to the vertices of $S_1$, that are disjoint in edges and inner vertices, and are internally disjoint from $\bigcup_{i=1}^{\ell}S_i$. We add the edge $(v_1,v_2)$ to the tree, and we let $\pset(e)$ be any subset of $w_2$ paths of $\hat{\pset}(v_2)$. Suppose now that we have processed vertices $v_2,\ldots,v_{i-1}$, and we would like to process vertex $v_i$. Consider any path $P\in \pset'''(v_i)$. Then $P$ does not contain the vertices $v_{i+1},\ldots,v_{\ell}$. Let $u(P)$ be the first vertex of $\set{v_1,v_2,\ldots,v_{i-1}}$ on $P$, and let $P'$ be the sub-path of $P$ from its first vertex to $u(P)$. Notice that $P'$ naturally defines a path in $G$, connecting a vertex of $S_i$ to a vertex of $S_j$, where $v_j=u(P)$, and moreover $\set{P'\mid P\in \pset'''(v_i)}$ is a collection of paths in $G$ that are edge-disjoint and internally node-disjoint. There is some index $1\leq i'< i$, so that for at least $\floor{|\pset'''(v_i)|/\ell}\geq w_2$ paths $P\in \pset'''(v_i)$, $u(P)=v_{i'}$. We add the edge $e=(v_i,v_{i'})$ to the tree $\tau^*$, and we let $\pset(e)$ be any set of $w_2$ paths in $G$, corresponding to the set $\set{P'\mid P\in \pset'''(v_i); u(P)=v_{i'}}$ of paths, so the paths in $\pset(e)$ connect vertices of $S_i$ to vertices of $S_{i'}$, and they are edge-disjoint and internally vertex-disjoint. 

In order to complete the definition of the \ToS, we need to define the set $S(v_0)\subseteq \tset$ of vertices corresponding to the root vertex $v_0$ of $\tau^*$, and to connect it to the remainder of the tree. Consider the set $\tpset'$ of $\floor{w_1/(3d\ell)}$ paths in $\nset$, connecting the terminals of $T^*$ to $v_1$. As before, for every path $P\in \tpset'$, we let $u(P)$ be the first vertex of $\set{v_1,v_2,\ldots,v_{\ell}}$ on $P$, and we let $P'$ be the sub-path of $P$ from its first vertex to $u(P)$. As before, $P'$ naturally defines a path in $G$, connecting a terminal of $\tset^*$ to a vertex of $S_j$, where $v_j=u(P)$, and $\set{P'\mid P\in \tpset'}$ is a collection of paths in $G$ that are edge-disjoint and internally node-disjoint. There is some index $1\leq i\leq \ell$, so that for at least $\floor{|\tpset'|/\ell}\geq w_2$ of paths $P\in \tpset'$, $u(P)=v_{i}$. We add the edge $e=(v_0,v_{i})$ to the tree $\tau^*$, and we let $\pset(e)$ be any set of $w_2$ paths in $G$, corresponding to the set $\set{P'\mid P\in \tpset'; u(P)=v_{i}}$ of paths, so the paths in $\pset(e)$ connect vertices of $\tset^*$ to vertices of $S_{i}$, and they are edge-disjoint and internally vertex-disjoint. We let $S(v_0)\subseteq \tset^*$ be the set of terminals that serve as endpoints of the paths in $\pset(e)$.

This procedure defines a tree $\tau^*$, and it almost gives us a \ToS, except for the following difficulty. We are guaranteed that the paths of $\bigcup_{e\in E(\tau^*)}\pset(e)$ are edge-disjoint and internally vertex-disjoint, but they may share endpoints. However, since we have assumed that the maximum vertex degree in $G$ is bounded by $d$, this is easy to resolve, by losing a factor of at most $4d^2$ in the sizes of the sets $\pset(e)$.
We do so by using the following simple observation.

\begin{observation}\label{obs: disjointness of endpoints}
Suppose we are given a set $A=\set{a_1,\ldots,a_N}$ of $N>0$ elements, and a collection $\rset=\set{R_1,\ldots,R_m}$ of multi-subsets of $A$,  so for $1\leq i\leq m$, $R_i$ may contain several copies of each element of $A$. Assume that each element $a_j$ appears at most $d$ times in the sets of $\rset$ in total. Then we can efficiently compute, for each $1\leq i\leq m$, a subset $R'_i\subseteq R_i$ of $\floor{\frac{|R_i|}{d}}$ elements, so that all sets $R_1',\ldots,R_m'$ are mutually disjoint, and every element appears at most once in each set.
\end{observation}

\begin{proof}
We build a flow network, whose vertex set consists of a source $s$; a set $U_1=\set{u_i\mid 1\leq i\leq m}$ of vertices representing the sets $R_i$; a set $U_2=\set{u'_j\mid 1\leq j\leq N}$ of vertices representing the elements of $A$; and a destination vertex $t$. We set the capacities of $s$ and $t$ to be infinite; the capacity of every vertex in $U_2$ is $1$, and the capacity of every vertex in $U_1$ is $\floor{\frac{|R_i|}{d}}$. We connect $s$ to every vertex of $U_1$, and every vertex of $U_2$ to $t$ via directed edges, and we add a directed edge $(u_i,u'_j)$ iff element $a_j$ belongs to $R_i$. It is easy to see that there is a valid $s$-$t$ flow in this network of value $\sum_{i=1}^m\floor{\frac{|R_i|}{d}}$, by sending $1/d$ flow units on each edge of the form $(u_i,u'_j)$ (we may need to lower the flows on some edges to satisfy the capacities of the vertices of $U_1$). Therefore, there is an integral flow of the same value. It is easy to see that every vertex $u_i\in U_1$ connects, via edges that carry non-zero flow, to exactly $\floor{\frac{|R_i|}{d}}$ vertices of $U_2$, and we define $R'_i$ to contain all these vertices.
\end{proof}

We root the tree $\tau^*$ at vertex $v_0$, and process all vertices of $\tau^*$ in the bottom-up fashion, so a vertex of $\tau^*$ is only processed after all its descendants have been processed. Suppose that a vertex $v_i$ is being processed, and let $e_1,\ldots,e_q$ be the set of edges incident on $v_i$ in $\tau^*$. We let $A$ be the set of all vertices of $S_i$, and for each $1\leq j\leq q$, we define a multi-subset $R_j$ of the vertices of $S_j$, containing all endpoints of the paths in $\pset(e_j)$. Then every element of $A$ appears at most $d$ times in these sets (counting multiplicities), and we can find subsets $R'_j\subseteq R_j$ that are mutually disjoint, as in Observation~\ref{obs: disjointness of endpoints}. We then discard from $\pset(e_j)$ all paths whose endpoints lie in $R_j\setminus R'_j$. Once every vertex of tree $\tau^*$ is processed, the resulting set $\pset=\bigcup_{e\in E(\tau^*)}\pset(e)$ contains paths that are completely disjoint from each other and internally disjoint from $\bigcup_{i=0}^{\ell}S(v_i)$. The cardinality of each set $\pset(e)$ of paths is at least $\frac{w_2}{4d^2} =\Omega(cwd^4\ell)$. We denote this value by $w_3$. This finishes the construction of a \ToS. Our last step is to turn it into a perfect \ToS, by boosting the well-linkedness inside each cluster $S_i$.

\paragraph{Step 3: Boosting Well-Linkedness}
Given a vertex $v_i\in V(\tau^*)\setminus\set{v_0}$, let $\delta(v)$ be the set of all edges of $\tau^*$ incident on $v_i$. 
Suppose we are given, for each edge $e\in E(\tau^*)$, a subset $\pset'(e)\subseteq \pset(e)$ of paths, and let $\pset'=\bigcup_{e\in E(\tau^*)}\pset'(e)$. For every vertex $v_i\in V(\tau^*)$, and every edge $e\in \delta(v_i)$, we let $U^{\pset'}(v_i,e)$ be the set of all endpoints of the paths in $\pset'(e)$ that belong to $S_i$. We also let $U^{\pset'}(v_i)=\bigcup_{e\in \delta(v_i)}U^{\pset'}(v_i,e)$.

Recall that for every cluster $S_i\in \sset$, $\Gamma(S_i)$ is $(w',\alpha)$-well-linked in $G[S_i]$. Since $|U^{\pset}(v_i)|\leq \ell w_3=\frac{\ell \cdot w_2}{4d^2}=\frac{\ell}{4d^2}\cdot \floor{\frac{w_1}{6d\ell ^2}}\leq w'$ and the vertices of $U^{\pset}(v_i)$ are contained in $\Gamma(S_i)$, set $U^{\pset}(v_i)$ is $\alpha$-well-linked in $G[S_i]$. 

In order to do boost the well-linkedness, we will define, for every edge $e\in E(\tau^*)$ a large subset $\pset'(e)\subseteq \pset(e)$ of paths, such that for every vertex $v_i\in V(\tau^*)\setminus\set{v_0}$: (i) for every edge $e\in \delta(v_i)$, $U^{\pset'}(v_i,e)$ is node-well-linked in $G[S_i]$; (ii) for every pair $e,e' \in \delta(v_i)$ of edges with $e\neq e'$, $U^{\pset'}(v_i,e)$ and $U^{\pset'}(v_i,e')$ are linked in $G[S_i]$; and $U^{\pset'}(v_i)$ is $1/5$-well-linked in $G[S_i]$.
% A simple way to achieve this is to apply the boosting theorem (Theorem~\ref{thm: grouping}) to each set $U^{\pset}(v_i,e)$ in turn, and then use Theorem~\ref{thm: linkedness from node-well-linkedness}.
%However, this approach loses a factor $\Theta(1/\alpha^3)=\Theta(r^3)$ in the sizes of the sets $\pset(e)$. The approach that we use here, while slightly more involved, only loses a factor of $O(1/\alpha)=O(r)$. Using our approach, we first boost the well-linkedness of the sets $U^{\pset'}(v_i)$ to a constant, using the tree-grouping technique (Theorem~\ref{thm: weak well-linkedness}), and only then apply Theorem~\ref{thm: grouping} and Theorem~\ref{thm: linkedness from node-well-linkedness}, each of which now only loses an additional constant factor.

Our first step is to apply Theorem~\ref{thm: weak well-linkedness} to each cluster $G[S_i]$ for $1\leq i\leq \ell$, with the set $U^{\pset}(v_i)$ of vertices serving as terminals. As a result, we obtain a collection $\fset_i$ of disjoint trees in $G[S_i]$, each of which contains at least $\ceil{1/\alpha}$ and at most $2d\cdot\ceil{1/\alpha}$ vertices of $U^{\pset}(v_i)$. 

Root the tree $\tau^*$ at vertex $v_0$. For every vertex $v_i\in V(\tau^*)\setminus\set{v_0}$, let $e_i$ be the edge of $\tau^*$ connecting $v_i$ to its parent. Let $w_4=\floor{\frac{w_3}{2d \ceil{1/\alpha}}}=\Omega(\frac{w_3}{d\ell})=\Omega(cwd^3)$.
Our next step is summarized in the following theorem.

\begin{theorem}\label{thm: choosing subsets of paths1}
We can find, for each $1\leq i\leq \ell$, a subset $\pset'(e_i)\subseteq \pset(e_i)$ of at least $w_4$ paths, such that for all $1\leq i\leq \ell$, for every tree $\tau\in \fset_i$, at most two vertices of $\tau$ belong to the paths of $\bigcup_{e\in E(\tau^*)}\pset'(e)$.
\end{theorem}

\begin{proof}
We prove the theorem by defining an appropriate single-source single-sink flow network $\nset$. We view each edge $e_i\in E(\tau^*)$, and the corresponding set $\pset(e_i)$ of paths, as directed towards the root of $\tau^*$, that is, away from $v_i$. For each $1\leq i\leq \ell$, for every tree $\tau\in \fset_i$, we introduce two vertices, $a(i,\tau)$ and $b(i,\tau)$. Both vertices are used to represent the tree $\tau$, but, intuitively, $a(i,\tau)$ will be used by the paths of $\pset$ originating from the vertices of $\tau$ (that is, the paths of $\pset(e_i)$), and $b(i,\tau)$ by paths terminating at the vertices of $\tau$ (that is, the paths of $\pset(e')$ for all $e'\in \delta(v_i)\setminus\set{e_i}$). We set the capacities of all such vertices to $1$. For every edge $e_i=(v_i,v_j)\in E(\tau^*)$, for every path $P\in \pset(e)$, if $P$ originates at a vertex of some tree $\tau'\in \fset_i$ and terminates at a vertex of some tree $\tau''\in \fset_j$, then we add an edge $e_P=(a(i,\tau'),b(j,\tau''))$ to $\nset$, and we view this edge as representing the path $P$. We introduce a destination vertex $t$ of infinite capacity, and connect every vertex $b(j,\tau'')$ for all $v_j\in V(\tau^*)$ and $\tau''\in \fset_j$ to $t$. Finally, we would like to ensure that enough paths from each set $\pset(e)$ are selected. In order to do so, we introduce, for every edge $e_i=(v_i,v_j)\in E(\tau^*)$ a vertex $u_i$, whose capacity is $w_4$. For every path $P\in \pset(e)$, if $P$ originates at a vertex of some tree $\tau\in \fset_i$, then we add the edge $(u_i,a(i,\tau))$ to our network. (Notice that we allow parallel edges. Notice also that for all $\tau\in \fset_i$, all edges of $\nset$ leaving $a(i,\tau)$ correspond to the paths of $\pset_{e_i}$.) Finally, we add a source vertex $s$ of infinite capacity, and connect it to every vertex $u_i$, for $1\leq i\leq \ell$, with a directed edge. 

It is easy to see that the current set $\pset$ of paths naturally defines an $s$-$t$ flow of value $\ell w_3$, with congestion at most $2d\ceil{1/\alpha}$ on the vertices of the form $a(i,\tau)$ and $b(i,\tau)$, since every such tree $\tau$ contains at most $2d\ceil{1/\alpha}$ vertices of $U^{\pset}(v_i)$. By scaling this flow down by factor $2d\ceil{1/\alpha}$ and using the integrality of flow, we obtain a valid integral $s$-$t$ flow of value $\ell w_4$, where the flow through every vertex $u_i$ is $w_4$. Notice that from our construction, each edge $e_P=(a(i,\tau),b(j,\tau'))$ corresponds to some path $P\in \pset(e_i)$, connecting a vertex of $\tau$ to a vertex of $\tau'$. For every edge $e\in E(\tau^*)$, we let $\pset'(e)$ be the set of all paths $P\in \pset(e)$, where $e_P$ carries one flow unit. Since the flow through each vertex  $u_i$, for $1\leq i\leq \ell$, is $w_4$, we get that $|\pset'(e)|=w_4$ for each $e\in E(\tau^*)$. Since we have two vertices, $a(i,\tau)$ and $b(i,\tau)$ representing each tree $\tau\in \fset_i$, for all $1\leq i\leq \ell$, at most two vertices of each such tree belong to the paths of $\bigcup_{e\in E(\tau^*)}\pset'(e)$.
\end{proof}

We need the following simple claim.

\begin{claim}\label{claim: weak well-linkedness}
For every vertex $v_i\in V(\tau^*)\setminus\set{v_0}$, the set $U^{\pset'}(v_i)$ of vertices is $1/5$-well-linked in $G[S_i]$.
\end{claim}

\begin{proof}
 Consider any pair $U',U''$ of disjoint equal-sized subsets of vertices of $U^{\pset'}(v_i)$, and assume that $|U'|=|U''|=\kappa$. 
It is enough to show that there is a flow $F:U'\sconnect_5U''$ in $G[S_i]$. In order to do so, we partition $U'$ into three subsets, $U'_1,U'_2$ and $U'_3$, and we similarly partition $U''$ into $U''_1,U''_2$ and $U''_3$. We then construct flow $F_j: U_j'\sconnect U''_j$ for all $j\in \set{1,2,3}$. For each tree $\tau\in \fset_i$, if there are two vertices, $u\in V(\tau)\cap U'$ and $u'\in V(\tau)\cap U''$, then we add $u$ to $U'_1$ and $u'$ to $U''_1$. This finishes the definition of the sets $U'_1$ and $U''_1$. Notice that there is a flow $F_1: U'_1\sconnect_1U''_1$, where we connect the pairs using their corresponding trees. We can partition the set $U'\setminus U'_1$ of the remaining vertices into two subsets, $U'_2$ and $U'_3$, where $|U'_2|=\floor{|U'\setminus U'_1|/2}$, and for each tree $\tau\in \fset_i$, at most one vertex of $\tau$ belongs to $U'_2$, and at most one vertex of $\tau$ belongs to $U'_3$. We partition $U''\setminus U''_1$ into $U''_2$ and $U''_3$ similarly. Notice that from Theorem~\ref{thm: weak well-linkedness}, set $U'_2\cup U''_2$ is $\half$-well-linked in $G[S_i]$, since for each tree $\tau\in \fset_i$, at most one vertex of $\tau$ belongs to $U'_2\cup U''_2$. From our definition of the sets, $|U'_2|=|U''_2|$. Therefore, there is a flow $F_2: U'_2\sconnect_2 U''_2$ in $G[S_i]$. Similarly, there is a flow $F_3: U'_3\sconnect_2 U''_3$ in $G[S_i]$. Combining the flows $F_1,F_2$ and $F_3$, we obtain a flow $F: U'\sconnect_5 U''$ in $G[S_i]$, proving that $U^{\pset'}(v_i)$ is $1/5$-well-linked in $G[S_i]$.
\end{proof}

In our final step, we process the edges of the tree $\tau^*$ one-by-one. Let $e=(v_i,v_j)$ be any such edge, and assume that $i\neq 0$. Using Theorem~\ref{thm: grouping}, we can find a subset $U'\subseteq U^{\pset'}(v_i,e)$ of $\Omega(w_4/d)$ vertices, such that $U'$ is node-well-linked in $G[S_i]$. We discard from $\pset'(e)$ all paths that do not have an endpoint in $U'$, obtaining a collection $\tpset(e)$ of $\Omega(w_4/d)$ paths. 
If $j\neq 0$, then we then apply Theorem~\ref{thm: grouping} to the endpoints of the paths in $\tpset(e)$ that lie in $S_j$, to obtain a subset $U''$ of $\Omega(|\pset'(e)|/d)=\Omega(w_4/d^2)$ vertices that are node-well-linked in $S_j$. We again discard from $\tpset(e)$ all paths that do not have an endpoint in $U''$, obtaining a subset $\tpset'(e)$ of $\Omega(w_4/d^2)$ paths. 
If $j=0$, then we let $\tpset'(e)\subseteq \tpset(e)$ be any collection of $\Omega(w_4/d^2)$ paths. 
Our final step is to select an arbitrary subset $\tpset''(e)\subseteq \tpset'(e)$ of $\floor{|\tpset'(e)|/20d}$ paths. Once we process all edges of tree $\tau^*$ in this way, and denote $\tpset''=\bigcup_{e\in E(\tau^*)}\tpset''(e)$, from Theorem~\ref{thm: linkedness from node-well-linkedness}, for every vertex $v_i\in V(\tau^*)\setminus\set{v_0}$, for every pair $e, e'\in \delta(v_i)$ of edges with $e\neq e'$, the vertices of $U^{\tpset''}(v_i,e)$ and $U^{\tpset''}(v_i,e')$ are node-well-linked in $G[S_i]$. We let $w_5=\Omega\left(\frac{w_4}{d^3}\right)=\Omega(cw)\geq w$ (if $c$ is chosen to be large enough), so that for each edge $e\in E(\tau^*)$, $|\tpset''(e)|=w_5$. This concludes the construction of the Tree-of-Sets System.

%-----------------------------------------------------------
%-----------------------------------------------------------
%-----------------------------------------------------------
%-----------------------------------------------------------
\label{----------------------------------extended construction---------------------------}
\section{Obtaining Better Bounds}\label{sec: extended construction}
%-----------------------------------------------------------
%-----------------------------------------------------------
%-----------------------------------------------------------
%-----------------------------------------------------------

In this section we prove that Theorem~\ref{thm: GMT} holds for $f(g)=\Theta(g^{19}\poly\log g)$. 
In this section, given a \ToS $(\tau,\set{S(v)}_{v\in V(\tau)},\set{\pset(e)}_{e\in E(\tau)})$ of width $w$ in graph $G$, it is convenient to think of it as a \emph{$w$-wide embedding of the tree $\tau$ into the graph $G$}. So every vertex $v\in V(\tau)$ is embedded into a cluster $S(v)$, and every edge $e\in E(\tau)$ is embedded into a set $\pset(e)$ of $w$ paths. This way we can discuss $w$-wide embeddings of specific trees into $G$, as opposed to general Tree-of-Sets systems, where we have no control over the structure of $\tau$. Later in this section we define such embeddings formally, and we make them more general, by allowing different edges $e$ of $\tau$ to have different width values $w(e)$, so $|\pset(e)|=w(e)$ holds.

We start with a high-level intuitive overview of our construction. We first note that, as seen from previous sections, we can build a Tree-of-Sets System with much better parameters than a Path-of-Sets system: a bounded-degree graph $G$ of treewidth $k$ is guaranteed to contain a perfect \ToS of size $\ell$ and width $w$, as long as $k\geq \Omega(w\ell^5)$. On the other hand, a perfect \PoS with the same parameters currently requires that $k\geq  \Omega(w\ell^{17})$. Unfortunately, the previous proofs require a \PoS system in order to construct the grid minor. In~\cite{CC14}, a \ToS was transformed into a \PoS before Corollary~\ref{cor: from path-set system to grid minor} was applied to obtain the grid minor. This step resulted in significant losses in the parameters of the final \PoS. Here we will use the \ToS directly (combined with a large number of Path-of-Sets Systems), in order to construct the grid minor. The idea is that, since we will be exploiting the \ToS, we will only need to construct relatively short Path-of-Sets Systems, and so save on the parameters.

Assume first that we are given a \PoS $(\sset,\bigcup \pset_i)$ in our input graph $G$, of length $\ell=\Omega(g^2)$ and width $w=\Omega(g^3)$. Then Corollary~\ref{cor: from path-set system to grid minor} guarantees the existence of the $(g\times g)$-grid minor in $G$ (in fact, width $\Omega(g^2)$ is sufficient). It would be instructive to consider a slightly different construction of the grid minor, that requires these weaker parameters, to motivate our final construction. As in the proof of Corollary~\ref{cor: from path-set system to grid minor}, we apply Corollary~\ref{cor: paths from the path-set system} to the \PoS, with parameters $h_1=g$ and $h_2=g^2$. If the outcome is a $(g\times g)$-grid minor, then we are done. Therefore, we assume that the outcome is a collection $\qset$ of $g^2$ node-disjoint paths, connecting vertices of $A_1$ to vertices of $B_{\ell}$, so that for all $1\leq i\leq \ell$, for every path $Q\in \qset$, $Q\cap S_i$ is a path, and for all $1\leq j\leq \floor{\ell/2}$, for every pair $Q,Q'\in \qset$ of paths, there is a path $\beta_{2j}(Q,Q')$ connecting a vertex of $Q$ to a vertex of $Q'$ in $G[S_{2j}]$, so that $\beta_{2j}(Q,Q')$ is internally disjoint from all paths in $\qset$. We now slightly depart from the proof of Corollary~\ref{cor: from path-set system to grid minor}: namely, we embed every vertex $v$ of the grid minor into a distinct path $P_v\in \qset$. For every edge $(u,v)$ of the grid minor, we select a distinct even-indexed cluster $S_{2i}$, that we use in order to embed the edge, via the path $\beta_{2i}(P_u,P_v)$. This immediately gives a model of the $(g\times g)$-grid minor into $G$.

Notice that in this proof, the embedding is ``sequential'' in some sense: the vertices of the grid are embedded into paths, that traverse the clusters of $\sset$ in a fixed order, and we use every other cluster in turn in order to embed a distinct edge of the grid minor. This naturally fits in with the Path-of-Sets system. In order to better exploit the Tree-of-Sets system, we observe that this construction can be ``parallelized'': we can break the $(g\times g)$ grid into sub-grids $Q_1,\ldots,Q_{(g/a)^2}$ of size $(a\times a)$ each, for some parameter $a$ (for simplicity of exposition we assume that $g/a$ is an integer), and then embed the edges contained in each such sub-grid independently. We then need to embed the edges connecting the different sub-grids in a coordinated fashion.

As an example of how we can exploit this idea, and avoid constructing a long \PoS, let $H$ be a spider graph with $(g/a)^2+1$ legs: that is, $H$ is a union of $(g/a)^2+1$ paths $L_0,\ldots,L_{(g/a)^2}$, that are completely disjoint, except for sharing the first vertex of each path, called the \emph{head of the spider}, and denoted by $v_0$.
We require that the length of path $L_0$ is $\Omega(g^2/a)$, and the lengths of all other paths $L_i$, for $1\leq i\leq g/a$, are $\Omega(a^2)$ each.
 Assume now that we are given a $w$-wide embedding of $H$ into our graph $G$, where $w=\Omega(g^3)$ (in other words, we are given a \ToS whose width is $w$ and the corresponding tree is $H$). We claim that we can find a model of the $(g\times g)$-grid minor in $G$. We break the $(g\times g)$-grid into $(g/a)^2$ sub-grids $Q_1,\ldots,Q_{(q/a)^2}$ of size $(a\times a)$ each. Notice that the embedding of $H$ into $G$ defines, for each $1\leq i\leq (g/a)^2$, a \PoS $(\sset_i,\bigcup_j\pset_j^i)$ of length $\Omega(a^2)$ and width $\Omega(g^3)$, corresponding to $L_i$. We discard the cluster corresponding to the head $v_0$ of the spider, and use the remaining \PoS in order to embed the grid $Q_i$, exactly like in the proof outlined above. In other words, we construct paths that traverse the clusters of $\sset_i$, and are used to embed the vertices of $Q_i$, and we use the even-indexed clusters of $\sset_i$ in order to embed the edges of $Q_i$. In this way, the subgrids $Q_i$ are embedded `` in parallel'', using different legs of the spider. Eventually, we need to connect the different subgrids to each other.  Let $U$ be the set of all vertices appearing on the boundaries of the subgrids $Q_i$, for $1\leq i\leq (q/a)^2$, so $|U|\leq 4g^2/a$. We extend the paths from the different Path-of-Sets systems that were used to embed the vertices of $U$ into the cluster corresponding to $v_0$, and from there to the clusters of the \PoS $(\sset_0,\bigcup_j\pset_j^0)$, corresponding to the leg $L_0$ of the spider. We then use the same construction as before, in order to embed the remaining edges of the grid minor, using the clusters of $\sset_0$.
 
 It is not hard to see (and we show it below), that we could use a similar proof, where instead of the spider $H$, we use any tree containing $\Omega(g^2/a^2)$ disjoint $2$-paths\footnote{Recall that a path $P$ in a graph $H$ is a $2$-path iff for every vertex $v\in P$, the degree of $v$ in $H$ is $2$} of length $\Omega(a^2)$ each, and one additional path of length at least $\Omega(g^2/a)$. We can construct such a tree, and its $w$-wide embedding into $G$, using tools we already have. First, we construct an $w'$-wide embedding of any tree with $\Omega(g^2/a^2)$ vertices into $G$, for a large enough parameter $w'$ (or, in other words, a \ToS of size $\Omega(g^2/a^2)$ and width $w'$), using Theorem~\ref{thm: build a ToS system}. Next, for every cluster $S$ of this embedding, corresponding to a leaf  or a degree-$2$ vertex of the tree, we perform a number of iterations that split $S$ into a \PoS of the desired length (either $\Omega(a^2)$ or $\Omega(g^2/a)$). The fact that we are now constructing much shorter Path-of-Sets Systems, and that Tree-of-Sets Systems are much cheaper to construct, allows us to improve the bounds of Theorem~\ref{thm: GMT}.  A natural way to push this approach even further is to parallelize more, by further partitioning the subgrids $Q_i$ into even smaller sub-grids.
 
 We now turn to a formal description of the proof.
We construct a family $\hset$ of trees, and show that at least one of the trees in $\hset$ can be appropriately embedded into $G$. In order to optimize the resulting parameters of the Excluded Grid theorem, we will use a variable-width embedding, instead of the width-$w$ embedding, that is defined as follows.

\begin{definition}
Let $H$ be any graph with non-negative integral width values $w(e)$ for each edge $e\in E(H)$, and let $G$ be any graph. A \emph{variable-width} embedding of $H$ into $G$ consists of the following:

\begin{itemize}
\item for every vertex $v\in V(H)$, a subset $S(v)\subseteq V(G)$ of vertices of $G$, such that $G[S(v)]$ is connected, and all resulting clusters $\set{S(v)}_{v\in V(H)}$ are disjoint; and

\item for every edge $e=(u,v)\in E(H)$, a set $\pset(e)$ of $w(e)$ disjoint paths in graph $G$, where each path $P\in \pset(e)$ connects a vertex of $S(u)$ to a vertex of $S(v)$, and it is internally disjoint from $\bigcup_{x\in V(H)}S(x)$. The paths in $\pset=\bigcup_{e\in E(H)}\pset(e)$ must be all mutually disjoint.
\end{itemize}

For every vertex $v\in V(H)$ and edge $e\in \delta_H(v)$, let $U(v,e)$ be the set of all endpoints of the paths in $\pset(e)$ that belong to $S(v)$, and let $U(v)=\bigcup_{e\in \delta_H(v)}U(v,e)$. We say that the embedding is perfect, if for every vertex $v\in V(H)$: (i) for every edge $e\in \delta_H(v)$, set $U(v,e)$ of vertices is node-well-linked in $G[S(v)]$; (ii) for every pair $e, e'\in \delta_H(v)$ of edges with $e\neq e'$, $(U(v,e),U(v,e'))$ are node-linked in $G[S(v)]$; and (iii) set $U(v)$ is $1/5$-well-linked in $G[S(v)]$.
\end{definition}

 We now proceed as follows. First, we formally define the graph family $\hset$, and show that a variable-width embedding of any graph from this family into $G$ is sufficient in order to construct the grid minor. We then show that if $G$ is a subcubic graph with a large enough treewidth, then at least one graph from $\hset$ can be embedded into $G$.

\subsection{Graph Family $\hset$}\label{subsec: the graph}
In this section we define the family $\hset$ of graphs.
We use two integral parameters: $g\geq 1$, which is an integral power of $2$, and $d\geq 3$. Intuitively, the goal is to find the $(g\times g)$-grid minor in a given input graph $G$, whose maximum vertex degree is bounded by $d$.
 We also use a parameter $q$, which is a large enough constant independent of $g$, that is an integral power of $2$, whose specific value we set later. The construction of the graph family is recursive. For $z=1,\ldots,\log_qg$, we construct a family $\hset_z$ of edge-weighted level-$z$ graphs. The final family $\hset$ of graphs is obtained from the last level, $\hset=\hset_{\log_qg}$. Every graph $H$ that we construct has one special vertex $r(H)$, that we refer to as \emph{the root of $H$}. For $z\geq 1$, let $W_z(d)=72dgq^{z+1}$.

We start by defining level-$1$ graphs. A level-$1$ graph $H_1$ is simply a path containing $4q^2+2$ vertices. One of the endpoints of the path is designated to be the root $r(H_1)$. Each edge $e$ of $H_1$ has width value $w(e)=W_1(d)$. Family $\hset_1$ then consists of a single graph $H_1$ defined above.

We now define the family $\hset_z$ of graphs, for $z>1$, assuming that family $\hset_{z-1}$ was already defined. In order to construct a graph $H_z\in \hset_z$, we start with a path $P_z$ containing $8q^{z+1}+1$ vertices, whose endpoints are denoted by $a_z$ and $a'_z$, respectively, and any tree $\tau_z$ that contains $4q^2$ vertices, where $P_z$ and $\tau_z$ are disjoint. We denote by $b_z$ any vertex of $\tau_z$, that is viewed as the root of the tree. Finally, we add an edge $(a'_z,b_z)$, obtaining an initial graph $H'_z$. Every edge $e$ of $H'_z$ has width $w(e)=W_z(d)$. %Every edge $e$ of path $P_z$, and the edge $(a'_z,b_z)$ have width $w(e)=72gq^{z+1}$, and every other edge of graph $H'_z$ has width $40dq^{z+1}$. 
Notice that graph $H'_z$ is also a tree, and we will view it as rooted at $a_z$. For every vertex $v\in V(H'_z)\setminus\set{a_z}$, we denote by $e_v$ the unique edge of $H'_z$, connecting $v$ to its parent.

Our final step is to select a subset $V_z$ of $q^2$ vertices of $V(\tau_z)\setminus \set{b_z}$, that have degrees $1$ or $2$ in $H'_z$, such that no edge of $\tau_z$ has both endpoints in $V_z$. This is done as follows. Let $L$ denote the number of leaves of tree $\tau_z$. If $L\geq q^2$, then we let $V_z$ be any set of $q^2$ leaves of $\tau_z$. Otherwise, there are at most $L-1$ vertices of degree greater than $2$ in $\tau_z$. Let $U\subseteq V(\tau_z)$ be the set of at least $3q^2$ vertices of $V(\tau_z)\setminus\set{b_z}$, whose degree is $1$ or $2$. Then $\tau_z[U]$ is a collection of disjoint paths, and so we can select a subset $V_z\subseteq U$ of $q^2$ vertices, such that no edge of $\tau_z$ connects a pair of vertices in $V_z$.

For each vertex $v\in V_z$, we select an arbitrary level-$(z-1)$ graph $H_{z-1}(v)\in \hset_{z-1}$, that is added to our graph, together with an edge $e'_v$, connecting its root $r(H_{z-1}(v))$ to $v$. The width of the edge $e'_v$ is $W_{z-1}(d)$, and the width of every edge in $H_{z-1}(v)$ remains the same as in the original graph $H_{z-1}(v)$. Let $H_z$ be this final graph (see Figure~\ref{fig: level z graph}). We designate the vertex $a_z$ to be the root of $H_z$. Notice that the degree of $r(H_z)$ is $1$.

\begin{figure}[h]
\scalebox{0.5}{\includegraphics{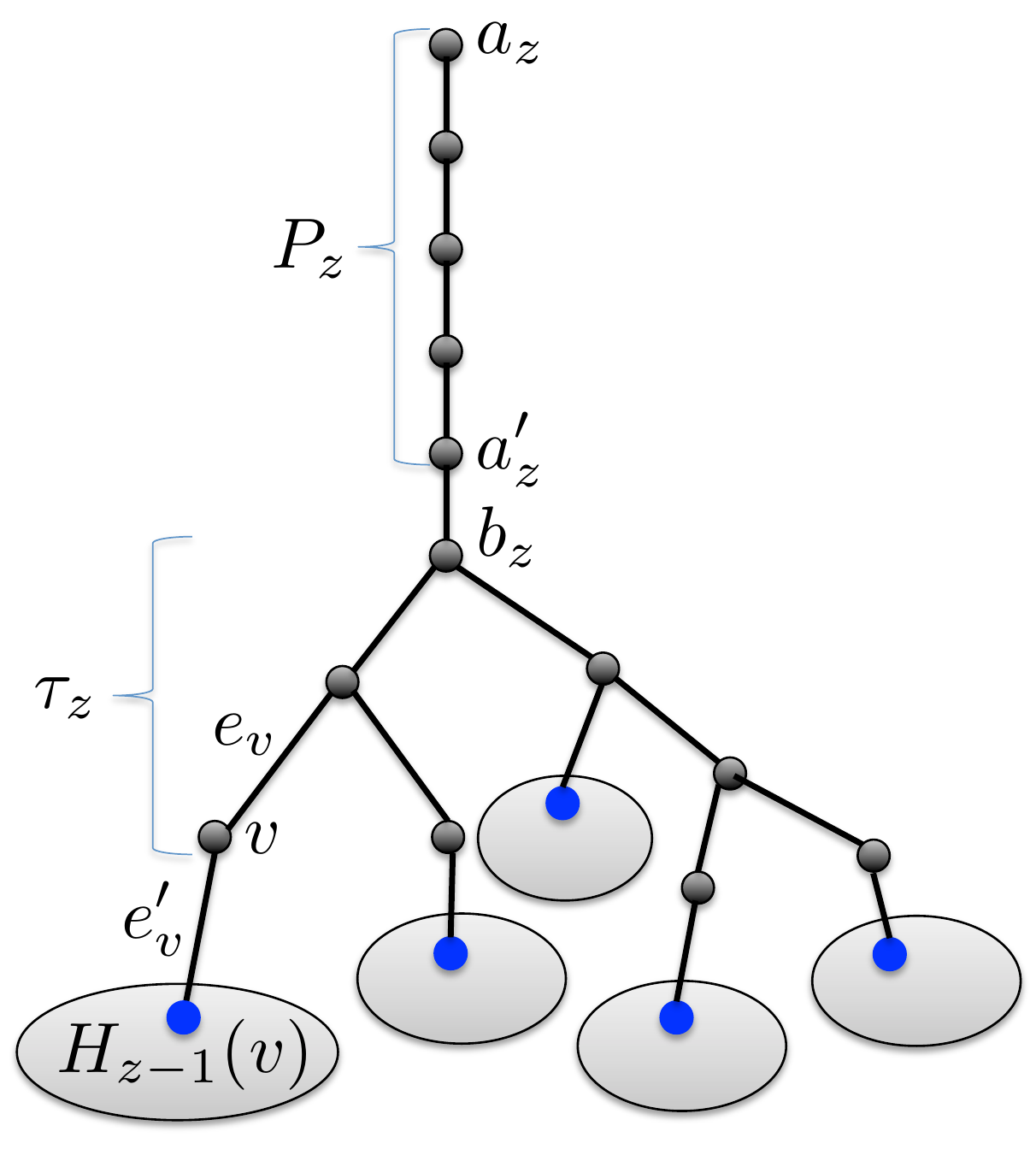}}
\caption{Construction of graph $H_z$. Path $P_z$ has $8q^{z+1}+1$ vertices, and tree $\tau_z$ has $4q^2$ vertices. All edges except those adjacent to the vertices lying in the level-$(z-1)$ graphs have width $W_z(d)$. The blue vertices are the roots of the level-$(z-1)$ graphs.\label{fig: level z graph}}
\end{figure}

This completes the construction of a level-$z$ graph $H_z$. Notice that every choice of the tree $\tau_z$, a subset $V_z$ of its vertices, and the graphs $\set{H_{z-1}(v)}_{v\in V_z}$ may result in a distinct graph $H_z$. The resulting family of all such graphs is denoted by $\hset_z$. We let $\hset=\hset_{\log_qg}$. It is easy to see that for every level $z$, every graph in $\hset_z$ is a tree.

%-----------------------------------------------------
%-----------------------------------------------------
%-----------------------------------------------------
%-----------------------------------------------------
%-----------------------------------------------------
%-----------------------------------------------------
\subsection{Constructing the Grid Minor}\label{subsec: building the minor}
%-----------------------------------------------------
%-----------------------------------------------------
%-----------------------------------------------------
%-----------------------------------------------------
%-----------------------------------------------------
%-----------------------------------------------------
In this section we show that if we are given any graph $H\in \hset$, and its perfect variable-width embedding into a graph $G$ whose maximum vertex degree is at most $d$, then $G$ contains the $(g\times g)$-grid as a minor. The proof is an induction on the graphs from different levels, where for each $z\geq 1$, we prove that if any graph $H_z\in \hset_z$ can be embedded into $G$ via the variable-width embedding, then $G$ contains the $(q^z\times q^z)$-grid as a minor.

For an integer $j>1$ and the $(j\times j)$-grid $R$, the \emph{boundary} of $R$ is the union of the first row, the $j$th row, the first column, and the $j$th column of $R$.

\begin{theorem}\label{thm: embedding to grid minor}
Given a graph $G$ with maximum vertex degree at most $d\geq 2$, an integer $z\geq 1$, and a graph $H_z\in \hset_z$, if there is a perfect variable-width embedding $\psi$ of $H_z$ into $G$, then either $G$ contains the $(g\times g)$-grid as a minor, or $G$ contains a model $\phi$ of the $(q^z\times q^z)$-grid $R$, that has the following additional property. Let $r=r(H_z)$, and let $e'_r$ be the unique edge of $H_z$ incident on $r$. Let $X$ be the set of all vertices of $R$ lying on its boundary. Then for each $x\in X$, $\phi(x)\cap S(r)$ contains exactly one vertex, and that vertex belongs to $U(r,e'_r)$, while for all $v\in V(R)\setminus X$, $\phi(v)\cap S(r)=\emptyset$. (Here $S(r)$, $U(r,e'_r)$ are defined with respect to the embedding $\psi$ of $H_z$ into $G$).
\end{theorem}

\begin{proof}
The proof is by induction on $z$. We start with the induction base, where $z=1$. Recall that $\hset_1$ contains a single graph $H_1$, which is a path containing $4q^2+2$ vertices. We denote the vertices of the path by $v_0,v_1,\ldots,v_{4q^2+1}$, where $v_0=r(H_1)$. For $0\leq i\leq 4q^2+1$, we denote $S(v_i)$ by $S_i$.
For all $0\leq i<4q^2+1$, let $e_i$ be the edge of $H_1$ connecting $v_i$ to $v_{i+1}$, and let $\pset_i=\pset(e_i)$. Let $B_i\subseteq S_i,A_{i+1}\subseteq S_{i+1}$ be the sets of vertices that serve as endpoints of the paths in $\pset(e_i)$. Let $\sset=(S_1,\ldots,S_{4q^2})$. We then obtain a perfect \PoS $(\sset,\bigcup_{i=1}^{4q^2-1}\pset_i)$ of width $W_1(d)=72dgq^2$ and length $4q^2$ in $G$.

From Corollary~\ref{cor: paths from the path-set system}, either there is a $(g\times g)$-grid minor in $G$, or there is a collection $\qset$ of $q^2$ node-disjoint paths in $G$, connecting  vertices of $A_1$ to  vertices of $B_{4q^2}$, such that for all $1\leq i\leq 4q^2$, for every path $Q\in \qset$, $S_i\cap Q$ is a path, and for every $1\leq j\leq 2q^2$, for every pair $Q,Q'\in \qset$ of paths, there is a path $\beta_{2i}(Q,Q')$ in $G[S_{2i}]$, connecting a vertex of $Q$ to a vertex of $Q'$, such that $\beta_{2i}(Q,Q')$ is internally disjoint from all paths in $\qset$. 

Consider some path $Q\in \qset$, and let $a\in A_1$ be the vertex where path $Q$ originates. Let $P_a\in \pset_0$ be the unique path containing vertex $a$, and let $u_Q$ be the other endpoint of path $P_a$, so $u_Q\in U(v_0,e_0)$. Finally, let $Q'$ be the concatenation of $P_a$ and $Q$. We are now ready to define the model $\phi$ of the $(q\times q)$-grid $R$ in $G$. 

Let $f: V(R)\rightarrow \qset$ be an arbitrary bijection  (recall that $|\qset|=q^2=|V(R)|$). For every vertex $v$ of $R$, if $v\not \in X$, then we let $\phi(v)=f(v)$; otherwise, we let $\phi(v)=Q'$, where $Q=f(v)$. In the latter case, we denote the corresponding vertex $u_Q\in U(v_0,e_0)$ by $\rho_v$. Notice that for each vertex $x\in X$, $\phi(x)\cap S(v_0)=\set{\rho_x}$, where $\rho_x\in U(v_0,e'_r)$, while for each vertex $v\in V(R)\setminus X$, $\phi(v)\cap S(v_0)=\emptyset$, as required.
Let $E(R)=\set{e^1,\ldots,e^{2q(q-1)}}$, where the ordering of the edges is arbitrary. Consider some edge $e^i=(x,y)\in E(R)$. 
We embed edge $e^i$ into the path $\beta_{2i}(Q_x,Q_y)$, contained in $G[S_{2i}]$, where $Q_x,Q_y\in\qset$ are the sub-paths of $\phi(x)$ and $\phi(y)$, respectively. Since the length of the \PoS is $4q^2$, and $|E(R)|<2q^2$, there are enough even-indexed clusters in the \PoS to accomplish this. This completes the construction of a model $\phi$ of $R$, for the case where $z=1$.

We now describe the induction step. Fix some integer $z>1$, and assume that the theorem holds for all levels $z'<z$. Let $H_z\in \hset_z$ be any level-$z$ graph. Recall that $H_z$ consists of a path $P_z$, whose endpoints are denoted by $a_z$ and $a'_z$, with $a_z=r(H_z)$, and a tree $\tau_z$, whose root vertex $b_z$ is connected to $a'_z$ with an edge (see Figure~\ref{fig: level z graph}). Additionally, we have selected a subset $V_z$ of $q^2$ vertices of $V(\tau_z)\setminus\set{b_z}$, and for each vertex $v\in V_z$, we have selected a level-$(z-1)$ graph $H_{z-1}(v)\in \hset_{z-1}$, connecting its root to $v$ with the edge $e'_v$. Recall that graph $H'_z$ is the union of $P_z,\tau_z$ and the edge $(a'_z,b_z)$, and it is a tree rooted at $a_z$. For every vertex $v\in V(H'_z)\setminus\set{a_z}$, we denote by $e_v$ the edge connecting $v$ to its parent in $H'_z$.

Given any subgraph $H'$ of $H_z$, we can use the variable-width embedding $\psi$ of $H$ into $G$ in order to define a corresponding subgraph $G'$ of $G$, into which $H'$ is embedded, as follows:

\[G'=\left(\bigcup_{v\in V(H')}G[S(v)]\right )\cup \left(\bigcup_{e\in E(H')}\pset(e)\right ).\]

We call $G'$ \emph{the subgraph of $G$ induced by the embedding  $\psi$ of $H'$ into $G$}. Clearly, the embedding $\psi$ of $H$ into $G$ immediately defines a perfect variable-width embedding of $H'$ into $G'$.

Let $R$ be the $(q^z\times q^z)$-grid. Partition $R$ into $q^2$ sub-grids of size $(q^{z-1}\times q^{z-1})$ each, and denote these sub-grids by $R_1,\ldots,R_{q^2}$ (the ordering is arbitrary). We also denote $V_z=\set{u_1,\ldots,u_{q^2}}$. For each $1\leq i\leq q^2$, we denote the graph $H_{z-1}(u_i)$ by $H^i$, and we let $G^i\subseteq G$ be the subgraph of $G$ induced by the embedding $\psi$ of $H^i$ into $G$. 
We also let $G',G''\subseteq G$ be the graphs induced by the embeddings of the path $P_z$ and the tree $\tau_z$ into $G$, respectively.
It is easy to verify that all graphs $G',G'',G^1,\ldots,G^{q^2}$ are mutually disjoint.

For each $1\leq i\leq q^2$, let $X_i$ be the set of vertices lying on the boundary of the grid $R_i$, so $|X_i|=4q^{z-1}$, and let $X'=\bigcup_{i=1}^{q^2}X_i$, so $|X'|=4q^{z+1}$.

The high-level idea is that for each $1\leq i\leq q^2$, we employ the induction hypothesis for $H^i$ and $G^i$, in order to construct a model of $R_i$ in $G^i$. 
We exploit the graph $G''$ in order to extend the embeddings of the vertices of $X'$, so that each of them contains a vertex of $S(b_z)$. Finally, we consider the \PoS defined by the embedding of the path $P_z$ into $G'$. We further extend the embeddings of the vertices of $X'$, so they traverse the clusters corresponding to the \PoS, and we then exploit the properties of the \PoS
 in order to embed the edges of $R$ whose endpoints lie in distinct sub-grids $R_i$.
 We now describe each of these steps in turn.
 
 \paragraph{Step 1: embedding the sub-grids $R_i$.}
Fix some $1\leq i\leq q^2$, and let $r_i=r(H^i)$. Let $e'_{r_i}$ be the unique edge of $H^i$ incident on $r_i$. Since we are given a perfect variable-width embedding of $H^i$ into $G^i$, by the induction hypothesis, either $G^i$ contains the $(g\times g)$-grid as a minor, or there is a model $\phi_i$ of $R_i$ in $G^i$, such that for every vertex $x\in X_i$, $\phi_i(x)\cap S(r_i)$ contains a single vertex, that belongs to $U(r_i,e'_{r_i})$, and for every vertex $v\in V(R_i)\setminus X_i$, $\phi_i(v)\cap S(r_i)=\emptyset$. If, for any $1\leq i\leq q^2$, graph $G^i$ contains the $(g\times g)$-grid as a minor, then we are done. Therefore, we assume from now on that for all $1\leq i\leq q^2$, there is a model $\phi_i$ of $R_i$ in $G^i$, as above. For every vertex $x\in X_i$, we denote by $\rho_x$ the unique vertex of $\phi_i(x)$ that lies in $U(r_i,e'_{r_i})$, and we let $Y_i=\set{\rho_x\mid x\in X_i}$, so $Y_i\subseteq U(r_i,e'_{r_i})$, and $|Y_i|=|X_i|=4q^{z-1}$ (see Figure~\ref{fig: step2.1}).

\begin{figure}[h]
\centering
\subfigure[Definitions of sets $Y_i$ and $Y_i'$. Recall that $|Y_i|=4q^{z-1}$, and $|Y_i'|=W_{z-1}(d)=72gdq^{z}$]{\scalebox{0.4}{\includegraphics{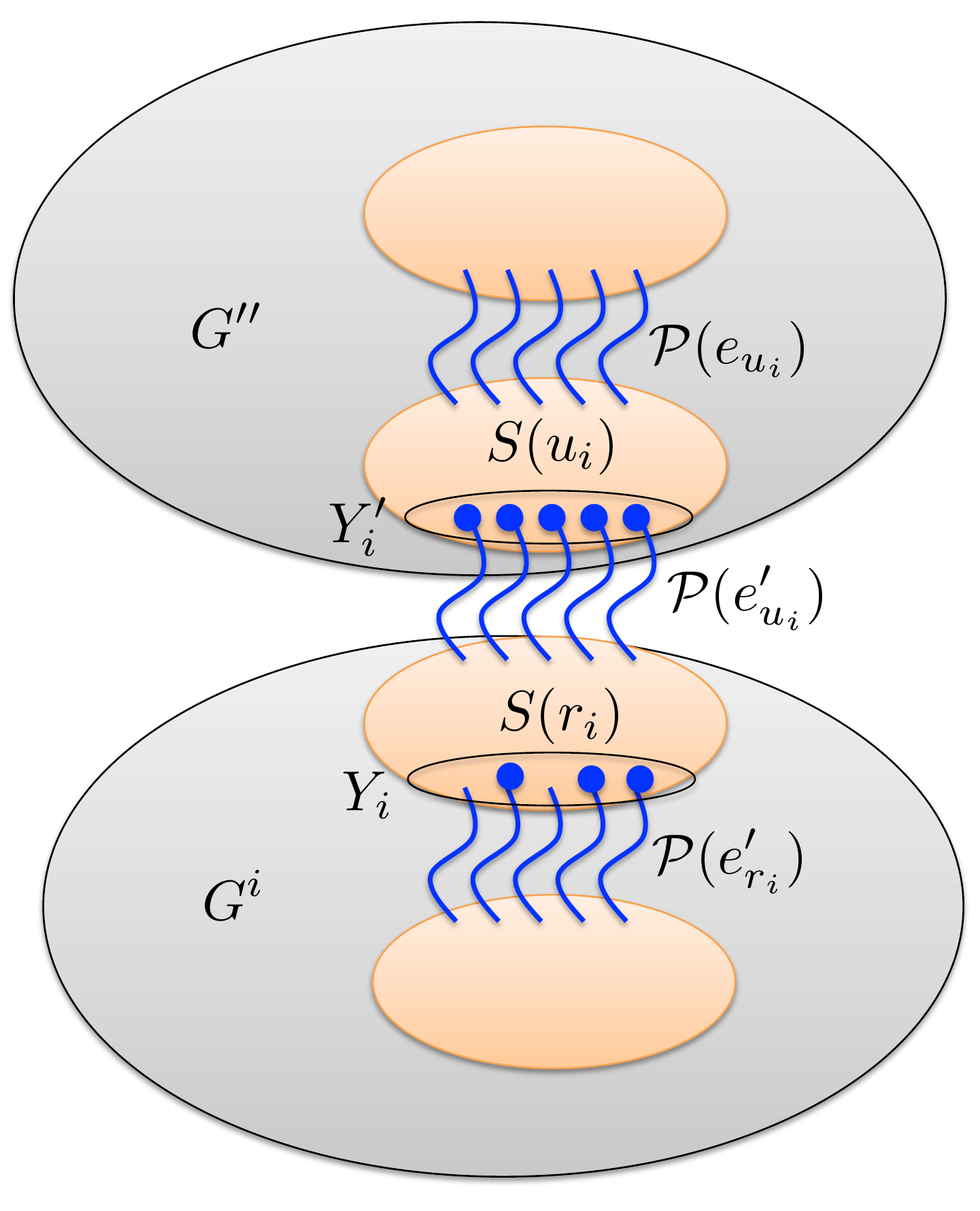}}\label{fig: step2.1}}
\hspace{1cm}
\subfigure[Definition of set $\Lambda$.]{
\scalebox{0.4}{\includegraphics{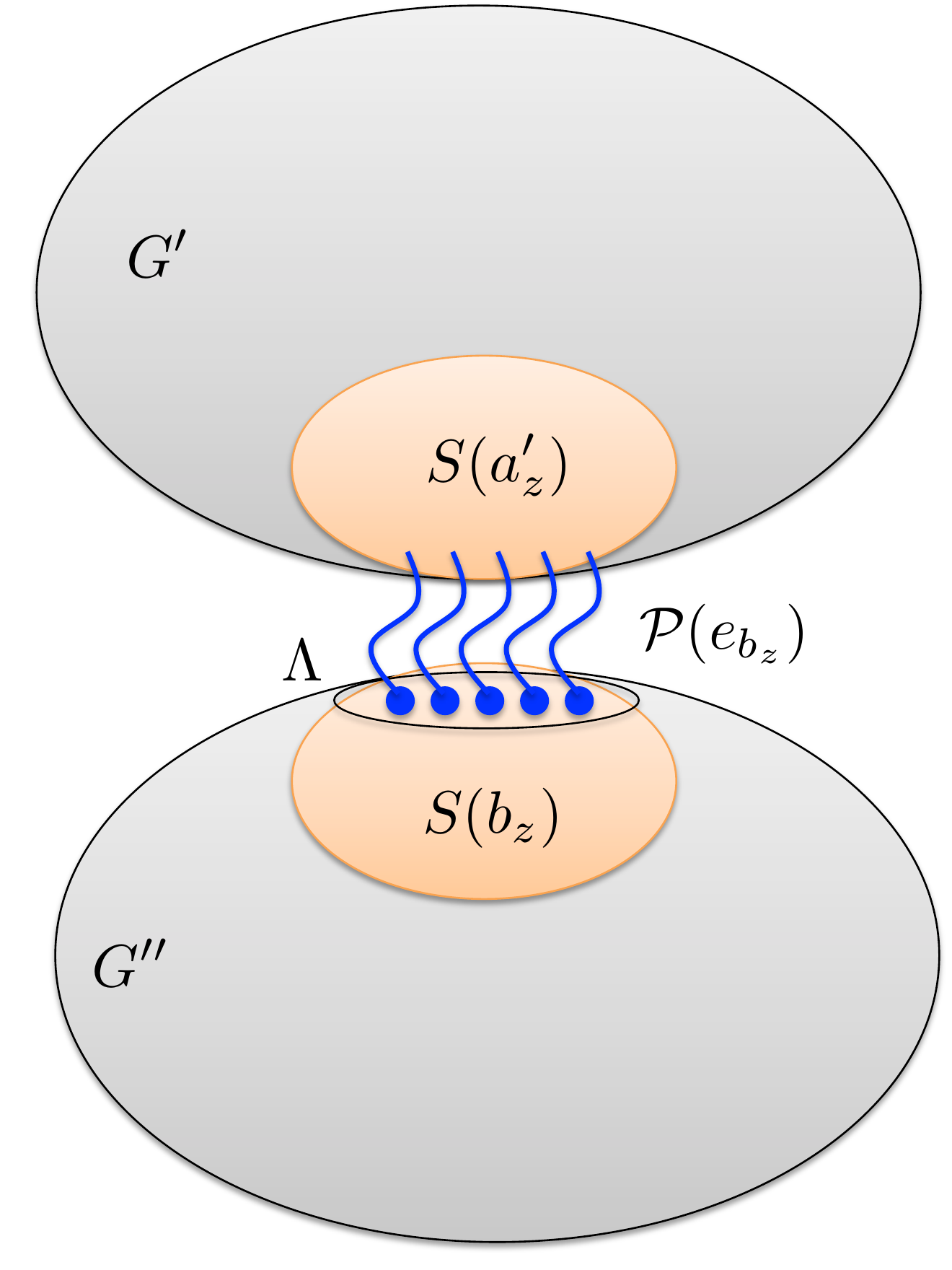}}\label{fig: step 2.2}}
\caption{Definitions of sets $Y_i$, $Y'_i$, and $\Lambda$. \label{fig: step 2}}
\end{figure}

\paragraph{Step 2: Routing inside $G''$.}
Consider some vertex $u_i\in V_z$, and let $e'_{u_i}$ be the edge connecting $u_i$ to $r_i$. Let $Y'_i=U(u_i,e'_{u_i})$ (see Figure~\ref{fig: step2.1}). Recall that $e_{b_z}$ is the edge connecting $b_z$ to $a'_z$. Let $\Lambda=U(b_z,e_{b_z})$  (see Figure~\ref{fig: step 2.2}), so $|\Lambda|=w(e_{b_z})=W_z(d)=72gdq^{z+1}$. The main result of this step is summarized in the following lemma.

\begin{lemma}\label{lem: extending the embedding for boundaries}
There is a set $\qset$ of node-disjoint paths in graph $G''$, where each path connects a vertex of $\bigcup_{i=1}^{q^2}Y'_i$ to a vertex of $\Lambda$, and for each $1\leq i\leq q^2$, the number of paths in $\qset$ originating from the vertices of $Y'_i$ is $4q^{z-1}$.
\end{lemma}

\begin{proof}
For each $1\leq i\leq q^2$, let $Y''_i\subseteq Y_i'$ be any subset containing $40dq^{z-1}$ vertices (since all edges $e'(u_i)$ have width $W_{z-1}(d)=72gdq^{z}$, such a set exists), and let $Y''=\bigcup_{i=1}^{q^2}Y''_i$, so $|Y''|=40dq^{z+1}$.
Recall that $|\Lambda|=W_z(d)=72gdq^{z+1}$. Let $\Lambda'\subseteq \Lambda$ be any subset of $40dq^{z+1}$ vertices. Our first step is to construct a collection $\qset':Y''\sconnect_5 \Lambda'$ of paths in $G''$.
The construction itself is somewhat tedious, but the intuition for it is simple: we ``grow'' the paths $\qset'$ from the vertices of $Y''$ towards the root of the tree $\tau'$, by exploiting the sets $\pset(e)$ of paths corresponding to the edges $e\in E(\tau)$, and the $1/5$-well-linkedness of the vertices of $U(v)$ in $G''[S(v)]$ for each $v\in V(\tau')$.

While the tree $\tau_z$ contains any leaf vertex that does not belong to $V_z$, we delete this vertex from $\tau_z$. Let $\tau'$ denote the resulting tree, so all leaves of $\tau'$ belong to $V_z$. We process the vertices of $\tau'$ one-by-one in the bottom-top order, so a vertex $v$ is processed only after all its descendants have been processed. Throughout the algorithm, we also gradually construct the set $\qset'$ of paths. For every vertex $v\in \tau'$, let $L(v)$ be the set of all descendants of $v$ (including possibly $v$ itself) that lie in $V_z$, and let $n(v)=|L(v)|$. Throughout the algorithm, we maintain the following invariants:

\begin{properties}{I}

\item All paths in $\qset'$ are contained in $G''$; they cause edge-congestion at most $5$ in $G''$, and all their endpoints are distinct;

\item If vertex $v$ was not processed yet, then all paths in $\qset'$ are disjoint from the paths in $\pset(e_v)$, and are internally disjoint from $S(v)$; and

\item If a vertex $v$ was already processed, and $v'$ is the parent of $v$ in $\tau'$, then set $\qset'$ contains a set of $40dq^{z-1}\cdot n(v)$ paths, connecting the vertices of $\bigcup_{u_i\in L(v)}Y''_i$ to distinct vertices of $U(v',e_v)$.\end{properties}

At the beginning of the algorithm, $\qset'=\emptyset$. It is easy to see that all invariants hold.
Consider some iteration of the algorithm, where some vertex $v$ of $\tau'$ is processed, and assume that all invariants hold so far.
We consider three cases.

The first case is when $v$ is a leaf of $\tau'$. Assume that $v=u_i$. Let $Z(u_i)$ be any set of $40dq^{z-1}$ vertices of $U(u_i,e_{u_i})$. Since we are given a perfect variable-width embedding of $H_z$ into $G$, $U(u_i,e_{u_i})$ and $U(u_i,e'_{u_i})$ are node-linked in $G[S(u_i)]$. %@@ (see Figure~\ref{fig: step2 case 1}). 
Therefore, there is a set $\tilde{\pset}(u_i): Y''_i\sconnect_1 Z(u_i)$ of paths in $G[S(u_i)]$. Let $\tilde{\pset}'(u_i)\subseteq \pset(e_{u_i})$ be the set of paths corresponding to the edge $e_{u_i}$ that connects $u_i$ to its parent, that originate at the vertices of $Z(u_i)$. We add to $\qset'$ a collection of $40dq^{z-1}$ edge-disjoint paths, obtained by concatenating the paths in $\tilde{\pset}(u_i)$ and the paths in $\tilde{\pset}'(u_i)$. If we denote by $v'$ the parent of $u_i$ in the tree, then each of these new paths connects a distinct vertex of $Y''_i$ to a distinct vertex of $U(v',e_{u_i})$. It is easy to verify that all invariants continue to hold.

\iffalse
@@@@@@@@@@@@
\begin{figure}[h]
\centering
\subfigure[Case 1: $v=u_i$ is a leaf]{\scalebox{0.3}{\includegraphics{step2-minor-constr-case1.pdf}}\label{fig: step2 case 1}}
\hspace{1cm}
\subfigure[Case 2: $v=u_i\in V_z$ is a degree-$2$ vertex.]{
\scalebox{0.3}{\includegraphics{step2-minor-constr-case2.pdf}}\label{fig: step 2 case 2}}
\hspace{1cm}
\subfigure[Case 3: $v\not\in V_z$. The vertices of $Z(v)$ are shown in green.]{
\scalebox{0.3}{\includegraphics{step3-minor-constr-case3.pdf}}\label{fig: step 2 case 3}}
\caption{Constructing the paths in $\qset'$. \label{fig: step 2 - 3 cases}}
\end{figure}
@@@@@@@@@@@@@@@@
\fi

The second case is when $v$ belongs to $V_z$, but it is not a leaf of $\tau'$. In this case, the degree of $v$ in $\tau'$ is $2$. We assume again that $v=u_i$. Recall that $e_{u_i}$ is the edge connecting $u_i$ to its parent, that we denote by $v'$, and $e'_{u_i}$ connects $u_i$ to $r_i$. Since the degree of $u_i$ in $\tau'$ is $2$, it has exactly one descendant in $\tau'$, that we denote by $v''$. Let $e''$ be the edge $(u_i,v'')$. %@@@ (See Figure~\ref{fig: step 2 case 2}.)
Notice that $n(u_i)=n(v'')+1$. Since we assumed that all invariants hold, there is a set $\qset(v'')\subseteq \qset'$ of $40 d q^{z-1}\cdot n(v'')$ paths of $\qset'$, connecting the vertices of $\bigcup_{u_j\in L(v'')}Y''_j$ to some vertices of $U(u_i,e'')$. The endpoints of the paths in $\qset(v'')$ are all distinct, and we denote the set of these endpoints that lie in $U(v_i,e'')$ by $Z(u_i)$.
Let $Z'(u_i)$ be any subset of $40dq^{z-1}n(u_i)$ vertices of $U(u_i,e_{u_i})$.  Notice that, since $n(u_i)\leq q^2$, and the width of every edge in $\tau_z$ is $W_z(d)=72gdq^{z+1}$, such a set exists. Since the embedding of $\tau_z$ into $G''$ is perfect, the set $U(u_i)$ of vertices is $1/5$-well-linked in $G[S(u_i)]$. Therefore, there is a set $\tilde{\pset}(u_i): Z(u_i)\cup Y''\sconnect_5 Z'(u_i)$ of paths contained in $G[S(u_i)]$. Let $\tilde{\pset}'(u_i)\subseteq \pset(e_{u_i})$ be the set of paths originating from the vertices of $Z'(u_i)$. We remove the paths of $\qset(v'')$ from $\qset'$, and replace them with the concatenation of the paths in $\qset(v''),\tilde{\pset}(u_i)$, and $\tilde{\pset}'(u_i)$. In other words, we have extended the paths in $\qset(v'')$, so they now terminate at the vertices of $U(v',e_{u_i})$, and we have connected the vertices of $Y''_i$ to the vertices of $U(v',e_{u_i})$.

The third case is when $v$ is an inner vertex of $\tau'$, and $v\not\in V_z$.
Let $v_1,\ldots,v_p$ be the children of $v$. %@@@ (See Figure~\ref{fig: step 2 case 3}).
For each $1\leq j\leq p$, let $\qset'_j\subseteq \qset'$ be the set of $40 d q^{z-1}\cdot n(v_j)$ paths of $\qset'$, connecting the vertices of $\bigcup_{u_i\in L(v_j)}Y''_i$ to some vertices of $U(v,e_{v_j})$. Let $\qset'(v)=\bigcup_{j=1}^p \qset'_j$. Then $|\qset'(v)|=40 d q^{z-1}\cdot n(v)$, and the endpoints of the paths in $\qset'(v)$ are all distinct, and lie in $\bigcup_{j=1}^pU(v,e_{v_j})$. We denote the set of these endpoints by $Z(v)$. 
Assume first that $v\neq b_z$. Let $Z'(v)\subseteq U(v,e_v)$ be any set of $40 d q^{z-1}\cdot n(v)$ vertices. Since $n(v)\leq q^2$, and the width of every edge in $\tau_z$ is $W_z(d)=72gdq^{z+1}$, such a set exists. Since the embedding of $\tau_z$ into $G''$ is perfect, the set $U(v)$ of vertices is $1/5$-well-linked in $G[S(v)]$. Therefore, there is a set $\tilde{\pset}(v): Z(v)\sconnect_5 Z'(v)$ of paths contained in $G[S(v)]$. Let $\tilde{\pset}'(v)\subseteq \pset(e_v)$ be the set of paths originating from the vertices of $Z'(v)$. We remove the paths of $\qset'(v)$ from $\qset'$, and replace them with the concatenation of the paths in $\qset'(v),\tilde{\pset}(v)$, and $\tilde{\pset}'(v)$. In other words, we have extended the paths in $\qset'(v)$, so they now terminate at the vertices of $U(v',e_v)$, where $v'$ is the parent of $v$.

Finally, if $v=b_z$, then $n(v)=q^2$, and so $|Z(v)|=40dq^{z+1}$. As before, since  the embedding of $\tau_z$ into $G''$ is perfect, the set $U(v)$ of vertices is $1/5$-well-linked in $G[S(v)]$, and so there is a set $\tilde{\pset}(v): Z(v)\sconnect_5 \Lambda'$ of paths contained in $G[S(v)]$. By concatenating the paths in $\qset'$ and the paths of $\tpset(v)$, we obtain the final set $\qset':Y''\connect_5 \Lambda$ of paths in $G''$.

Our last step is to build a collection of node-disjoint paths, using standard techniques. We construct a directed flow network $\nset$, by starting from $G''$ and bi-directing all its edges. All vertices of $G''$ are assigned unit capacities. For each $1\leq i\leq q^2$, we add a vertex $s_i$ of capacity $4q^{z-1}$, and connect it with a directed edge to every vertex of $Y''_i$. We also add a source vertex $s$ of infinite capacity, and connect it with a directed edge to every vertex $s_i$, for $1\leq i\leq q^2$. Finally, we add a destination vertex $t$ of infinite capacity, and connect every vertex of $\Lambda$ to it with a directed edge. The set $\qset'$ of paths induces an $s$-$t$ flow in $\nset$ with vertex-congestion at most $10d$. By scaling this flow down by factor $10d$, we obtain a valid $s$-$t$ flow of value $4q^{z+1}$ in $\nset$. Using the integrality of flow, we obtain a valid integral $s$-$t$ flow of the same value. This flow naturally induces a collection $\qset$ of node-disjoint paths, connecting some vertices of $\bigcup_{i=1}^{q^2}Y'_i$ to some vertices of $\Lambda$, where for each $1\leq i\leq q^2$, exactly $4q^{z-1}$ paths of $\qset$ originate from the vertices of $Y'_i$.
\end{proof}

Fix some $1\leq i\leq q^2$, and let $\qset_i\subseteq \qset$ be the set of paths originating from the vertices of $Y'_i$. Let $Y^*_i\subseteq Y'_i$ be the set of $4q^{z-1}$ vertices, where the paths of $\qset_i$ originate. Consider the graph $\hat{H}_i$, obtained by the union of $G[S(r_i)]$ and $\pset(e'_{u_i})$, where $e'_{u_i}=(u_i,r_i)$. Since the embedding of $H_z$ into $G$ is perfect, $U(r_i,e_{u_i})$ and $U(r_i,e'_{u_i})$ are node-linked in $G[S(r_i)]$. Therefore, there is a set $\pset'_i: Y_i\sconnect Y^*_i$ of node-disjoint paths in graph $\hat{H}_i$. Let $\qset^*_i$ be the set of paths obtained by concatenating the paths in $\qset_i$ and $\pset'_i$, and let $\qset^*=\bigcup_{i=1}^{q^2}\qset^*_i$. Then $\qset^*$ is a set of node-disjoint paths in $G\setminus G'$, connecting every vertex of $\bigcup_{i=1}^{q^2}Y_i$ to some vertex of $\Lambda$. 

Consider now some vertex $x\in X'$, and assume that $x\in X_i$, for $1\leq i\leq q^2$. Let $Q_x\in \qset^*$ be the unique path originating from the vertex $\rho_x$, and let $\rho'_x$ be its other endpoint. We define a new embedding $\phi'(x)$ of $x$ as follows: $\phi'(x)=\phi_i(x)\cup Q_x$. 
Let $\Lambda^*\subseteq \Lambda$ be the set of vertices $\set{\rho'_x\mid x\in X'}$, so $|\Lambda^*|=|X'|=4q^{z+1}$.

\paragraph{Step 3: Completing the Construction of the Model of the Grid.} In this step we complete the construction of the model of $R$ in $G$. We observe that the embedding of $P_z$ into $G'$ defines a perfect \PoS, which we exploit in order to embed the edges connecting different sub-grids $R_i$ of $R$, while extending the embeddings $\phi'(x)$ of all vertices $x\in X'$, so that they traverse the cluster of the \PoS.

 Recall that $P_z\subseteq H_z$ is a path containing $8q^{z+1}+1$ vertices. We denote these vertices by $a_z=v_0,v_1,\ldots,v_{8q^{z+1}}=a'_z$, and we assume that they appear on $P_z$ in this order. For $0\leq i\leq 8q^{z+1}$, we denote $S(v_i)$ by $S_i$.
For all $0\leq i<8q^{z+1}$, let $\hat{e}_i$ be the edge of $P_z$ connecting $v_i$ to $v_{i+1}$, and let $\hat{\pset}_i=\pset(\hat e_i)$. Let $B_i\subseteq S_i,A_{i+1}\subseteq S_{i+1}$ be the sets of vertices that serve as endpoints of the paths in $\hat{\pset_i}$. We also let $B_{8q^{z+1}}=U(a'_z,e_{b_z})$, so $|B_{8q^{z+1}}|=W_z(d)=72gdq^{z+1}$.  Letting $\sset=(S_1,\ldots,S_{8q^{z+1}})$, we then obtain a perfect \PoS $(\sset,\bigcup_{i=1}^{8q^{z-1}-1}\hat{\pset}_i)$ of width $72gdq^{z+1}$ and length $8q^{z+1}$ in $G'$.

From Corollary~\ref{cor: paths from the path-set system}, either there is a $(g\times g)$-grid minor in $G'$, or there is a collection $\hat{\qset}$ of $4q^{z+1}$ node-disjoint paths in $G'$, connecting the vertices of $A_1$ to the vertices of $B_{8q^{z+1}}$, such that for all $1\leq i\leq 8q^{z+1}$, for every path $Q\in \hat{\qset}$, $S_i\cap Q$ is a path, and for every $1\leq j\leq 4q^{z+1}$, for every pair $Q,Q'\in \hat{\qset}$ of paths, there is a path $\beta_{2i}(Q,Q')$ in $G'[S_{2i}]$, connecting a vertex of $Q$ to a vertex of $Q'$, such that $\beta_{2i}(Q,Q')$ is internally disjoint from all paths in $\hat{\qset}$. If $G'$ contains the $(g\times g)$-grid as a minor, then we are done, so we assume that the latter happens. We shorten every path $Q\in \hat \qset$, so it is internally disjoint from the last cluster, $S_{8q^{z+1}}$. That is, every path in $\hat Q$ now connects a vertex of $A_1$ to a vertex of $A_{8q^{z+1}}$. Let $\hat{\Lambda}\subseteq A_{8q^{z+1}}$ be the set of vertices where the paths of $\hat \qset$ now terminate, so $|\hat{\Lambda}|=4q^{z+1}$.

Let $\hat{G}$ be the subgraph of $G$, obtained by the union of $\pset(e_{b_z})$ and $G[S(a'_z)]=G[S_{8q^{z+1}}]$. Since the embedding of $H_z$ into $G$ is perfect, and $|\hat {\Lambda}|=|\Lambda^*|=4q^{z+1}$, we can find a set $\hat{\qset}'$ of node-disjoint paths, connecting every vertex of $\Lambda^*$ to a distinct vertex of $\hat{\Lambda}$ in $\hat G$. Let $\hat{\qset}^*$ be the set of paths obtained by concatenating the paths of $\hat{\qset}$ and the paths of $\hat{\qset'}$. Then for every vertex $x\in X'$ of the grid $R$, there is a unique path $\hat Q\in \hat{\qset }^*$, originating from the vertex $\rho'_x$. We are now ready to define the model $\phi$ of $R$ in $G$.

Consider first any vertex $v\in V(R)$. Assume first that $v\not \in X'$, and let $R_i$ be the sub-grid of $R$ to which $v$ belongs. Then we define $\phi(v)=\phi_i(v)$. Assume now that $v\in X'\setminus X$. Then we let $\phi(v)$ be the concatenation of $\hat Q_v$ and $\phi'(v)$.
Finally, if $v\in X$, then we let $\phi(v)$ be the concatenation of $\hat Q_v$, $\phi'(v)$, and the unique path $P_v\in\hat{\pset}_0$, that shares an endpoint with $\hat Q_v$. Let $\rho''_v$ denote the other endpoint of  $P_v$. Then it is easy to see that for all $v\in X$, $\phi(v)\cap S(a_z)=\set{\rho''_v}$, and $\rho''_v\in U(a_z,e'_{a_z})$, where $e'_{a_z}$ is the unique edge incident on $a_z$ in $H_z$,  while for all $v\in V(R)\setminus X$, $\phi(v)\cap S(a_z)=\emptyset$. 

Consider now some edge $e=(v,v')\in E(R)$. Assume first that $e\in E(R_i)$ for some $1\leq i\leq q^2$. We then set $\phi(e)=\phi_i(e)$. Let $E'\subseteq E(R)\setminus\bigcup_{i=1}^{q^2}E(R_i)$ be the set of all remaining edges. Then every edge in $E'$ connects a pair of vertices that both belong to $X'$, and $|E'|< |X'|=4q^{z+1}$. We assign a distinct integer $1\leq j(e)< 4q^{z+1}$ to each edge $e\in E'$. Consider now some edge $e=(x,x')\in E'$, and assume that $j(e)=j$. We then embed the edge $e$ into the path $\beta_{2j}(Q,Q')$, where $Q,Q'\in \hat \qset$ are the paths with $Q\subseteq \phi(x)$ and $Q'\subseteq \phi(x')$.
It is easy to verify that this gives a valid model $\phi$ of $R$ in graph $G$, that has the required properties.
\end{proof}

%------------------------------------------------------------------
%------------------------------------------------------------------
%------------------------------------------------------------------
\subsection{Embedding $H$ into $G$}

%------------------------------------------------------------------
%------------------------------------------------------------------
%------------------------------------------------------------------
Before we provide a construction of a variable-width embedding of some graph $H\in\hset$ into $G$, we need one last ingredient, which is a new way to split a cluster. %We discuss it in the following section.

%\subsubsection{Another Way to Split a Cluster}
Recall that in Section~\ref{sec: splitting a cluster} we proved Theorem~\ref{thm: main advanced splitting}, that allowed us to split one cluster into two. The splitting itself can be viewed as being ``sequential'': that is, we start with some cluster $C$, and two large disjoint subsets $T_1,T_2$ of its vertices, so that $T_1\cup T_2$ is well-linked in $G[C]$ (we omit precise technical details in this informal overview). We then showed that there are two disjoint clusters $X,Y\subseteq C$, such that there are large subsets $\tT_1 \subseteq T_1\cap X$ and $\tT_2 \subseteq T_2\cap Y$ of vertices, and a large set $E'\subseteq E(X,Y)$ of edges with the following properties. Let  $\Upsilon_X$ and $\Upsilon_Y$ denote the endpoints of the edges of $E'$ that belong to $X$ and $Y$, respectively. Then $\tT_1\cup \Upsilon_X$ is well-linked in $G[X]$, and $\tT_2\cup \Upsilon_2$ is well-linked in $G[Y]$ (see Figure~\ref{fig: splitting-sequential}). In this section, we need a different way to split $C$ into clusters $X$ and $Y$, that we can think of as being ``parallel'', as opposed to the ``sequential'' splitting described above. Our goal is to find two disjoint clusters $X,Y\subseteq C$, together with two large subsets $\tT_1\subseteq T_1\cap X$ and $\tT_2\subseteq T_2\cap X$ of vertices, and a large collection $\rset$ of node-disjoint paths, connecting vertices of $X$ to vertices of $Y$, such that, if we denote by $\Upsilon_X$ and $\Upsilon_Y$ the endpoints of the paths in $\rset$ that lie in $X$ and $Y$, respectively, then $\Upsilon_Y$ is well-linked in $G[Y]$, and $\tT_1\cup \tT_2\cup \Upsilon_X$ are well-linked in $G[X]$ (see Figure~\ref{fig: splitting-parallel}). 

\begin{figure}[h]
\centering
\subfigure[Starting Point]{\scalebox{0.4}{\includegraphics{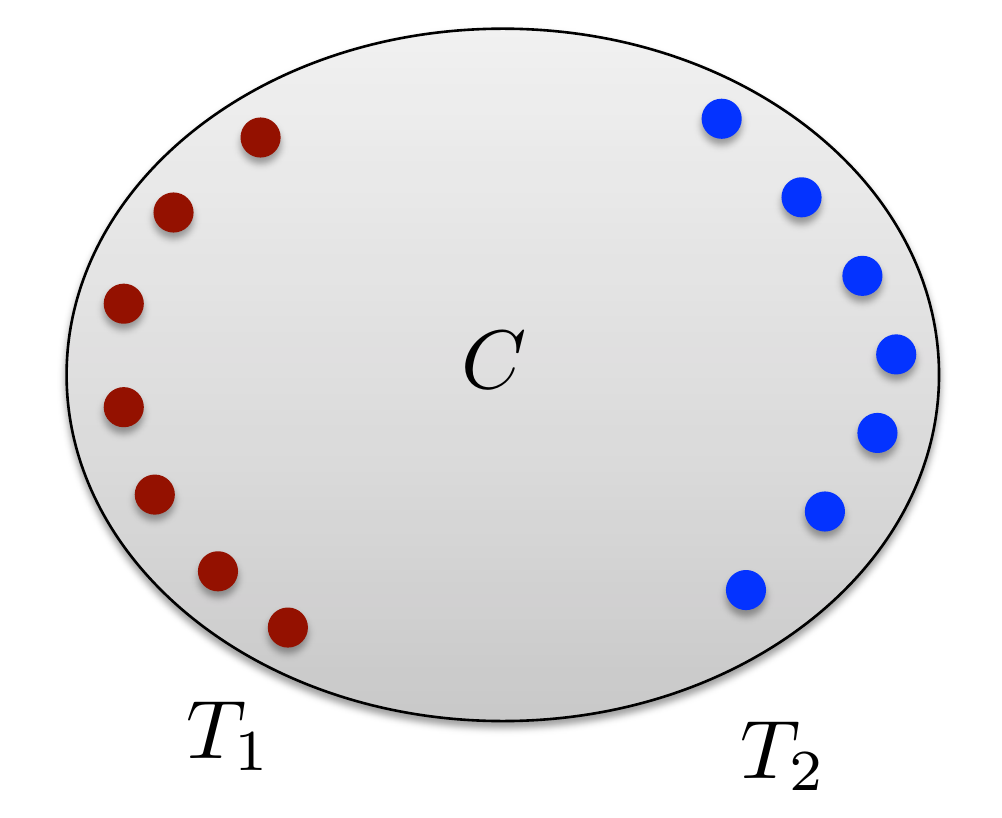}}\label{fig: splitting-start}}
\hspace{0.5cm}
\subfigure[Sequential Splitting]{\scalebox{0.4}{\includegraphics{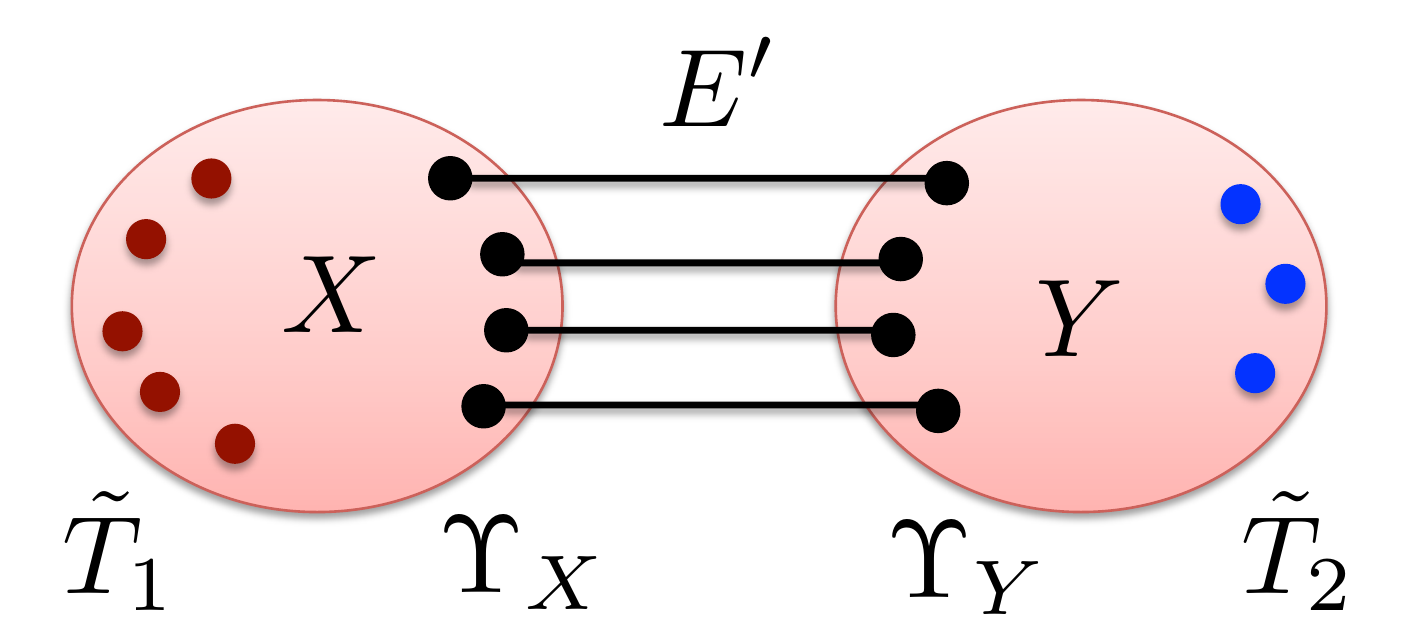}}\label{fig: splitting-sequential}}
\hspace{0.5cm}
\subfigure[Parallel Splitting]{\scalebox{0.4}{\includegraphics{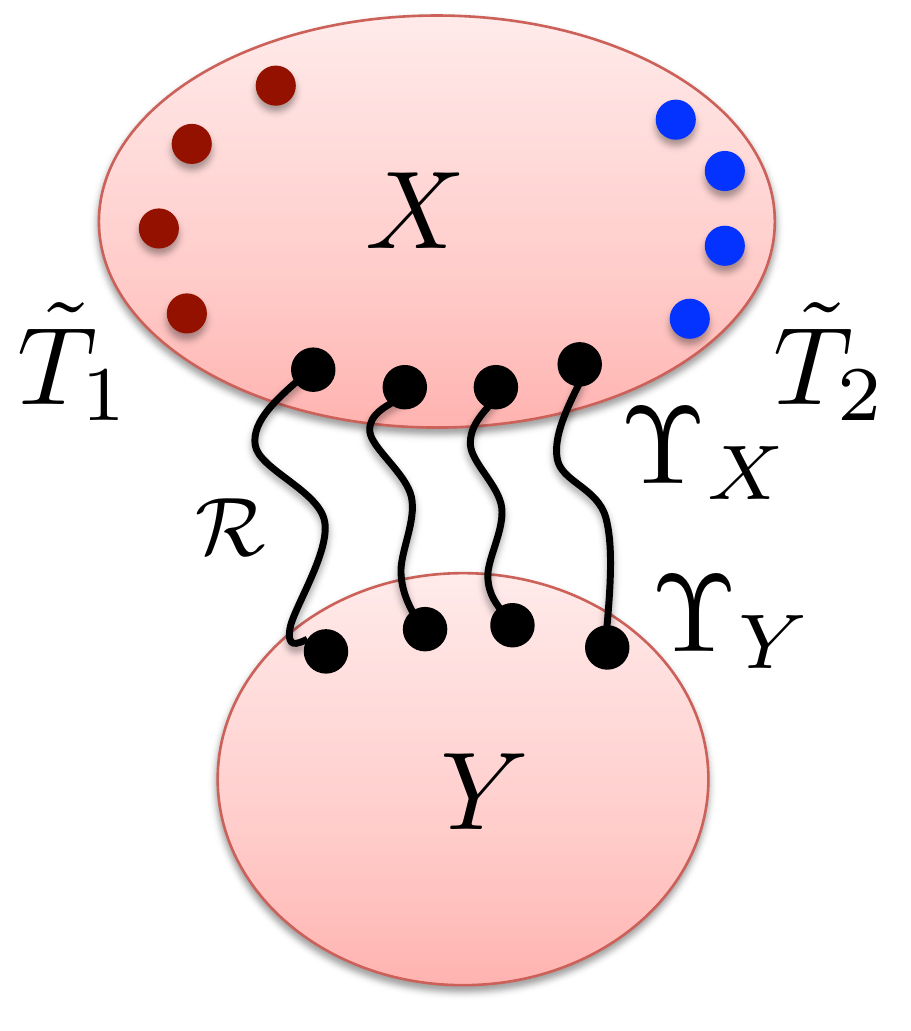}}\label{fig: splitting-parallel}}
\caption{Two ways to split a cluster \label{fig: splitting-cluster}}
\end{figure}

We summarize the new splitting procedure in the following theorem. As before, since we only focus on the subgraph $G[C]$ of $G$, we can ignore the rest of the graph, so we will denote $G[C]$ by $G$ in the next theorem.

\begin{theorem}\label{thm: parallel splitting}
For every integer $d>0$, there is an integer $c(d)>0$ depending only on $d$, such that the following holds. Let $G$ be any graph of maximum vertex degree at most $d$, and let $T_1,T_2$ be two disjoint subsets of vertices of $G$, with $|T_1|=|T_2|=\kappa$, such that $T_1$ and $T_2$ are each node-well-linked in $G$, and $(T_1,T_2)$ are node-linked in $G$. Then there are two disjoint clusters $X,Y\subseteq V(G)$, a set $\rset$ of at least $\kappa/c(d)$ node-disjoint paths connecting vertices of $X$ to vertices of $Y$, so that the paths of $\rset$ are internally disjoint from $X\cup Y$, and two subsets $\tT_1\subsetq T_1\cap X$, $\tT_2\subseteq T_2\cap X$ of at least $\kappa/c(d)$ vertices each such that, if we denote by $\Upsilon_X$ and $\Upsilon_Y$ the endpoint of the paths of $\rset$ lying in $X$ and $Y$ respectively, then:

\begin{itemize}
\item set $\Upsilon_Y$ is node-well-linked in $G[Y]$;

\item each of the tree sets $\tT_1,\tT_2$ and $\Upsilon_X$ is node-well-linked in $G[X]$; 

\item $\tT_1\cup \tT_2\cup \Upsilon_X$ is $1/5$-well-linked in $G[X]$;
and

\item every pair of sets in $\set{\tT_1,\tT_2,\Upsilon_X}$ is node-linked in $G[X]$.
\end{itemize}
\end{theorem}

We defer the proof of Theorem~\ref{thm: parallel splitting} to Section~\ref{sec: parallel splitting}.
%
%------------------------------------------------------------------
%------------------------------------------------------------------
%------------------------------------------------------------------
%------------------------------------------------------------------
%------------------------------------------------------------------
%------------------------------------------------------------------
%
%------------------------------------------------------------------
%------------------------------------------------------------------
%------------------------------------------------------------------
%\subsubsection{Finding the Embedding}
%------------------------------------------------------------------
%------------------------------------------------------------------
%------------------------------------------------------------------
%------------------------------------------------------------------
%------------------------------------------------------------------
%------------------------------------------------------------------
%------------------------------------------------------------------
%------------------------------------------------------------------
%------------------------------------------------------------------
%

From Theorem~\ref{thm: embedding to grid minor}, in order to complete the proof of the Excluded Grid Theorem, it is now enough to show that we can find a variable-width embedding of some graph $H\in \hset$ into $G$. We again employ an induction on the number of levels, and for each $1\leq z\leq \log_qg$, we show that if $G_z$ is a graph of a sufficiently large treewidth, then we can find a variable-width embedding of some graph $H_z\in \hset_z$ into $G_z$. Since this inductive proof involves combining several embeddings together, we need a slightly more general notion of variable-width embeddings, called an \emph{anchored embedding}. This is very similar to the definition of the $T$-anchored \ToS, except that for convenience of notation we define it slightly differently. As before, the graph $G$ that we embed into has a special set $T$ of vertices called terminals, and, in addition to embedding $H$ into $G$ using the standard variable-width embedding, we require that there are many node-disjoint paths connecting the root of $H$ to the vertices of $T$. 
Given a standard variable-width embedding $\phi=\left (\set{S(v)}_{v\in V(H)},\set{\pset(e)}_{e\in E(H)}\right )$ of $H$ into $G$, let $G_{\phi}\subseteq G$ denote the graph $G_\phi=\left(\bigcup_{v\in V(H)}G[S(v)]\right )\cup \left (\bigcup_{e\in E(H)}\pset(e)\right )$.
We use the following definition.

\begin{definition}
Let $G$ be any graph, with a special subset $\tset\subseteq V(G)$ of vertices called terminals, and let $w_0\geq 1$ be an integer. Let $H$ be any graph with non-negative integral width values $w(e)$ for edges $e\in E(H)$, and a special vertex $r\in V(H)$ called the root of $H$, that has degree $1$ in $H$. A $T$-anchored variable-width embedding of $H$ into $G$ with parameter $w_0$ consists of:

\begin{itemize}
\item a variable-width embedding $\phi=\left(\set{S(v)}_{v\in V(H)},\set{\pset(e)}_{e\in E(H)}\right )$ of $H$ into $G$;

%Let $\tG\subseteq G$ be the subgraph of $G$ induced by this embedding, that is, $\tG$ is the union of $G[S(v)]$ for all $v\in V(H)$ and $\pset(e)$ for all $e\in E(H)$.

\item a subset $\tset'\subseteq \tset$ of $w_0$ terminals with $\tset'\cap V(G_{\phi})=\emptyset$; and

\item a set $\pset^*$ of $w_0$ node-disjoint paths in $G$, connecting the vertices of $\tset'$ to some vertices of $S(r)$, such that the paths in $\pset^*$ are internally disjoint from $V(G_{\phi})$.
\end{itemize}

Let $e$ be the unique edge of $H$ incident on $r$, let $U(r)\subseteq S(r)$ be the set of the endpoints of the paths of $\pset(e)$ lying in $S(r)$, and let 
$U'(r)\subseteq S(r)$ be the set of the endpoints of the paths of $\pset^*$ lying in $S(r)$. We say that the embedding is perfect, if the embedding $\phi$ is perfect, and additionally $U'(r)$ is node-well-linked  in $G[S(r)]$ and $(U(r),U'(r))$ are node-linked in $G[S(r)]$.
\end{definition}

%Let $G$ be any graph of treewidth $\kappa=\Omega(g^{19}\poly\log g)$. We use Theorem~\ref{thm: degree reduction} to obtain a subgraph $G'$ of $G$, whose maximum vertex degree is $3$, together with a set $\tset$ of $\kappa^*=\Omega(\kappa/\poly\log \kappa)$ terminals, such that the terminals of $\tset$ are node-well-linked in $G'$, and the degree of every terminal is $1$. From now on we will be working with graph $G'$ only, so to simplify the notation, we denote $G'$ by $G$, and $\kappa^*$ by $\kappa$. We assume that $\kappa>c^*g^{19}$ for some large constant $c^*$. 

Throughout the proof we will be working with a graph $G$ whose maximum vertex degree is at most $3$.
We use three tools in order to embed a graph $H\in \hset$ into $G$. First, we use Corollary~\ref{cor: PoS parameters}, that asserts that, given any graph $G'$ of maximum vertex degree at most $3$, and two disjoint subsets $\tset_1,\tset_2$ of its vertices of cardinality $\kappa'$ each, such that each of $\tset_1,\tset_2$ is node-well-linked in $G'$, $(\tset_1,\tset_2)$ are node-linked in $G'$, and the degrees of the vertices in $T_1\cup T_2$ are at most $2$ in $G'$, there is a perfect \PoS of width $w$ and length $\ell$ in $G'$, if $\kappa'\geq c_p\ell^{17}w$, where $c_p>0$ is some fixed constant. Moreover, if $A_1,B_{\ell}$ are the anchors of the resulting \PoS, then $A_1\subseteq \tset_1$ and $B_{\ell}\subseteq\tset_2$. We note that the corollary requires that $\ell$ is an integral power of $2$, but we can ignore this condition. Indeed, let $\ell'$ be the smallest integral power of $2$ with $\ell'\geq \ell$, so $\ell\leq \ell'\leq 2\ell$. We can obtain a perfect \PoS of width $w$ and length $\ell'$ in $G'$, if for some constant $c$, $\kappa'\geq c w (2\ell)^{17}\geq c_pw(\ell')^{17}$, and then discard the extra clusters.

The second tool is Theorem~\ref{thm: build a ToS system}. For $d=3$, the theorem guarantees that there is some universal constant $c_t$, such that given any graph $G'$ with maximum vertex degree at most $3$, and a set $\tset$ of $\k>2$ terminals that are $1$-well-linked in $G'$, there is a perfect $T$-anchored Tree-of-Sets System of size $\ell$ and width $w$ in $G'$, if $\k'\geq c_t w \ell^5$.

The final tool is Theorem~\ref{thm: parallel splitting}. Since our graph $G$ is sub-cubic, we denote the constant $c(d)$ for $d=3$ from Theorem~\ref{thm: parallel splitting} by $c_s$. We assume without loss of generality that $c_s\geq 1$.

We are now ready to set the value of the constant $q$. We let $q$ be a large enough constant, so $q>2^{200}c^2_pc_sc_t$. Recall that we have defined, for $z\geq 1$, a parameter $W_z(d)=72dgq^{z+1}$, that is used in defining width values for edges in graphs in $\hset_z$. Since the bound $d$ on the maximum vertex degree is $3$ throughout the rest of this proof, we will denote by $W_z=W_z(3)=216gq^{z+1}$.
The main result of this section is the following theorem.

\begin{theorem}\label{theorem: find var-width embedding}
For all $z\geq 1$, if $G_z$ is a graph with maximum vertex degree at most $3$, with a subset $T_z\subseteq V(G_z)$ of $\rho_z=gq^{18z+19}$ vertices, such that $T_z$ is node-well-linked in $G_z$, and the degree of every vertex of $T_z$ is at most $2$ in $G_z$, then there is some graph $H_z\in \hset_z$, and a  perfect $T$-anchored variable-width embedding of $H_z$ into $G_z$, with parameter $w_0=W_z$.
\end{theorem}

Before we prove Theorem~\ref{theorem: find var-width embedding}, observe that for $z=\log_qg$, we get that $\rho_z=gq^{18\log_qg+19}=q^{19}g^{19}$. Therefore, if $G$ is a graph with maximum vertex degree at most $3$, that contains a subset $T\subseteq V(G)$ of $q^{19}g^{19}$ vertices, such that $T$ is node-well-linked in $G$, then there is a variable-width embedding of some graph $H\in \hset$ into $G$, and so from Theorem~\ref{thm: embedding to grid minor}, $G$ contains the $(g\times g)$-grid as a minor. Using Theorem~\ref{thm: degree reduction} to reduce the maximum vertex degree of the input graph $G$ to $3$, we conclude that for some large enough constants $c,c'$, for every $g\geq 2$, if $G$ is a graph of treewidth at least $cg^{19}(\log g)^{c'}$, then $G$ contains the $(g\times g)$-grid as a minor. We now turn to prove Theorem~\ref{theorem: find var-width embedding}.

\subsection{Proof of Theorem~~\ref{theorem: find var-width embedding}}
The proof is by induction on $z$. 

\subsubsection*{Induction base} 
The induction base is when $z=1$. Recall that in this case, $\hset_1$ contains a single graph $H_1$, which is a path containing $4q^2+2$ vertices. Every edge on the path has width $W_1=216gq^{2}$. Let $G_1$ be any graph with maximum vertex degree at most $3$, and a set $T_1\subseteq V(G_1)$ of $\rho_1=gq^{37}$ terminals, which are node-well-linked in $G_1$. Let $T_1',T_1''\subseteq T_1$ be any pair of disjoint subsets of $T_1$ of cardinality $\floor{\rho_1/2}$ each. Since $T_1$ is node-well-linked in $G_1$, it is easy to verify that each of $T_1',T_1''$ is node-well-linked, and $(T_1',T_1'')$ are node-linked in $G_1$. Using  Corollary~\ref{cor: PoS parameters}, we can now obtain  a perfect \PoS of width $w=W_1=216gq^2$ and length $\ell=4q^2+3$, since $\rho_1=gq^{37}$, while $c_p\ell^{17}w=216 c_p gq^2(4q^2+3)^{17}=\Theta(gq^{36})$. Since $q$ is a large enough constant, we can assume that $\floor{\rho_1/2}\geq c_p\ell^{17}w$.

Let $S_1,S_2,\ldots,S_{4q^2+3}$ be the clusters of the \PoS, and for each $1\leq i<4q^2+3$, let $\pset_i$ be the corresponding set of $W_1$ paths, connecting the vertices $B_i\subseteq S_i$ to the vertices $A_{i+1}\subseteq S_{i+1}$. Recall that we are guaranteed that $A_1\subsetq T_1'$. Let the vertices of $H_1$ be denoted by $v_1,\ldots,v_{4q^2+2}$, and assume that they appear on the path in this order, with $v_1=r(H_1)$. For $1\leq i<4q^2+2$, we denote by $e_i$ the edge $(v_i,v_{i+1})$ of $H_1$.

The embedding of $H_1$ into $G_1$ is defined as follows. For each $1\leq i\leq 4q^2+2$, we let $S(v_i)=S_{i+1}$, and for each $1\leq i<4q^2+2$, we let $\pset(e_i)=\pset_{i+1}$. In order to define the set $\pset^*$ of paths, connecting the terminals to $S(v_1)$, recall that $(A_1,B_1)$ are node-linked in $G_1[S_1]$. Therefore, there is a set $\qset: A_1\sconnect B_1$ of node-disjoint paths in $G_1[S_1]$. We then let $\pset^*$ be the concatenation of the paths in $\qset$ and $\pset_1$. These paths are internally disjoint from all clusters $S_2,\ldots,S_{4q^2+3}$, and they are disjoint from all paths in $\bigcup_{i=2}^{4q^2+2}\pset_i$. The properties of the \PoS ensure the required well-linkedness properties inside each cluster $G_1[S(v_i)]$.

\subsubsection*{Induction Step}
We now consider some integer $z>1$, and we assume that the theorem holds for all $z'<z$. Recall that we are given a graph $G_z$ with maximum vertex degree at most $3$, with a subset $T_z\subseteq V(G_z)$ of $\rho_z=gq^{18z+19}$ vertices that we call terminals, such that $T_z$ is node-well-linked in $G_z$,  and the degree of every vertex of $T_z$ is at most $2$. For simplicity of notation, we denote $G_z$ by $G$ and $T_z$ by $T$ from now on.

Let $T',T''\subseteq T$ be any pair of disjoint subsets of $T$ of cardinality $\floor{\rho_z/2}$ each. As before, each set $T',T''$ is node-well-linked,  $(T',T'')$ are node-linked, and the degree of every vertex in $T'\cup T''$ is at most $2$ in $G$. Using  Corollary~\ref{cor: PoS parameters}, we can obtain a perfect \PoS of width $w=\frac{\rho_z}{4^{18}\cdot c_p}$ and length $4$.
Let $\tS_1,\ldots,\tS_4$ be the clusters of the \PoS, and for $1\leq i<4$, let $\tpset_i$ be the set of paths connecting the set $\tB_i\subsetq S_i$ of vertices to the set $\tA_{i+1}\subseteq \tS_{i+1}$ of vertices, and let $\tA_1\subseteq T',\tB_{4}\subseteq T''$ be the anchors of the \PoS (see Figure~\ref{fig: var width embedding 1}). 

At a high level, we will use Corollary~\ref{cor: PoS parameters} in order to split cluster $\tS_2$ into a perfect \PoS, and use it to embed the path $P_z$ of $H_z$. We then use Theorem~\ref{thm: build a ToS system} in graph $G[\tS_4]$, with the vertices of $\tA_4$ serving as terminals to obtain a perfect embedding of some tree $\tau_z$ into $G[\tS_4]$, by constructing a \ToS. Cluster $\tS_3$ is then used in order to embed the edge $(a'_z,b_z)$, and cluster $\tS_1$ is used to construct the set $\pset^*$ of paths, connecting the cluster corresponding to  $r(H_z)$ to the terminals (see Figure~\ref{fig: var width embedding 2}). Finally, we select a subset $V_z$ of $q^2$ vertices of $\tau_z$, and for each such vertex $v\in V_z$, we use 
 Theorem~\ref{thm: parallel splitting}, in order to split $S(v)$ into two clusters, $X(v)$ and $Y(v)$. Cluster $Y(v)$ is then used in order to embed a copy $H_{z-1}(v)$ of some level-$(z-1)$ graph, via the induction hypothesis, while vertex $v$ itself is embedded into $X(v)$.
We now describe these steps in more detail.

\begin{figure}[h]
\centering
\subfigure[The \PoS of length $4$ and width $w=\frac{\rho_z}{4^{18}\cdot c_p}$.]{\scalebox{0.3}{\includegraphics{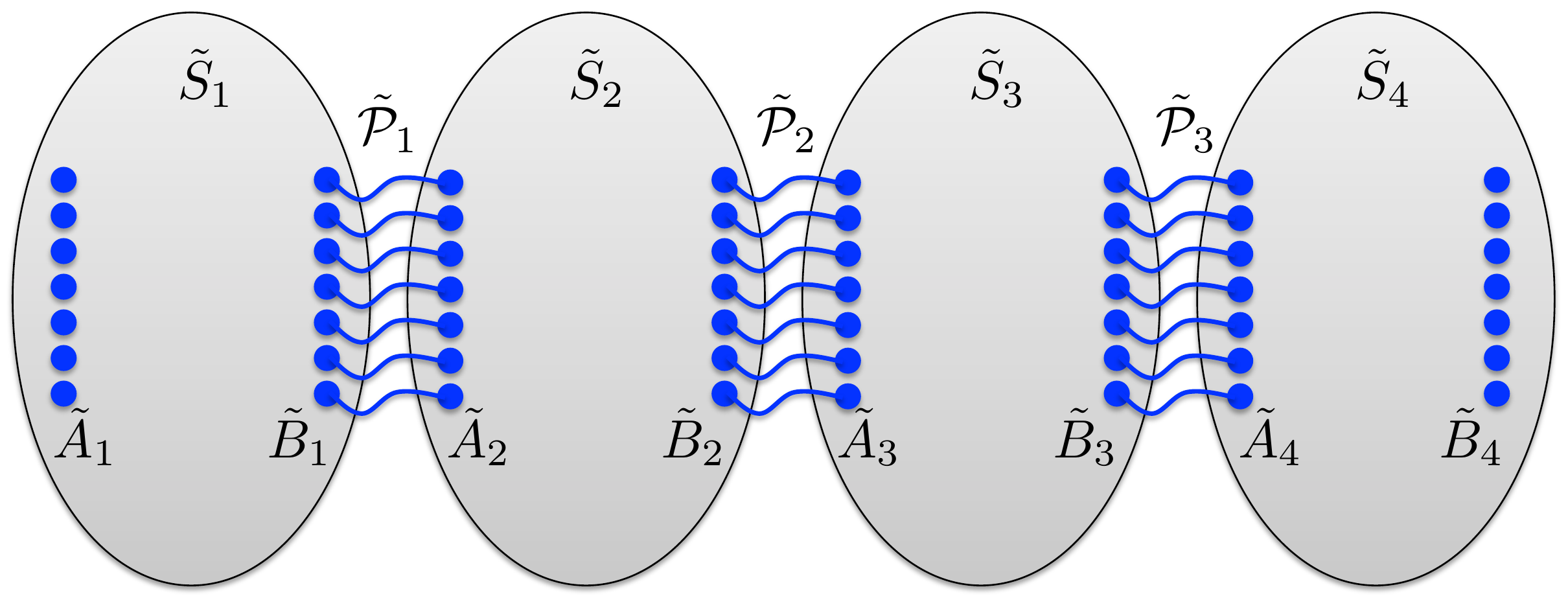}}\label{fig: var width embedding 1}}
\hspace{0.5cm}
\subfigure[The embeddings of path $P_z$ and tree $\tau_z$]{
\scalebox{0.3}{\includegraphics{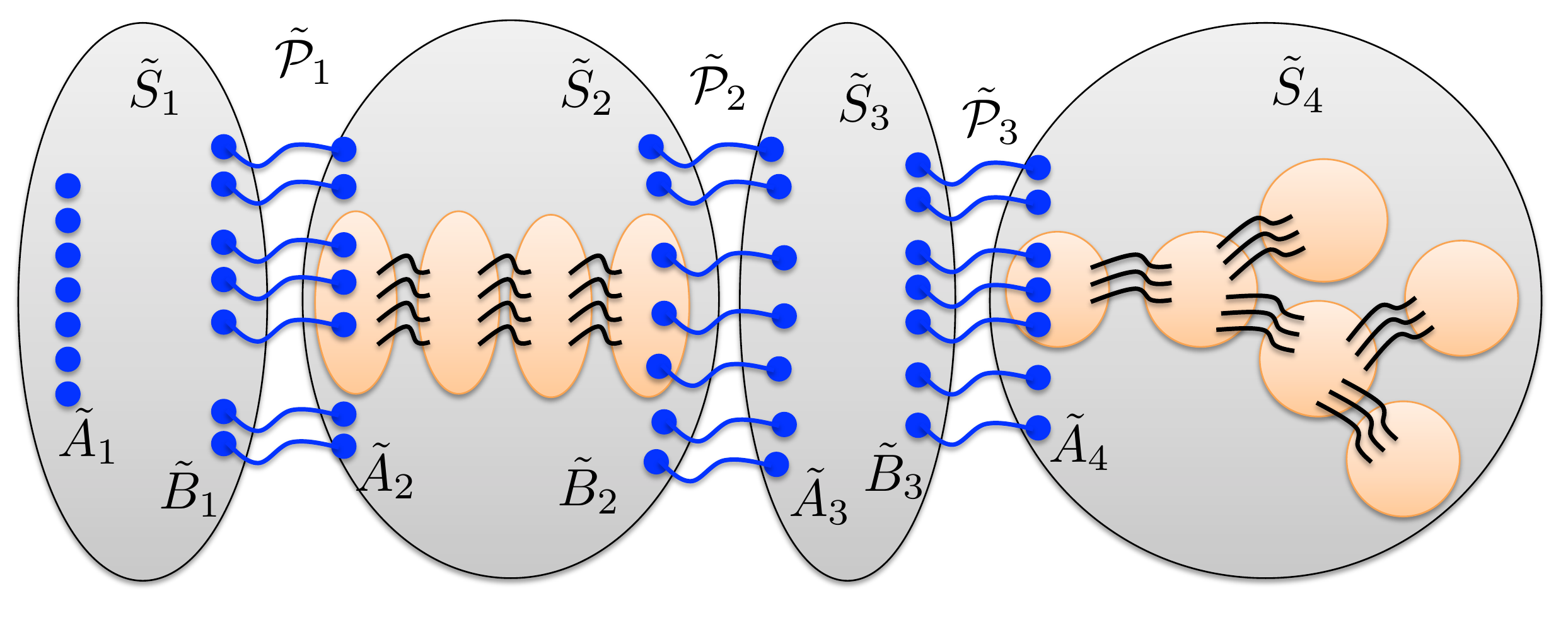}}\label{fig: var width embedding 2}}
\hspace{1cm}
\subfigure[The embedding of the set $\pset^*$ of paths. The sets $\tpset'\subseteq \tpset_1$ and $\qset$ of paths, and the set $X\subseteq \tilde{B}_1$ of vertices are shown in green. ]{
\scalebox{0.33}{\includegraphics{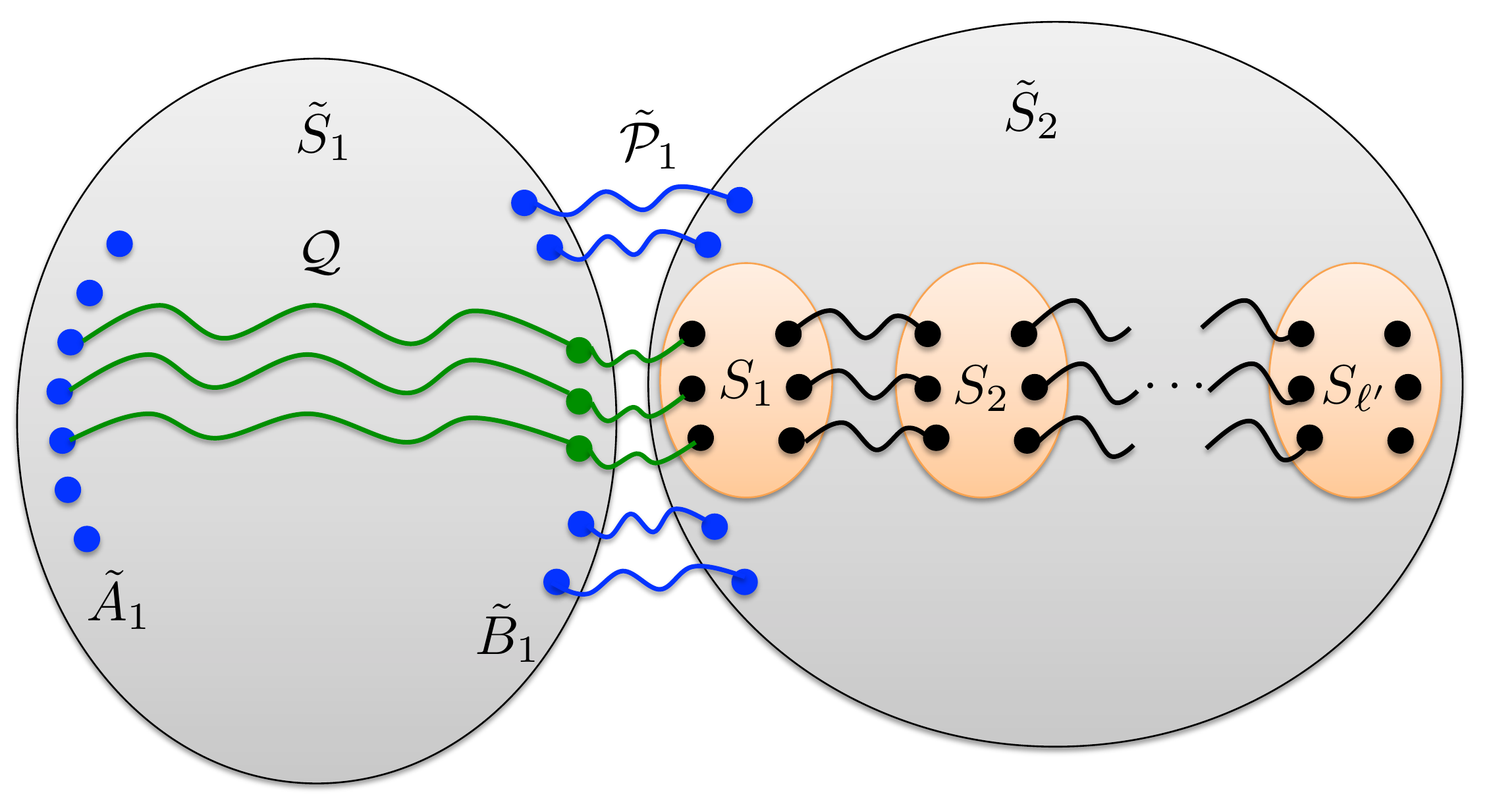}}\label{fig: var width embedding 3}}
\hspace{0.5cm}
\subfigure[Connecting the embeddings of path $P_z$ and tree $\tau_z$. Set $\qset_1\subseteq \pset'(e^*)$ of paths, and the vertices of $R_1$ are shown in green. The set $\qset_2\subseteq \tpset_3$ of paths and the set $R_2\subseteq \tilde B_3$ of vertices are shown in orange. The set $B_{\ell'}\subseteq \tilde B_2$ of vertices and the set $\qset_4\subseteq \tilde P_2$ of paths are shown in red. The vertices of $R_3\subseteq \tilde A_3$ and the paths of $\qset_3$ are shown in brown.]{
\scalebox{0.33}{\includegraphics{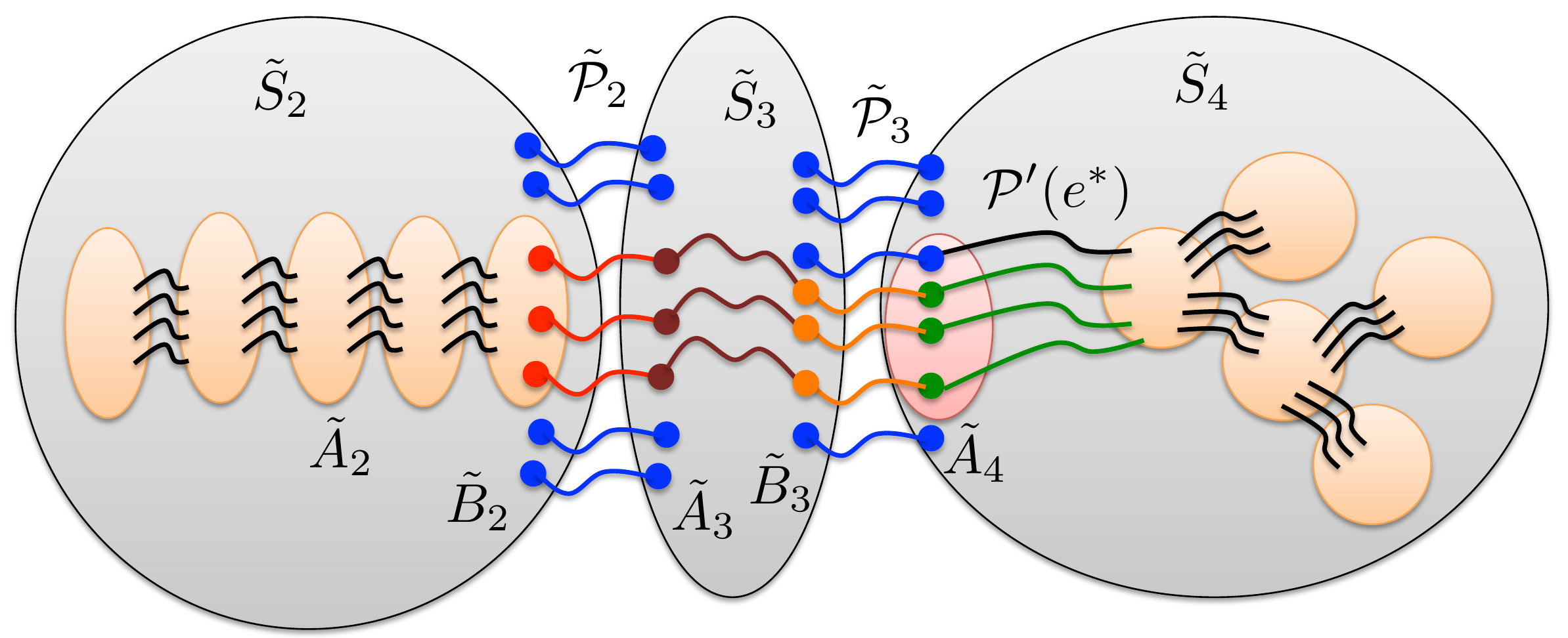}}\label{fig: var width embedding 4}}
\caption{Constructing the variable-width embedding of $H_z$ into $G_z$. \label{fig: var-width embedding}}
\end{figure}

%-----------------------------------------------------------
%-----------------------------------------------------------
%-----------------------------------------------------------
%-----------------------------------------------------------
\paragraph{Step 1: embedding the path $P_z$.}
%-----------------------------------------------------------
%-----------------------------------------------------------
%-----------------------------------------------------------
%-----------------------------------------------------------
 This step is very similar to the induction base. Let $G'=G[\tS_2]$, and recall that $\tA_2,\tB_2$ are disjoint sets of $w=\frac{\rho_z}{ 4^{18}\cdot c_p}$ vertices; sets $(\tA_2,\tB_2)$ are node-linked in $G'$, and each of the two sets is node-well-linked in $G'$. Since all vertex degrees in $G$ are bounded by $3$, every vertex in $\tilde A_2\cup \tilde B_2$ has degree at most $2$ in $G'$. Let $\ell'=8q^{z+1}+1$ and $w'=W_z=216gq^{z+1}$. Since $c_p(\ell')^{17}w'=c_p(8q^{z+1}+1)^{17}\cdot 216gq^{z+1}\leq 2^{76}c_pgq^{18z+18}$, while $w=\frac{\rho_z}{ 4^{18}\cdot c_p}=\frac{gq^{18z+19}}{4^{18}}$, and $q$ is a large enough constant, $w>c_p(\ell')^{17}w'$.
Therefore, from Corollary~\ref{cor: PoS parameters} there is a perfect \PoS $\left(\set{S_i}_{i=1}^{\ell'},\bigcup_{i=1}^{\ell'-1}\pset_i\right )$ of width $w'$ and length $\ell'$ in $G'$, with the anchors $A_1\subseteq \tA_2$ and $B_{\ell'}\subseteq \tB_2$.
Let $v_1,\ldots,v_{\ell'}$ be the vertices of the path $P_z$ in the order in which they appear on the path, with $a_z=v_1$ and $a'_z=v_{\ell'}$.
For $1\leq i\leq \ell'$, we define the embedding $S(v_i)=S_i$, and for all $1\leq i<\ell'$, if $e'_i$ denotes the edge $(v_i,v_{i+1})$, then we embed $\pset(e'_i)=\pset_i$. 

We now define the set $\pset^*$ of $W_z$ paths, connecting $S(a_z)$ to $\tset'$ (see Figure~\ref{fig: var width embedding 3}). Let $\tpset'\subseteq \tpset_1$ be the subset of paths of the length-$4$ \PoS computed in the previous step, that terminate at the vertices of $A_1\subsetq \tA_2$, and let $X\subseteq \tB_1$ denote the endpoints of the paths in $\tpset'$ that lie in $\tS_1$. Since the vertices of $\tA_1,\tB_1$ are node-linked in $G[\tS_1]$, there is a set $\qset$ of node-disjoint paths in $G[\tS_1]$, connecting every vertex of $X$ to some vertex of $\tA_1\subseteq T'$. By concatenating the paths in $\qset$ and $\tpset'$, we obtain a collection $\pset^*$ of $W_z$ node-disjoint paths, connecting the vertices of $A_1\subseteq S[a_z]$ to some vertices of $\tset'$. The properties of the \PoS ensure the desired well-linkedness properties inside every cluster $S(v_i)$ for $1\leq i\leq \ell'$, including for cluster $S(v_1)$ with respect to the sets of the endpoints of the paths of $\pset^*$ and $\pset(e_1')$ that lie in $S(v_1)$.

%-----------------------------------------------------------
%-----------------------------------------------------------
%-----------------------------------------------------------
%-----------------------------------------------------------
\paragraph{Step 2: defining and embedding the tree $\tau_z$.}
%-----------------------------------------------------------
%-----------------------------------------------------------
%-----------------------------------------------------------
%-----------------------------------------------------------
In this step, we focus on graph $G''=G[\tS_4]$, and the set $\tA_4$ of $w=\frac{\rho_z}{ 4^{18}\cdot c_p}$ vertices, that are node-well-linked in $G''$. We let $w''=2\rho_{z-1}\cdot c_S$, and $\ell''=4q^2$. Notice that
$c_t w''(\ell'')^5=2c_tc_s\rho_{z-1}\cdot (4q^2)^5=c_tc_s\cdot 2^{11} gq^{18z+11}$, while $w=\frac{\rho_z}{ 4^{18}\cdot c_p}=\frac{gq^{18z+19}}{4^{18}\cdot c_p}$. Since $q$ is a large enough constant, we get that $w>c_t w''(\ell'')^5$. 
Observe also that $\rho_{z-1}=gq^{18z+1}>W_z=216gq^{z+1}$. From Theorem~\ref{thm: build a ToS system}, there is a perfect $\tA_4$-anchored Tree-of-Sets system of width $w''$ and size $\ell''$ in $G''$. We denote the corresponding tree by $\tau$; the embedding of every vertex $v\in \tau$ is denoted by $S'(v)$, and the embedding of every edge $e\in \tau$ is denoted by $\pset'(e)$.
Let $v^*$ denote the root of the tree $\tau$, and let $b_z$ denote its unique neighbor in $\tau$. Recall that $S'(v^*)\subseteq \tA_4$. Let $e^*$ be the edge $(b_z,v^*)$ of the tree $\tau$. We define $\tau_z=\tau\setminus v^*$.

We select a subset $V_z\subseteq V(\tau_z)\setminus \set{b_z}$ of $q^2$ vertices, as in the construction of the graph $H_z$, so that every vertex in $V_z$ has degree $1$ or $2$, and no edge of $\tau_z$ connects any pair of such vertices. We denote $V_z=\set{u_1,\ldots,u_{q^2}}$. We now complete the embedding of the tree $\tau_z$ into $G''$. 

For every vertex $v\in V(\tau_z)\setminus V_z$, the embedding $S(v)=S'(v)$ (that is, we use the embedding given by the \ToS). Similarly, for every edge $e\in E(\tau_z)$, such that none of the endpoints of $e$ belong to $V_z$, we let $\pset(e)$ contain any subset of $W_z$ paths in $\pset'(e)$ (since $|\pset'(e)|=w''=2\rho_{z-1}\cdot c_S>W_z$, such a set exists).

Consider now some vertex $u_i\in V_z$, and assume first that $u_i$ has degree $2$ in $\tau_z$. Let $e$ be the edge connecting $u_i$ to its parent, and $e'$ be the edge connecting it to its child. Let $T^1_i,T^2_i\subseteq S'(u_i)$ be the sets of the endpoints of the paths of $\pset'(e),\pset'(e')$, respectively, that lie in $S'(u_i)$ (see Figure~\ref{fig: embedding ui - start}). Recall that $|T_i^1|,|T^2_i|=w''=2\rho_{z-1}c_S$; each of the two sets is node-well-linked in $G''[S(v_i)]$, and both sets are node-linked in $G''[S(v_i)]$. We use Theorem~\ref{thm: parallel splitting} in order to compute two disjoint clusters $X_i,Y_i\subseteq S'(v_i)$, with the corresponding subsets $\hat{T}_i^1\subseteq T_i^1\cap X_i$, $\hat{T}_i^2\subseteq T_i^2\cap X_i$ of vertices, of cardinality $\rho_{z-1}$ each, such that each of these sets is node-well-linked in $G''[X_i]$, and both sets are node-linked in $G''[X_i]$. We define the embedding $S(u_i)=X_i$. Edge $e$ is embedded into a subset $\pset(e)\subseteq \pset'(e)$ of $W_z<\rho_{z-1}$ paths that contain vertices of $\hat{T}_i^1$ as their endpoints, and edge $e'$ is embedded into a subset $\pset(e')\subseteq \pset'(e')$ of $W_z$ paths that contain vertices of $\hat{T}_{i}^2$ as their endpoints (see Figure~\ref{fig: embedding ui - finish}).
We denote by $\qset_i$ the set of $2\rho_{z-1}$ paths connecting $X_i$ to $Y_i$ that are given by the splitting procedure. We then discard paths from $\qset_i$, until $|\qset_i|=\rho_{z-1}$ holds, and we denote by $\Upsilon_i\subseteq X_i$ and $\Upsilon'_i\subseteq Y_i$ the endpoints of the paths of $\qset_i$ lying in $X_i$ and $Y_i$, respectively.

\begin{figure}[h]
\centering
\subfigure[The beginning. ]{
\scalebox{0.3}{\includegraphics{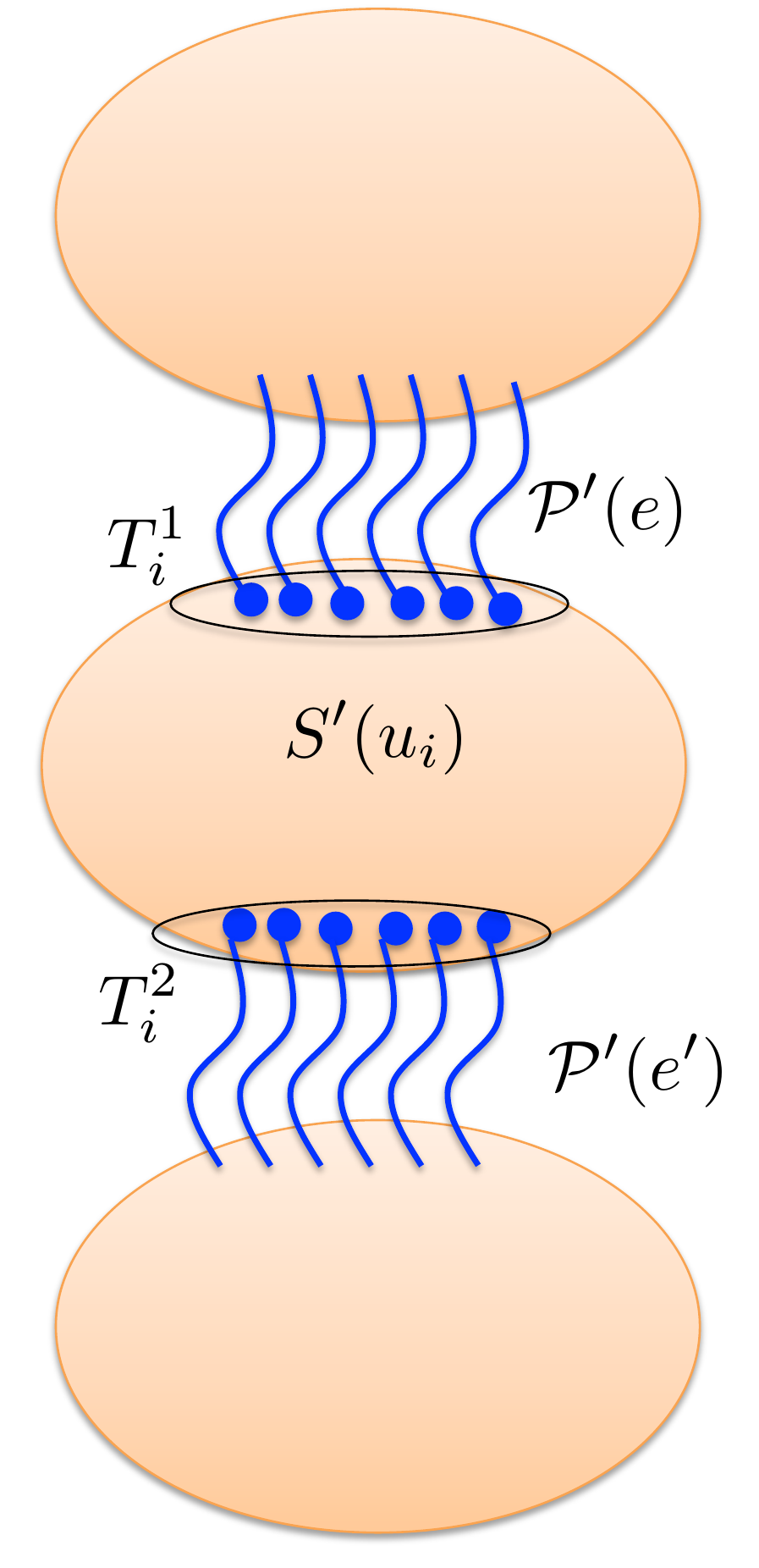}}\label{fig: embedding ui - start}}
\hspace{1cm}
\subfigure[The end. The sets $\pset(e)$, $\pset(e')$ of paths are shown in red.]{
\scalebox{0.3}{\includegraphics{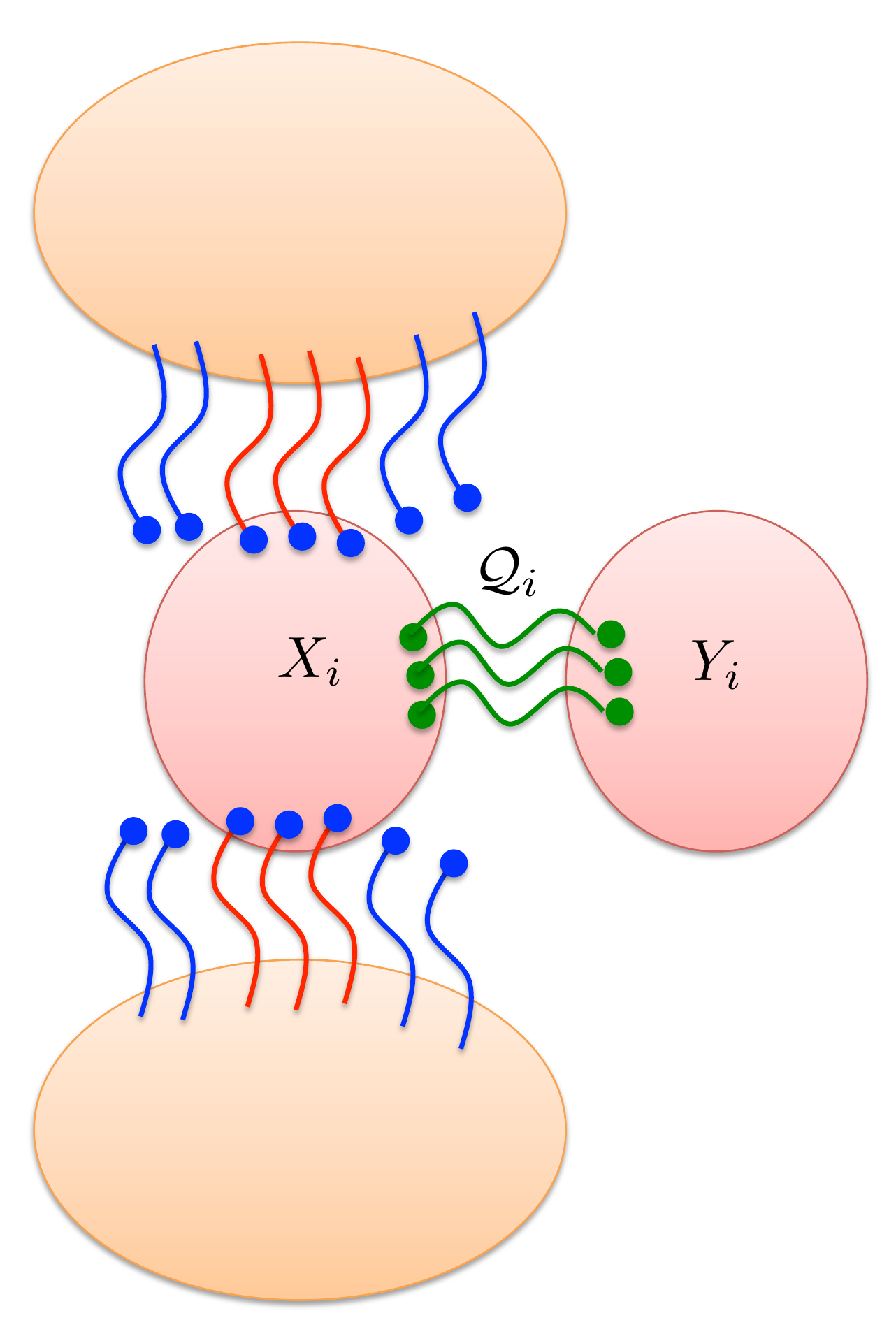}}\label{fig: embedding ui - finish}}
\hspace{1cm}
\subfigure[Connecting the embedding of $H^i_{z-1}$; the paths of $\pset(e'_{u_i})$ are shown in blue]{
\scalebox{0.3}{\includegraphics{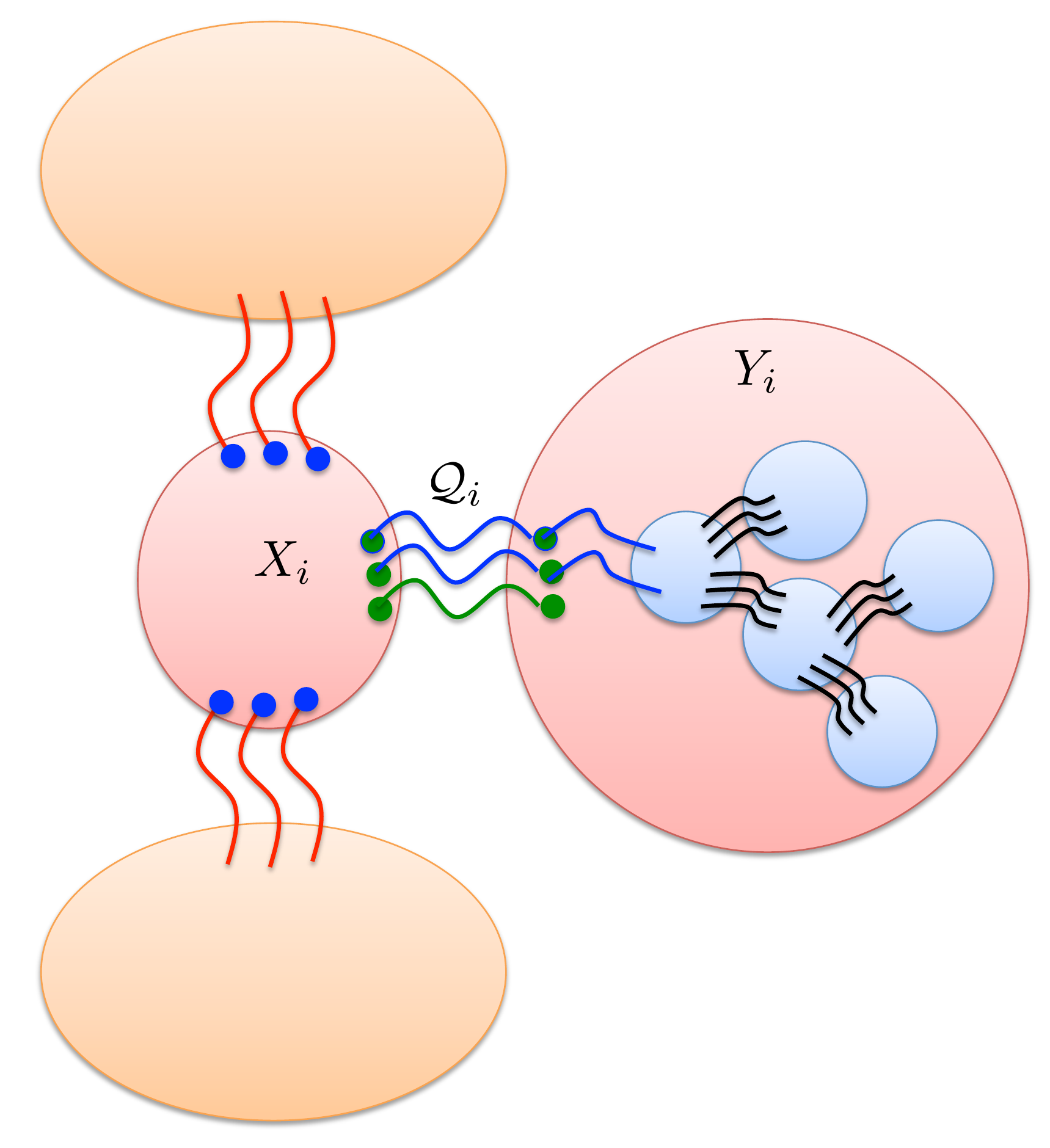}}\label{fig: embedding ui - add next level}}
\caption{{Embedding vertex $u_i\in V_z$.\label{fig: splitting for embedding}}}
\end{figure}

Assume now that $u_i$ has degree $1$ in $\tau_z$, and let $e$ be the unique edge incident on $u_i$. We let $T^1_i,T^2_i\subseteq S'(u_i)$ be two arbitrary disjoint subsets of the endpoints of the paths $\pset'(e)$ that lie in $S'(u_i)$ of cardinality $|\pset'(e)|/2=\rho_{z-1}\cdot c_S$ each. We compute the two clusters $X_i,Y_i\subseteq S'(u_i)$, define the embeddings of $u_i$ and of the edge $e$, the set $\qset_i$ of paths, and the sets $\Upsilon_i,\Upsilon_i'$ of vertices exactly as before.

This completes the definition of the embedding of the tree $\tau_z$ into $G[S_4]$. Let $\hat e$ be the edge $(a'_z,b_z)$ of $H_z$. We now combine the embeddings of $P_z$ and $\tau_z$, by defining an embedding of the edge $\hat e$. Recall that $P_z$ is embedded into $G[\tilde S_2]$, and the set $\pset^*$ of edges is embedded into $\tpset_1\cup G[\tilde S_1]$, while $\tau_z$ is embedded into $G[\tilde S_4]$ (see Figure~\ref{fig: var-width embedding}). We will embed the edge $\hat e$ into the graph $\tpset_2\cup G[\tilde S_3]\cup \tpset_3$ (see Figure~\ref{fig: var width embedding 4}).

Recall that the anchors of the \PoS system computed in Step 1 are $A_1,B_{\ell'}$, with $B_{\ell'}\subseteq \tB_2$, and $|B_{\ell'}|=W_z$. Recall also that the root vertex $v^*$ of the tree $\tau$ has $S'(v^*)\subseteq \tA_4$, and we are given a set $\pset'(e^*)$ of $w''=2c_S\rho_{z-1}>W_z$ paths connecting the vertices of $S'(v^*)$ to some vertices of $S(b_z)$.

Let $\qset_1\subseteq \pset'(e^*)$ be any subset of $W_z$ paths, and let $R_1\subseteq \tA_4$ be the set of their endpoints lying in $S'(v^*)$. Let $\qset_2\subseteq \tpset_3$ be the set of paths terminating of the vertices of $R_1$, and let $R_2\subseteq \tB_3$ be the set of their endpoints lying in $\tilde S_3$. Let $\qset_4\subseteq \tpset_2$ be the set of paths originating from the vertices of $B_{\ell'}$ (the anchors of the \PoS), and let $R_3\subseteq \tA_3$ be the set of their endpoints lying in $\tS_3$. Since the sets $(\tA_3,\tB_3)$ are node-linked in $G[\tS_3]$, there is a set $\qset_3:R_2\sconnect R_3$ of node-disjoint paths in $G[\tS_3]$ (see Figure~\ref{fig: var width embedding 4}). By combining the paths of $\qset_1,\qset_2,\qset_3$ and $\qset_4$, we obtain a collection of node-disjoint paths, that we denote by $\pset(\hat e)$, connecting the vertices of $B_{\ell'}\subseteq S(a'_z)$ to the vertices of $S(b_z)$. We use the set $\pset(\hat e)$ of paths in order to embed the edge $\hat e=(a'_z,b_z)$. It is easy to see that the paths in $\pset(\hat e)$ are completely disjoint from all other paths into which the edges of $P_z\cup \tau_z$ were embedded so far, and from the paths of $\pset^*$, and they are internally disjoint from all clusters into which the vertices of $P_z\cup \tau_z$ were embedded. The paths in $\pset(\hat e)$ are also disjoint from $\bigcup_{i=1}^{q^2}Y_i$.

%-----------------------------------------------------------
%-----------------------------------------------------------
%-----------------------------------------------------------
%-----------------------------------------------------------
\paragraph{Step 3: embedding level-$(z-1)$ graphs.}
%-----------------------------------------------------------
%-----------------------------------------------------------
%-----------------------------------------------------------
%-----------------------------------------------------------
We now fix some $1\leq i\leq q^2$, and denote $G_{z-1}^i=G_z[Y_i]$. Recall that we are given a set $\Upsilon_i'$ of $\rho_{z-1}$ vertices of $Y_i$, that are node-well-linked in $G_{z-1}^i$. It is easy to see that the degree of every vertex in $\Upsilon_i'$ is at most $2$ in $G_{z-1}^i$. By the induction hypothesis, there is some graph $H_{z-1}^i\in \hset_{z-1}$, and a perfect $\Upsilon'_i$-anchored embedding of $H_{z-1}^i$ into $G_{z-1}^i$, with parameter $w_0=W_{z-1}$. We add $H_{z-1}^i$ to the graph $H_z$ that we are constructing, and connect its root vertex $r(H_{z-1}^i)$ to $u_i$ with an edge that we denote by $e_{u_i}'$. The embedding of every vertex and edge of $H_{z-1}^i$ remains unchanged. It now only remains to define the embedding of $e_{u_i}'$ (see Figure~\ref{fig: embedding ui - add next level}). Recall that the $\Upsilon'_i$-anchored embedding of $H_{z-1}^i$ into $G_{z-1}^i$ defines a set $\pset^*_i$ of $W_{z-1}$ paths, connecting some vertices of $S(r(H_{z-1}^i))$ to some subset $\Upsilon''_i\subseteq \Upsilon'_i$ of $W_{z-1}$ vertices. Let $\qset'_i\subseteq \qset_i$ be the subset of paths (connecting $X_i$ to $Y_i$, that we have computed when splitting $S'(u_i)$), that contain the vertices of $\Upsilon''_i$ as their endpoints. We then set the embedding of $e_{u_i}'$ to be the set  $\pset(e_{u_i}')$ of paths, obtained by concatenating the paths of $\pset^*$ with the paths of $\qset'_i$.
This concludes the construction of the level-$z$ graph $H_z\in \hset_z$, and its $T_z$-anchored variable-width embedding into $G_z$. It is immediate to verify that the resulting embedding is perfect, due to the properties of the perfect \PoS, the perfect Tree-of-Sets system, and the guarantees given by Theorem~\ref{thm: parallel splitting}.

%-----------------------------------------------------------
%-----------------------------------------------------------
%-----------------------------------------------------------
%-----------------------------------------------------------
%-----------------------------------------------------------
%-----------------------------------------------------------
\label{-------------------------------------------sec: parallel splitting---------------------------------------}
\section{Parallel Cluster Splitting}\label{sec: parallel splitting}
%-----------------------------------------------------------
%-----------------------------------------------------------
%-----------------------------------------------------------
%-----------------------------------------------------------
%-----------------------------------------------------------
%-----------------------------------------------------------
This section is devoted to the proof of Theorem~\ref{thm: parallel splitting}.
%
%\begin{proof}
Let $T=T_1\cup T_2$. We refer to the vertices of $T$ as terminals. The following observation follows immediately from the well-linkedness properties of $T_1$ and $T_2$ in $G$.

\begin{observation}\label{obs: well-linkedness of T}
The vertex set $T$ is $1/3$-well-linked in $G$.
\end{observation}

 Let $G'$ be the smallest (with respect to edge-deletion) subgraph of $G$, such that $T$ is $1/3$-well-linked in $G'$. Notice that it is enough to find the clusters $X$ and $Y$ and the set $\rset$ of paths with the required properties (including the well-linkedness) in graph $G'$. In order to simplify the notation, we denote $G'$ by $G$ from now on.
We use the following lemma, which slightly generalizes and strengthens the Deletable Edge Lemma of Chekuri, Khanna and Shepherd~\cite{deletable-edge-original}.

\begin{lemma}\label{lemma: deletable edge2}
Let $H$ be any graph, $T\subseteq V(H)$ any subset of its vertices, such that for some $0<\alpha<1$, $T$ is $\alpha$-well-linked in $H$, and $H$ is a minimal graph with respect to edge-deletion in which $T$ is $\alpha$-well-linked. Let $\Gamma\subseteq V(H)\setminus T$ be another subset of vertices, so that $\Gamma$ is $\alpha'$-well-linked in $H$, for some $0<\alpha'<1$. Then there is a set $\pset$ of $\floor{\alpha'|\Gamma|/3}$ edge-disjoint paths that connect vertices of $T$ to vertices of $\Gamma$ in $H$.
\end{lemma}

%-----------------------------------------------
%-----------------------------------------------
%-----------------------------------------------
%\subsection{Proof of Lemma~\ref{lemma: deletable edge}}
%-----------------------------------------------
%-----------------------------------------------
%-----------------------------------------------
\begin{proof}
We follow the proof of~\cite{deletable-edge} almost exactly (slightly tightening their bounds).
We can assume that $|\Gamma|\geq 3/\alpha'$, as otherwise the claim is trivial.
Assume for contradiction that no such set $\pset$ of paths exists. Let $E'$ be a minimum-cardinality set of edges, such that no path connects a vertex of $T$ to a vertex of $\Gamma$  in $H\setminus E'$, and denote $|E'|=\gamma$. Since we have assumed that $\pset$ does not exist, $\gamma\leq \floor{\alpha'|\Gamma|/3}-1$. Let $U$ be the union of all connected components of $H\setminus E'$, that contain vertices of $\Gamma$. Then $T\cap U=\emptyset$, and $|\out_H(U)|\leq |E'|=\gamma$.

Let $M\subseteq V(H)$ be a set of vertices that has the following properties:

\begin{itemize}
\item $M\cap T=\emptyset$;

\item $|M\cap \Gamma|\geq |\Gamma|/2$;

\item $|\out_H(M)|\leq \gamma$;

\item $|M|$ is minimum among all sets satisfying the above properties.
\end{itemize}

Notice that set $U$ satisfies the first three properties, so $M$ is well-defined. Since $|M\cap \Gamma|\geq |\Gamma|/2$, while $|\out_H(M)|\leq \gamma <\alpha'|\Gamma|/3$, from the $\alpha'$-well-linkedness of $\Gamma$, $H[M]$ must contain at least one edge. We claim that any such edge is deletable, in the following claim.

\begin{claim}\label{claim: deletable}
Let $e$ be any edge with both endpoints in $M$. Then $T$ remains $\alpha$-well-linked in $H\setminus \set{e}$.
\end{claim}

Notice that the above claim contradicts the minimality of $H$. Therefore, in order to complete the proof of Lemma~\ref{lemma: deletable edge2}, it is now enough to prove Claim~\ref{claim: deletable}.

\begin{proof}
Assume otherwise.
Then from Observation~\ref{obs: well-linkedness alt def}, there is a partition $(A,B)$ of $V(H)$, with $|A\cap \tset|\leq |B\cap \tset|$, and $|E_H(A,B)|<\alpha |\tset\cap A|+1$, such that $e\in E_H(A,B)$. Let $Z=A\cap M$ and $Z'=B\cap M$ (see Figure~\ref{fig: deletable edge}). 

\begin{figure}[h]
\scalebox{0.4}{\includegraphics{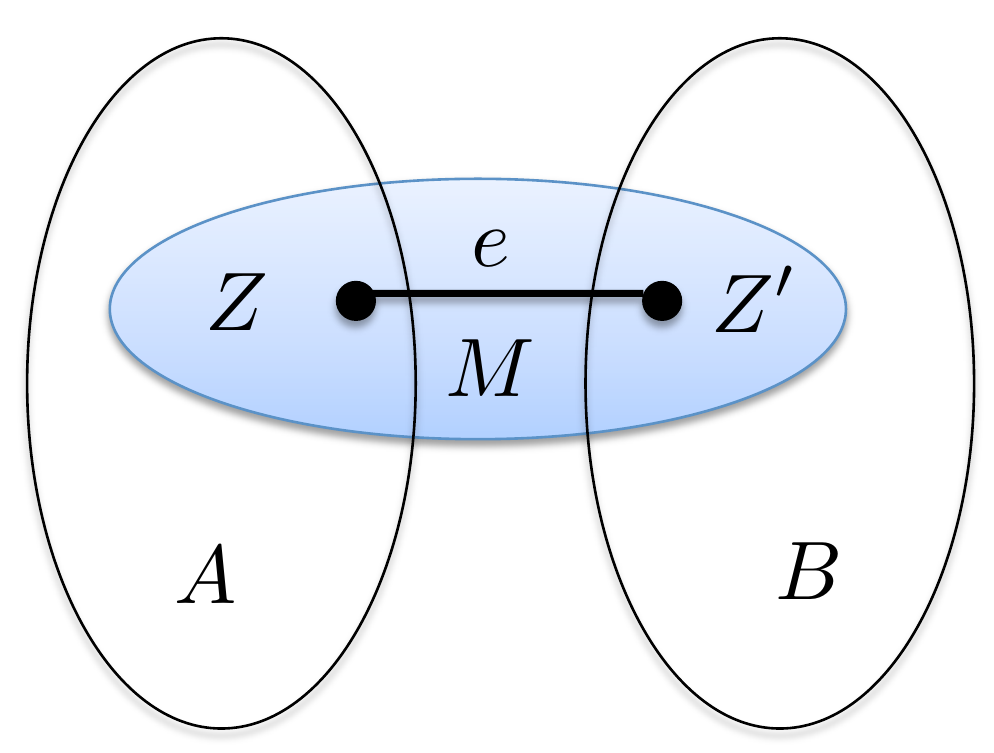}}
\caption{Illustration for Claim~\ref{claim: deletable}\label{fig: deletable edge}}
\end{figure}

From the sub-modularity of cuts,

\[|\out_H(A)|+|\out_H(M)|\geq |\out_H(A\cup M)|+|\out_H(A\cap M)|=|\out_H(A\cup M)|+|\out_H(Z)|,\]

and

\[|\out_H(A)|+|\out_H(M)|\geq |\out_H(A\setminus M)|+|\out_H(M\setminus A)|=|\out_H(A\setminus M)|+|\out_H(Z')|.\]

Recall that $|\out_H(A)|=|E_H(A,B)|<\alpha |\tset\cap A|+1$, while $|\out(M)|\leq \gamma$. On the other hand, $|\out(A\cup M)|\geq \alpha |\tset\cap A|$, from the well-linkedness of $\tset$, and since $M\cap \tset=\emptyset$. Therefore, $|\out_H(Z)|<\gamma+1$. Since $|\out(Z)|$ and $\gamma$ are both integers, $|\out_H(Z)|\leq \gamma$ must hold. Similarly, since $|\out_H(A\setminus M)|\geq \alpha|\tset\cap A|$ (from the well-linkedness of $\tset$), we get that $|\out_H(Z')|<\gamma+1$, and $|\out_H(Z')|\leq \gamma$. From the minimality of $M$, $|Z\cap \Gamma|,|Z'\cap \Gamma|<|\Gamma|/2$ must hold. But since $|\out_H(M)|\leq \gamma$ and $|M\cap \Gamma|\geq |\Gamma|/2$, we get that $M$ must contain at least $|\Gamma|-\gamma/\alpha'$ vertices of $\Gamma$ (from the $\alpha'$-well-linkedness of $\Gamma$). Therefore, $|Z\cap \Gamma|+|Z'\cap \Gamma|\geq |\Gamma|-\gamma/\alpha'$. Assume w.l.o.g. that $|Z\cap \Gamma|\geq |Z'\cap \Gamma|$. Then $|Z\cap \Gamma|\geq \frac{|\Gamma|}{2}-\frac{\gamma}{2\alpha'}$, and since $|Z\cap \Gamma|\leq |\Gamma|/2$, from the well-linkedness of $\Gamma$, $|\out_H(Z)|\geq \alpha'|Z\cap \Gamma|\geq \frac{\alpha'|\Gamma|}{2}-\frac{\gamma}{2}>\gamma$, a contradiction.
\end{proof}
\end{proof}

\begin{corollary}\label{cor: small out-degree of well-linked cluster}
Let $S\subseteq V(G)$ be any cluster, such that $\Gamma(S)$ is $\alpha'$-well-linked in $G[S]$ for some $0<\alpha'<1$. Then $|\Gamma(S)|\leq 12\kappa d/\alpha'$.
\end{corollary}

\begin{proof}
Assume otherwise. Since  $\Gamma(S)$ is $\alpha'$-well-linked in $G[S]$, it is also $\alpha'$-well-linked in $G$. From Lemma~\ref{lemma: deletable edge2}, there must be a set of more than $2\kappa d$ edge-disjoint paths connecting the terminals in $T$ to $\Gamma(S)$. But $|T|=2\kappa$, and every terminal is incident on at most $d$ edges, so this is impossible.
\end{proof}

We now turn to complete the proof of Theorem~\ref{thm: parallel splitting}. Let $\rho=\frac{\kappa}{64}$. Given any subset $S\subseteq V(G)$ of vertices of $G$, we denote by $\Gamma'(S)=\Gamma(S)\cup (T\cap S)$.

The rest of the proof consists of three steps. In the first step, we compute the set $Y\subseteq V(G)$, so that $\rho/4\leq |\Gamma'(Y)|\leq \rho$, and $\Gamma'(Y)$ is $\alpha(d)$-well-linked in $G[Y]$ for some constant $\alpha(d)$ (that depends on $d$). In the second step, we compute the set $X$, and an initial set $\rset'$ of paths connecting $X$ to $Y$. In the final step, we boost well-linkedness inside the two resulting clusters, to obtain the final set $\rset$ of paths, and the final subsets $\tT_1$ and $\tT_2$ of terminals.

\paragraph{Step 1: Constructing Cluster $Y$.}
  This step is summarized in the following lemma.

\begin{lemma}\label{lemma: finding Y}
There is a cluster $Y\subseteq V(G)$, such that $\Gamma'(Y)$ is $\alpha(d)$-well-linked in $G[Y]$ for some parameter $\alpha(d)=\Omega(1/\poly(d))$, with $0<\alpha(d)<1$, and $\rho/4\leq |\Gamma'(Y)|\leq \rho$. 
\end{lemma}

\begin{proof}
 Let $S_0\subseteq V(G)$ be the smallest, inclusion-wise, subset of vertices, such that $|\Gamma'(S_0)|\geq \rho/4$, and $\Gamma'(S_0)$ is $1/9$-well-linked in $G[S_0]$. Note that such a set $S_0$ exists, since, for example, we can take $S_0=V(G)$. From Corollary~\ref{cor: small out-degree of well-linked cluster}, $|\Gamma'(S)|\leq 110\kappa d$. For integers $i\geq 0$, let $\alpha_i=\frac{1}{3^{i+2}}$.
 
We perform a number of iterations. The input to the $i$th iteration is a set $S_{i-1}\subseteq S_0$ of vertices, such that $\Gamma'(S_{i-1})$ is $\alpha_{i-1}$-well-linked in $G[S_{i-1}]$, and $\rho/4\leq |\Gamma'(S_{i-1})|\leq |\Gamma'(S_0)|\cdot \left(\frac 7 8\right )^{i-1}$. Notice that $S_0$ is a valid input to the first iteration. The algorithm terminates once we compute a set $S_i$ with $\rho/4\leq |\Gamma'(S_i)|\leq \rho$.

The $i$th iteration is executed as follows. If $|\Gamma'(S_{i-1})|\leq \rho$, then we terminate the algorithm. Assume now that $|\Gamma'(S_{i-1})|> \rho$, and let $(A_i,B_i)$ be the minimum $1/4$-balanced cut of $S_{i-1}$ in $G$ with respect to $\Gamma'(S_{i-1})$. Assume without loss of generality that $|A_i\cap \Gamma'(S_{i-1})|\geq |B_i\cap \Gamma'(S_{i-1})|$, and let $F\subseteq A_i$ be the set of vertices incident on the edges of $E(A_i,B_i)$. 
We need the following two observations.

\begin{observation}
Set $\Gamma'(A_i)$ is $\alpha_i$-well-linked in $G[A_i]$.
\end{observation}

\begin{proof}
Let $H$ be the graph obtained from $G[S_{i-1}]$, by adding, for every vertex $v\in \Gamma'(S_{i-1})$, an edge $e_v$, whose one endpoint is $v$, and the other endpoint is a new vertex. Then $\Gamma'(S_{i-1})=\Gamma_H(S_{i-1})$, and so $S_{i-1}$ has the $\alpha_{i-1}$-bandwidth property in $H$, and moreover $(A_i,B_i)$ is the minimum $1/4$-balanced cut of $S_i$ in $H$ with respect to $\Gamma_H(S_{i-1})$. From Lemma~\ref{lemma: balanced cut large piece wl}, $A_i$ has the $\alpha'$-bandwidth property in $H$, for $\alpha'=\frac{\alpha_{i-1}}{2-\alpha_{i-1}}\geq \frac{\alpha_{i-1}}{3}=\alpha_i$, and since $\Gamma'(A_i)=\Gamma_H(A_i)$, set $\Gamma'(A_i)$ is $\alpha_i$-well-linked in $G[A_i]$.
\end{proof}

\begin{observation}\label{obs: size goes down}
$|F|\leq |\Gamma'(S_{i-1})|/8$.
\end{observation}

Assume for now that Observation~\ref{obs: size goes down} is correct; we prove it below. Then $|\Gamma'(A_i)|\leq 3|\Gamma'(S_{i-1})|/4+|F|\leq 7|\Gamma'(S_{i-1})|/8$, and, since $|\Gamma'(S_{i-1})|>\rho$, we get that $|\Gamma'(A_i)|\geq |\Gamma'(S_{i-1})|/2\geq \rho/2$. Therefore, $S_i$ is a valid input to the next iteration. Let $z$ be the index of the last iteration, and let $Y=S_z$ be the output of the algorithm. Since $|\Gamma'(S_0)|\leq 110\kappa d$, and for all $0< i\leq z$, $|\Gamma'(S_i)|\leq 7|\Gamma'(S_{i-1})|/8$, while $\rho=\kappa/64$, we get that $z\leq \log_{8/7} 7040d$. We are then guaranteed that $\rho/4\leq |\Gamma'(Y)|\leq \rho$, and $\Gamma'(Y)$ is $\alpha(d)$-well-linked in $G[Y]$, where $\alpha(d)=\frac{1}{3^{z+2}}=\Omega\left(\frac{1}{\poly(d)}\right )$.

It now remains to prove Observation~\ref{obs: size goes down}.

\begin{proofof}{Observation~\ref{obs: size goes down}}
For convenience, we denote $S_{i-1}$ by $S$ and $\Gamma'(S_{i-1})$ by $\Gamma'$. Assume for contradiction that $|F|> |\Gamma'|/8$. Then $|E(A_i,B_i)|> |\Gamma'|/8$, and so for every $1/4$-balanced partition $(A,B)$ of $S$ in $G$ with respect to $\Gamma'$, $|E(A,B)|>|\Gamma'|/8$ must hold. We will show that there is a subset $S'\subsetneq S$, such that $|\Gamma'(S')|\geq\rho/4$, and $\Gamma'(S')$ is $1/9$-well-linked in $G[S']$. Since $S\subseteq S_0$, this will contradict the choice of $S_0$.

In order to do so, we employ standard well-linked decomposition techniques. Throughout the algorithm, we maintain a subset $U\subseteq S$ of vertices, starting with $U=S$, and a subset $E'\subseteq E(S)$ of edges (that we delete from $G[S]$), also starting with $E'=\emptyset$. For accounting purposes, we associate every vertex $v\in S$ with a budget $\beta(v)$, that may change throughout the algorithm as set $U$ changes, as follows: if $v\in \Gamma'(U)$, then $\beta(v)=1/8$, and otherwise $\beta(v)=0$. Throughout the algorithm, we will maintain the invariant that $\sum_{v\in U}\beta(v)+|E'|\leq |\Gamma'|/8$. The invariant clearly holds at the beginning of the algorithm.

An iteration is executed as follows. While $\Gamma'(U)$ is not $1/9$-well-linked in $G[U]$, let $(A,B)$ be any violating partition: that is, if we assume that $|A\cap \Gamma'(U)|\geq |B\cap \Gamma'(U)|$, then $|E(A,B)|<|B\cap \Gamma'(U)|/9$. We add the edges of $E(A,B)$ to $E'$, set $U=A$, and continue to the next iteration. We now verify that the invariant continues to hold. Let $\Gamma_1=A\cap \Gamma'(U)$, and let $\Gamma_2$ be the endpoints of the edges of $E(A,B)$ that lie in $A$. Let $\Gamma_3=B\cap \Gamma'(U)$. Notice that $\Gamma'(A)=\Gamma_1\cup \Gamma_2$. The changes to the budgets of the vertices are the following: the budget of every vertex in $\Gamma_3$ decreases by $1/8$, and the budget of every vertex in $\Gamma_2$ increases by $1/8$. Therefore, in total, the budgets of the vertices in $\Gamma_3$ decrease by $|\Gamma_3|/8$, and, since $|\Gamma_2|\leq |E(A,B)|< |\Gamma_3|/9$, the total increase in the budgets of the vertices in $\Gamma_2$, and the number of edges in $E'$ is bounded by:

\[|\Gamma_2|/8+|E(A,B)|\leq 9|E(A,B)|/8\leq |\Gamma_3|/8,\]

and so the invariant continues to hold. Let $S'$ be the final set $U$ obtained at the end of the algorithm, so $\Gamma'(S')$ is $1/9$-well-linked in $G[S']$. We claim that $|\Gamma'(S')|\geq |\Gamma'|/2$. Indeed, assume otherwise. Consider the last iteration of the algorithm, such that $|\Gamma'(U)|>|\Gamma'|/2$ held at the beginning of the iteration, and let $(A,B)$ be the partition of $U$ computed in that iteration. Then $|\Gamma'(A)|\geq |\Gamma'(U)|/2\geq |\Gamma'|/4$, and $|\Gamma'(A)|\leq 3|\Gamma'|/4$. Therefore, $(A,S\setminus A)$ is a $1/4$-balanced partition of $S$ with respect to $\Gamma'(S)$. However, $E(A,S\setminus A)\subseteq E'$, so $|E(A,S\setminus A)|\leq |\Gamma'|/8$, contradicting our assumption that for every $1/4$-balanced cut $(A,B)$ of $S$ with respect to $\Gamma'(S)$, $|E(A,B)|>|\Gamma'|/8$. Therefore, $|\Gamma'(S')|\geq |\Gamma'|/2\geq \rho/4$ must hold.

We conclude that there is a subset $S'\subsetneq S$, such that $|\Gamma'(S')|\geq \rho/4$, and $\Gamma'(S')$ is $1/9$-well-linked in $G[S']$, contradicting the choice of $S_0$. Therefore, $|F|\leq |\Gamma'(S_{i-1})|/8$ must hold.
\end{proofof}
\end{proof}

\paragraph{Step 2: Constructing Cluster $X$.}
In this step, we prove the following lemma.

\begin{lemma}\label{lemma: finding X}
There is a cluster $X\subseteq V(G)\setminus Y$, such that $|T_1\cap X|,|T_2\cap X|\geq \kappa/2$, and $\Gamma'(X)$ is $1/33$-well-linked in $G[X]$.
\end{lemma}

\begin{proof}
The proof uses standard techniques, and is very similar to the proof of Observation~\ref{obs: size goes down}. Throughout the algorithm, we maintain a subset $U\subseteq V(G)\setminus Y$ of vertices, starting with $U=V(G)\setminus Y$, and a subset $E'\subseteq E(G)$ of edges, starting with $E'=\out_G(Y)$. For accounting purposes, we associate every vertex $v\in V(G)$ with a budget $\beta(v)$, that may change throughout the algorithm as set $U$ changes, as follows: if $v\in \Gamma'(U)$, then $\beta(v)=1/32$, and otherwise $\beta(v)=0$. Throughout the algorithm, we will maintain the invariant that $\sum_{v\in U}\beta(v)+|E'|<\kappa/12$, and $\out(U)\subseteq E'$.
Notice that at the beginning of the algorithm, $|E'|\leq \rho=\kappa/64$, and $|\Gamma'(U)|\leq |\tset|+\rho\leq 2\kappa+\kappa/64=129\kappa/64$. It is now easy to verify that the invariant holds at the the beginning of the algorithm.

An iteration is executed as follows. While $\Gamma'(U)$ is not $1/33$-well-linked in $G[U]$, let $(A,B)$ be any violating partition: that is, if we assume that $|A\cap \Gamma'(U)|\geq |B\cap \Gamma'(U)|$, then $|E(A,B)|<|B\cap \Gamma'(U)|/33$. We add the edges of $E(A,B)$ to $E'$, set $U=A$, and continue to the next iteration. We now verify that the invariant continues to hold. If $\out(U)\subseteq E'$ held at the beginning of the iteration, then clearly $\out(A)\subseteq E'$ holds at the end of the iteration. Let $\Gamma_1=A\cap \Gamma'(U)$, and let $\Gamma_2$ be the endpoints of the edges of $E(A,B)$ that lie in $A$. Let $\Gamma_3=B\cap \Gamma'(U)$. Notice that $\Gamma'(A)=\Gamma_1\cup \Gamma_2$. The changes to the budgets of the vertices are the following: the budget of every vertex in $\Gamma_3$ decreases by $1/32$, and the budget of every vertex in $\Gamma_2$ increases by $1/32$. Therefore, in total, the budgets of the vertices in $\Gamma_3$ decrease by $|\Gamma_3|/32$, and, since $|\Gamma_2|\leq |E(A,B)|< |\Gamma_3|/33$, the total increase in the budgets of the vertices in $\Gamma_2$, and the number of edges in $E'$ is bounded by:

\[|\Gamma_2|/32+|E(A,B)|\leq 33|E(A,B)|/32\leq |\Gamma_3|/32,\]

and so the invariant continues to hold. Let $X$ be the final set $U$ obtained at the end of the algorithm, so $\Gamma'(X)$ is $1/33$-well-linked in $G[X]$. Since $|T_1|=|T_2|=\kappa$, it is now enough to show that $|X\cap T|\geq 3\kappa/2$. We do so using the following claim.

\begin{claim}\label{claim: X has many terminals}
$|X\cap T|\geq 3\kappa/2$.
\end{claim}

%Before we provide the proof of Claim~\ref{claim: X has many terminals}, we show that it completes the proof of Lemma~\ref{lemma: finding X}. Indeed, assume for contradiction that $|X\cap T|<1.5\kappa$, and denote $X'=V\setimnus X$. Then $|X'\cap T|\geq \kappa/2$ must hold, and, since the terminals are $1/3$-well-linked, and $X$ contains at least half the terminals from Claim~\ref{claim: X has many terminals}, $|\out(X')|\geq \kappa/6$ must hold. However, $\out(X')\subseteq \out(X)\subseteq E'$, and from our invariant, $|E'|<\kappa/6$, a contradiction. It now remains to complete the proof of Claim~\ref{claim: X has many terminals}.

\begin{proof}
Assume otherwise. Notice that $|Y\cap \tset|\leq |\Gamma'(Y)|\leq \rho=\kappa/64$, so at the beginning of the algorithm, $|U\cap \tset|>3\kappa/2$. Consider the last iteration of the algorithm, such that at the beginning of the iteration $|U\cap T|\geq 3\kappa/2$ held. Recall that $\out(U)\subseteq E'$, and so $|\Gamma(U)|\leq |\out(U)|\leq |E'|<\kappa/12$, while $|\Gamma'(U)|\geq |U\cap \tset|\geq 3\kappa/2$. Consider the partition $(A,B)$ of $U$ computed in the iteration, so $|\Gamma'(A)\cap \Gamma'(U)|\geq |\Gamma'(U)|/2\geq 3\kappa/4$. 
Since $|\Gamma'(A)\cap \Gamma(U)|\leq |\Gamma(U)|<\kappa/12$, we get that $|\Gamma(A)\cap T|> 3\kappa/4-\kappa/12=2\kappa/3$. 
Since the terminals are $1/3$-well-linked, $A$ contains at least $2\kappa/3$ terminals, and $V(G)\setminus A$ contains at least $\kappa/2$ terminals, we get that $|\out(A)|\geq \kappa/6$ must hold. But from our invariant, $\out(A)\subseteq E'$ and $|E'|<\kappa/12$, a contradiction.
\end{proof}
\end{proof}

\paragraph{Step 3: Connecting the Clusters and Boosting Well-Linkedness.}
 To summarize, so far we have shown the existence of two disjoint clusters $X,Y\subseteq V(G)$, such that $\rho/4\leq |\Gamma'(Y)|\leq \rho$ for $\rho=\kappa/64$, and $\Gamma'(Y)$ is $\alpha(d)$-well-linked in $G[Y]$, for some constant $0<\alpha(d)<1$ that depends on $d$. Additionally, $|T_1\cap X|,|T_2\cap X|\geq \kappa/2$, and $\Gamma'(X)$ is $1/33$-well-linked in $G[X]$.
 We need the following lemma.

\begin{lemma}\label{lemma: paths}
There is a set $\qset'$ of $\kappa_1=\Omega(\frac{\kappa\alpha(d)}{d^2})$ node-disjoint paths, connecting vertices of $X$ to vertices of $Y$ in $G$, so that the paths in $\qset'$ are internally disjoint from $X\cup Y$.
\end{lemma}

\begin{proof}
We use the following observation.

\begin{observation}
$\abs{\Gamma(Y)}\geq \frac{\rho}{24 d}$.
\end{observation}

\begin{proof}
Assume first that  $|\tset\cap Y|\geq \rho/8$. Then, since the terminals of $T$ are $1/3$-well-linked, and $X$ contains at least half the terminals, $\abs{\out(Y)}\geq \rho/24$ must hold. Since the maximum vertex degree in $G$ is bounded by $d$, and every edge in $\out(Y)$ is incident on some vertex of $\Gamma(Y)$, we get that $\abs{\Gamma(Y)}\geq \frac{\rho}{24d}$.

Therefore, we can assume that $\abs{\tset\cap Y}<\rho/8$. Then $\abs{\Gamma(Y)}\geq \abs{\Gamma'(Y)}-\abs{T\cap Y}\geq \rho/4-\rho/8\geq \rho/8$.
\end{proof}

Notice that it is possible that $\Gamma(Y)\cap \tset\neq \emptyset$. If $|\Gamma(Y)\cap \tset|\geq \frac{\rho}{48d}$, then there is a set $\qset_1$ of $\frac{\rho}{48d}$ node-disjoint paths connecting vertices of $\Gamma(Y)$ to vertices of $\tset$, where every path consists of a single vertex. Otherwise, we can use Lemma~\ref{lemma: deletable edge2} to conclude that there is a set $\qset_1$ of at least $\floor{\frac{\alpha(d) \rho}{144d}}=\Omega(\frac{\kappa\alpha(d)}{d})$ edge-disjoint paths connecting vertices of $\Gamma(Y)$ to vertices of $T$. We select a subset $\qset_1'\subseteq \qset_1$ of $\Omega(\frac{\kappa\alpha(d)}{d^2})$ paths, so that they all terminate at distinct vertices of $T$, and we denote by $\kappa'=|\qset_1'|$. Let $\tset'\subseteq \tset$ be the subset of vertices where the paths of $\qset_1'$ terminate, and let $\tset''\subseteq \tset\cap X$ be any subset of $\kappa'$ terminals distinct from the terminals of $\tset'$ (since $X$ contains at least $3\kappa/2$ terminals, such a set exists). Since the terminals are $1/3$-well-linked, there is a set $\qset_2: \tset'\sconnect_3\tset''$ of paths in $G$. Combining the paths in $\qset_1'$ and $\qset_2$, we obtain a collection of $\kappa'$ paths, connecting some vertices of $\Gamma(Y)$ to the terminals of $\tset''$, with total edge-congestion at most $4$.  Using Observation~\ref{obs: low cong flow to NDP}, there is a set $\qset'$ of at least $\frac{\kappa'}{4d}$ node-disjoint paths, connecting some vertices of $\Gamma(Y)$ to some vertices of $X$. By suitably truncating these paths, we can ensure that they are internally disjoint from $Y$ and $X$.
\end{proof}

So far we have constructed two clusters, $X$ and $Y$, and a collection $\qset'$ of $\kappa_1$ node-disjoint paths, connecting vertices of $\Gamma(Y)$ to vertices of $\Gamma(X)$, so that the paths in $\qset'$ are internally disjoint from $X\cup Y$. Our construction also guarantees that $\Gamma(Y)$ is $\alpha(d)$-well-linked in $G[Y]$, and $\Gamma'(X)$ is $1/33$-well-linked in $G[X]$. Moreover, $|T_1\cap X|,|T_2\cap X|\geq \kappa/2$ must hold. Our last step is to boost well-linkedness. 

Let $\Upsilon_Y'\subseteq Y$ be the subset of vertices of $\Gamma(Y)$ that serve as endpoints of the paths in $\qset'$. Since the vertices of $\Upsilon_Y'$ are $\alpha(d)$-well-linked in $G[Y]$, we can use Theorem~\ref{thm: grouping} in order to find a subset $\Upsilon_Y''\subseteq \Upsilon_Y'$ of $\Omega(\alpha(d)\kappa_1/d)$ vertices, that are node-well-linked in $G[Y]$. Let $\qset''\subseteq \qset'$ be the set of paths in which the vertices of $\Upsilon_Y''$ participate. Let $\Upsilon_X''\subseteq \Gamma(X)$ be the set of vertices of $X$ that serve as endpoints of the paths in $\qset''$. Since the vertices of $\Upsilon_X''$ are $1/33$-well-linked in $G[X]$, using Theorem~\ref{thm: grouping}  we can find a subset $\Upsilon'''_X\subseteq \Upsilon_X''$ of $\Omega(\alpha(d)\kappa_1/d^2)$ vertices that are node-well-linked in $G[X]$. Let $\kappa_3=|\Upsilon'''_X|$, and let $\qset'''\subseteq\qset''$ be the subset of paths in which the vertices of $\Upsilon'''_X$ participate.

Similarly, we can find subsets $\ttset_1'\subseteq \tset_1\cap X$,  $\ttset_2'\subseteq \tset_2\cap X$, such that $\ttset_1',\ttset_2'$ are each node-well-linked in $G[X]$, and $|\ttset_1'|=\abs{\ttset_2'}=\kappa_3$.

Let $\kappa_4=\frac{\kappa_3}{132d}$
In our final step, we select arbitrary subsets $\ttset_1\subseteq \ttset_1',\ttset_2\subseteq\ttset_2'$ containing $\kappa_4$ vertices each, and set $\qset\subseteq \qset'''$ containing $\kappa_4$ paths each. Let $\Upsilon_X,\Upsilon_Y$ denote the endpoints of the paths in $\qset$ that lie in $X$ and $Y$ respectively. Then from Theorem~\ref{thm: linkedness from node-well-linkedness}, every pair of sets in $\set{\ttset_1,\ttset_2,\Upsilon_X}$ is node-linked in $G[X]$. We have already ensured that $\ttset_1,\ttset_2$ and $\Upsilon_X$ are each node-well-linked in $G[X]$, and $\Upsilon_Y$ is node-well-linked in $G[Y]$. The cardinalities of the sets $\ttset_1,\ttset_2$ and $\qset$ are $\kappa_4=\Omega((\alpha(d))^2\kappa/d^2)=\Omega(\kappa/\poly(d))$.
We now only need to verify that $\ttset_1\cup \ttset_2\cup \Upsilon_X$ is $1/5$-well-linked in $G[X]$.
\begin{observation}
Set $\ttset_1\cup \ttset_2\cup \Upsilon_X$ is $1/5$-well-linked in $G[X]$.
\end{observation}

\begin{proof}
Let $A,B$ be any pair of disjoint equal-sized subsets of $\ttset_1\cup \ttset_2\cup \Upsilon_X$, and assume that $|A|=|B|=z$.

Let $z_1=\min\set{|A\cap \ttset_1|, |B\cap \ttset_1|}$, and let $A_1\subseteq A\cap \ttset_1,B_1\subseteq B\cap \ttset_2$ be any subsets of $z_1$ vertices each. Similarly, we let $z_2=\min{|A\cap \ttset_2|, |B\cap \ttset_2|}$, and $z_3=\min\set{|A\cap \Upsilon_X|, |B\cap \Upsilon_X|}$,
and define subsets $A_2\subseteq A\cap \ttset_2,B_2\subseteq B\cap \ttset_2$ of cardinality $z_2$ each, and subsets $A_3\subseteq A\cap \Upsilon_X$ and $B_3\subseteq B\cap \Upsilon_X$ similarly.

Since each of the three sets is node-well-linked, for each $1\leq i\leq 3$, there is a set $\pset_i: A_i\sconnect B_i$ of node-disjoint paths in $G$.

Let $A'=A\setminus(A_1\cup A_2\cup A_3)$, and let $B'=B\setminus (B_1\cup B_2\cup B_3)$. Then at least one of the two sets $A'$ or $B'$ must be contained in one of the three sets $\ttset_1,\ttset_2$ or $\Upsilon_X$. We assume without loss of generality that $A'$ only contains vertices of $\ttset_1$, and so $B'$ only contains vertices of $\ttset_2$ and $\Upsilon_X$. 

Let $B'_1=B'\cap \ttset_2$ and $B'_2=B'\setminus B'_1$. Let $z'_1=|B'_1|$. We partition $A'$ arbitrarily into two subsets $A'_1$ of cardinality $z'_1$, and $A'_2$ of cardinality $|A'|-z'_1$. Since every pair of sets $\ttset_1,\ttset_2,\Upsilon_X$ is node-linked in $G[X]$, there are two sets of node-disjoint paths in $G$: $\pset_4: A'_1\sconnect B'_1$ and $\pset_5: A'_2\sconnect B'_2$. Overall, we obtain a set $\pset=\bigcup_{i=1}^5\pset_i$ of paths, connecting every vertex of $A$ to a distinct vertex of $B$ with edge-congestion at most $5$. 
\end{proof}

{\bf Acknowledgement.} The author thanks Chandra Chekuri for many extensive discussions.

%##################################################
%----------------------------------------
%----------------------------------------

%----------------------------------------
%----------------------------------------
%----------------------------------------
%----------------------------------------
%----------------------------------------
%----------------------------------------
%----------------------------------------
%----------------------------------------
%----------------------------------------
%----------------------------------------
%----------------------------------------
%----------------------------------------
%----------------------------------------
%----------------------------------------
%----------------------------------------
%----------------------------------------
%----------------------------------------
%----------------------------------------

\iffull
\bibliographystyle{alpha}
\fi

\ifabstract
\bibliographystyle{plain}
\fi

\bibliography{improved-improved-GMT-v7}

\begin{thebibliography}{DJGT99}

\bibitem[And10]{Andrews}
Matthew Andrews.
\newblock Approximation algorithms for the edge-disjoint paths problem via
  {Raecke} decompositions.
\newblock In {\em Proceedings of IEEE FOCS}, pages 277--286, 2010.

\bibitem[CC14]{CC14}
Chandra Chekuri and Julia Chuzhoy.
\newblock Polynomial bounds for the grid-minor theorem.
\newblock In {\em Proceedings of the 46th Annual ACM Symposium on Theory of
  Computing}, STOC '14, pages 60--69, New York, NY, USA, 2014. ACM.

\bibitem[CC15]{tw-sparsifiers}
Chandra Chekuri and Julia Chuzhoy.
\newblock Degree-3 treewidth sparsifiers.
\newblock In Piotr Indyk, editor, {\em Proceedings of the Twenty-Sixth Annual
  {ACM-SIAM} Symposium on Discrete Algorithms, {SODA} 2015, San Diego, CA, USA,
  January 4-6, 2015}, pages 242--255. {SIAM}, 2015.

\bibitem[CE13]{ChekuriE13}
Chandra Chekuri and Alina Ene.
\newblock Poly-logarithmic approximation for maximum node disjoint paths with
  constant congestion.
\newblock In {\em Proceedings of the Twenty-Fourth Annual ACM-SIAM Symposium on
  Discrete Algorithms}, pages 326--341. SIAM, 2013.

\bibitem[CHR03]{CHR}
M.~Conforti, R.~Hassin, and R.~Ravi.
\newblock Reconstructing edge-disjoint paths.
\newblock {\em Operations Research Letters}, 31(4):273--276, 2003.

\bibitem[Chu12]{Chuzhoy11}
Julia Chuzhoy.
\newblock Routing in undirected graphs with constant congestion.
\newblock In {\em Proc.\ of ACM STOC}, pages 855--874, 2012.

\bibitem[Chu15]{GMT-STOC}
Julia Chuzhoy.
\newblock Excluded grid theorem: Improved and simplified.
\newblock In {\em Proceedings of the Forty-Seventh Annual ACM on Symposium on
  Theory of Computing}, STOC '15, pages 645--654, New York, NY, USA, 2015. ACM.

\bibitem[CK09]{ChekuriK09}
Chandra Chekuri and Nitish Korula.
\newblock A graph reduction step preserving element-connectivity and
  applications.
\newblock In {\em Proc.\ of ICALP}, pages 254--265, 2009.

\bibitem[CKS04]{deletable-edge-original}
Chandra Chekuri, Sanjeev Khanna, and F.~Bruce Shepherd.
\newblock A deletable edge lemma for general graphs.
\newblock Manuscript, 2004.

\bibitem[CKS05]{CKS}
Chandra Chekuri, Sanjeev Khanna, and F.~Bruce Shepherd.
\newblock Multicommodity flow, well-linked terminals, and routing problems.
\newblock In {\em Proc.\ of ACM STOC}, pages 183--192, 2005.

\bibitem[CKS13]{ANF}
Chandra Chekuri, Sanjeev Khanna, and F.~Bruce Shepherd.
\newblock The all-or-nothing multicommodity flow problem.
\newblock {\em SIAM Journal on Computing}, 42(4):1467--1493, 2013.

\bibitem[CL12]{ChuzhoyL12}
Julia Chuzhoy and Shi Li.
\newblock A polylogarithimic approximation algorithm for edge-disjoint paths
  with congestion 2.
\newblock In {\em Proc.\ of IEEE FOCS}, 2012.

\bibitem[CNS13]{deletable-edge}
Chandra Chekuri, Guyslain Naves, and F.~Bruce Shepherd.
\newblock Maximum edge-disjoint paths in k-sums of graphs.
\newblock In Fedor~V. Fomin, Rusins Freivalds, Marta~Z. Kwiatkowska, and David
  Peleg, editors, {\em Automata, Languages, and Programming - 40th
  International Colloquium, {ICALP} 2013, Riga, Latvia, July 8-12, 2013,
  Proceedings, Part {I}}, volume 7965 of {\em Lecture Notes in Computer
  Science}, pages 328--339. Springer, 2013.

\bibitem[DH07a]{DemaineH-survey}
E.D. Demaine and M~Hajiaghayi.
\newblock {The Bidimensionality Theory and Its Algorithmic Applications}.
\newblock {\em The Computer Journal}, 51(3):292--302, November 2007.

\bibitem[DH07b]{DemaineH07}
Erik~D Demaine and MohammadTaghi Hajiaghayi.
\newblock {Quickly deciding minor-closed parameters in general graphs}.
\newblock {\em European Journal of Combinatorics}, 28(1):311--314, January
  2007.

\bibitem[DHK09]{DemaineHK09}
Erik Demaine, MohammadTaghi Hajiaghayi, and Ken-ichi Kawarabayashi.
\newblock Algorithmic graph minor theory: Improved grid minor bounds and
  {Wagner's} contraction.
\newblock {\em Algorithmica}, 54:142--180, 2009.

\bibitem[Die12]{Diestel-book}
Reinhard Diestel.
\newblock {\em Graph Theory, 4th Edition}, volume 173 of {\em Graduate texts in
  mathematics}.
\newblock Springer, 2012.

\bibitem[DJGT99]{DiestelJGT99}
Reinhard Diestel, Tommy~R. Jensen, Konstantin~Yu. Gorbunov, and Carsten
  Thomassen.
\newblock Highly connected sets and the excluded grid theorem.
\newblock {\em J. Comb. Theory, Ser. B}, 75(1):61--73, 1999.

\bibitem[FST11]{FominST11}
Fedor~V. Fomin, Saket Saurabh, and Dimitrios~M. Thilikos.
\newblock Strengthening {Erdos-P\'osa} property for minor-closed graph classes.
\newblock {\em Journal of Graph Theory}, 66(3):235--240, 2011.

\bibitem[HO96]{element-connectivity}
H.~R. Hind and O.~Oellermann.
\newblock Menger-type results for three or more vertices.
\newblock {\em Congressus Numerantium}, 113:179--204, 1996.

\bibitem[KK12]{KawarabayashiK-grid}
K.~Kawarabayashi and Y.~Kobayashi.
\newblock {Linear min-max relation between the treewidth of H-minor-free graphs
  and its largest grid minor}.
\newblock In {\em Proc.\ of STACS}, 2012.

\bibitem[KRV09]{KRV}
Rohit Khandekar, Satish Rao, and Umesh Vazirani.
\newblock Graph partitioning using single commodity flows.
\newblock {\em J. ACM}, 56(4):19:1--19:15, July 2009.

\bibitem[KT10]{KreutzerT10}
Stephan Kreutzer and Siamak Tazari.
\newblock On brambles, grid-like minors, and parameterized intractability of
  monadic second-order logic.
\newblock In {\em Proc.\ of ACM-SIAM SODA}, pages 354--364, 2010.

\bibitem[LS14]{LeafS12}
Alexander Leaf and Paul Seymour.
\newblock Tree-width and planar minors.
\newblock {\em Journal of Combinatorial Theory, Series B}, 2014.

\bibitem[Mad78]{edge-connectivity}
W.~Mader.
\newblock A reduction method for edge connectivity in graphs.
\newblock {\em Ann. Discrete Math.}, 3:145--164, 1978.

\bibitem[R{\"a}c02]{Raecke}
Harald R{\"a}cke.
\newblock Minimizing congestion in general networks.
\newblock In {\em Proc.\ of IEEE FOCS}, pages 43--52, 2002.

\bibitem[Ree97]{Reed-chapter}
Bruce Reed.
\newblock {\em Surveys in Combinatorics}, chapter Treewidth and Tangles: A New
  Connectivity Measure and Some Applications.
\newblock London Mathematical Society Lecture Note Series. Cambridge University
  Press, 1997.

\bibitem[RS86]{RS-grid}
Neil Robertson and P~D Seymour.
\newblock {Graph minors. V. Excluding a planar graph}.
\newblock {\em Journal of Combinatorial Theory, Series B}, 41(1):92--114,
  August 1986.

\bibitem[RS95]{flat-wall-RS}
Neil Robertson and Paul~D Seymour.
\newblock Graph minors. {XIII}. the disjoint paths problem.
\newblock {\em Journal of Combinatorial Theory, Series B}, 63(1):65--110, 1995.

\bibitem[RST94]{RobertsonST94}
N~Robertson, P~Seymour, and R~Thomas.
\newblock {Quickly Excluding a Planar Graph}.
\newblock {\em Journal of Combinatorial Theory, Series B}, 62(2):323--348,
  November 1994.

\bibitem[RW12]{ReedW-grid}
Bruce~A Reed and David~R Wood.
\newblock {Polynomial treewidth forces a large grid-like-minor}.
\newblock {\em European Journal of Combinatorics}, 33(3):374--379, April 2012.

\bibitem[RZ10]{RaoZhou}
Satish Rao and Shuheng Zhou.
\newblock Edge disjoint paths in moderately connected graphs.
\newblock {\em SIAM J. Comput.}, 39(5):1856--1887, 2010.

\bibitem[Sch03]{Schrijver}
Alexander Schrijver.
\newblock {\em Combinatrial Optimization: Polyhedra and Efficiency}, volume~24
  of {\em Algorithms and Combinatorics}.
\newblock Springer, 2003.

\bibitem[Sey]{PS-comm}
Paul Seymour.
\newblock Personal communication.

\bibitem[SL07]{Singh-Lau}
Mohit Singh and Lap~Chi Lau.
\newblock Approximating minimum bounded degree spanning trees to within one of
  optimal.
\newblock In David~S. Johnson and Uriel Feige, editors, {\em STOC}, pages
  661--670. ACM, 2007.

\bibitem[Tho88]{Thomassen88}
C~Thomassen.
\newblock {On the presence of disjoint subgraphs of a specified type}.
\newblock {\em Journal of Graph Theory}, 12(1):101--111, 1988.

\end{thebibliography}

%----------------------------------------
%----------------------------------------
\iffull
%----------------------------------------
%----------------------------------------
%----------------------------------------
%----------------------------------------
%----------------------------------------
%----------------------------------------
%----------------------------------------
%----------------------------------------
%----------------------------------------
%----------------------------------------
%----------------------------------------
%----------------------------------------
%----------------------------------------
%----------------------------------------
%----------------------------------------
%----------------------------------------

\label{--------------------------------------------Appendix--------------------------------------------}
\appendix
%----------------------------------------
%----------------------------------------
%----------------------------------------
%----------------------------------------
%----------------------------------------
%----------------------------------------
%----------------------------------------
%----------------------------------------
%----------------------------------------
%----------------------------------------
%----------------------------------------
%----------------------------------------
%----------------------------------------
%----------------------------------------
%----------------------------------------
%----------------------------------------
%----------------------------------------
%----------------------------------------
%----------------------------------------
%----------------------------------------
\section{Proofs Omitted from Section~\ref{sec: prelims}}\label{sec: proofs from Prelims}
%----------------------------------------
%----------------------------------------
%--------------------------------------------------------
%--------------------------------------------------------
%--------------------------------------------------------
\subsection{Proof of Observation~\ref{obs: EDP to NDP in degree-3}}
%--------------------------------------------------------
%--------------------------------------------------------

%\begin{proof}
Assume for contradiction that there are two paths $P,P'\in \pset$ that share the same vertex $v$. Assume first that $v\in T_1\cup T_2$. The endpoints of the paths in $\pset$ are all distinct, and so if $v$ is an endpoint of, say, $P$, then it is an inner vertex on $P'$. Then $P'$ contains two edges incident on $v$, and $P$ contains one such edge, contradicting the fact that the degree of $v$ is at most $2$. Therefore, $v\not\in T_1\cup T_2$, and it is an inner vertex on both $P$ and $P'$. But then $P$ contains two edges incident on $v$, and $P'$ also contains two such edges, contradicting the fact that the degree of $v$ is at most $3$.
%\end{proof}

%--------------------------------------------------------
%--------------------------------------------------------
\subsection{Proof of Observation~\ref{obs: well-linkedness alt def}}
%--------------------------------------------------------
%--------------------------------------------------------

Let $T',T''\subseteq T$ be any pair of disjoint equal-sized vertex subsets, such that there is no flow $F:\tset'\sconnect_{1/\alpha}\tset''$ in $G$.
We construct a directed flow network $H$ from $G$, by replacing every edge of $G$ with a pair of bi-directed edges, and setting the capacity $c(e)$ of each such edge $e$ to $1/\alpha$. We then add two special vertices to the graph: a source $s$, that connects with a capacity-$1$ edge to every vertex of $\tset'$, and a destination $t$, to which every vertex of $\tset''$ connects with a capacity-$1$ edge.  Let $k=|\tset'|=|\tset''|$, and let $F$ be the maximum $s$-$t$ flow in $H$. Clearly, the value of $F$ is less than $k$, since otherwise we can use $F$ to define a flow $F':\tset'\sconnect_{1/\alpha}\tset''$ (as we can assume w.l.o.g. that for every pair $e',e''$ of anti-parallel edges, only one of these edges carries non-zero flow). 

Let $(A',B')$ be the minimum $s$-$t$ cut in $H$, so $\sum_{e\in E_{H}(A',B')}c(e)<k$, and let $A=A'\setminus\set{s}$ and $B=B'\setminus\set{t}$. We assume that $|\tset\cap A|\leq |\tset\cap B|$ - the other case is symmetric. Let $k_1=|\tset'\cap A'|$ and $k_2=|\tset'\cap B'|$. Then $\sum_{e\in E_H(A',B')}c(e)\geq k_2+|E_G(A,B)|/\alpha$. Therefore, $|E_G(A,B)|< \alpha (k-k_2)= \alpha k_1= \alpha |\tset'\cap A|\leq \alpha|\tset\cap A|$.

%--------------------------------------------------------
\subsection{Proof of Observation~\ref{obs: generalized well-linkedness alt def}}
%--------------------------------------------------------
%--------------------------------------------------------

%\begin{proof}
Let $T',T''\subseteq T$ be any pair of disjoint equal-sized vertex subsets, with $|\tset'|=|\tset''|\leq k'$, such that there is no flow $F:\tset'\sconnect_{1/\alpha}\tset''$ in $G$. We define a flow network $H$, and partitions $(A',B')$ of $V(H)$ and $(A,B)$ of $V(G)$ exactly as in the proof of Observation~\ref{obs: well-linkedness alt def}. Let $\ttset=\tset'\cup \tset''$, so $|\ttset|\leq 2 k'$. As in the proof of Observation~\ref{obs: well-linkedness alt def}, $|E_G(A,B)|<\alpha\cdot |\ttset\cap A|\leq \alpha\cdot |\tset\cap A|$, and $|E_G(A,B)|<\alpha\cdot |\ttset\cap B|\leq \alpha\cdot |\tset\cap B|$.
On the other hand,  since $|\ttset|\leq 2k'$, either $|\ttset\cap A|\leq k'$, or $|\ttset\cap B|\leq k'$ must hold, and so $|E_G(A,B)|<\alpha\cdot k'$.
%\end{proof}

%--------------------------------------------------------------------
%--------------------------------------------------------------------
%--------------------------------------------------------------------
%--------------------------------------------------------------------
\subsection{Proof of Theorem~\ref{thm: grouping}}
%--------------------------------------------------------------------
%--------------------------------------------------------------------
%--------------------------------------------------------------------
%--------------------------------------------------------------------
A separation in graph $G$ is a pair $(Y,Z)$ of subgraphs of $G$, such that
every edge of $G$ belongs to exactly one of $Y,Z$. The order of the
separation is $|V(Y)\cap V(Z)|$. We say that a separation $(Y,Z)$ is
\emph{balanced} iff $|V(Y)\cap \tset|,|V(Z)\cap \tset|\geq
|\tset|/4$. Let $(Y,Z)$ be a balanced separation of $G$ of minimum
order, and let $X=V(Y)\cap V(Z)$. Assume without loss of generality
that $|V(Y)\cap \tset|\geq |V(Z)\cap \tset|$, so $|V(Y)\cap \tset|\geq
|\tset|/2$. We claim that $X$ is node-well-linked in graph $Y$.

\begin{claim}
Set $X$ of vertices is node-well-linked in graph $Y$.
\end{claim}

\begin{proof}
  Let $A,B$ be any two disjoint equal-sized subsets of $X$, and assume that
  $|A|=|B|=z$. It is enough to show that there is a set $\pset$ of $z$
  disjoint paths connecting $A$ to $B$ in $Y$. Assume otherwise. Then
  there is a set $S\subseteq V(Y)$ of at most $z-1$ vertices separating $A$ from $B$
  in $Y$.

  Let $\cset$ be the set of all connected components of $Y\setminus
  S$. We partition $\cset$ into three subsets: $\cset_1$ contains all
  clusters containing the vertices of $A$; $\cset_2$ contains all
  clusters containing the vertices of $B$, and $\cset_3$ contains all
  remaining clusters (notice that all three sets of clusters are
  pairwise disjoint). Let $R_1=\bigcup_{C\in \cset_1}C$, and define
  $R_2$ and $R_3$ for $\cset_2$ and $\cset_3$, respectively. Assume
  without loss of generality that $|R_1\cap \tset|\geq |R_2\cap
  \tset|$. We define a new separation $(Y',Z')$, as follows. The set
  of vertices $V(Y')=R_1\cup R_3\cup S$, and $V(Z')=V(Z)\cup R_2\cup
  S$. Let $X'=V(Y')\cap V(Z')$. The edges of $Y'$ include all edges of
  $G$ with both endpoints in $V(Y')\setminus X'$, and all edges of $G$
  with one endpoint in $V(Y')\setminus X'$ and the other endpoint in
  $X'$. The edges of $Z'$ include all edges with both endpoints in
  $V(Z')$.

  We claim that $(Y',Z')$ is a balanced separation. Clearly,
  $|V(Z')\cap \tset|\geq |\tset|/4$, since $V(Z)\subseteq V(Z')$, and
  $|V(Z)\cap \tset|\geq |\tset|/4$.  We next claim that $|V(Y')\cap
  \tset|\geq |\tset|/4$. Assume otherwise. Then $|(R_1\cup R_3\cup S)\cap \tset|\leq |\tset|/4$, and, from our assumption,
  $|R_2\cap \tset|<|\tset|/4$, and so $|V(Y)\cap \tset|<|\tset|/2$, a
  contradiction. Therefore, $(Y',Z')$ is a balanced
  separator. Finally, we claim that its order is less than $|X|$,
  contradicting the minimality of $X$. Indeed, $|V(Y')\cap V(Z')|\leq
  |X|-|B|+|S|<|X|$.
\end{proof}

Let $\tset_1=\tset\cap V(Z)$, $\tset_2=\tset\cap V(Y)$, and let $\tset_1'\subseteq \tset_1$, $\tset_2'\subseteq \tset_2$ be two disjoint subsets containing $\ceil{\kappa /4}$ vertices each. Since the terminals are $\alpha$-well-linked, there is a flow $F:\tset_1'\sconnect_{1/\alpha}\tset_2'$ in $G$. We now bound the vertex-congestion caused by the flow $F$. For every vertex $v\in V(G)$, let $F_1(v)$ be the total amount of flow on all paths that originate or terminate at $v$, and let $F_2(v)$ be the total amount of flow on all paths that contain $v$ as an inner vertex. It is immediate to verify that $F_1(v)\leq 1$, while $F_2(v)\leq \frac{\Delta}{2\alpha}$, since every flow-path $P$ that contains $v$ as an inner vertex contributes flow $F(P)$ to two edges incident to $v$. Therefore, the total flow through $v$ is at most $\frac{\Delta}{2\alpha}+1\leq \frac{5\Delta}{6\alpha}$, as $\Delta\geq 3$ and $\alpha\leq 1$.
By sending $\frac{6\alpha}{5\Delta}\cdot F(P)$ flow units via every path $P$, we obtain a flow of value at least $\frac{\kappa}{4}\cdot \frac{6\alpha}{5\Delta}=\frac{3\alpha \kappa}{10\Delta}$ from vertices of $\tset_1'$ to vertices of $\tset_2'$, that causes vertex-congestion at most $1$. From the integrality of flow, there is a set 
%
%
%. Recall that
%we have assumed that $|\tset_1|\leq |\tset_2|$.  Add a source vertex
%$s$ and connect it to every terminal in $\tset_1$ with a directed
%edge, and add a sink vertex $t$ and connect every terminal in
%$\tset_2$ to it with a directed edge. Set all edge and vertex
%capacities to $1$, except for $s$ and $t$ whose capacities are
%infinite.  Since $\tset$ is $\alpha$-well-linked in $G$, there is an
%$s$-$t$ flow $F$ of value $|\tset_1|\geq \kappa/4$ in this network,
%with edge-congestion at most $1/\alpha$. Scaling this flow down by
%factor $\Delta/(2\alpha)$, we obtain an $s$-$t$ flow of value at least
%$\frac{\kappa \alpha}{4\Delta}$ and vertex congestion at most $1$.
%From the integrality of flow, there is a set of $\floor{\frac{\kappa
%    \alpha}{4\Delta}}$ internally disjoint $s$-$t$ paths. This gives a
%collection 
%
$\pset'$ of $\kappa'=\ceil{\frac{3\alpha \kappa}{10\Delta}}$
node-disjoint paths connecting terminals in $\tset_1'$ to terminals in
$\tset_2'$ in $G$. Each such path has to go through a vertex of $X$. For each
path $P'\in \pset'$, we truncate the path $P'$ to the first vertex of
$X$ on $P'$ (where the path is directed from $\tset_1$ to
$\tset_2$). Let $\pset$ be the resulting set of truncated paths. Then
$\pset$ is a set of $\kappa'$ disjoint paths, connecting  vertices
of $\tset_1'$ to vertices of $X$; every path in $\pset$ is
completely contained in graph $Z$, and is disjoint from $X$ except for
its last endpoint that belongs to $X$.

Let $\tset'\subseteq \tset_1$ be the set of terminals from which the
paths of $\pset$ originate, and let $X'\subseteq X$ be the set of
vertices where they terminate. We claim that $\tset'$ is
node-well-linked in $G$. Indeed, let $A,B\subseteq\tset'$ be any pair of disjoint
equal-sized subsets of terminals.  Denote $|A|=|B|=z$.

We define a set $\tilde{A}\subset X'$ of vertices as follows: for each terminal
$t\in A$, let $P_t\in \pset$ be the path originating at $t$, and let
$x_t$ be its other endpoint, that belongs to $X$. We then set
$\tilde{A}=\set{x_t\mid t\in A}$. We define a set $\tB\subseteq X$
similarly for $B$. Let $\pset_A\subseteq \pset$ be the set of paths
originating at the vertices of $A$, and let $\pset_B\subseteq \pset$
be the set of paths originating at the vertices of $B$.  Notice that
both sets of paths are contained in $Z$, and are internally disjoint
from $X$. The paths in $\pset_A\cup \pset_B$ are also mutually
disjoint.

Consider the two subsets $\tA,\tB\subseteq X$ and recall that $|\tA|=|\tB|=z$.  Since
$X$ is node-well-linked in $Y$, there is a set $\qset$ of $z$ node-disjoint
paths connecting $\tA$ to $\tB$ in $Y$. The paths in $\qset$ are then
completely disjoint from the paths in $\pset_1,\pset_2$ (except for
sharing endpoints with them). The final set of paths connecting $A$ to
$B$ is obtained by concatenating the paths in $\pset_1,\qset$, and $\pset_2$.

 %----------------------------------------------
%----------------------------------------------
%----------------------------------------------
%----------------------------------------------
\subsection{Proof of Theorem~\ref{thm: linkedness from node-well-linkedness}}
%----------------------------------------------
%----------------------------------------------
%----------------------------------------------
%----------------------------------------------
%If $|\tset_1|>\kappa$, then we discard from $\tset_1$ vertices that belong to $\tset_1\setminus\tset_1'$, until $|\tset_1|=\kappa$ holds, and we do the same for $\tset_2$. Notice that $\tset_1\cup \tset_2$ continues to be $\alpha$-well-linked,  and each of the resulting sets $\tset_1,\tset_2$ is still node-well-linked. 
%
Let $\tset=\tset_1\cup \tset_2$. We refer to the vertices of $\tset$ as terminals. Assume for contradiction that $(\tset'_1,\tset_2')$ are not node-linked in $G$. Then there are two sets $A\subseteq\tset_1',B\subseteq \tset_2'$, with $|A|=|B|=\kappa'$ for some $\kappa'\leq \frac{\alpha\kappa}{2\Delta}$, and a set $S$ of $\kappa'-1$ vertices, separating $A$ from $B$ in $G$. 

Let $A'\subseteq \tset_1$ be the set of all terminals $t\in \tset_1$, such that $t$ lies in the same component of $G\setminus S$ as some vertex of $A$. We claim that $|A'|> \kappa-\kappa'$. Indeed, assume otherwise, and let $A''\subseteq \tset_1\setminus A'$ be any set of $\kappa'$ vertices. Notice that $A''\cap A=\emptyset$, as $A\subseteq A'$. Since $\tset_1$ is node-well-linked in $G$, there is a set $\pset$ of $\kappa'$ node-disjoint paths, connecting the vertices of $A$ to the vertices of $A'$ in $G$. At most $\kappa'-1$ of these paths may contain the vertices of $S$, and so at least one vertex of $\tset_1\setminus A'$ is connected to some vertex of $A$ in $G\setminus S$, a contradiction.

Similarly, we let $B'\subseteq \tset_2$ be the set of all terminals $t\in \tset_2$, such that $t$ lies in the same component of $G\setminus S$ as some vertex of $B$. As before, $|B'|\geq \kappa-\kappa'$. Finally, we show that there is some pair $a\in A'$, $b\in B'$ of vertices that lie in the same connected component of $G\setminus S$. Indeed, since the terminals of $\tset$ are $\alpha$-well-linked in $G$, there is a set $\qset$ of at least $\kappa-\kappa'$ paths in $G$, where each path originates at a distinct vertex of $A'$ and terminates at a distinct vertex of $B'$, and every edge of $G$ participates in at most $\ceil{1/\alpha}$ paths. At most $(k'-1)\Delta\cdot \ceil{1/\alpha}$ of the paths in $\qset$ may contain the vertices of $S$. Since $|\qset|=\kappa-\kappa'>(\kappa'-1)\Delta\cdot \ceil{1/\alpha}$, at least one path of $\qset$ belongs to $G\setminus S$. Therefore, there is a path in $G\setminus S$ from a vertex of $A$ to a vertex of $B$, a contradiction.
\subsection{Proof of Theorem~\ref{thm from weak to perfect PoS}}
%----------------------------------------
%----------------------------------------
We note that using by applying  Theorems~\ref{thm: grouping} and \ref{thm: linkedness from node-well-linkedness} to each of the clusters of the \PoS in turn, we can obtain a significantly simpler proof of the theorem, but with weaker parameters: the width of the new \PoS becomes roughly $\Theta(w\alpha^2)$.

\paragraph{Obtaining a good \PoS.}
We assume that we are given an $\alpha$-weak \PoS $(\sset,\bigcup_{i=1}^{\ell-1}\pset_i)$, where $\sset=(S_1,\ldots,S_{\ell})$, and each set $S_i$ is associated with a pair $(A_i,B_i)$ of disjoint vertex subsets of $S_i$ of cardinality $w$ each, such that $A_i\cup B_i$ is $\alpha$-well-linked in $G[S_i]$. %Our goal is to construct a $1$-weak \PoS of width $r$ and height $\ceil{\alpha h/2}$.
 In order to simplify the notation, we add a new set $\pset_{\ell}$ of $w$ artificial paths, where each path in $\pset_{\ell}$ consists of a single new edge, connecting a distinct vertex in $B_{\ell}$ to a distinct vertex in $A_1$. Let $E_{\ell}$ denote this corresponding new set of edges. Then the edges of $E_{\ell}$ define a matching between the vertices of $A_1$ and the vertices of $B_{\ell}$. Since the vertices of $A_1\cup  B_{\ell}$ have degree at most $2$ each, the maximum vertex degree in the resulting graph is at most $3$. We will use addition modulo $\ell$ when indexing sets of vertices or sets of paths, so, for example, by $A_{\ell+1}$ we mean $A_1$.

In our final \PoS, the sequence of clusters $(S_1,\ldots,S_{\ell})$ will remain the same. For each $1\leq i\leq \ell$, we will select a subset $\pset'_i\subseteq \pset_i$ of $\ceil{\alpha w/4}$ paths, and we will denote by $B_i'$ the set of the endpoints of the paths in $\pset'_i$ that lie in $S_i$, and by $A_{i+1}'$ the set of their endpoints that lie in $S_{i+1}$. We will ensure that for each $1\leq i\leq \ell$, the vertices of $B_i'$ are $1$-well-linked in $G[S_i]$, and $(A_i',B_i')$ are $\half$-linked in $G[S_i]$. The final good \PoS is then $(\sset,\bigcup_{i=1}^{\ell-1}\pset'_i)$.% In order to find the desired subsets $\pset_i'\subseteq \pset_i$ of paths, we use the following lemma, that allows us to find subsets of paths with slightly weaker properties.

Fix some $1\leq i\leq \ell$, and let $(X_i,Y_i)$ be the minimum $1/4$-balanced cut of $G[S_i]$ with respect to $B_i$, where $|X_i\cap B_i|\geq |Y_i\cap B_i|$. Let $Z_i\subseteq X_i$ be the set of vertices incident on the edges of $E(X_i,Y_i)$. We need the following simple claim.

\begin{claim}\label{claim: wl in balanced cut}
For all $1\leq i\leq \ell$, set $Z_i$ is $1$-well-linked in $G[X_i]$, and $|Z_i|\geq \alpha \ceil{w/4}$.
\end{claim}

\begin{proof}
Assume for contradiction that $Z_i$ is not $1$-well-linked in $G[X_i]$. Then there is a partition $(J,J')$ of $X_i$, with $|E(J,J')|<\min\set{|Z_i\cap J|,|Z_i\cap J'|}$. Assume without loss of generality that $|J\cap B_i|\geq |J'\cap B_i|$. Then $(J, J'\cup Y_i)$ is a $1/4$-balanced cut of $G[S_i]$ with respect to $B_i$, and $|E(J,J'\cup Y_i)|\leq |E(X_i,Y_i)|-|E(J',Y_i)|+|E(J,J')|\leq |E(X_i,Y_i)|-|Z_i\cap J'|+|E(J,J')|<|E(X_i,Y_i)|$, contradicting the minimality of the cut $(X_i,Y_i)$.

To see that $|Z_i|\geq \alpha \ceil{w/4}$, recall that the vertices of $B_i$ are $\alpha$-well-linked in $G[S_i]$, and $|B_i\cap X_i|,|B_i\cap Y_i|\geq \ceil{w/4}$. Let $T_1\subseteq B_i\cap X_i$, $T_2\subseteq B_i\cap Y_i$ be any pair of sets of cardinality $\ceil{w/4}$ each. Then there is a flow $F: T_1\sconnect_{1/\alpha}T_2$ in $G[S_i]$. Scaling this flow by factor $\alpha$, we obtain a flow of value $\alpha\ceil{w/4}$ from vertices of $T_1$ to vertices of $T_2$, that causes no edge-congestion, such that for every vertex $v\in T_1\cup T_2$, the total flow on paths originating from or terminating at $v$ is exactly $\alpha$. From the integrality of flow, there is a 
set $\pset$ of $\alpha \ceil{w/4}$ edge-disjoint paths from vertices of $T_1$ to vertices of $T_2$ in $G[S_i]$, whose endpoints are all distinct. Since the degree of every vertex of $B_i$ is at most $2$ in $G[S_i]$, from Observation~\ref{obs: EDP to NDP in degree-3}, the paths in $\pset$ are also node-disjoint.  Each such path must contain a vertex of $Z_i$, so $|Z_i|\geq  \alpha \ceil{ w/4}$.
\end{proof}

Fix some $1\leq i\leq \ell$.
Let $\tilde{B_i}$ be any subset of $\ceil {w/4}$ vertices of $B_i\cap Y_i$, and let $\tilde{\pset}_i\subseteq \pset_i$ be the set of paths originating at the vertices of $\tilde{B_i}$. Let $\tilde A_{i+1}\subseteq A_{i+1}$ be the set of vertices where the paths in $\tpset_i$ terminate (see Figure~\ref{fig: flow-network1}). The following lemma is central to our proof.

\begin{lemma}\label{lemma: select the As and Bs}
For each $1\leq i\leq \ell$,  there is a subset $\pset_i'\subseteq \tpset_i$ of $\ceil{\alpha w/4}$ paths, such that,  if $B_i'\subseteq \tilde B_i$ and $A_{i+1}'\subseteq\tilde A_{i+1}$ denote the sets of the endpoints of paths in $\pset_i'$ lying in clusters $S_i$ and $S_{i+1}$ respectively, then:

\begin{itemize}
\item there is a set $\qset_i'$ of node-disjoint paths  in $G[Y_i\cup Z_i]$, connecting every vertex of $B'_i$ to a distinct vertex of $Z_i$, where the paths in $\qset_i'$ are internally disjoint from $Z_i$; and

\item there is a set $\qset_{i+1}''$ of edge-disjoint paths  in $G[S_{i+1}]$, connecting every vertex of $A_{i+1}'$ to a distinct vertex of $Z_{i+1}$. 
\end{itemize}
\end{lemma}

We prove Lemma~\ref{lemma: select the As and Bs} below, after we complete the construction of the good \PoS using it.
Our final \PoS consists of the sequence $(S_1,\ldots,S_{\ell})$ of clusters and the sets $\pset_1',\ldots,\pset_{\ell-1}'$ of paths given by Lemma~\ref{lemma: select the As and Bs}. Sets $A'_i,B'_i\subseteq S_i$ of vertices are defined as in the lemma. It now remains to show that for each $1\leq i\leq \ell$, $B_i'$ is $1$-well-linked, and $(A_i',B_i')$ are $\half$-linked  in $G[S_i]$. We do so in the following two observations.

\begin{observation} For all $1\leq i\leq \ell$, $B_i'$ is $1$-well-linked in $G[S_i]$.
\end{observation}
\begin{proof}
The observation follows immediately from the fact that $Z_i$ is $1$-well-linked in $G[X_i]$, and there is a set $\qset_i'$ of node-disjoint paths, connecting every vertex of $B'_i$ to a distinct vertex of $Z_i$, that are contained in $G[Y_i\cup Z_i]$, and are internally disjoint from $Z_i$.
\end{proof}

\begin{observation} For all $1\leq i\leq \ell$,  $(A_i',B_i')$ are $\half$-linked in $G[S_i]$.
\end{observation}
\begin{proof}
Fix some $1\leq i\leq \ell$.
 Let $T'\subseteq A_i'$, $T''\subseteq B_i'$ be any pair of equal-sized vertex sets. We need to show that there is a flow $F: T'\sconnect_2 T''$ in $G[S_i]$. Let $\rset'\subseteq \qset_i'$ be the set of paths originating at the vertices of $T''$, and let $Z_i'$ be the set of vertices where the paths of $\rset'$ terminate. Recall that $G[S_i]$ contains a set $\qset_i'':A_i'\connect Z_i$ of edge-disjoint paths, where every path terminates at a distinct vertex of $Z_i$. Let $\rset''\subseteq \qset_i''$ be the set of paths originating at the vertices of $T'$, and let $Z_i''$ be the set of vertices where the paths of $\qset_i''$ terminate. Since set $Z_i$ is $1$-well-linked in $G[X_i]$, there is a set $\tilde{\qset}: Z_i'\setminus Z_i''\sconnect Z_i''\setminus Z_i'$ of edge-disjoint paths in $G[X_i]$. The final flow $F$ is obtained by concatenating the paths in $\rset',\tilde{\qset}$, and $\rset''$, and sending one flow unit via each resulting path. We claim that $F$ causes edge-congestion at most $2$ in $G[S_i]$. Consider an edge $e\in E_G(S_i)$. This edge may participate in at most one path in $\rset''$. If $e\in E(X_i)$, then the edge may participate in at most one path in $\tilde{\qset}$, but it does not belong to any path in $\rset'$. If $e\in E(Y_i)\cup E(X_i,Y_i)$, then it may participate in at most one path in $\rset'$, but it does not belong to any path in $\tilde{\qset}$. Therefore, overall, edge $e$ participates in at most two paths.
\end{proof}

In order to complete the proof of Theorem~\ref{thm from weak to perfect PoS} for a good \PoS, it is now enough to prove Lemma~\ref{lemma: select the As and Bs}.

\begin{proofof}{Lemma~\ref{lemma: select the As and Bs}}
%------------------------------------------------
%------------------------------------------------
%------------------------------------------------
%------------------------------------------------
%------------------------------------------------
%------------------------------------------------
Fix some $1\leq i\leq \ell$. %Let $\pset^*_i\subseteq \tilde \pset_i$ be the subset of paths terminating at the vertices of $A_{i+1}^*$, and let $B_i^*\subseteq \tilde B_i$ be the set of vertices where the paths of $\pset_i^*$ originate. Recall that $|B_i^*|=|\pset_i^*|=| A_{i+1}^*|=\ceil{h/8}$.
Let $G_i=G[S_{i+1}]\cup \tpset_i\cup G[Y_i]\cup E(Z_i,Y_i)$. We build a directed flow network $H_i$, as follows (see Figure~\ref{fig: flow-network2}). Start with $G_i$, and replace every edge in $G_i[S_{i+1}]\cup G_i[Y_i\cup Z_i]$ with a bi-directed pair of edges. Direct all edges on paths $P\in \tilde{\pset}_i$ along the path, from $\tilde B_i$ toward $\tilde A_{i+1}$. For each vertex $v\in Z_i$, delete all incoming edges. Add a source vertex $s$, that connects to every vertex in $Z_i$ with a directed edge, and a destination vertex $t$, to which every vertex in $Z_{i+1}$ is connected. We set the capacity of every edge to $1$. 

\begin{figure}[h]
\centering
\subfigure[Balanced cuts in $S_i$ and $S_{i+1}$. The vertices of $Z_i$ and $Z_{i+1}$ are shown in blue, the vertices of $\tilde{B}_i$ in red, and the vertices of $\tilde{A}_{i+1}$ in green.]{\scalebox{0.45}{\includegraphics{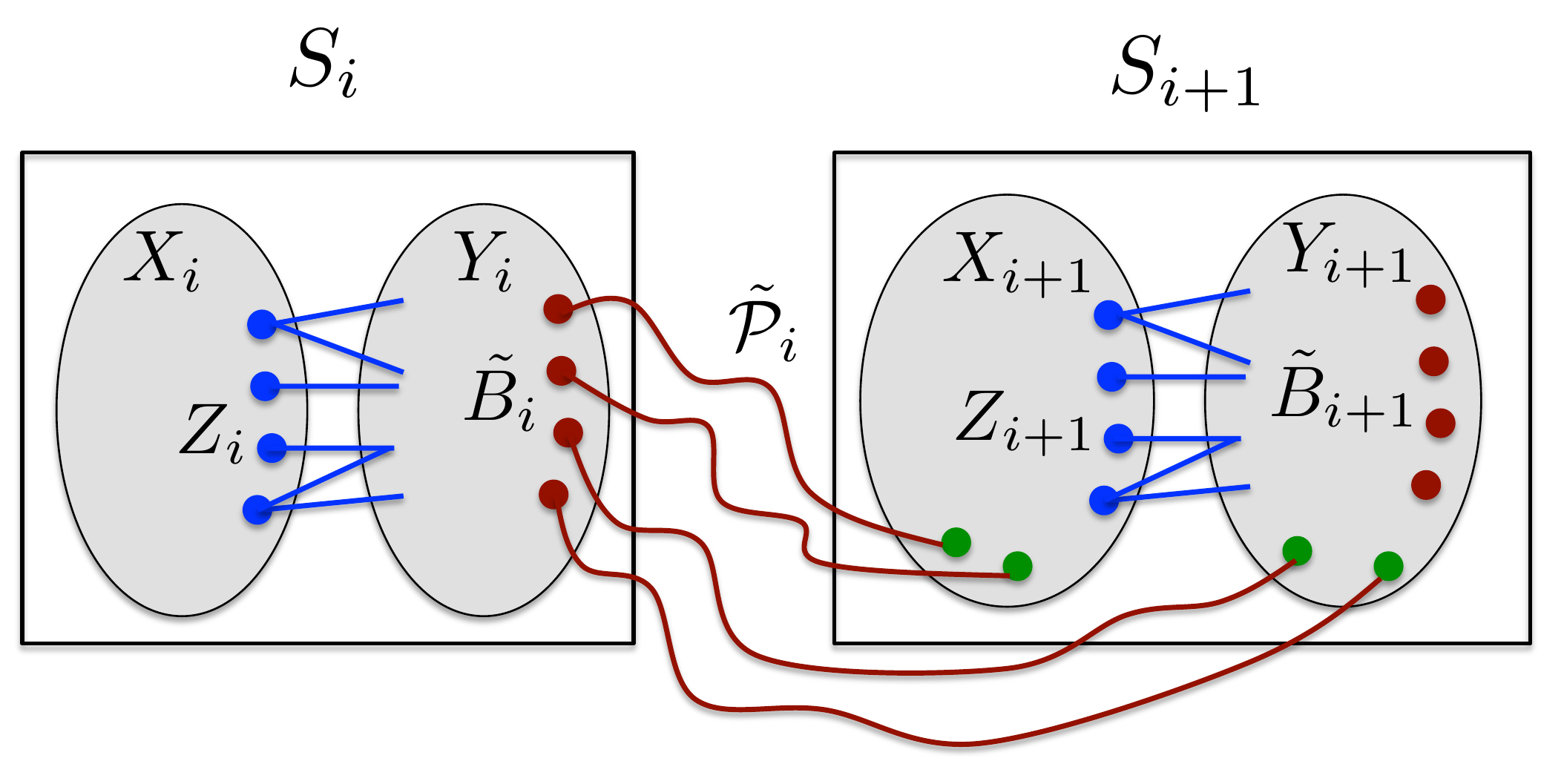}}\label{fig: flow-network1}}
\hspace{1cm}
\subfigure[Constructing flow network $H_i$.]{
\scalebox{0.45}{\includegraphics{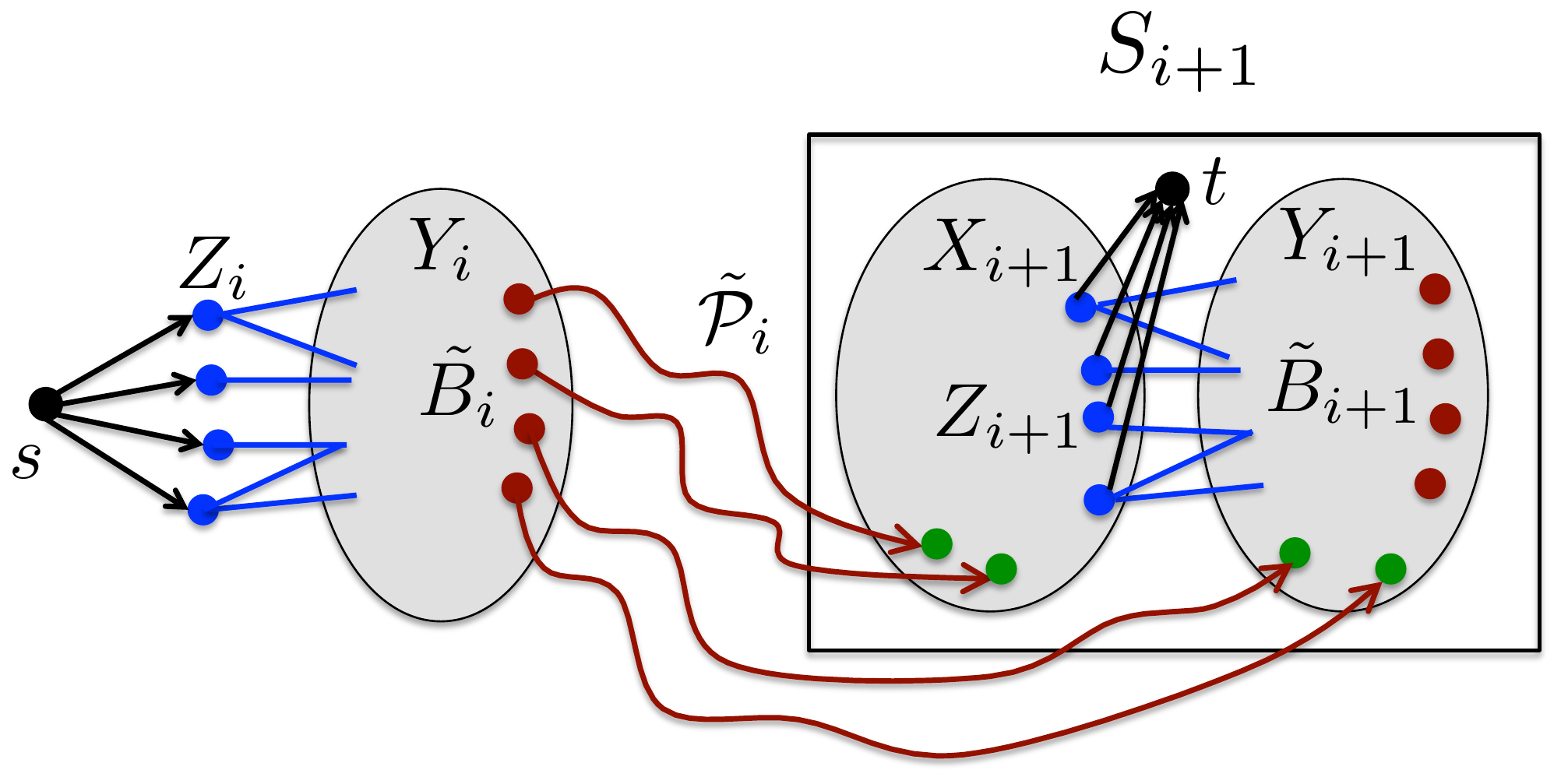}}\label{fig: flow-network2}}
%\hspace{1cm}
%\subfigure[Clusters $X'$ and $Y'$.]{\scalebox{0.35}{\includegraphics{advanced-splitting3.pdf}}\label{fig: advanced-splitting3}}
\caption{Illustration to the proof of Theorem~\ref{thm from weak to perfect PoS}.\label{fig: flow-network}}
\end{figure}

\begin{figure}[h]
%\scalebox{0.8}{\includegraphics{flow-network-cut.pdf}}\caption{Flow network $H_i$\label{fig: flow-network}}
\end{figure}

We claim that the resulting flow network has a valid $s$-$t$ flow $F$ of value $\alpha \ceil{w/4}$.
In order to construct the flow, we need the following two claims.

\begin{claim}\label{claim: first flow}
There is a flow in $G[Z_i\cup Y_i]$ from the vertices of $\tilde B_i$ to the vertices of $Z_i$, where every vertex in $\tilde B_i$ sends exactly $\alpha$ flow units, every vertex in $Z_i$ receives at most one flow unit, and the flow causes no edge-congestion. Moreover, vertices of $Z_i$ do not serve as inner vertices on flow-paths that carry non-zero flow.
\end{claim}

\begin{claim}\label{claim: third flow}
There is a set $\qset^*_{i+1}:\tilde A_{i+1} \connect_{1/\alpha}Z_{i+1}$ of paths in $G[S_{i+1}]$, where every vertex of $Z_{i+1}$ is an endpoint of at most $1/\alpha$ paths.
\end{claim}

We prove both claims below, after we complete the proof of Lemma~\ref{lemma: select the As and Bs} using them.
Flow $F$ is constructed by concatenating three different flows, $F_1,F_2$, and $F_3$. 
Flow $F_1$ is contained in $H_i[s\cup Z_i\cup Y_i]$, and it originates at $s$ and terminates at the vertices of $\tilde B_i$, where every vertex of $\tilde B_i$ receives exactly $\alpha$ flow units. This flow is constructed using the flow given by Claim~\ref{claim: first flow}, after we reverse the direction of the flow.
For the second flow, $F_2$, we simply send $\alpha$ flow units along each path in $\tilde{\pset}_i$, and flow $F_3$ is obtained by sending $\alpha$ flow units along each path in $\qset_{i+1}^*$.
By concatenating the flows in $F_1,F_2$ and $F_3$, we obtain a valid $s$--$t$ flow of value $\alpha \ceil{w/4}\geq \alpha w/4$ in $H_i$. Using the integrality of flow, there is a set $\rset_i$ of $\ceil{\alpha w/4}$ edge-disjoint paths in $H_i$, connecting $s$ to $t$, such that for every pair $(e,e')$ of anti-parallel edges, at most one edge participates in the paths in $\rset_i$. Each path in $\rset_i$ contains exactly one path in $\tilde{\pset}_i$. We let $\pset_i'\subseteq \tilde{\pset}_i$ be the set of paths contained in the paths of $\rset_i$. Let $B_i'\subseteq \tilde B_i$ be the set of vertices where the paths of $\pset_i'$ originate, and let $A_{i+1}'\subseteq \tilde A_{i+1}$ be
the set of vertices where the paths of $\pset_i'$ terminate. Note that the set $\rset_i$ of paths defines two sets of paths: a set $\qset_i'$ of edge-disjoint paths contained in $G_i[Y_i\cup Z_i]$, connecting every vertex of $B_i'$ to a distinct vertex of $Z_i$, and a set $\qset_{i+1}''$ of edge-disjoint paths contained in $G[S_{i+1}]$, connecting every vertex of $A_{i+1}'$ to a distinct vertex of $Z_{i+1}$. Moreover, since each vertex $v\in Z_i$ has no incoming edges in $H_i$, except for the edge $(s,v)$, the paths in $\qset_i'$ are internally disjoint from $Z_i$. Since all vertices of $G$ have degree at most $3$ in $G$, and the degree of every vertex $v\in B_i'$ in graph $G[S_i]$ is at most $2$, using reasoning similar to that used in the proof of Observation~\ref{obs: EDP to NDP in degree-3}, we conclude that the paths in $\qset_i'$ are node-disjoint. It now remains to complete the proofs of Claims~\ref{claim: first flow} and~~\ref{claim: third flow}.

\begin{proofof}{Claim~\ref{claim: first flow}}
Recall that $B_i$ is $\alpha$-well-linked in $G[S_i]$, and $|X_i\cap B_i|\geq |B_i|/2\geq \ceil{w/4}$. Let $T_i\subseteq X_i\cap B_i$ be any subset of $\ceil{w/4}$ such vertices. From the $\alpha$-well-linkedness of $B_i$ in $G[S_i]$, there is a flow $F': \tilde B_i\sconnect_{1/\alpha} T_i$ in $G[S_i]$. Let $G'$ be a directed graph obtained from $G[S_i]$ by bi-directing all its edges. Then there is a flow $F'':\tilde B_i\sconnect_{1/\alpha} T_i$ in $G'$, where for every pair $(e,e')$ of anti-parallel edges, at most one of these edges carries non-zero flow.
Since the maximum vertex degree is at most $3$ in $G$, for every vertex $v\in Z_i$, either there is at most one incoming edge with non-zero flow, or there is at most one outgoing edge with non-zero flow in $G'$. Therefore, the total amount of flow through $v$, $\sum_{P: v\in P}F''(P)\leq 1/\alpha$. We can then obtain the desired flow $F$ as follows. For every path $P$ carrying no-zero flow, we truncate $P$ at the first vertex of $Z_i$ that lies on it (where we view the path as directed from a vertex of $\tilde B_i$ to a vertex of $T_i$), and then send $\alpha\cdot F'(P)$ flow units via $P$.\end{proofof}

\begin{proofof}{Claim~\ref{claim: third flow}}
We construct a flow network $\nset$, as follows (see Figure~\ref{fig: st-flow1}). Start with graph $G[S_{i+1}]$, and set the capacity of every edge in $G[S_{i+1}]$ to be $1/\alpha$. Add a source $s$, that connects to every vertex of $\tilde A_{i+1}$ with an edge of capacity $1$, and add a destination vertex $t$, connecting every vertex of $Z_{i+1}$ to $t$ with an edge of capacity $1/\alpha$. Since $1/\alpha$ is an integer, from the integrality of flow, it is enough to show that there is an $s$--$t$ flow of value $|\tilde A_{i+1}|=\ceil{w/4}$ in $\nset$ -- such a flow immediately defines the desired collection $\qset^*_{i+1}: \tilde A_{i+1}\connect_{1/\alpha} Z_{i+1}$ of paths, where for every vertex $v\in Z_{i+1}$, at most $1/\alpha$ of the paths terminate at $v$. Assume for contradiction that such a flow does not exist. Then there is an $s$--$t$ cut $(S,T)$ in $\nset$, with $s\in S$, $t\in T$, and $\sum_{e\in E(S,T)}c(e)<\ceil{w/4}$. Let $S'=S\setminus \set{s}$, and $T'=T\setminus\set{t}$, so $(S',T')$ is a partition of $S_{i+1}$. Let $z_1=|\tilde A_{i+1}\cap S'|$, $z_2=|\tilde A_{i+1}\cap T'|$, and similarly $z'_1=|Z_{i+1}\cap S'|$, and $z'_2=|Z_{i+1}\cap T'|$ (see Figure \ref{fig: st-flow1}). Let $E'=E_G(S',T')$.

\begin{figure}[h]
\centering
\subfigure[Network $\nset$. Edges incident on $s$ have capacity $1$; all other edges have capacity $1/\alpha$.]{\scalebox{0.45}{\includegraphics{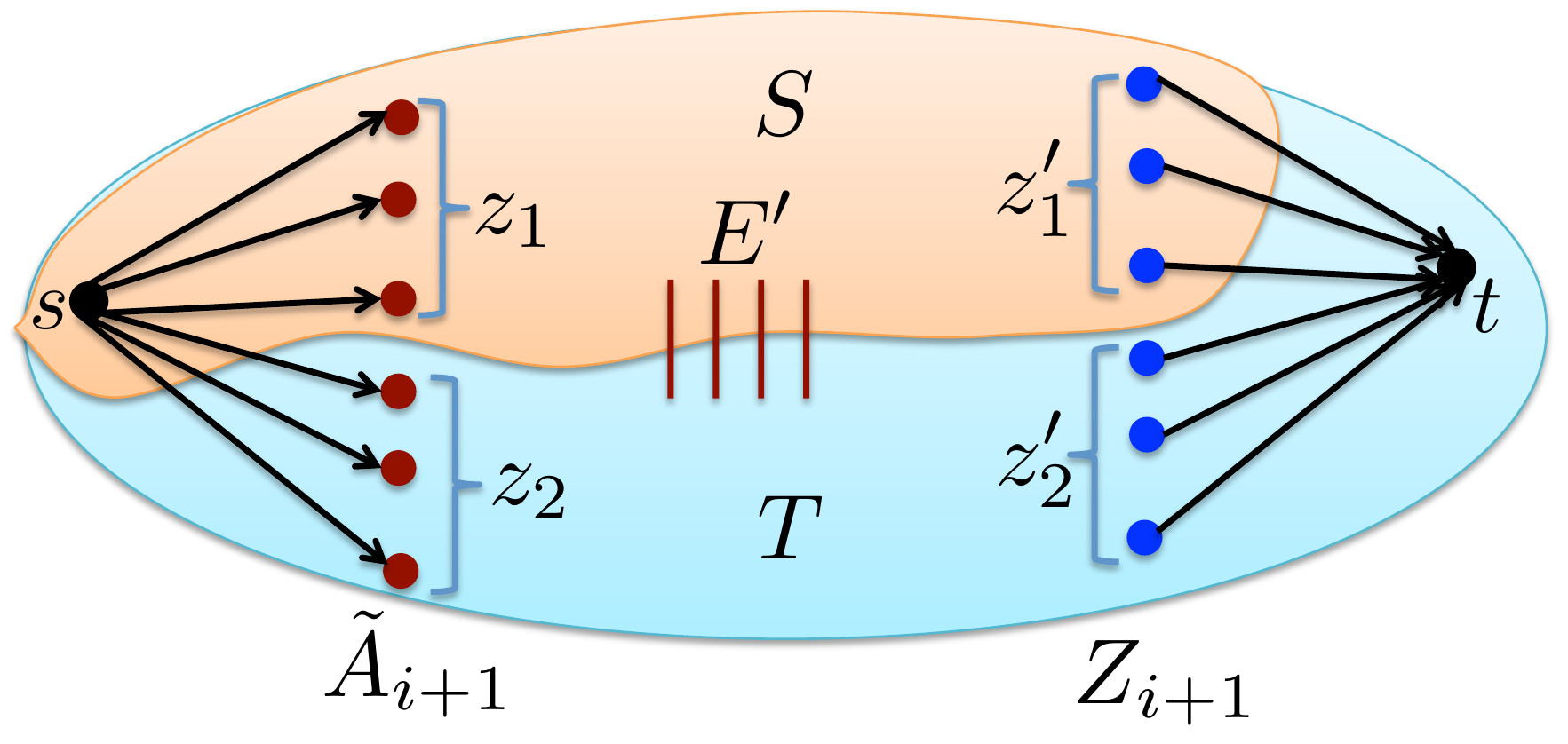}}\label{fig: st-flow1}}
\hspace{1cm}
\subfigure[Improving the balanced cut $(X_{i+1},Y_{i+1})$. The red edges are edges of $E'$.]{
\scalebox{0.5}{\includegraphics{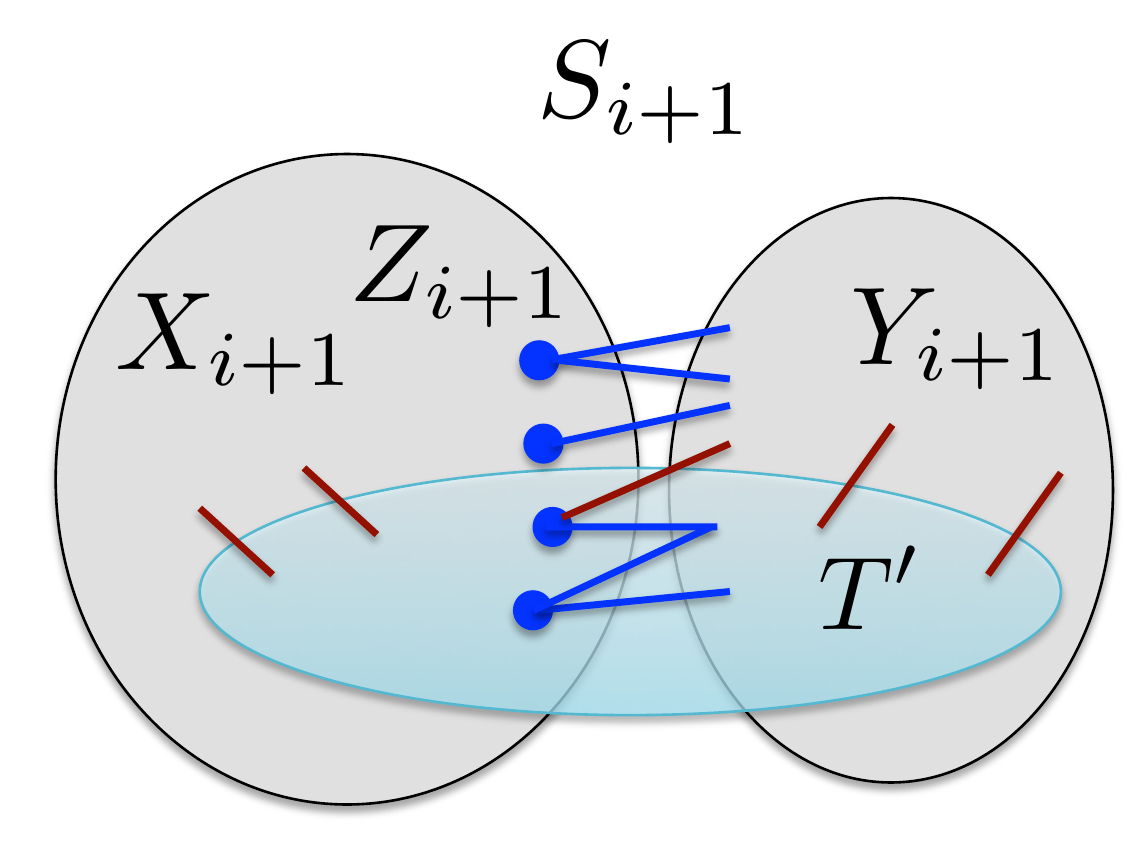}}\label{fig: st-flow2}}
\caption{Illustration to the proof of Claim~\ref{claim: third flow}.\label{fig: st-flow}}
\end{figure}

Notice that $|E'|<z_2'$, since otherwise $\sum_{e\in E_{\nset}(S,T)}c(e)\geq (|E'|+z_1')/\alpha\geq (z_2'+z_1')/\alpha= |Z_{i+1}|/\alpha\geq \ceil{w/4}$ from Claim~\ref{claim: wl in balanced cut}. Similarly, $|E'|<\alpha z_1$ must hold, since otherwise $\sum_{e\in E_{\nset}(S,T)}c(e)\geq |E'|/\alpha+z_2\geq z_1+z_2=\ceil{w/4}$, a contradiction.

Finally, we claim that $|T'\cap B_{i+1}|<z_1$. Indeed, otherwise, from the $\alpha$-well-linkedness of $A_{i+1}\cup B_{i+1}$, there must be a flow of value at least $z_1$ between vertices of $A_{i+1}\cap S'$ and $B_{i+1}\cap T'$ with edge-congestion at most $1/\alpha$, implying that $|E'|\geq \alpha z_1$, contradicting our conclusion above.

Consider now a new partition $(X',Y')$ of $S_{i+1}$, where $X'=X_{i+1}\setminus T'$, and $Y'=Y_{i+1}\cup T'$ (see Figure~\ref{fig: st-flow2}). Since $|B_{i+1}\cap T'|<z_1\leq \ceil{w/4}$, this is a $1/4$-balanced cut with respect to $B_{i+1}$. We now claim that $|E_G(X',Y')|<|E_G(X_{i+1},Y_{i+1})|$, contradicting the minimality of the cut $(X_{i+1},Y_{i+1})$. Indeed, every vertex of $Z_{i+1}\cap T'$ is incident on some edge $e\in E_G(X_{i+1},Y_{i+1})$. If $e\not \in E'$, then $e$ does not belong to $E_G(X',Y')$. Therefore, $|E_G(X',Y')|\leq |E_G(X_{i+1},Y_{i+1})|-|Z_{i+1}\cap T'|+|E'|=|E_G(X_{i+1},Y_{i+1})|-z_2'+|E'|<|E(X_{i+1},Y_{i+1})|$, since $|E'|<z_2'$.
\end{proofof}
\end{proofof}

%\end{proofof}

\paragraph{Obtaining a perfect \PoS.}
We now turn to construct a perfect \PoS, from a good \PoS of width $\ceil{\alpha w/4}$ and length $\ell$. Let $w'=\ceil{\alpha w/4}$.
To simplify notation, we assume that we are given a good \PoS $(\sset,\bigcup_{i=1}^{\ell-1}\pset_{i})$ of width $w'$ with $\sset=(S_1,\ldots,S_{\ell})$, and we denote the corresponding sets of vertices contained in $S_i$, for $1\leq i\leq \ell$, by $A_i$ and $B_i$. As before, we add an artificial set $\pset_{\ell}$ of $w'$ disjoint paths, each containing a single edge, connecting a distinct vertex of $B_{\ell}$ to a distinct vertex of $A_1$.

It is immediate to verify that there is some constant $0<\alpha'\leq 1$, such that for all $1\leq i\leq \ell$, $A_i\cup B_i$ is $\alpha'$-well-linked in $G[S_i]$, since $B_i$ is $1$-well-linked and $(A_i,B_i)$ are $\half$-linked in $G[S_i]$. Fix some $1\leq i\leq \ell$. From Theorem~\ref{thm: grouping}, there is a subset $\tilde B_i\subseteq B_i$ of at least $\alpha' w'/12$ vertices, such that $\tilde B_i$ is node-well-linked in $G[S_i]$. Let $\tpset_i\subseteq \pset_i$ be the subset of paths originating from $\tilde B_i$, and let $\tilde A_{i+1}\subseteq A_{i+1}$ be the set of vertices where they terminate. By Theorem~\ref{thm: grouping}, there is a subset $\tilde A_{i+1}'\subseteq \tilde A_{i+1}$ of at least $\alpha' |\tilde A_{i+1}'|/12\geq (\alpha')^2w'/144$ vertices, such that $\tilde A_{i+1}'$ is node-well-linked in $G[S_{i+1}]$. Let $\tpset'_i\subseteq \tpset_i$ be the subset of paths terminating at $\tilde A_{i+1}$. Finally, let $\pset'_i\subseteq \tpset'_i$ be any subset of $\floor{\frac{\alpha'}{2\Delta} \cdot |\tpset_i'|}=\Omega(w')$ paths. Denote by $B_i'\subseteq \tilde B_i$, $A_{i+1}'\subseteq \tilde{A}_{i+1}'$ the sets of endpoints of the paths in $\pset_i'$ that lie in $S_i$ and $S_{i+1}$, respectively. Then $B_i'$ is node-well-linked in $G[S_i]$ and $A_{i+1}'$ is node-well-linked in $G[S_{i+1}]$. Moreover, from Theorem~\ref{thm: linkedness from node-well-linkedness}, for each $1\leq j\leq \ell$, $(A'_j,B'_j)$ are node-linked in $G[S_j]$. Our final \PoS is  $(\sset,\bigcup_{i=1}^{\ell-1}\pset_{i}')$.

\subsection{Proof of Theorem~\ref{thm: find grid minor or good linkage}}

Suppose we are given any connected $n$-vertex graph $Z$. A simple path
$P$ in $Z$ is called a \emph{$2$-path} iff every vertex
$v\in P$ has degree $2$ in $Z$. In particular, $P$ must be an induced path in $Z$. The following theorem, due to Leaf and
Seymour~\cite{LeafS12} states that we can find either a spanning tree
with many leaves or a long $2$-path in $Z$. 

\begin{theorem}\label{thm: many leaves or a long 2-path}
  Let $Z$ be any connected $n$-vertex graph, and $L\geq 1,p\geq 1$ integers
  with $\frac n {2L}\geq p+5$. Then there is an efficient
  algorithm that either finds a spanning tree $T$ with at least $L$
  leaves in $Z$, or a $2$-path of length at least $p$ in $Z$.
\end{theorem}

A set $\lset$ of $w$ disjoint paths that connect the vertices of $A$ to the vertices of $B$ is called an $A$--$B$ linkage. Since $(A,B)$ are node-linked in $G$, such a linkage must exist.

Given an $A$--$B$ linkage $\lset$, we associate a graph $H_{\lset}$ with it, as follows.  The vertices of $H_{\lset}$ are $U=\set{u_P\mid P\in \lset}$, and there is an edge between $u_P$ and $u_{P'}$ iff $P\neq P'$ and there is a path $\gamma_{P,P'}$ in $G$, whose first vertex belongs to $P$,
last vertex belongs to $P'$, and all inner vertices are disjoint from $\bigcup_{P''\in \lset}V(P'')$. Notice that since $G$ is a connected graph, so is $H_{\lset}$, for any $A$--$B$ linkage $\lset$. We say that an $A$--$B$ linkage $\lset$ is \emph{good} if and only if the length of the longest $2$-path in the corresponding graph $H_{\lset}$ is less than $8h_1+1$. 

Assume first that there is a good linkage $\lset$ in $G$. Then Theorem~\ref{thm: many leaves or a long 2-path} guarantees that there is a spanning tree $\tau$ in $H_{\lset}$ with at least $\frac{w}{2(8h_1+5)}\geq h_2$ leaves. We let $\pset$ contain all paths $P\in \lset$ whose corresponding vertex $u_P$ is a leaf of $\tau$. Then $\pset$ contains at least $h_2$ node-disjoint paths, connecting vertices of $A$ to vertices of $B$. Consider any pair $P,P'\in \pset$ of paths with $P\neq P'$, and let $Q$ be the path connecting $u_P$ to $u_{P'}$ in $\tau$. Let $H_{P,P'}\subseteq G$ be the graph consisting of the union of all paths $P''$ with $u_{P''}\in V(Q)$, and paths $\gamma_{P_1,P_2}$ where $(u_{P_1},u_{P_2})$ is an edge of $Q$. Then graph $H_{P,P'}$ contains a path $\beta_{P,P'}$, connecting a vertex of $P$ to a vertex of $P'$, such that all inner vertices of $\beta_{P,P'}$ are disjoint from $\bigcup_{P''\in \pset}V(P'')$.

Therefore, we assume from now on that $G$ contains no good $A$--$B$ linkage. Among all $A$--$B$ linkages $\lset$ in $G$, choose the one minimizing the number of degree-$2$ vertices in the corresponding graph $H_{\lset}$.

Since $\lset$ is not a good $A$--$B$ linkage, there is a $2$-path $R^*=(u_{P_0},\ldots,u_{P_{8h_1}})$ of length $8h_1+1$ in the corresponding graph $H=H_{\lset}$. Consider the following four subsets of paths: $\pset_1=\set{P_1,\ldots,P_{2h_1}}$, $\pset_2=\set{P_{{2h_1}+1},\ldots,P_{4{h_1}}}$, $\pset_3=\set{P_{4{h_1}+1},\ldots,P_{6{h_1}}}$, and $\pset_4=\set{P_{6{h_1}+1},\ldots,P_{8{h_1}}}$, whose corresponding vertices participate in the $2$-path $R^*$. (Notice that $P_0\not\in \pset_1$, but the degree of $u_{P_0}$ is $2$ in $H$ - we use this fact later). Let $X\subseteq A$ be the set of the endpoints of the paths in $\pset_2$ that belong to $A$, and let $Y\subseteq B$ be the set of paths in $\pset_4$ that belong to $B$ (see Figure~\ref{fig-setup}). Since $(A,B)$ are node-linked in $G$, we can find a set $\qset$ of ${2h_1}$ node-disjoint paths connecting $X$ to $Y$ in $G$. We view the paths in $\qset$ as directed from $X$ to $Y$.

Let $Q\in \qset$ be any such path. Observe that, since $R^*$ is a $2$-path in $H$, path $Q$ has to either intersect all paths in $\pset_1$, or all paths in $\pset_3$, before it reaches $Y$. Therefore, it must intersect $P_{2{h_1}+1}$ or $P_{4{h_1}}$. Let $v$ be the last vertex of $Q$ that belongs to $P_{2{h_1}+1}\cup P_{4{h_1}}$. Let $Q'$ be the segment of $Q$ starting from $v$ and terminating at a vertex of $Y$. Assume first that $v\in P_{2{h_1}+1}$. We say that $Q$ is a type-1 path in this case. Let $u$ be the first vertex on $Q'$ that belongs to $P_0$. (Such a vertex must exist again due to the fact that $R^*$ is a $2$-path.) Let $Q^*$ be the segment of $Q'$ between $v$ and $u$. Then
$Q^*$ intersects every path in $\pset_1\cup\set{P_0,P_{2{h_1}+1}}$, and does not intersect any other path in $\lset$, while $|V(Q^*)\cap V(P_0)|=|V(Q^*)\cap P_{2{h_1}+1}|=1$. (see Figure~\ref{fig-setup}).

Similarly, if $v\in P_{4{h_1}}$, then we say that $Q$ is a type-2 path. Let $u$ be the first vertex of $Q'$ that belongs to $P_{6{h_1}+1}$, and let $Q^*$ be the segment of $Q'$ between $u$ and $v$. Then $Q^*$ intersects every path in $\pset_3\cup\set{P_{4{h_1}}\cup P_{6{h_1}+1}}$, and does not intersect any other path in $\lset$, while $|V(Q^*)\cap V(P_{4{h_1}})|=|V(Q^*)\cap V(P_{6{h_1}+1})|=1$. 

\begin{figure}[h]
\scalebox{0.6}{\includegraphics{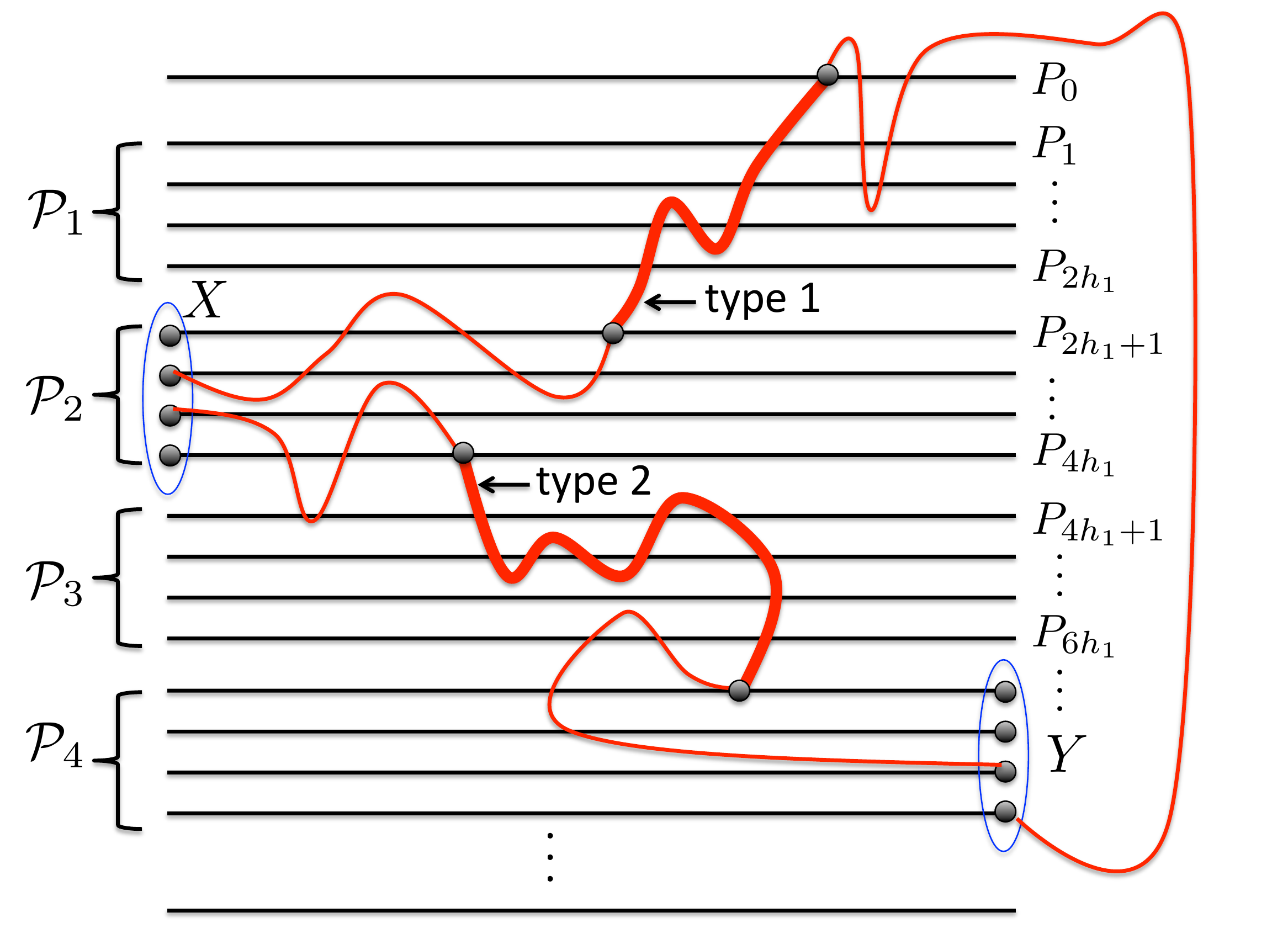}}
\caption{Two examples for paths in $\qset$ - a type-1 and a type-2 path - are shown in red, with the $Q^*$ segment highlighted.\label{fig-setup}}
\end{figure}

 Clearly, either at least half the paths in $\qset$ are type-1 paths, or at least half the paths in $\qset$ are type-2 paths. Assume without loss of generality that the former is true.
 Let $\qset'$ be the set of the sub-paths $Q^*$ for all type-1 paths $Q\in \qset$, that is, $\qset'=\set{Q^*\mid Q\in \qset \mbox{ and $Q$ is type-1}}$, so $|\qset'|\geq {h_1}$.

The rest of the proof is based on the following idea. We will iteratively simplify the intersection pattern of the paths in $\qset'$ and $\pset_1$, until the graph obtained from their union becomes planar, and then find a model of the grid minor in the resulting planar graph directly.

%The rest of the proof is based on the following idea. We will show that either the graph obtained from the union of the paths in $\qset'\cup \pset_1$ is a planar graph, in which case we recover a grid-minor directly, or we will find a new linkage $\lset'$, such that $H_{\lset'}$ contains fewer degree-2 vertices than $H_{\lset}$, leading to a contradiction. To accomplish this we will iteratively simplify the intersection pattern of the paths in $\qset'$ and $\pset_1$.

The algorithm performs a number of iterations. Throughout the algorithm, the set $\qset'$ of paths remains unchanged. The input to every iteration consists of a set $\pset_1'$ of paths, such that  the following hold:

\begin{itemize}
\item $\lset'=(\lset\setminus \pset_1)\cup \pset_1'$ is a valid $A$-$B$ linkage;

\item The graphs $H$ and $H'=H_{\lset'}$ are isomorphic to each other, where the vertices $u_P$ for $P\not\in \pset_1$ are mapped to themselves; 

\item Every path $Q\in \qset'$ intersects every path in $\pset_1'\cup \set{P_0,P_{2{h_1}+1}}$, and no other paths of $\lset'$. Moreover, $Q$ originates at a vertex of $P_0$, terminates at a vertex of $P_{2h_1+1}$, and is internally disjoint from $V(P_0)\cup V(P_{2h_1+1})$.
 \end{itemize}

 The input to the first iteration is $\pset_1'=\pset_1$. Throughout the algorithm, we maintain a graph $\tilde H$ - the subgraph of $G$ induced by the edges participating in the paths of $\pset_1'\cup \qset'$. We define below two combinatorial objects: a bump and a cross. We show that if $\tilde H$ contains either a bump or a cross, then we can find a new set $\pset''_1$ of paths, such that $\lset''=(\lset'\setminus \pset_1')\cup \pset_1''$ is a valid $A$--$B$ linkage. Moreover, $H_{\lset''}$ is isomorphic to $H_{\lset'}$, and we show that we obtain a valid input to the next iteration, while $|E(\qset')\cup E(\pset_1')|> |E(\qset')\cup E(\pset_1'')|$. In other words, the number of edges in the graph $\tilde H$ goes down in every iteration. We also show that if $\tilde H$ contains no bump and no cross, then a large subgraph of $\tilde H$ is planar, and contains a grid minor of size $(h_1\times h_1)$. Therefore, after $|E(G)|$ iterations the algorithm is guaranteed to terminate with the desired output. We now proceed to define the bump and the cross, and their corresponding actions.
A useful observation is that for any $A$--$B$ linkage $\lset'$, the corresponding graph $H_{\lset'}$ is a connected graph, since $G$ is connected.

\paragraph{A bump.} Let $\pset_1'$ be the current set of paths, and let $\lset'=(\lset\setminus \pset_1)\cup \pset_1'$ be the corresponding linkage.
 We say that the corresponding graph $\tilde H$ contains a bump, if there is a sub-path $Q'$ of some path $Q\in \qset'$, whose endpoints, $s$ and $t$, both belong belong to the same path $P_j\in \pset_1'$, and all inner vertices of $Q'$ are disjoint from all paths in $\pset_1'$. (See Figure~\ref{fig: bump}). 
Let $a_j\in A,b_j\in B$ be the endpoints of $P_j$, and assume that $s$ appears before $t$ on $P_j$, as we traverse it from $a_j$ to $b_j$. Let $P'_j$ be the path obtained from $P_j$, by concatenating the segment of $P_j$ between $a_j$ and $s$, the path $Q'$, and the segment of $P_j$ between $t$ and $b_j$.

\begin{figure}[h]
\scalebox{0.6}{\includegraphics{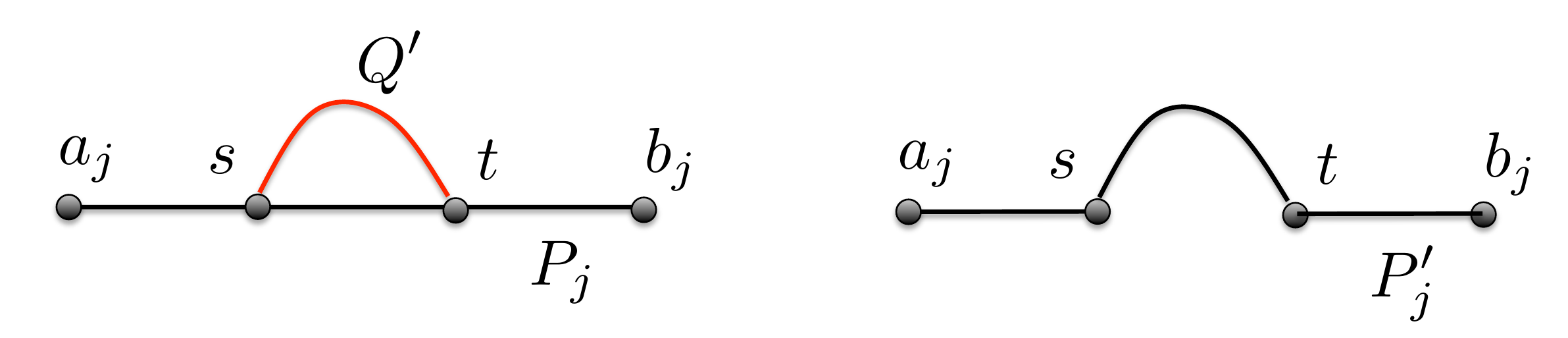}}
\caption{A bump and the corresponding action.\label{fig: bump}}
\end{figure}

Let $\pset_1''$ be the set of paths obtained by replacing $P_j$ with $P'_j$ in $\pset_1'$, and let $\lset''=(\lset'\setminus \pset_1')\cup \pset_1''=(\lset \setminus \pset_1)\cup \pset_1''$. It is immediate to verify that $\lset''$ is a valid $A$--$B$ linkage. Let $H'=H_{\lset'}$, and $H''=H_{\lset''}$, and let $E'$ be the set of edges in the symmetric difference of the two graphs (that is, edges, that belong to exactly one of the two graphs). Then for every edge in $E'$, both endpoints must belong to the set $\set{u_{P_{j-1}},u_{P_j},u_{P_{j+1}}}$. In particular, the only vertices whose degree may be different in the two graphs are $u_{P_{j-1}},u_{P_j},u_{P_{j+1}}$. If the degree of any one of these three vertices is different in $H''$ and $H'$, then, since their degrees are $2$ in both $H'$ and the original graph $H$, we obtain a new $A$--$B$ linkage $\lset''$, such that $H_{\lset''}$ contains fewer degree-$2$ vertices than $H$, which is impossible. Therefore, we assume that the degrees of all three vertices remain equal to $2$. Then it is immediate to verify that $H''$ is isomorphic to $H'$, where each vertex is mapped to itself, except that we replace $u_{P_j}$ with $u_{P_{j'}}$. It is easy to verify that all invariants continue to hold in this case. Let $\tilde H$ be the graph obtained by the union of the paths in $\pset_1'$ and $\qset'$, and define $\tilde H'$ similarly for $\pset_1''$ and $\qset'$. Then $\tilde H'$ contains fewer edges than $\tilde H$, since one of the edges of $P_j$ that is incident on $s$ belongs to $\tilde H$ but not to $\tilde H'$.

\paragraph{A cross.}
Suppose we are given two disjoint paths $Q_1',Q_2'$, where $Q_1'$ is a sub-path of some path $Q_1\in \qset'$, and $Q_2'$ is a sub-path of some path $Q_2\in \qset'$ (where possibly $Q_1=Q_2$). Assume that the endpoints of $Q_1'$ are $s_1,t_1$ and the endpoints of $Q_2'$ are $s_2,t_2$. Moreover, suppose that $s_1,s_2$ appear on some path  $P_j\in \pset_1'$ in this order, and  $t_2,t_1$ appear on $P_{j+1}\in \pset_1'$ in this order (where the paths in $\pset_1'$ are directed from $A$ to $B$), and no inner vertex of $Q_1'$ or $Q_2'$ belongs to any path in $\pset_1'$. We then say that $(Q_1',Q_2')$ is a cross. (See Figure~\ref{fig: cross}.)

\begin{figure}[h]
\scalebox{0.6}{\includegraphics{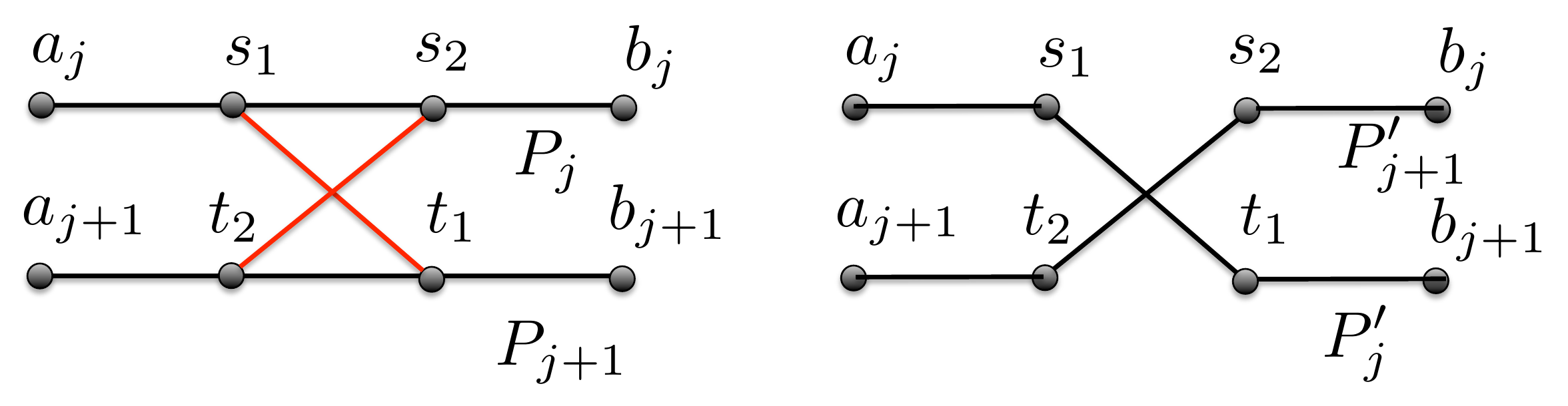}}
\caption{A cross and the corresponding action.\label{fig: cross}}
\end{figure}

Given a cross as above, we define two new paths, as follows. Assume that the endpoints of $P_j$ are $a_j\in A$, $b_j\in B$, and similarly the endpoints of $P_{j+1}$ are $a_{j+1}\in A$, $b_{j+1}\in B$. Let $P'_j$ be obtained by concatenating the segment of $P_j$ between $a_j$ and $s_1$, the path $Q_1'$, and the segment of $P_{j+1}$ between $t_1$ and $b_{j+1}$. Let $P'_{j+1}$ be obtained by concatenating the segment of $P_{j+1}$ between $a_{j+1}$ and $t_2$, the path $Q_2'$, and the segment of $P_j$ between $s_2$ and $b_j$. We obtain the new set $\pset''_1$ of paths by replacing $P_j,P_{j+1}$ with $P'_j,P'_{j+1}$ in $\pset_1'$. Let $\lset''=(\lset'\setminus \pset_1')\cup \pset_1''=(\lset \setminus \pset_1)\cup \pset_1''$. It is immediate to verify that $\lset''$ is a valid $A$--$B$ linkage. As before, let $H'=H_{\lset'}$, and $H''=H_{\lset''}$, and let $E'$ be the set of edges in the symmetric difference of the two graphs. Then for every edge in $E'$, both endpoints must belong to the set $\set{u_{P_{j-1}},u_{P_j},u_{P_{j+1}},u_{P_{j+2}}}$. The only vertices whose degree may be different in the two graphs are $u_{P_{j-1}},u_{P_j},u_{P_{j+1}},u_{P_{j+2}}$. If the degree of any one of these four vertices is different in $H''$ and $H'$, then, since their degrees are $2$ in both $H'$ and the original graph $H$, we obtain a new linkage $\lset''$, such that $H_{\lset''}$ contains fewer degree-$2$ vertices than $H$, which contradicts the choice of $\lset$. Therefore, the degrees of all four vertices remain equal to $2$.  It is now immediate to verify that $H''$ is isomorphic to $H'$, where each vertex is mapped to itself, except that we replace $u_{P_j}, u_{P_{j+1}}$ with $u_{P_{j}'}, u_{P_{j+1}'}$ (possibly swapping them). It is easy to verify that all invariants continue to hold in this case. Let $\tilde H$ be the graph obtained by the union of the paths in $\pset_1'\cup \qset'$, and define $\tilde H'$ similarly for $\pset_1''\cup \qset'$. Then $\tilde H'$ contains fewer edges than $\tilde H$, since one of the edges of $P_j$ incident on $s_1$ belongs to $\tilde H$ but not to $\tilde H'$.

We are now ready to complete the description of our algorithm. We start with $\pset_1'=\pset_1$, and then iterate. In every iteration, we construct the graph $\tilde H$, the subgraph of $G$ induced by $\pset_1'\cup \qset'$. If $\tilde H$ contains a bump or a cross, we apply the appropriate action, and obtain a valid input to the next iteration. Moreover, the number of edges in the new graph $\tilde H$ strictly decreases. Therefore, we are guaranteed that within $O(|E(G)|)$ iterations, the graph $\tilde H$ contains no bump and no cross. 

Consider the final graph $\tilde H$. For each path $Q\in \qset'$, let $v_Q$ be the last vertex of $Q$ that belongs to $V(P_1)$, and let $u_Q$ be the first vertex of $Q$ that belongs to $V(P_{2h_1})$ (we view $Q$ as directed from $P_0$ to $P_{2h_1+1}$). Let $\tilde Q$ be the sub-path of $Q$ between $v_Q$ and $u_Q$. Delete from $\tilde H$ all vertices of $V(Q)\setminus V(\tilde Q)$ for all $Q\in \qset'$, and let $\tilde H'$ denote the resulting graph. Let $\tilde \qset$ be the set of the sub-paths $\tilde Q$ for all $Q\in \qset$. We need the following claim.

\begin{claim} If $\tilde H$ contains no cross and no bump, then $\tilde H'$ is planar.
\end{claim}
\begin{proof}
Consider some path $\tilde Q\in \tilde \qset$. Delete from $\tilde Q$ all edges that participate in the paths in $\pset_1'$, and let $\Sigma(\tilde Q)$ be the resulting set of sub-paths of $\tilde Q$. While any path $\sigma\in \Sigma(\tilde Q)$ contains a vertex $v\in V(\pset_1')$ as an inner vertex, we replace $\sigma$ with two sub-paths, where each subpath starts at one of the endpoints of $\sigma$ and terminates at $v$. Let $\Sigma=\bigcup_{\tilde Q\in \tilde \qset}\Sigma(\tilde Q)$ be the resulting set of paths.
Then for each path $\sigma\in \Sigma$, both endpoints of $\sigma$ belong to $V(\pset_1')$, and the inner vertices are disjoint from $V(\pset_1')$. Moreover, since the paths in $\pset_1'$ induce a $2$-path in the corresponding graph $H_{\lset'}$, and since there are no bumps, the endpoints of each such path $\sigma$ connect two consecutive paths in $\pset_1'$. Since no crosses are allowed, it is easy to see that the graph $\tilde H'$ is planar.
\end{proof}

We now show how to construct a grid minor in graph $\tilde H'$. We start from the union of the paths in $\pset_1'$ and $\tilde \qset$, and perform the following transformation. We say that a segment $\sigma$ of a path $Q\in \tilde \qset$ is a \emph{hill} if and only if (i) the endpoints $s,t$ of $\sigma$ lie on some path $P_i\in \pset_1'$; (ii) the segment $\sigma'$ of $P_i$ whose endpoints are $s$ and $t$ does not contain any vertex of $V(\pset_1'\setminus\set{P_i})$; and (iii) $\sigma$ intersects $P_{i-1}$ and is internally disjoint from all vertices of $V(\pset_1'\setminus\set{P_i})$.  While there is a hill in $\pset_1'\cup \tilde \qset$, we modify the corresponding path $Q$ by replacing the segment $\sigma$ with $\sigma'$. If this creates a cycle on $Q$ (this can happen if $\sigma'$ contained a vertex of $Q$), we discard all such cycles until $Q$ becomes a simple path. We continue performing such transformations, until there is no hill in the set $\pset_1'\cup \tilde \qset$ of paths. Notice that this transformation cannot create any bumps. 
We need the following claim:

\begin{claim}\label{claim: grid}
When the above algorithm terminates, for all $P_i\in \pset_1'$ and $Q\in \tilde \qset$, $P_i\cap Q$ is a path.
\end{claim}

Notice that it is now immediate to obtain the $(h_1\times h_1)$-grid minor from the union of the paths in $\pset_1'\cup \tilde \qset$, by first contracting every path $P_i\cap Q$ for all $P_i\in \pset_1'$ and $Q\in \tilde \qset$, and then suppressing all degree-2 vertices, after which we discard the $h_1$ extra rows.  It is now enough to prove Claim~\ref{claim: grid}.

\begin{proof}
Assume otherwise. Then there must be some path $Q\in \qset$, and some segment $\sigma$ of $Q$, whose two endpoints $s,t$ lie on some path $P_i\in \pset_1'$, such that $\sigma$ intersects $P_{i-1}$, and it is internally disjoint from  $V(\pset_1'\setminus\set{P_i})$. Notice that since there are no bumps, the intersection of $Q$ and $P_{i-1}$ is a path. Among all such pairs $(Q,P_i)$, choose the one maximizing $i$. Since $\sigma$ is not a hill, there must be some path $Q'\neq Q$ in $\tilde \qset$ that intersects the segment $\sigma'$ of $P_i$, lying between $s$ and $t$. Let $v$ be any vertex in $Q'\cap \sigma'$, and let $\sigma''$ be the longest contiguous sub-path of $Q'$ contained in $\sigma'$. Let $u$ be the last vertex of $Q'$ before $\sigma''$ that belongs to $V(\pset_1')$, and let $u'$ be the first vertex of $Q'$ after $\sigma''$ that belongs to $V(\pset_1')$. Then it is easy to verify that $u,u'\in V(P_{i+1})$ (and in particular $i\neq 2h_1$). But then we should have chosen the pair $(Q',P_{i+1})$ instead of $(Q,P_i)$, a contradiction.
\end{proof}

%--------------------------------------------------------------
%--------------------------------------------------------------
%--------------------------------------------------------------
%--------------------------------------------------------------
\subsection{Proof of Corollary~\ref{cor: paths from the path-set system}}

%--------------------------------------------------------------
%--------------------------------------------------------------
%--------------------------------------------------------------
%--------------------------------------------------------------
We say that a cluster $S_i\in \sset$ is \emph{even} if $i$ is even, and otherwise we say that $S_i$ is \emph{odd}.
We apply Theorem~\ref{thm: find grid minor or good linkage} to graph $G[S_i]$ for every even cluster $S_i$, using $A=A_i$ and $B=B_i$. If, for any even cluster $S_i$, the outcome is the $(h_1\times h_1)$-grid minor, then we terminate the algorithm and return the model of this minor. Therefore, we assume that for every even index $i$, Theorem~\ref{thm: find grid minor or good linkage} returns a collection $\lset_i$ of $h_2$ node-disjoint paths contained in $G[S_i]$, that connect some subset $A_i'\subseteq A_i$ of $h_2$ vertices to a subset $B_i'\subseteq B_i$ of $h_2$ vertices, such that for every pair $P,P'\in \lset_i$, there is a path $\beta_i(P,P')$ in $G[S_i]$, connecting a vertex of $P$ to a vertex of $P'$, where $\beta_i(P,P')$ is internally disjoint from $V(\lset_i)$.

Fix some $1\leq i\leq \floor{\ell/2}$. Let $A'_{2i}\subseteq A_{2i}$ and $B'_{2i}\subseteq B_{2i}$ be the sets of endpoints of the paths in $\lset_{2i}$. Let $\lset^-_{2i-1}\subseteq \pset_{2i-1}$ be the set of paths terminating at the vertices of $A'_{2i}$. If $2i<\ell$, then let $\lset^+_{2i}\subseteq \pset_{2i}$ be the set of paths originating at the vertices of $B'_{2i}$; otherwise, let $\lset^+_{2i}$ contain $h_2$ paths, each of which consists of a single distinct vertex of $B_{2i}'$.

Consider now some odd-indexed cluster $S_{i}$. If $i\neq 1$, then let $A'_{i}\subseteq A_{i}$ be the set of vertices where the paths of $\lset^+_{i-1}$ terminate, and otherwise let $A'_{i}$ be any set of $h_2$ vertices of $A_{i}$. If $i<\ell$, then let $B_{i}'\subseteq B_{i}$ be the set of vertices where the paths of $\lset^-_{i+1}$ originate, and otherwise let $B_i'$ be any set of $h_2$ vertices of $B_{i}$. Since $(A_{i},B_{i})$ are node-linked in $G[S_i]$, there is a set $\rset_{i}$ of $h_2$ node-disjoint paths, that are contained in $G[S_{i}]$, and connect $A_{i}'$ to $B_{i}'$.

We now define a the set $\qset$ of  paths, obtained by the concatenation of all paths $\lset^-_i,\lset_i,\lset^+_i$ where $S_i$
is an even cluster, and paths $\rset_j$, where $S_j$ is an odd cluster.  The resulting set $\qset$ contains $h_2$ disjoint paths, originating at the vertices of $A_1$ and terminating at the vertices of $B_{\ell}$, where for every $1\leq i\leq \ell$, for every path $Q\in \qset$, $Q\cap S_i$ is a path. Moreover, for every even integer $1\leq i\leq \ell$, for every pair $Q,Q'\in \qset$ of paths, there is a path $\beta_i(Q,Q')$ contained in $G[S_i]$, that connects a vertex of $Q$ to a vertex of $Q'$, and is internally disjoint form all paths in $\qset$. It is immediate to verify that all paths in $\qset$ are contained in $G'$.

%--------------------------------------------------------------
%--------------------------------------------------------------
%--------------------------------------------------------------
%--------------------------------------------------------------
\subsection{Proof of Corollary~\ref{cor: from path-set system to grid minor}}

%--------------------------------------------------------------
%--------------------------------------------------------------
%--------------------------------------------------------------
%--------------------------------------------------------------
We apply Corollary~\ref{cor: paths from the path-set system} to the \PoS, with parameters $h_1=g$ and $h_2=g$, so $w\geq 16g^2+10g$ as required. If the outcome is the $(g\times g)$-grid minor, then we terminate the algorithm and return its model. Therefore, we assume that the outcome of Corollary~\ref{cor: paths from the path-set system} is a set $\qset$ of $g$ paths connecting vertices of $A_1$ to vertices of $B_{\ell}$, that we denote by $\qset=\set{Q_1,\ldots,Q_g}$.
 We will embed the rows of the grid into the paths in $\qset$, where the $j$th row is embedded into $Q_j$. Let $E'$ be the set of the vertical edges of the $(g\times g)$ grid. We define the following ordering of the edges in $E'$: given any pair $e,e'$ of edges, where $e$ belongs to column $C_i$ and $e'$ belongs to column $C_j$, if $i<j$, then $e$ appears before $e'$; and if $i=j$, but $e$ appears above $e'$ in the grid, then $e$ also appears before $e'$ in this ordering. In other words, we order the edges by their column index, and inside each column in their natural top-to-bottom ordering. Notice that $|E'|=g(g-1)$, as is the number of the clusters $S_i\in \sset$ where $i$ is even. For each $1\leq i\leq g(g-1)$, we will embed the $i$th edge of $E'$ inside $G[S_{2i}]$. Let $e_i$ be the $i$th edge of $E'$, and assume that it connects a vertex on row $j$ to a vertex on row $j+1$ of the grid. Recall that there is a path $\beta_i(Q_j,Q_{j+1})$ in $G[S_{2i}]$, connecting a vertex of $Q_j$ to a vertex of $Q_{j+1}$, such that $\beta_i(Q_j,Q_{j+1})$ is internally disjoint from all paths in $\qset$. We use the path $\beta_i(Q_j,Q_{j+1})$ to embed the edge $e_i$. It is immediate now to complete an embedding of the $(g\times g)$-grid into the graph.

%--------------------------------------------------------------------------------------
%--------------------------------------------------------------------------------------
%--------------------------------------------------------------------------------------
%--------------------------------------------------------------------------------------
\label{-----------------------------------proofs from new sec------------------------}
\section{Proofs Omitted from Section~\ref{sec: splitting a cluster}}
%--------------------------------------------------------------------------------------
%--------------------------------------------------------------------------------------
%--------------------------------------------------------------------------------------
%--------------------------------------------------------------------------------------
%---------------------------------------------------------------------------------------
%---------------------------------------------------------------------------------------
\subsection{Proof of Lemma~\ref{lemma: re-routing of vertex-disjoint paths}}
%-------------------------------------------------------------------------------------
The proof we provide here was suggested by Paul Seymour~\cite{PS-comm}. A different proof, using stable matching was shown by
Conforti, Hassin and Ravi~\cite{CHR}.

Let $\hat G'\subseteq \hat G$ be obtained from the union of the paths in $\xset_1\cup \xset_2$. Let $U_1'\subseteq U_1$ be the set of vertices where the paths of $\xset_1$ originate, and define $U_2'\subseteq U_2$ for the set $\xset_2$ of paths similarly.
Let $E_1$ be the set of all edges participating in the paths in $\xset_1$. While there is an edge $e\in E(\hat G')\setminus E_1$, such that graph $\hat G'\setminus\set{e}$ contains a set of $\ell_2$ nearly-disjoint $U_2'$--$s$ paths, we delete $e$ from $\hat G'$. At the end of this procedure, the final graph $\hat G'$ has the property that for every edge $e\in E(\hat G')\setminus E_1$, the largest number of nearly-disjoint $U_2'$--$s$ paths in graph $\hat G'\setminus\set{e}$ is less than $\ell_2$. Notice that $\xset_1\subseteq \hat G'$, and graph $\hat G'$ contains $\ell_2$ nearly-disjoint $U_2'$--$s$ paths. We need the following claim.

\begin{claim}\label{claim: augmenting}
There is a set $\xset'$ of $\ell_1$ nearly-disjoint $(U_1'\cup U_2')$--$s$ paths in graph $\hat G'$, such that exactly $\ell_2$ paths of $\xset'$ originate at the vertices of $U_2'$.
\end{claim}

Before we prove Claim~\ref{claim: augmenting}, we show that the set $\xset'$ of paths has the properties required by Lemma~\ref{lemma: re-routing of vertex-disjoint paths}. Let $\xset'_1\subseteq \xset'$ be the set of paths originating from the vertices of $U_1'$, and let $\xset'_2=\xset'\setminus\xset'_1$. We only need to show that $\xset'_1\subseteq \xset_1$. Assume otherwise. Then there is some edge $e\in E(\hat G')\setminus E_1$, that lies on some path in $\xset'_1$. But then $\xset'_2\subseteq \hat G'\setminus \set{e}$, and edge $e$ should have been removed from graph $\hat G'$. It now only remains to prove Claim~\ref{claim: augmenting}.

\begin{proofof}{Claim~\ref{claim: augmenting}}
The proof follows standard arguments. We construct a directed node-capacitated flow network $H$: start from graph $\hat G'$, and assign capacity $1$ to each vertex of $\hat G'$, except for vertex $s$, whose capacity is $\ell_1$. We add two new vertices: vertex $t_1$ of capacity $\ell_1-\ell_2$, that connects to every vertex of $U_1'$ with a directed edge, and vertex $t_2$ of capacity $\ell_2$ that connects to every vertex of $U_2'$ with a directed edge. Finally, we add a vertex $t$ of capacity $\ell_1$, that connects to $t_1$ and $t_2$ with directed edges. It is enough to show that there is a $t$--$s$ flow of value $\ell_1$ in the resulting flow network: we can then use the integrality of flow to obtain an integral flow of the same value, which in turn immediately defines the desired set $\xset'$ of paths.

Assume for contradiction that there is no $t$--$s$ flow of value $\ell_1$ in $H$. Then there is a set $Z$ of vertices, whose total capacity is less than $\ell_1$, such that $H\setminus Z$ contains no path connecting $t$ to $s$. Since the capacities of $t$ and $s$ are $\ell_1$ each, $s,t\not\in Z$. Also, since the capacities of $t_1$ and $t_2$ sum up to $\ell_1$, both these vertices cannot simultaneously belong to $Z$.

Assume first that $t_1\in Z$. Then, since $t_2\not \in Z$, set $Z$ contains at most $\ell_2-1$ vertices, each of which must have capacity $1$. Since there is a set of $\ell_2$ nearly-disjoint $U_2'$--$s$ paths in $\hat G'$, at least one such path $P$ is disjoint from $Z$, and so $H\setminus Z$ must contain a path connecting $t$ to $s$, a contradiction.

Similarly, if $t_2\in Z$, then $t_1\not\in Z$, and set $Z$ contains at most $\ell_1-\ell_2-1$ vertices, whose capacities must be all unit. But then at least one path in $\xset_1$ is disjoint from $Z$, giving a path connecting $t$ to $s$ in $H\setminus Z$, a contradiction.

Therefore, we assume that all vertices of $Z$ are capacity-$1$ vertices, that belong to $\hat G'$. But then $Z$ contains at most $\ell_1-1$ vertices, so at least one path in $\xset_1$ is disjoint from $Z$, giving again a path connecting $t$ to $s$ in $H\setminus Z$, a contradiction.
\end{proofof}

%--------------------------------------------------------------------------------------
%--------------------------------------------------------------------------------------
%--------------------------------------------------------------------------------------
\section{Proof of Theorem~\ref{thm: good cluster w small boundary is enough}}\label{type-2 perfect clusters}
%--------------------------------------------------------------------------------------
%--------------------------------------------------------------------------------------
%--------------------------------------------------------------------------------------

If $C$ contains a perfect cluster $\tC$, with $|\out(\tC)|\leq k+k'+1$, such that $\tC$ has the $(k'',\alpha^*)$-bandwidth property, then, from Theorem~\ref{thm: perfect cluster gives chain}, $G$ contains a strong $2$-cluster chain. Therefore, we assume from now on that there is no such cluster $\tC\subseteq C$.

Assume that there is a good cluster $C'\subseteq C$, that has the $(k'',\alpha^*)$-bandwidth property, with $|\out(C')|\leq k+k'+1$ and $|\Gamma(C')|\leq 7k/8$. Among all such clusters, choose the one minimizing $|C'|$. We claim that $|\Gamma(C')|=|\out(C')|$. Indeed, assume otherwise, that is, $|\out(C')|>|\Gamma(C')|$. Then there is some vertex $v\in \Gamma(C')$, that is incident on exactly two edges of $\out(C')$ (since all vertex degrees are at most $3$ and, from Observation~\ref{obs: bandwidth gives connectivity}, $G[C']$ is connected). Consider the cluster $C''=C'\setminus\set{v}$. It is easy to see that $C''$ has the $(k'',\alpha^*)$-bandwidth property,  $|\out(C'')|\leq k+k'+1$, and $|\Gamma(C'')|\leq 7k/8$. From our assumption, $C''$ cannot be a perfect cluster. But from Observation~\ref{obs: cutting good cluster}, it is a good cluster, contradicting the choice of $C'$.
Therefore, we assume from now on that $|\out(C')|= |\Gamma(C')|\leq 7k/8$. For simplicity of notation, we denote $C'$ by $C$ in the rest of this proof.
Let $L=V(G)\setminus C$, and observe that $|\Gamma(L)|\leq |\out(L)|=|\out(C)|=|\Gamma(C)|\leq 7k/8$. The following lemma is central to the proof of Theorem~\ref{thm: good cluster w small boundary is enough}.

\begin{lemma}\label{lemma: type-2 perfect inner}
There is a cluster $S^*\subseteq L$, with $k/4\leq |S^*\cap T_1|\leq 3k/4$, such that $|E(S^*,L\setminus S^*)|\leq k/16$, and $\Gamma(S^*)\cup (T\cap S^*)$ is $(k'',\alpha^*)$-well-linked in $G[S^*]$.
\end{lemma}

We prove Lemma~\ref{lemma: type-2 perfect inner} below, and complete the proof of Theorem~\ref{thm: good cluster w small boundary is enough} here. We construct a weak $2$-cluster chain, with $X'=C$ and $Y'=S^*\setminus T$. Notice that both clusters have the $(k'',\alpha^*)$-bandwidth property. We let $\pset_2$ be a set of $2k'=k/32$ paths, each of which consists of a single edge, connecting a terminal of $T_1\cap S^*$ to a vertex of $Y'$. It now remains to construct the set $\pset_1$ of paths. In order to do so, denote $R=L\setminus S^*$. Then $|T_1\cap S^*|,|T_1\cap R|\geq k/4$, so there is a set $\qset$ of $k/4$ node-disjoint paths, connecting the vertices of $T_1\cap R$ to the vertices of $T_1\cap S^*$ in $G$. Let $\qset'\subseteq \qset$ be the subset of paths that do not use the edges of $E(R,S^*)$. Since $|E(R,S^*)|\leq k/16$, $|\qset'|\geq 3k/16$. By truncating each path in $\qset'$ at the first vertex of $C$ on each such path (where we view the paths as directed from the vertices of $R\cap T_1$ to the vertices of $S^*\cap T_1$), we obtain a collection $\qset''$ of node-disjoint paths, connecting the vertices of $R\cap T_1$ to the vertices of $C$, so that the paths in $\qset''$ are internally disjoint from $C\cup S^*$. 
We let $\pset_1\subseteq \qset''$ be any subset of $2k'=k/32$ paths, completing the construction of a weak $2$-cluster chain. From Theorem~\ref{thm: from weak to strong $2$-cluster chain}, there is a strong $2$-cluster chain in $G$. It now only remains to prove Lemma~\ref{lemma: type-2 perfect inner}.

\begin{proofof}{Lemma~\ref{lemma: type-2 perfect inner}}
We show an algorithm to compute the cluster $S^*$. Throughout the algorithm, we maintain a partition $\sset$ of $L$, and a set $E'=(\bigcup_{S\in \sset}\out(S))\setminus\out(L)$ of edges. We say that a cluster $S\in \sset$ is a level-1 cluster, iff $|S\cap T_1|>k/2$, we say that it is a level-2 cluster, iff $k/4<|S\cap T_1|\leq k/2$, and we say that it is a level-3 cluster otherwise. The algorithm is executed as follows. 

Recall that since $C$ is a good cluster, there is a $1/4$-balanced partition $(A,B)$ of $G\setminus C$ with respect to $T_1$, with $|E(A,B)|\leq k/28$. We start with $\sset=\set{A,B}$, and $E'=E(A,B)$.
While there is any level-1 or level-2 cluster $S\in \sset$, such that $\Gamma(S)\cup (S\cap T)$ is not $(k'',\alpha^*)$-well-linked in $G[S]$, let $(Z,Z')$ be any $(k'',\alpha^*)$-violating partition of $S$: that is, $|E(Z,Z')|<\alpha^*\cdot\min\set{|Z\cap(\Gamma(S)\cup T)|,|Z'\cap (\Gamma(S)\cup T)|, k''}$. We replace $S$ with $Z$ and $Z'$ in $\sset$, add the edges of $E(Z,Z')$ to $E'$, and continue to the next iteration. The algorithm terminates when for each level-1 and level-2 cluster $S\in \sset$, $\Gamma(S)\cup (S\cap T)$ is  $(k'',\alpha^*)$-well-linked in $G[S]$. The following claim is central to our proof.

\begin{claim}\label{claim: bound E'}
When the algorithm terminates, $|E'|< k/16$.
\end{claim}

Before we prove Claim~\ref{claim: bound E'}, we complete the proof of Lemma~\ref{lemma: type-2 perfect inner} using it. We start by showing that there is some cluster $S\in \sset$ with $|S\cap T_1|\geq k/4$. Indeed, assume otherwise. Then, since $T_1$ is $(k/4,1)$-well-linked in $G$, for each $S\in \sset$, $|\out(S)|\geq |S\cap T_1|$, and $\sum_{S\in \sset}|\out(S)|\geq k$. On the other hand, $\sum_{S\in \sset}|\out(S)|=|\out(L)|+2|E'|< 7k/8+k/8=k$. Therefore, there is some cluster $S$ with $|S\cap T_1|\geq k/4$.  We let $S^*$ be any such cluster $S$. From our algorithm, $\Gamma(S^*)\cup (S^*\cap T)$ is $(k'',\alpha^*)$-well-linked in $G[S^*]$. It is also easy to see that $|S^*\cap T_1|\leq 3k/4$, since $S^*\subseteq A$ or $S^*\subseteq B$ must hold, and $|A\cap T_1|,|B\cap T_1|\leq 3k/4$, as $(A,B)$ is a $1/4$-balanced cut of $G[L]$ with respect to $T_1$.
Finally, since $E(S^*,L\setminus S^*)\subseteq E'$, $|E(S^*,L\setminus S^*)|< k/16$. It now remains to prove Claim~\ref{claim: bound E'}.

\begin{proofof}{Claim~\ref{claim: bound E'}}
Throughout the algorithm, we maintain budgets $\beta(v)$ for all vertices $v\in L$, defined as follows. For each level-$i$ cluster $S\in \sset$, where $i\in \set{1,2}$, for each vertex $v\in \Gamma(S)$, we set $\beta(v)=\alpha^*/(1-\alpha^*)$. For each terminal $v\in T\cap S$, we set $\beta(v)=\alpha^*/(1-\alpha^*)$ if $S$ is a level-2 cluster, and $\beta(v)=\alpha^*/(1-\alpha^*)+1/512$ if it is a level-1 cluster. All other vertex budgets are $0$.

Observe that at the beginning, the budget of each terminal is at most $\frac{\alpha^*}{1-\alpha^*}+\frac{1}{512}$, and the budget of each vertex in $\Gamma(L)$, and every vertex that serves as an endpoint of an edge in $E(A,B)$ is $\frac{\alpha^*}{1-\alpha^*}$. Since $|\Gamma(L)|\leq 7k/8$, and $|E(A,B)|\leq k/28$,  at the beginning of the algorithm:

\[
\sum_{v\in L}\beta(v)+|E'|\leq \frac{\alpha^*}{1-\alpha^*}\left (k+k'+\frac{7k}{8}+\frac{k}{14}\right )+\frac{k}{28}+\frac{k+k'}{512}.\]

We denote this bound by $\mu$. We now claim that throughout the algorithm the following invariant holds: $\sum_{v\in L}\beta(v)+|E'|\leq \mu$. We already showed that the invariant holds at the beginning of the algorithm.

Assume now that the invariant holds at the beginning of the current iteration, when some cluster $S\in \sset$ is partitioned into clusters $Z$ and $Z'$. We assume w.l.o.g. that $|Z\cap (\Gamma(S)\cup T_1)|\leq |Z'\cap (\Gamma(S)\cup T_1)|$. If $S$ belongs to level $i\in\set{1,2}$,  then $Z$ must belong to level $i+1$ or higher. We now consider two cases. The first case happens when $Z$ belongs to level $2$. Then $|Z\cap T_1|\geq k/4$, and the budget of each vertex in $Z\cap T_1$ decreases by at least $1/512$. On the other hand, $|E(Z,Z')|\leq \alpha^* k''$, and the budget of each endpoint of each edge in $E(Z,Z')$ increases by at most $\alpha^*/(1-\alpha^*)$. Therefore, the total increase in the budgets of the vertices that serve as endpoints of the edges in $E(Z,Z')$, and in $|E'|$, is bounded by:

\[|E'|\left (1+\frac{2\alpha^*}{1-\alpha^*}\right )\leq \alpha^* k''\cdot \frac{1+\alpha^*}{1-\alpha^*}=\frac{1}{64}\cdot \frac{k}{512}\cdot\frac{65}{63}<\frac{k}4\cdot\frac 1 {512}.\]
The budgets of all other vertices do not increase, and so the invariant continues to hold.

The second case is when $Z$ belongs to level $3$. In this case, the budget of each vertex in $Z\cap (\Gamma(S)\cup T)$ decreases by at least $\alpha^*/(1-\alpha^*)$, while the budget of every vertex in $Z'$ that serves as an endpoint of an edge in $E(Z,Z')$ increases by at most $\alpha^*/(1-\alpha^*)$. Therefore, the total increase in the budgets of the vertices of $Z'$, and in $|E'|$, is bounded by:

\[|E'|\left (1+\frac{\alpha^*}{1-\alpha^*}\right )\leq \alpha^*\cdot |Z\cap (\Gamma(S)\cup T)|\cdot \frac{1}{1-\alpha^*}.\]

The budgets of all other vertices do not increase, and so the invariant continues to hold.

Let $\sset^*$ denote the final partition, and $E^*$ the final set $E'$ of edges. Our goal is to show that $|E^*|\leq k/16$. Let $\sset'$ be the set $\sset$ at the beginning of the last iteration, and let $E''$ be the corresponding set $E'$. Let $S$ be the set that was partitioned in the last iteration, into subsets $Z$ and $Z'$. Notice that $|E(Z,Z')|<\alpha^*k''$, since $(Z,Z')$ is the partition of $S$ violating the $(k'',\alpha^*)$-well-linkedness of $\Gamma(S)\cup (S\cap T)$ in $G[S]$. Therefore, $|E^*|\leq |E''|+\alpha^* k''$. On the other hand, $S$ is a level-1 or a level-2 cluster. Therefore, $|S\cap T_1|\geq k/4$. Moreover, from the $(k/4,1)$-well-linkedness of the terminals in $T_1$, $|\Gamma(S)|\geq k/4$. Therefore, before the last iteration started, $\sum_{v\in L}\beta(v)\geq \frac{\alpha^*}{1-\alpha^*}\cdot \frac k 2$ held.

From our invariant,

\[\begin{split}
|E^*|&\leq |E'|+\alpha^* k''\\
&\leq \mu-\frac{\alpha^*}{1-\alpha^*}\cdot \frac k 2 +\alpha^* k''\\
&=\frac{\alpha^*}{1-\alpha^*}\left (k+k'+\frac{7k}{8}+\frac{k}{14}\right )+\frac{k}{28}+\frac{k+k'}{512}-\frac{\alpha^*}{1-\alpha^*}\cdot \frac k 2 +\alpha^* k''\\
&=\frac{k}{63}\left(\half+\frac 7 8 +\frac 1 {14}+\frac{1}{64}\right )+\frac k {28}+\frac{65k}{64\cdot 512}+\frac{k}{64\cdot 512}\\
&< \frac{k}{16}.
\end{split}
\]
 (We have used the fact that $\alpha^*=1/64$, $k'=k/64$, and $k''=k/512$).
\end{proofof}
\end{proofof}

%--------------------------------------------------------------------------------------
%--------------------------------------------------------------------------------------
%--------------------------------------------------------------------------------------
\section{Proof of Theorem~\ref{thm: balanced-cut-large-smaller-side}}\label{type-3 perfect clusters}
%--------------------------------------------------------------------------------------
%--------------------------------------------------------------------------------------
%--------------------------------------------------------------------------------------
Throughout the proof, we maintain a partition $\sset$ of $B$, together with a set $E'=\left(\bigcup_{S\in \sset}\out(S)\right )\setminus\out(B)$ of edges. 

We say that a cluster $S\in \sset$ is a \emph{level-1} cluster, iff $(k+k'+1)/4<|S\cap \Gamma(C)|\leq (k+k'+1)/2$. We say that it is a \emph{level-2} cluster iff $(k+k'+1)/8<|S\cap \Gamma(C)|\leq (k+k'+1)/4$. We say that it is a \emph{level-3} cluster iff $(k+k'+1)/16<|S\cap \Gamma(C)|\leq (k+k'+1)/8$, and we say that it is a \emph{level-4} cluster otherwise.

We start with $\sset=\set{B}$, and $E'=\emptyset$. Observe that $|B\cap \Gamma(C)|\leq (k+k'+1)/2$. While there is a cluster $S\in \sset$, such that $S$ belongs to levels $1,2$ or $3$, and it does not have the $(k'', \alpha^*)$-bandwidth property, let $(Z,Z')$ be the $(k'',\alpha^*)$-violating partition of $S$ with respect to $\Gamma(S)$, that is, $|E(Z,Z')|< \alpha^*\cdot \min\set{|Z\cap \Gamma(S)|,|Z'\cap \Gamma(S)|,k''}$. We then replace $S$ with $Z$ and $Z'$ in $\sset$, add the edges of $E(Z,Z')$ to $E'$, and continue to the next iteration. The algorithm terminates, when each cluster $S\in \sset$ that belongs to levels $1,2$ or $3$ has the $(k'', \alpha^*)$-bandwidth property.
The key to the analysis of the algorithm is the following claim:

\begin{claim}\label{claim: number of edges cut}
When the algorithm terminates, $|E'|< |E(A,B)|/2$.
\end{claim}

We prove the claim below, and complete the proof of Theorem~\ref{thm: balanced-cut-large-smaller-side}  here. For each cluster $S\in \sset$, let $d_1(S)=|\out(S)\cap E(A,B)|$, and let $d_2(S)=|\out(S)\setminus E(A,B)|$. Then $\sum_{S\in \sset}d_2(S)=2|E'|< |E(A,B)|$. Therefore, there is at least one cluster $S^*\in \sset$ with $d_2(S^*)< d_1(S^*)$. We claim that $|S^*\cap \Gamma(C)|> (k+k'+1)/16$. Assume otherwise, and consider the cut $(A\cup S^*,B\setminus S^*)$. Notice that $|E(A\cup S^*,B\setminus S^*)|=|E(A,B)|-d_1(S^*)+d_2(S^*)< |E(A,B)|$. Moreover,

\[|\Gamma(C)\cap (B\setminus S)|\geq |\Gamma(C)\cap B|-|\Gamma(C)\cap S| \geq\frac{27k}{80}-\frac{k+k'+1}{16}\geq\frac{k+k'+1}{4}\geq \rho|\Gamma(C)| ,\]

since we have assumed that $\rho |\Gamma(C)|\leq (k+k'+1)/4$. We conclude that $(A\cup S^*,B\setminus S^*)$ is a $\rho$-balanced cut, contradicting the minimality of $(A,B)$. Therefore, $|\Gamma(C)\cap S^*|> (k+k'+1)/16> 3k/64$ must hold, and $S^*$ has the $(k'', \alpha^*)$-bandwidth property. It now remains to prove Claim~\ref{claim: number of edges cut}.

\begin{proofof} {Claim~\ref{claim: number of edges cut}}
Throughout the algorithm, we maintain non-negative budgets $\beta(v)$ for all vertices $v\in B$, as follows. Let $S\in \sset$ be any cluster, and assume that it belongs to level $i$, for $1\leq i\leq 3$. Then every vertex $v\in \Gamma(S)$ has budget $\beta(v)=\frac{ \alpha^*}{1-2 \alpha^*}\cdot \frac{11-i}{8}$ if $v\in \Gamma(C)$, and budget $\beta(v)=\frac{ \alpha^*}{1-2 \alpha^*}$ if $v\in \Gamma(S)\setminus \Gamma(C)$. All other vertices have budget $0$. It is now enough to prove that throughout the algorithm, the following invariant holds:

\[\sum_{v\in B}\beta(v)+|E'|< |E(A,B)|/2.\]

At the beginning of the algorithm, the vertices in $\Gamma(C)\cap B$ have budgets at most $\frac{5 \alpha^*}{4(1-2 \alpha^*)}$ each, and there are at most $23|E(A,B)|$ such vertices, due to the $1/23$-bandwidth property of $C$. The vertices of $B$ incident on the edges of $E(A,B)$ have budgets $\frac{ \alpha^*}{1- \alpha^*}$ each. Therefore,

\[\sum_{v\in B}\beta(v)\leq \frac{5 \alpha^*}{4(1-2 \alpha^*)}\cdot 23|E(A,B)|+\frac{ \alpha^*}{1-2 \alpha^*}|E(A,B)|< \frac{|E(A,B)|}{2},\]

since $ \alpha^*=1/64$.

Assume now that the invariant holds at the beginning of the current iteration, when some cluster $S\in \sset$ is partitioned into clusters $Z$ and $Z'$. We assume w.l.o.g. that $|Z\cap \Gamma(C)|\leq |Z'\cap \Gamma(C)|$. If $S$ belongs to levels $1\leq i\leq 3$, then $Z$ must belong to level $i+1$ or higher. We now consider two cases. The first case happens when $Z$ belongs to level $2$ or $3$. Then $|Z\cap \Gamma(C)|\geq (k+k'+1)/16$, and the budget of each vertex in $Z\cap \Gamma(C)$ decreases by at least $\frac{ \alpha^*}{8(1-2 \alpha^*)}$. On the other hand, $|E(Z,Z')|\leq  \alpha^* k''\leq  \alpha^* k/512$, and the budget of each endpoint of each edge in $E(Z,Z')$ increases by at most $ \alpha^*/(1-2 \alpha^*)$. The budgets of all other vertices do not increase. The total increase in the budgets of the vertices that serve as endpoints of the edges in $E(Z,Z')$, and in $|E'|$, is bounded by:

\[|E(Z,Z')|\left (1+\frac{2 \alpha^*}{1-2 \alpha^*}\right )\leq \frac{ \alpha^* k}{512}\cdot \frac{1}{1-2 \alpha^*}<\frac{k+k'+1}{16}\cdot \frac{ \alpha^*}{8(1-2 \alpha^*)},\]
and so the invariant continues to hold.

The second case is when $Z$ belongs to level $4$. In this case, the budget of each vertex in $Z\cap \Gamma(S)$ decreases by at least $ \alpha^*/(1-2 \alpha^*)$, while the budget of every vertex in $Z'$ that serves as an endpoint of an edge in $E(Z,Z')$ increases by at most $ \alpha^*/(1-2 \alpha^*)$. The budgets of all other vertices do not increase. The total increase in the budgets of the vertices of $Z'$, and in $|E'|$, is bounded by:

\[|E(Z,Z')|\left (1+\frac{ \alpha^*}{1-2 \alpha^*}\right )\leq  \alpha^*\cdot |Z\cap \Gamma(S)|\cdot \frac{1}{1-2 \alpha^*},\]

and the invariant continues to hold.
\end{proofof}

%--------------------------------------------------------------------------------------
%--------------------------------------------------------------------------------------
%--------------------------------------------------------------------------------------
\subsection{Proof of Lemma~\ref{lem: balanced-cut-simple}}\label{subsec: proof of balanced cut lemma}
%--------------------------------------------------------------------------------------
%--------------------------------------------------------------------------------------
%--------------------------------------------------------------------------------------
We show an algorithm to compute $C'$.
Throughout the algorithm, we maintain a partition $\cset$ of $C$ into clusters, and a special cluster $S\in \cset$. We also maintain a collection $E'$ of edges, where $E'=\left(\bigcup_{R\in \cset}\out_G(R)\right )\setminus \out_G(C)$. At the beginning, $S=C$, $\cset=\set{S}$, and $E'=\emptyset$. While $|S\cap \Gamma(C)|>3|\Gamma(C)|/4$, and $S$ does not have the $\alpha'$-bandwidth property, let $(R,R')$ be an $\alpha'$-violating partition of $S$, that is, $|E(R,R')|<\alpha'\cdot \min\set{|\Gamma(S)\cap R|,|\Gamma(S)\cap R'|}$. Assume w.l.o.g. that $|R\cap \Gamma(C)|\geq |R'\cap \Gamma(C)|$. We then replace $S$ with $R$ and $R'$ in $\cset$, set $S=R$, add the edges of $E(R,R')$ to $E'$, and continue to the next iteration. Notice that $|\Gamma(R)|<|\Gamma(S)|$, since $\alpha'<1$. The algorithm terminates when either $|S\cap \Gamma(C)|\leq 3|\Gamma(C)|/4$, or $S$ has the $\alpha'$-bandwidth property. Notice that $|\out(S)|$ strictly decreases in each iteration, and since $C$ does not have the $\alpha'$-bandwidth property, at least one iteration is executed. Therefore, when the algorithm terminates, $|\out(S)|<|\out(C)|$.
We need the following claim.

\begin{claim}\label{claim: bound on E'}
When the algorithm terminates, $|E'|\leq \frac{\alpha'}{1-\alpha'}|\Gamma(C)|$.
\end{claim}
\begin{proof}
Throughout the algorithm, we maintain a non-negative budget $\beta(v)$ for each vertex $v\in C$, defined as follows. For $v\in \Gamma(S)$, $\beta(v)=\frac{\alpha'}{1-\alpha'}$, and for all other vertices $v\in C$, $\beta(v)=0$. It is enough to show that throughout the algorithm, $\sum_{v\in C}\beta(v)+|E'|\leq   \frac{\alpha'}{1-\alpha'}|\Gamma(C)|$. It is clear that the inequality holds at the beginning of the algorithm.

Assume now that the inequality holds at the beginning of the current iteration, where cluster $S$ is partitioned into $R$ and $R'$. Then the budgets of the vertices in $\Gamma(S)\cap R'$ decrease from  $\frac{\alpha'}{1-\alpha'}$ to $0$. Let $U\subseteq R$ be the subset of vertices incident on the edges of $E(R,R')$. The budget of every vertex in $U$ increases by $\frac{\alpha'}{1-\alpha'}$, and the total increase in the vertex budgets, and in $|E'|$ is bounded by:

\[\frac{\alpha'}{1-\alpha'}\cdot |U|+|E(R,R')|\leq \left (\frac{\alpha'}{1-\alpha'}+1\right )|E(R,R')|< \frac{\alpha'}{1-\alpha'}\cdot |\Gamma(S)\cap R'|.\]
%
%We now claim that $\left (\frac{\alpha}{1-\alpha}+1\right )|E(R,R')|\leq \frac{\alpha}{1-\alpha}|\Gamma(S)\cap R'|$.
%Recall that $\alpha'=a/b$ for some positive integers $a$ and $b$, and so $ \left (\frac{\alpha'}{1-\alpha'}+1\right )|E(R,R')|$ is an integral multiple of $\frac{1}{b-a}$. On the other hand, $\frac{\alpha''}{1-\alpha'}\cdot |\Gamma(S)\cap R'|=\left (\frac{a}{b-a}+\frac{a}{n^2(b-a)}\right )\cdot |\Gamma(S)\cap R'|$ is not an integral multiple of $1/(b-a)$, since $|\Gamma(S)\cap R'|\leq n$. Therefore, $ \left (\frac{\alpha'}{1-\alpha'}+1\right )|E(R,R')|$ is bounded by the largest integral multiple of $1/(b-a)$ that is smaller than $\frac{\alpha''}{1-\alpha'}\cdot |\Gamma(S)\cap R'|$, and so $ \left (\frac{\alpha'}{1-\alpha'}+1\right )|E(R,R')|\leq \frac{a}{b-a}\cdot |\Gamma(S)\cap R'|= \frac{\alpha'}{1-\alpha'}|\Gamma(S)\cap R'|$. We conclude that the total increase in the vertex budgets and in $|E'|$ is bounded by $\frac{\alpha'}{1-\alpha'}|\Gamma(S)\cap R'|$, and 
Since the budget of every vertex in $\Gamma(S)\cap R'$ decreases by $\alpha'/(1-\alpha')$, the invariant that $\sum_{v\in C}\beta(v)+|E'|\leq   \frac{\alpha'}{1-\alpha'}|\Gamma(C)|$ continues to hold.
\end{proof}

It is easy to see that when the algorithm terminates, $|S\cap \Gamma(C)|\geq |\Gamma(C)|/4$: before the last iteration started, $|S\cap \Gamma(C)|>3|\Gamma(C)|/4$ held. Therefore, $|R\cap \Gamma(C)|\geq 3|\Gamma(C)|/8\geq |\Gamma(C)|/4$. Since the final set $S=R$, we get that $|S\cap \Gamma(C)|\geq |\Gamma(C)|/4$. We claim that $|\Gamma(C)\setminus S|<|\Gamma(C)|/4$. Indeed, assume otherwise. Then $(S,C\setminus S)$ is a $1/4$-balanced partition of $C$, with respect to $\Gamma(C)$.
Since $|E(S,C\setminus S)|\leq |E'|\leq \frac {\alpha'}{1-\alpha'}|\Gamma(C)|$, the value of the minimum balanced cut in $C$ with respect to $\Gamma(C)$ is at most $\frac{\alpha'}{1-\alpha'}|\Gamma(C)|$, a contradiction.

 We conclude that when the algorithm terminates, $|S\cap \Gamma(C)|\geq 3|\Gamma(C)|/4$, and so $S$ has the $\alpha'$-bandwidth property. We return $C'=S$. As observed above, $|\out(C')|<|\out(C)|$, and $C'\subsetneq C$.

%--------------------------------------------------------------------------------------
%--------------------------------------------------------------------------------------
%--------------------------------------------------------------------------------------
%--------------------------------------------------------------------------------------
\subsection{Proof of Claim~\ref{claim: C' has bw prop}}

Let $(Z,Z')$ be any partition of $C'$, and assume w.l.o.g. that $|Z\cap \Gamma(C')|\leq |Z'\cap \Gamma(C')|$. It is enough to show that $|E(Z,Z')|\geq |Z\cap \Gamma(C')|/23$. Since $v$ is a non-cut vertex of $G[C]$, $|E(Z,Z')|>0$.  If $|Z\cap \Gamma(C')|\leq 23$, then, since $|E(Z,Z')|\geq 1$, we get that $|E(Z,Z')|\geq |Z\cap \Gamma(C')|/23$, as required. Therefore, we assume from now on that $|Z\cap \Gamma(C
')|>23$.

Consider the partition $(Z,Z'\cup\set{v})$ of $C$. From the $(k/4,1)$-bandwidth property of $C$, $|E(Z,Z'\cup\set{v})|\geq \min\set{k/4,|Z\cap\Gamma(C)|, |Z'\cap \Gamma(C)|}$. 

Since $\Gamma(C')\setminus \Gamma(C)$ contains at most two vertices - the neighbors of $v$ in $C$, $|Z\cap \Gamma(C')|\leq |Z\cap \Gamma(C)|+2$, and similarly, $|Z'\cap \Gamma(C')|\leq |Z'\cap \Gamma(C)|+2$. Therefore, 

\[\begin{split}
|Z\cap \Gamma(C')|&=\min\set{|Z\cap \Gamma(C')|,|Z'\cap \Gamma(C')|}\\
&\leq \min\set{|Z\cap \Gamma(C)|,|Z'\cap \Gamma(C)|}+2\\
&\leq \frac{k+k'}{2}+2\\
&\leq \frac{65k}{128}+2.\end{split}\]

(We have used the fact that $|\Gamma(C)|\leq k+k'$, and $k'= k/64$). We conclude that:

\[|Z\cap \Gamma(C')|\leq \frac{65}{32}\min\set{|Z\cap \Gamma(C)|,|Z'\cap \Gamma(C)|,\frac{k}{4}}+2.\]

On the other hand, since $v$ is incident on at most two edges of $G[C]$, $|E(Z,Z')|\geq |E(Z,Z'\cup\set{v})|-2$. Therefore, altogether:

\[\begin{split}
|E(Z,Z')|&\geq |E(Z,Z'\cup\set{v})|-2\\
&\geq  \min\set{k/4,|Z\cap \Gamma(C)|,|Z'\cap \Gamma(C)|}-2\\
&\geq \frac{32}{65}(|Z\cap \Gamma(C')|-2)-2\\
&\geq \frac{|Z\cap \Gamma(C')|}{23},\end{split}\]

since $|Z\cap \Gamma(C')|\geq 23$.

%----------------------------------------
%----------------------------------------
\fi
%----------------------------------------
%----------------------------------------

\end{document}